\documentclass[a4paper,10pt]{article}
\usepackage[utf8]{inputenc}
\usepackage[margin=1in]{geometry}
\usepackage{amssymb,amsmath,amsthm}
\usepackage{verbatim,hyperref}
\usepackage{dsfont}
\usepackage{algorithm}
\usepackage[noend]{algpseudocode}
\usepackage{thmtools,thm-restate}

\usepackage{thmtools}
\usepackage{mathtools}
\usepackage{enumerate}
 
\newcommand{\Q}{\widehat{Q}}

\newcommand{\m}{\widehat{m}}
\newcommand{\G}{\mathcal{G}}
\newcommand{\N}{\mathbb{N}}

\newcommand{\R}{\mathbb{R}}

\newtheorem{theorem}{Theorem}
\newtheorem{lemma}{Lemma}

\newtheorem{corollary}{Corollary}

\newtheorem{claim}{Claim}
\newtheorem{proposition}{Proposition}
\newtheorem{fact}{Fact}
\newtheorem{remark}{Remark}

\theoremstyle{definition}
\newtheorem{definition}{Definition}

\newcommand{\e}{\epsilon}

\newcommand{\wh}[1]{\widehat{#1}}
\newcommand{\poly}{\text{poly}}

\usepackage{todonotes}
\newcounter{todocounter}

\newcommand{\gggg}[1]{\stepcounter{todocounter}
  {\color{blue!90} Grzegorz: \thetodocounter: #1}}
\newcommand{\adp}[1]
  {\ensuremath{\left\langle #1 \right\rangle_{\scriptscriptstyle apx}}}
\newcommand{\rdp}[1]
  {\ensuremath{\left\langle #1 \right\rangle}} 
\newcommand{\an}[1]
  {\ensuremath{\left\lVert#1\right\rVert_{\scriptscriptstyle apx}}}
\newcommand{\rn}[1]
  {\ensuremath{|| #1 ||}} 
\newcommand{\Ca}[1]
  {\ensuremath{C^{\scriptscriptstyle apx}_{#1}}} 
\newcommand{\Cr}[1]
  {\ensuremath{C_{#1}}}

\newcommand{\wt}[1]{\widetilde{#1}}

\def\polylog{\operatorname{polylog}}

\title{Spectral Clustering Oracles in Sublinear Time}
\author{Grzegorz Gluch\\EPFL
\and Michael Kapralov\\EPFL 
\and Silvio Lattanzi\\Google Research
\and Aida Mousavifar \\EPFL 
\and Christian Sohler\\University of Cologne \footnote{Work was partially done while author was visiting 
researcher at Google Research, Switzerland}}
\date{}

{\makeatletter
	\gdef\xxxmark{%
		\expandafter\ifx\csname @mpargs\endcsname\relax 
		\expandafter\ifx\csname @captype\endcsname\relax 
		\marginpar{xxx}
		\else
		xxx 
		\fi
		\else
		xxx 
		\fi}
	\gdef\xxx{\@ifnextchar[\xxx@lab\xxx@nolab}
	\long\gdef\xxx@lab[#1]#2{{\bf [\xxxmark #2 ---{\sc #1}]}}
	\long\gdef\xxx@nolab#1{{\bf [\xxxmark #1]}}
}

\begin{document}

\maketitle

\begin{abstract}
Given a graph $G$ that can be partitioned into $k$ disjoint expanders with outer conductance upper bounded by $\epsilon\ll 1$, can we efficiently construct a small space data structure that allows quickly classifying vertices of $G$ according to the expander (cluster) they belong to? Formally, we would like an efficient local computation algorithm that misclassifies at most an $O(\epsilon)$ fraction of vertices in every expander. We refer to such a data structure as a \emph{spectral clustering oracle}.

Our main result is a spectral clustering oracle with query time $O^*(n^{1/2+O(\epsilon)})$ and preprocessing time $2^{O(\frac{1}{\e} k^4 \log^2(k))} n^{1/2+O(\epsilon)}$ that provides misclassification error $O(\epsilon \log k)$ per cluster for any $\epsilon \ll 1/\log k$. More generally, query time can be reduced at the expense of increasing the preprocessing time appropriately (as long as the product is about $n^{1+O(\epsilon)}$) -- this in particular gives a nearly linear time spectral clustering primitive. 

The main technical contribution is a sublinear time oracle that provides dot product access to the spectral embedding of $G$ by estimating distributions of short random walks from vertices in $G$. The distributions themselves provide a poor approximation to the spectral embedding, but we show that an appropriate linear transformation can be used to achieve high precision dot product access. We give an estimator for this linear transformation and analyze it using spectral perturbation bounds and a novel upper bound on the leverage scores of the spectral embedding matrix of a $k$-clusterable graph. We then show that dot product access to the spectral embedding is sufficient to design a clustering oracle.  At a high level our approach amounts to hyperplane partitioning in the spectral embedding of $G$, but crucially operates on a nested sequence of carefully defined subspaces in the spectral embedding to achieve per cluster recovery guarantees.
\end{abstract}
\setcounter{page}{0}

\newpage

\tableofcontents
\setcounter{page}{0}

\newpage

\section{Introduction}

As a central problem in unsupervised learning, graph clustering has been extensively studied in the past decades. Several formalizations of the problem have been considered in the literature. In this paper, we focus on the following (informal) variant of graph clustering:  Given a graph $G$ and an integer $k$, we are interested in finding $k$ nonoverlapping sets $C_1,C_2,\dots,C_k$ that are internally well-connected and that have a sparse cut to the outside. 
A popular approach to this problem is spectral clustering~\cite{KVV04,ng2002spectral,shi2000normalized,von2007tutorial}:
One embeds vertices of the graph into $k$ dimensional Euclidean space using the bottom $k$ eigenvectors of the Laplacian, and clusters the points in Euclidean space using the $k$-means algorithm (in practice), or using a more careful space partitioning approach (in theory). Spectral clustering has been applied in the context of a wide variety of problems, for example, image segmentation \cite{shi2000normalized}, speech separation \cite{BJ06}, clustering of protein sequences \cite{paccanaro2006spectral}, and predicting landslides in geophysics~\cite{bellugi2015spectral}. 
Spectral clustering usually requires to process the graph in two steps. First one computes the spectral embedding and then one clusters the resulting point set. This two stage approach seems to be highly non-local and it seems to be hard to obtain faster methods, if one only has to determine the cluster membership for a small subset of the vertices. However, such a sublinear time access is desirable in some applications. As a basic step towards such a sublinear time 
clustering algorithm, we need a way to quickly access the spectral embedding in some way. Therefore, we ask the following question, where we use $f_x \in \R^k$ to denote the spectral embedding of vertex $x$:

\vspace{0.1in}
\fbox{
\parbox{0.9\textwidth}{
\begin{center}
Is it possible to obtain dot product access to the spectral embedding of a graph in sublinear time? In other words, given a pair of vertices $x, y\in V$, can we quickly approximate the dot product $\langle f_x, f_y\rangle$ in $o(n)$ time?
\end{center}
}
}
\vspace{0.1in}

If such access is possible, it appears plausible that one can design a {\em sublinear spectral clustering oracle}, a small space data structure that provides fast query access to a good clustering of the graph. Our main result in this paper is {\bf (a)} a small space data structure that provides query access to dot products in the spectral embedding, as above, and {\bf (b)} a sublinear time spectral clustering oracle that uses this data structure.

We study a popular version of the spectral clustering problem where one assumes the existence of a planted solution, namely that the input graph can be partitioned into clusters $C_1,\ldots, C_k$ whose internal connectivity is nontrivially higher than the external connectivity. The goal is to recover the clusters approximately.  An average case version of this problem, where the clusters induce Erd\H{o}s-R\'{e}nyi graphs (or random regular graphs), and the edges across clusters are similarly random, has been studied extensively in the literature on the stochastic block model (SBM)~\cite{DBLP:journals/ftcit/Abbe18} for its close relationship to the community detection problem. In this work we study a worst-case version of this problem:

\vspace{0.1in}

\begin{center}
\fbox{
\parbox{0.9\textwidth}{
\begin{center}
Given a graph $G=(V, E)$ that admits a partitioning into a disjoint union of $k$ induced expanders  $C_1,\ldots, C_k$ with outer conductance bounded by $\e\ll 1$, output an approximation to $C_1,\ldots, C_k$ that is correct up to a $O(\e)$ error {\bf on every cluster}.

\end{center}
}
}
\end{center}

We define a {\em spectral clustering oracle with per cluster error $\delta\in (0, 1)$} as a small space data structure that  implicitly defines disjoint subsets $\wh{C}_1,\ldots, \wh{C}_k$ of $V$ such that for some permutation $\pi$ on $k$ elements one has $|C_i\Delta \wh{C}_{\pi(i)}| \leq \delta |C_i|$ for every $i=1,\ldots, k$. The oracle must provide fast query access to such a clustering. The focus of this paper is:
 
\vspace{0.1in} 

\fbox{
\parbox{0.9\textwidth}{
\begin{center}
Design a sublinear time spectral clustering oracle with per cluster error $\approx O(\e)$.
\end{center}
}
}

\vspace{0.1in}

Our main result is a spectral clustering oracle as above, with a slight loss in error parameter. Specifically, our spectral clustering oracle is correct up to $O(\e \log k)$ error  on every cluster:
\begin{theorem}[Informal]\label{thm:mainresult-inf}
There exists a spectral clustering oracle that for every graph $G=(V, E)$ that admits a partitioning into a disjoint union of $k$ induced expanders  $C_1,\ldots, C_k$ with outer conductance bounded by $\e\ll \frac1{\log k}$ achieves error $O(\e \log k)$ per cluster, query time $\approx n^{1/2+O(\e)}$, preprocessing time $\approx 2^{O(\frac{1}{\e} k^4 \log^2(k))} n^{1/2+O(\e)}$ and space $\approx n^{1/2+O(\e)}$. 

Query times can be made faster at the expense of increased space and prepropcessing time, as long as the product of query time and preprocessing time is $\approx n^{1+O(\e)}$, leading in particular to a nearly linear time algorithm for spectral clustering.
\end{theorem}

As byproduct of our main result we also obtain new efficient clustering algorithms in the Local Computation Algorithms (LCA) model (see \cite{DBLP:journals/corr/abs-1104-1377} for introduction of the model and \cite{DBLP:conf/soda/AlonRVX12} for LCA with limited randomness). 

A very important feature of the problem above is the fact that our algorithms recovers a $1-O(\e \log k)$ fraction of {\bf every cluster} as opposed to just classifying a $1-O(\e\log k)$ fraction of vertices of the graph correctly (this latter question allows one to output fewer than $k$ clusters, and is much easier to solve). 
To put this in perspective, it is instructive to apply multiway Cheeger inequalities (e.g.,~\cite{lee2014multiway}, \cite{KwokMultiwaySpectrumGap}) to our setting, noting that the $k$-th eigenvalue $\lambda_k$ of the normalized Laplacian of a graph that can be partitioned into $k$ clusters as above is bounded by $O(\e)$. This means that multiway Cheeger inequalities can be used to recover $k$ clusters with outer conductance $k^2 \sqrt{\e}$ (see~\cite{lee2014multiway}), which becomes trivial unless $\e<1/k^4$ (we note that our problem admits a much simpler solution when $\e\ll 1/k$). One may note that multiway Cheeger inequalities can also recover $0.9k$ clusters with outer conductance $\log^{O(1)} k \sqrt{\e}$ in our setting (e.q. \cite{LouisMultiwayCheegerFraction}), but, as mentioned above, recovering most clusters is much easier that recovering each cluster to $1\pm O(\e)$ multiplicative error, and does not solve our problem.  The most relevant prior result is due to Sinop~\cite{DBLP:conf/soda/Sinop16}, where the author achieves error $O(\sqrt{\e})$ per cluster using spectral techniques. Sinop's result improves up on previous work of~\cite{AwasthiS12}, which achieved per cluster error of $O(\e k)$ (or, rather, is somewhat incomparable to~\cite{AwasthiS12} due to the worst dependence on $\e$, but a lack of dependence on $k$). As we argue below, Sinop's techniques are hard to extend to the sublinear time regime. At the same time, one should note that our result improves on~\cite{AwasthiS12} under the assumption that cluster sizes are comparable while using only sublinear time in the size of the input graph.

\paragraph{Main challenges and comparison to results on testing cluster structure.} This problem is related the well-studied expansion testing problem~\cite{KaleS08,NachmiasS10,DBLP:books/sp/goldreich2011/GoldreichR11,DBLP:journals/cpc/CzumajS10,KalePS13}, which corresponds to the setting of one or two clusters, as well as to the problem of testing cluster structure of graphs, where one essentially wants to determine $k$, the number of clusters in $G$. The problem of {\em testing} cluster structure has recently been considered in the literature~\cite{DBLP:conf/stoc/CzumajPS15, chiplunkar2018testing}: given access to a graph $G$ as above, compute the value of $k$ (in fact, both results~\cite{DBLP:conf/stoc/CzumajPS15} and~\cite{chiplunkar2018testing} apply to the harder property testing problem of distinguishing between graphs that are $k$-clusterable according to the definition above and graphs that are $\e$-far from $k$-clusterable, but a procedure for computing $k$ is the centerpiece of both results). It is interesting to note that the work of~\cite{DBLP:conf/stoc/CzumajPS15} also yields an algorithm for our problem, but only under very strong assumptions on the outer conductance of the clusters (one needs $\e\ll \frac1{\text{poly}(k)\log n}$). The recent work of Peng~\cite{Peng20} considers a robust version of testing cluster structure, but requires $\e\ll \frac1{\text{poly}(k)\log n}$, just like the work of~\cite{DBLP:conf/stoc/CzumajPS15}.

The recent work of~\cite{chiplunkar2018testing} on testing cluster structure yields an optimal tester, which works for any $\e$ smaller than a constant and achieves essentially optimal runtime, but unfortunately their techniques do no extend to the `learning' version of the problem. The reason is very simple: the algorithm of~\cite{chiplunkar2018testing} needs to distinguish between the graph $G$ being a union of $k$ clusters and $k+1$ clusters, and their approach amounts to verifying whether a graph can be partitioned into $k$ clusters. To do so it suffices to check whether the spectral embedding is effectively $k$-dimensional, i.e. whether it spans a nontrivial $(k+1)$-dimensional volume. In order to certify this, however, it suffices to exhibit $k+1$ vertices that span a nontrivial $(k+1)$-dimensional volume. For that, one essentially only needs to locate at least one `typical' point in every cluster, which is much easier than our task of correctly recovering almost all, i.e. a $1-O(\e)$ fraction of vertices in every cluster. In other words, testing graph cluster structure requires only a rather basic access to and control of the spectral embedding. The main technical contribution of our paper is a set of tools for getting precise dot product access to this embedding, together with several new structural claims about it that enable our clustering algorithm.

\paragraph{Comparison to the work of Sinop~\cite{DBLP:conf/soda/Sinop16}.} The work of Sinop~\cite{DBLP:conf/soda/Sinop16} gives a nearly linear time algorithm for recovering every cluster up to error of $1\pm O(\sqrt{\e})$ using spectral techniques\footnote{One must note that the work of~\cite{DBLP:conf/soda/Sinop16} does not require the bounded degree assumption, and can handle clusters of significantly different size.}, for sufficiently small $\e$. The algorithm would be very hard to implement in sublinear time, since one of its central tools (the \textsc{Round} procedure, which controls propagation of error i.e., Lemma 5.4 of ~\cite{DBLP:conf/soda/Sinop16}) heavily relies on the ability to have explicit access to the eigendecomposition of the Laplacian. Specifically, Sinop's algorithm first finds a crude approximation $S$ to a cluster to be recovered, and then improves the approximation by explicitly constructing the corresponding submatrix of the spectral embedding and performing an SVD. One could plausibly envision implementing this using random walks, but that would be challenging, since one would need to consider a random walk induced on a rather unstructured subset of vertices of the graph.

\paragraph{Our contributions: sublinear time access to the spectral embedding.} 

Let $G=(V,E)$ be a $d$-regular graph with $n=|V|$. 
Without loss of generality we assume that $V=\{1,\ldots,n\}$. We assume that $n$ and $d$ are given to the algorithm and that we have oracle access to $G$: We can specify a 
vertex $x\in V$ and a number $i, 1 \le i \le d$, and we will be given in constant time the $i$-th neighbor of $x$. This is also called the bounded degree
graph model.

In this paper we will consider $d$-regular graphs that have a certain cluster structure. We parameterize this cluster structure using the internal and external conductance parameters.

\begin{definition}[\textbf{Internal and external conductance}]
\label{def:conductance}
Let $G = (V,E)$ be a graph. For a set $S\subseteq C\subseteq V$, let $E(S,C\setminus S)$ be the set of edges with one endpoint in $S$ and the other in $C\setminus S$.
The \textit{conductance of a set $S$ within $C$} is $\phi^G_C(S)=\frac{|E(S,C\setminus S)|}{d|S|}$. The \textit{external-conductance} of set $C$ is defined to be $\phi^G_{V}(C)=\frac{|E(C,V\setminus C)|}{d|C|}$.
The \textit{internal-conductance} of set $C\subseteq V$, denoted by $\phi^G(C)$, is  $$\min_{S\subseteq C\text{,} 0<|S|\leq \frac{|C|}{2}}\phi^G_C(S)$$ if $|C|>1$ and one otherwise. 
\end{definition}

\begin{remark}
For simplicity we present all the proofs for $d$-regular graphs, even though all the proofs also work for $d$-bounded graphs, with the same definition of conductance as in Definition~\ref{def:conductance} (i.e., with normalization by $d |S|$ as opposed to the volume of $S$; the two notions of conductance can in the worst case differ by a factor of $d$).
Note that this is equivalent to converting a $d$-bounded degree graph $G$ to a $d$-regular graph $G^{\text{reg}}$ by adding $d-\text{deg}(v)$ self-loops to each vertex $v$ with degree $\text{deg}(v)$. Let $L^{\text{reg}}$ be the \textit{normalized Laplacian} of $G^{\text{reg}}$. Then the random walk on graph $G$ is exactly same as a lazy random
walk on graph $G^{\text{reg}}$ and the definition of conductance is consistent.
\end{remark}

Based on the conductance, clusterability of graphs is defined as follows.

\begin{definition}[$(k,\varphi,\epsilon)$-\textbf{clustering}]
\label{def:clusterable}
Let $G=(V,E)$ be a $d$-regular graph. A $(k,\varphi,\epsilon)$-clustering of $G$ is a partition of vertices $V$ into disjoint subsets $C_1\cup \ldots \cup C_k$ such that for all $i\in [k]$, $\phi^G(C_i)\geq\varphi$, $\phi^G_V(C_i)\leq\epsilon$ and for all $i,j \in [k]$ one has $\frac{|C_i|}{|C_j|} \in O(1)$.
$G$ is called $(k,\varphi,\epsilon)$-clusterable if there exists a $(k,\varphi,\epsilon)$-clustering for $G$.
\end{definition}

We also need for formally define spectral embedding.
\begin{definition}[Spectral embedding]\label{def:embedding}
For a $d$-regular graph $G=(V, E)$ and integer $2\leq k\leq n$ we define the spectral embedding of $G$ as follows. Let $U\in \R^{k\times n}$ denote the matrix of the bottom $k$ eigenvectors of the  normalized Laplacian of $G$ (this choice is not unique; fix any such matrix $U$). Then for every $x\in V$ the spectral embedding $f_x\in \R^k$ of $x$ is the $x$-th column of the matrix $U$, which we write as $U=(f_y)_{y\in V}$.
\end{definition}

\begin{remark} We note that the spectral embedding $f_x, x\in V$ is not uniquely defined. However, in this paper we are only interested in obtaining dot product access to this embedding, i.e. in fast algorithms for computing $\langle f_x, f_y\rangle$ for $x, y\in V$. Such dot products are in fact uniquely defined for any $G$ that is $(k, \varphi, \e)$-clusterable with $\e/\varphi^2$ smaller than an absolute constant -- see Remark~\ref{rm:dot-products-well-defined} below.
\end{remark}

Our first algorithmic result is a sublinear time spectral dot product oracle:
\begin{restatable}{theorem}{thmdot} [Spectral Dot Product Oracle]
\label{thm:dot}
Let $\epsilon, \varphi\in (0,1)$ with $\epsilon \leq \frac{\varphi^2}{10^5}$.
Let $G=(V,E)$ be a $d$-regular graph that admits a $(k,\varphi,\epsilon)$-clustering $C_1,\dots,C_k$.
Let  $ \frac{1}{n^5} < \xi <1$. Then \textsc{InitializeOracle($G,1/2,\xi$)} (Algorithm \ref{alg:LearnEmbedding})
computes in time $(\frac{k}{\xi})^{O(1)} \cdot n^{1/2+O(\epsilon/\varphi^2)}\cdot (\log n)^3 \cdot \frac{1}{\varphi^2}$ a sublinear space data structure $\mathcal{D}$ of size $(\frac{k}{\xi})^{O(1)}\cdot n^{1/2+O(\epsilon/\varphi^2)}\cdot (\log n)^3 $
such that with probability at least $1-n^{-100}$ the following property is satisfied: 

For every pair of vertices $x,y\in V$, \textsc{SpectralDotProduct($G,x,y,1/2,\xi,\mathcal{D}$)} (Algorithm \ref{alg:dotProduct}) 
computes an output value $\adp{f_{x},f_y}$ such that with probability at least
$1-n^{-100}$ 
\[
\left|\ \adp{f_{x},f_y} - \langle f_x , f_y \rangle   \right|\leq \frac{\xi}{n}.
\]
The running time of \textsc{SpectralDotProduct($G,x,y,1/2,\xi,\mathcal{D}$)} is $(\frac{k}{\xi})^{O(1)} \cdot n^{1/2+O(\epsilon/\varphi^2)}\cdot (\log n)^2 \cdot \frac{1}{\varphi^2}$.

Furthermore, for any $0\leq \delta \leq 1/2 $, one can obtain the following trade-offs between preprocessing time and query time: 
Algorithm \textsc{SpectralDotProduct($G,x,y,\delta,\xi,\mathcal{D}$)} requires $(\frac{k}{\xi})^{O(1)} \cdot n^{\delta+O(\epsilon/\varphi^2)}\cdot (\log n)^2 \cdot \frac{1}{\varphi^2}$ per query when the prepressing time of 
Algorithm \textsc{InitializeOracle($G,\delta,\xi$)} is increased to $(\frac{k}{\xi})^{O(1)} \cdot n^{1-\delta+O(\epsilon/\varphi^2)}\cdot (\log n)^3\cdot \frac{1}{\varphi^2}$.
\end{restatable}

\paragraph{Our results: a spectral clustering oracle.}
 Our goal is to compute a data structure that provides sublinear time access to a $(k,\varphi,\e)$-clustering of $G$. Such a data structure is called a $(k,\varphi,\e)$-clustering oracle.
 We now formally define a spectral clustering oracle in the Local Computation (LCA) model: 
\begin{definition}[Spectral clustering oracle]\label{def:oracle}
A randomized algorithm $\mathcal{O}$ is a $(k,\varphi,\e)$-clustering oracle if, when given query access to a $d$-regular graph $G=(V,E)$ that admits a $(k,\varphi,\epsilon)$-clustering  $C_1, \ldots , C_k$,  the algorithm $\mathcal{O}$ provides
consistent query access to a partition $\widehat{P}=(\widehat{C}_1,\ldots, \widehat{C}_k)$ of $V$. The partition $\widehat{P}$ is determined solely by $G$ and the algorithm's random seed.
Moreover, with probability at least $9/10$ over the random bits of $\mathcal{O}$ the partition $\widehat{P}$ has the following property:
for some permutation $\pi$ on $k$ elements one has for every $i \in [k]$:
$$|C_i \triangle \widehat{C}_{\pi(i)}| \leq O\left(\frac{\epsilon \cdot \log(k)}{\varphi^3}\right) |C_i|\text{.}$$
\end{definition}

\begin{remark}
Note that it is crucial that $\mathcal{O}$ provides consistent answers, i.e. classifies a given $x\in V$ in the same way every time it is queried (for a fixing of its random seed).
\end{remark}

We are interested in clustering oracles that perform few probes per query. Our main contribution is:

\begin{restatable}{theorem}{thmmainresult}\label{thm:mainresult}
For every integer $k\geq 2$, every $\varphi\in (0, 1)$, every $\epsilon \ll \frac{\varphi^3}{\log k}$, every $\delta \in (0,1/2]$ there exists a $(k,\varphi,\e)$-clustering oracle that:
\begin{itemize}
    \item has $\widetilde{O}_{\varphi} \left( 2^{O \left(\frac{\varphi^2}{\e} k^4 \log^2(k) \right)} \cdot n^{1 - \delta + O(\e/\varphi^2)} \right)$ preprocessing time,
    \item has $\widetilde{O}_{\varphi} \left( \left(\frac{k}{\e}  \right)^{O(1)} \cdot n^{\delta + O(\e/\varphi^2)} \right)$ query time,
    \item uses $\widetilde{O}_{\varphi} \left(\left(\frac{k}{\e} \right)^{O(1)} \cdot n^{1 - \delta + O(\e/\varphi^2)} \right)$ space,
    \item uses $\widetilde{O}_{\varphi} \left( \left(\frac{k}{\e} \right)^{O(1)} \cdot n^{O(\e/\varphi^2)} \right)$ random bits,
\end{itemize}
where $O_{\varphi}$ suppresses dependence on $\varphi$ and $\widetilde{O}$ hides all $\polylog(n)$ factors.


\end{restatable}

To the best of our knowledge, our algorithm is the first sublinear spectral clustering algorithm in literature. We hope that our main technique for providing sublinear time access to the spectral embedding will have further applications in sublinear time spectral graph theory.  Our simple algorithm for recovering clusters using hyperplane partitioning in a carefully defined sequence of subspaces may also be of independent interest in spectral partitioning problems. We provide a detailed overview of the analysis and the main ideas are involved in Section~\ref{sec:tech-overview}.

\paragraph{Other related work.} Besides the work on property testing and the work on clustering with labelled, data another closely related area is local clustering. In local clustering one is interested of finding the entire cluster around a node $v$ in time proportional to the size of the cluster. Several algorithms are known for this problem~\cite{DBLP:journals/im/AndersenCL08, DBLP:journals/jacm/AndersenGPT16, DBLP:conf/soda/OrecchiaZ14, DBLP:journals/siammax/SpielmanT14, DBLP:conf/icml/ZhuLM13} but unfortunately they cannot be applied to solve our problem because when the clusters have linear size they take linear time (in addition, the output clusters may overlap). In this paper instead we focus on solving the problem using strictly sublinear time. 


\section{Preliminaries}\label{sec_prelim}

In this paper we mostly use the matrix notation to represent graphs.  For a vertex $x\in V$, we say that $\mathds{1}_x\in\mathbb{R}^{n}$ is the indicator of $x$, that is, the vector which is $1$ at index $x$ and $0$ elsewhere. For a (multi) set $I_S=\{x_1,\ldots,x_s\}$ of vertices from $V$ we abuse notation and also denote by $S$ the $n \times s$ matrix whose $i^{\text{\tiny{th}}}$ column is $\mathds{1}_{x_i}$. For $i \in \mathbb{N}$ we use $[i]$ to denote the set $\{1,2, \dots, i \}$.

 For a symmetric matrix $A$, we write $\nu_i(A)$ (resp. $\nu_{\max}(A), \nu_{\min}(A))$ to denote the $i^{\text{th}}$ largest (resp. maximum, minimum) eigenvalue of $A$. 

Let $m\leq n$ be integers.  For any matrix $A\in \R^{n\times m}$ with singular value decomposition (SVD) $A=Y\Gamma Z^T$ we assume $Y\in \R^{n\times n}$, $\Gamma \in \R^{n\times n}$ is a diagonal matrix of singular values and $Z\in \R^{m\times n}$ (this is a slightly non-standard definition of the SVD, but having $\Gamma$ be a square matrix will be convenient).  $Y$ has orthonormal columns, the first $m$ columns of $Z$ are orthonormal, and the rest of the columns of $Z$ are zero. For any integer $q\in[m]$ we denote $Y_{[q]}\in \R^{n \times q}$ as the first $q$ columns of $Y$ and $Y_{-[q]}$ to denote the matrix of the remaining columns of $Y$.  We also denote by $Z_{[q]}\in \R^{m \times q}$ as the first $q$ columns of $Z$ and $Z_{-[q]}$ to denote the matrix of the remaining $n-q$ columns of $Z$. Finally we denote by $\Gamma_{[q]}\in \R^{q \times q}$ the submatrix of $\Gamma$ corresponding to the first $q$ rows and columns of $\Gamma$ and we use $\Gamma_{-[q]}$ to denote the submatrix corresponding to the last $n-q$ rows and $n-q$ columns of $\Gamma$. So for any $q\in[m]$ the span of $Y_{-[q]}$ is the orthogonal complement of the span of $Y_{[q]}$ in $\R^n$, also the span of the columns of $Z_{-[q]}$ is the orthogonal complement of the span of $Z_{[q]}$ in $\R^m$. Thus we can write $A=Y_{[q]}\Gamma_{[q]}Z^T_{[q]} + Y_{-[q]}\Gamma_{-[q]}Z^T_{-[q]}$. 

We also denote with $A_G$ the adjacency matrix of $G$ and with $L$ the \textit{normalized Laplacian} of $G$ where $L=I-\frac{A_G}{d}$. For $L$ we denote its eigenvalues with  $0\leq\lambda_1\leq\ldots\leq\lambda_n\leq 2$ and we write $\Lambda$ to refer to the diagonal matrix of these eigenvalues in ascending order. We also denote with $(u_1,\ldots,u_n)$  an orthonormal basis of eigenvectors of $L$ and with $U \in \R^{n\times n}$  the matrix whose columns are the orthonormal eigenvectors of $L$ arranged in increasing order of eigenvalues. Therefore the eigendecomposition of $L$ is $L=U \Lambda U^T$. We  write $U_{[k]}\in \R^{n\times k}$ for the matrix whose columns are the first $k$ columns of $U$ and also define $F=U_{[k]}^T$. For every vertex $x$ we denote the spectral embedding of vertex $x$ on the bottom $k$ eigenvectors of $L$ with $f_x\in \R^k$, i.e. $f_x=F  \mathds{1}_x $.  For pairs of vertices $x, y\in V$ we use the notation 
$$
\langle f_x, f_y\rangle:=f_x^Tf_y
$$
to denote the dot product in the embedded domain. 
\begin{remark}\label{rm:dot-products-well-defined}
We note that if $G$ is a $(k,\varphi,\epsilon)$-clusterable graph with $\e/\varphi^2$ smaller than a constant, the space spanned by the bottom $k$ eigenvectors of the normalized Laplacian of $G$ is uniquely defined, i.e. the choice of $U_{[k]}$ is unique up to multiplication by an orthonormal matrix $R\in \R^{k\times k}$ on the right.  Indeed, by Lemma~\ref{lem:bnd-lambda} below one has $\lambda_k\leq 2\e$ and by Lemma~\ref{lem_Cheeger}  below one has $\lambda_{k+1}\geq \varphi^2/2$. Thus, since we assume that $\e/\varphi^2$ is smaller than an absolute constant, we have $2\e<\varphi^2/2$, and therefore the subspace spanned by the bottom $k$ eigenvectors of the Laplacian, i.e. the space of $U_{[k]}$, is uniquely defined, as required. We note that while the choice of $f_x$ for $x\in V$ is not unique, but  the dot product between the spectral embedding of $x\in V$ and $y\in V$ is well defined, since for every orthonormal $R\in \R^{k\times k}$ one has 
$\langle Rf_x, Rf_y\rangle=(Rf_x)^T(Rf_y)=f_x^T (R^TR) f_y=f_x^Tf_y.$
\end{remark}

In this paper we also consider the transition matrix of the \textit{random walk associated with} $G$ $M=\frac{1}{2}\cdot \left(I+\frac{A}{d}\right)$. From any vertex $v$, this random walk takes every edge incident to $v$ with probability $\frac{1}{2d}$, and stays on $v$ with the remaining probability which is at least $\frac{1}{2}$. Note that this random walk is exactly same as a lazy random walk on $G$ and that $M=I-\frac{L}{2}$. Observe that $\forall i$ $u_i$ is also an eigenvector of $M$, with eigenvalue $1-\frac{\lambda_i}{2}$.
We denote with $\Sigma$ the diagonal matrix of the eigenvalues of $M$ in descending order.  Therefore the eigendecomposition of $M$ is $M=U \Sigma U^T$. We  write $\Sigma_{[k]}\in \R^{k\times k}$ for the matrix whose columns are the first $k$ rows and columns of $\Sigma$. Furthermore, for any $t$, $M^t$ is a transition matrix of random walks of length $t$. For any vertex $x$, we denote the probability distribution of a $t$-step random walk starting from $x$ by $m_x=M^t \mathds{1}_x$. For a (multi) set $I_S=\{x_1,\ldots,x_s\}$ of vertices from $V$, let matrix $M^t S \in \R^{n\times s}$ is a matrix whose columns are probability distributions of $t$-step random walks starting from vertices in $I_S$. More formally the $i$th column of $M^t S$ is $m_{x_i}$. For any vertex $x\in V$ let $\mathcal{N}(x):\{y\in V: \{x,y\}\in E\}$ denote the set of vertices that are adjacent to the vertex $x$.


\begin{definition}[\textbf{Cluster Centers}]
\label{def:centers}
Let $G=(V,E)$ be a $d$-regular graph. Let $C_1, \ldots, C_k$ be a $(k,\varphi,\epsilon)$-clustering of $G$. We define the \emph{spectral center}
of cluster $C_i$ as
\[\mu_i := \frac{1}{|C_i|}\sum_{x \in C_i} f_x \text{.}\] 
For vertex $x \in V$, we define $\mu_x$ as the cluster center of the cluster  which $x$ belongs to.
\end{definition}

In our analysis we use the following standard results on eigenvalues and matrix norms.  Recall that for any $m\times n$ matrix $A$, the multi-sets of nonzero eigenvalues of $AA^\top$ and $A^\top A$ are equal. 

\begin{lemma}[\cite{chiplunkar2018testing}]
\label{lem_Cheeger}
Let $G$ be any graph which is composed of $k$ components $C_1,\ldots C_k$ such that $\phi^G(C_i)\geq\varphi$ for any $i\in [k]$. Let $L$ be the normalized Laplacian matrix of $G$, and $\lambda_{k+1}$ be the $(k+1)$st smallest eigenvalue of $L$. Then $\lambda_{k+1}\geq\frac{\varphi^2}{2}$.
\end{lemma}
For a $d$-regular graph $G$, let $\rho_G(k)$ denote the minimum value of the maximum conductance over any possible $k$ disjoint nonempty subsets. That is
\[\rho_G(k)\leq \min_{\text{disjoint }S_1,\ldots, S_k}\max_{i}\phi_{G}(S_i)\]
\begin{lemma}[\cite{lee2014multiway}]
\label{lem:multi}
For any  $d$-regular graph $G$ and any $k \geq 2$, it holds that
\[\lambda_k\leq 2\rho_G(k)\text{.}\]
\end{lemma}

\begin{lemma}
\label{lem:bnd-lambda}
Let $G=(V,E)$ be a $d$ regular graph that admits a $(k,\varphi,\epsilon)$-clustering $C_1, \ldots , C_k$. Let $L$ be the normalized Laplacian matrix of $G$. Let $\lambda_1\leq \ldots \leq \lambda_n$ be eigenvalues of $L$, then we have $\lambda_{k+1}\geq \frac{\varphi^2}{2}$ and $\lambda_k\leq 2\epsilon$. 
\end{lemma}
\begin{proof}
Note that $G$ is composed of $k$ components $C_1, \ldots C_k$ such that for all $1\leq i\leq k$ we have $\phi^G(C_i)\geq \varphi$. Hence, by Lemma \ref{lem_Cheeger} we get $\lambda_{k+1}\geq \frac{\varphi^2}{2}\text{.}$ Moreover for all $1\leq i\leq k$, we have $\phi^G_V(C_i)\leq \epsilon$. Thus by Lemma \ref{lem:multi} we have $\lambda_k\leq 2\epsilon\text{.}$
\end{proof}

Since we assume that the maximum ratio of cluster sizes is bounded by a constant, we have
\begin{proposition}
\label{lem:min_size}
Let $G=(V,E)$ be a $d$ regular graph that admits a $(k,\varphi,\epsilon)$-clustering $C_1, \ldots , C_k$. Then we have $\min_{i\in \{1,\ldots,k\}}|C_i|= \Omega\left(\frac{ n}{k}\right)$ and $\max_{i\in \{1,\ldots,k\}}|C_i|= O\left(\frac{ n}{k}\right)$.
\end{proposition}

A symmetric $n \times n$ matrix is positive semi-definite, 
if and only if
 all its eigenvalues are non-negative. 
The spectral norm of matrix $A\in \R^{n\times n}$ is defined as
$\max_{x \in \R^n, x \not= 0} \frac{\|Ax\|_2}{\|x\|_2}$ that equals the square root of the largest eigenvalue of the matrix $A^TA$. The Frobenius norm of a matrix $A$ is defined as $\sqrt{\sum_{i,j} (A_{i,j})^2}$. For matrices $A,\widetilde{A}\in \R^{n\times n}$, we write $A\preccurlyeq \widetilde{A}$, if $\forall x \in \R^n$ we have $x^TAx \leq x^T \widetilde{A} x$.

\section{Technical overview}\label{sec:tech-overview}
In this section we give an overview of the analysis and the main technical contributions of the paper. Recall that we denote the matrix of bottom $k$ eigenvectors of the normalized Laplacian of $G$ by $U_{[k]}$.
The spectral embedding of a vertex $x\in V$, denoted by $f_x\in \R^k$, is simply the $x$-th column of $U_{[k]}^T$. The main intuition behind spectral clustering is that the points $f_x\in \R^k$ are well-concentrated around cluster means $\mu_i\in \R^k$, defined for every $i=1,\ldots, k$ by 
\begin{equation}\label{eq:cluster-means}
\mu_i=\frac1{|C_i|}\sum_{x\in C_i} f_x.
\end{equation}
See Fig.~\ref{fig:example} for an illustration.

\begin{figure}
\centering
\pgfmathsetseed{5}
\begin{tikzpicture}[scale=0.6] 
\draw [fill] (0,0) circle [radius=0.04];
\draw [very thick,->] (0,0,0) -- (10,0,0);
\draw [very thick,->] (0,0,0) -- (0,10,0);
\draw [very thick,->] (0,0,0) -- (0,0,10);




\foreach \i in {0,...,60}
        \shadedraw [ball color= gray] (7 - 0.5 + rnd,rnd - 0.5, rnd - 0.5) circle (0.1cm);
\foreach \i in {0,...,60}
        \shadedraw [ball color= gray] (rnd - 0.5, rnd - 0.5, 7 + rnd - 0.5) circle (0.1cm);
\foreach \i in {0,...,60}
        \shadedraw [ball color= gray] (rnd - 0.5, 7+  rnd - 0.5, rnd - 0.5) circle (0.1cm);
        
\node[] at (7,1.1,0) {$\mu_1$};
\node[] at (0,1.1,7) {$\mu_2$};
\node[] at (-0.8,7.8,0) (t) {$\mu_3$};


\node[] at (7,0,4) (n) {$C_1$};
\node[] at (0,3,12.5) (n) {$C_2$};
\node[] at (-1,11.5,0) (n) {$C_3$};


\draw [->] (7,0,3.3) -- (7,0,1.5); 
\draw [->] (0,2.4,11.4) -- (0,1,8.5); 
\draw [->] (-1,11,0) -- (-0.2,7.7,0); 

\end{tikzpicture}

\caption{\label{fig:example} Example of a spectral embedding where points are concentrated around means.}
\end{figure}
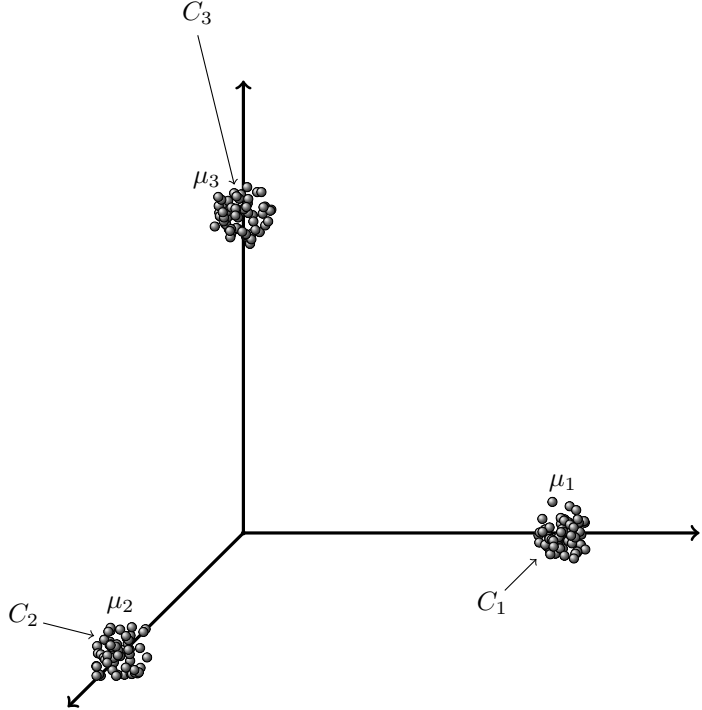

The contributions of our paper are twofold. Our first contribution is a primitive that provides dot product access to the spectral embedding of a graph in sublinear time: we show in Theorem~\ref{thm:dot} how, given any pair of vertices $x, y\in V$ one can compute 
\begin{equation}\label{eq:dot-tech-overview}
\langle f_x, f_y\rangle_{apx}\approx \langle f_x, f_y\rangle,
\end{equation}
in time $\approx n^{1/2+O(\e)}$ per evaluation (see~Algorithm \ref{alg:dotProduct-tt} in Section~\ref{sec:dot} for the formal definition of $\langle\cdot, \cdot \rangle_{apx}$ and its analysis).

Our second contribution is to show how dot product access as in~\eqref{eq:dot-tech-overview} above allows one to solve the cluster recovery problem.  Both of these contributions are based on a new property of the spectral embedding that we establish. This property allows us to quantify the intuitive statement that vertices in the embedding concentrate around cluster means defined in~\eqref{eq:dot-tech-overview} above in a very strong formal sense. 

In the rest of this section we first present our sublinear time dot product oracle (in Section~\ref{sec:dp-overview}) and then outline how access to such an oracle can be used to design a simple spectral clustering algorithm (in Section~\ref{sec:sp-overview}). We assume that the inner conductance of the clusters $\varphi$ is constant for the purposes of this overview to simplify notation.

\subsection{Sublinear time dot product access to the spectral embedding}\label{sec:dp-overview}

We start with a description of the main underlying ideas underlying the proof of Theorem~\ref{thm:dot}. Our starting point from earlier work is the observation that  
collision statistics of random walks can be used to exhibit the structure of a $(k,\varphi,\epsilon)$-clusterable graph. In particular, in 
$(k,\varphi, \epsilon)$-clusterable graphs, there is a gap between $\lambda_k$ and $\lambda_{k+1}$, and the behavior of random walks is essentially determined 
by the bottom $k$ eigenvectors of the Laplacian and the corresponding eigenvalues. This suggests that we can potentially use random walks to determine the
spectral embedding. The spectral embedding is of course not necessarily unique (for example, if not all of the 
bottom $k$ eigenvalues are unique). However, the dot product of the embedded vertices is still well-defined as a function of the subspace spanned by the bottom $k$ eigenvectors of the Laplacian, as the subspace itself is uniquely defined because of the aforementioned gap between $\lambda_k$ and $\lambda_{k+1}$. See Remark~\ref{rm:dot-products-well-defined} for more details. We now give an overview of our approach.

Fix two vertices $x, y\in V$. We would like to compute
$$
\langle f_x, f_y\rangle=(F\mathds{1}_x)^T(F\mathds{1}_y)=\mathds{1}_x^TU_{[k]} U_{[k]}^T \mathds{1}_y.
$$

The direct approach to this would amount to computing an eigendecomposition of $M$ to obtain $U_{[k]}$, but that would take at least $\Omega(n)$ time and is too expensive for our purposes. On the other hand, it is well-known that we are able to estimate, in about $n^{1/2}$ time, the dot product 
$$
(M^t \mathds{1}_x)^T (M^t \mathds{1}_y)=\mathds{1}_x^T M^{2t} \mathds{1}_y.
$$
Note that $\mathds{1}_x^T M^{2t} \mathds{1}_y=\mathds{1}_x^T U\Sigma^{2t}U^T \mathds{1}_y$. Thus to get $U_{[k]} U_{[k]}^T$ from  $\mathds{1}_x^T M^{2t} \mathds{1}_y$ we need to remove the matrix $\Sigma^{2t}$ from the middle. Specifically, we can estimate the quantity above as follows. For some precision parameter $\xi\in (0, 1)$ we first run $\approx n^{1/2+O(\e/\varphi^2)}/\xi^2$ random walks from $x$, letting $\wh{m}_x\in  \mathbb{R}^n$ denote a vector whose $a$'th component is the fraction of random walks from $x$ that end up at $a$. Similarly, we run $\approx n^{1/2+O(\e/\varphi^2)}/\xi^2$ random walks from $y$, letting $\wh{m}_y\in  \mathbb{R}^n$ denote a vector whose $a$'th component is the fraction of random walks from $y$ that end up at $a$. One can show\footnote{This calculation is mostly amounts to a rather standard collision counting calculation that relies on the birthday paradox if one wants to establish the claim {\bf for most vertices} $x, y\in V$ (this was done in~\cite{DBLP:conf/stoc/CzumajPS15} and~\cite{chiplunkar2018testing} for example). Our new moment bounds for the spectral embedding (see Lemmas~\ref{lem:tail_bound} and~\ref{lem:l-inf-bnd} in Section~\ref{sebsec:moment-bounds}) allow us to establish such a claim {\bf for all vertices $x, y\in V$} -- see Lemma~\ref{lem:Mt-bnd}.} that with high (constant) probability we have 
\begin{equation}\label{eq:94hg9jg3gwgadg}
\left|\wh{m}_x^T\wh{m}_y-\mathds{1}_x^T M^{2t} \mathds{1}_y\right|\leq \xi\cdot \frac1{n}.
\end{equation}
While~\eqref{eq:94hg9jg3gwgadg} is not directly useful, a primitive for constructing empirical distributions $\wh{m}_x$ and $\wh{m}_y$ as above is a central part of our approach.  We formalize it as Algorithm~\ref{alg:random-walk-tt}  (\textsc{RunRandomWalks}) below:
\begin{algorithm}[H]
\caption{\textsc{RunRandomWalks}($G,R,t,x$)}
\label{alg:random-walk-tt}
\begin{algorithmic}[1]
	\State Run $R$ random walks of length $t$ starting from $x$
	\State Let $\m_x(y)$ be the fraction of random walks that ends at $y$ \Comment vector $\m_x$ has support at most $R$ 
	\State \Return $\m_x$
\end{algorithmic}
\end{algorithm}

Even if we cannot apply~\eqref{eq:94hg9jg3gwgadg} directly, it lets us compute a seemingly related to quantity $\mathds{1}_x^T M^{2t} \mathds{1}_y$ quickly by invoking Algorithm~\ref{alg:random-walk-tt} and computing one dot product. 
In order to get from $\mathds{1}_x^T M^{2t} \mathds{1}_y$ to $\mathds{1}_x^TU_{[k]} U_{[k]}^T \mathds{1}_y$, we need to somehow apply a linear transformation on the random walk distributions {\bf before} computing the dot product between them, i.e. we need a different dot product operation. It is easy to see that the correct linear transformation is given by the matrix $U_{[k]} {\Sigma}_{[k]}^{-2t} U_{[k]}^T$, where $M^t=U\Sigma^tU^T$ is the eigendecomposition of $M$ and $U_{[k]}$ stands for the matrix of bottom $k$ eigenvectors of the Laplacian\footnote{Note that this matrix is not well defined in the presence of repeated eigenvectors, but any fixed choice of this matrix suffices for our purposes. It is also interesting to note that while we use a canonical choice of the eigendecomposition of $M$ throughout the paper, all our bounds are oblivious to the choice of this basis, and hold for the {\em subspace} of bottom $k$ eigenvectors, which is well defined since there is a gap between the $k$-th and $(k+1)$-th eigenvalues in $k$-clusterable graphs.}.  Specifically, we have
$$
(M^t \mathds{1}_x)^T (U_{[k]} {\Sigma}_{[k]}^{-2t} U_{[k]}^T) (M^t \mathds{1}_y)=\mathds{1}_x^T U_{[k]} U_{[k]}^T \mathds{1}_y=\langle f_x, f_y\rangle,
$$
which is exactly the quantity we are interested in. Of course, there is a major problem with this approach, since $U_{[k]} {\Sigma}_{[k]}^{-2t} U_{[k]}^T$ is an $n\times n$ matrix! To get around this issue, we approximate $U_{[k]} {\Sigma}_{[k]}^{-2t} U_{[k]}^T$ by a sparse low rank matrix, as we describe below. Specifically, we let $I_S$ be a multiset of $s\ll n$ vertices selected uniformly at random. Let $S$ be the $n\times s$ matrix whose $j$-th column equals $\mathds{1}_{i_j}$  and let $\widetilde{W} \widetilde{\Sigma}^{2t} \widetilde{W}^T$ denote the eigendecomposition of $\frac{n}{s}\cdot (M^t S)^T(M^tS)$\footnote{We abuse notation somewhat by writing $S$ to denote the $n\times s$ matrix whose $(a, j)$-th entry equals $1$ if the $j$-th sampled vertex equals $a$ and $0$ otherwise.}. We show that with an appropriate choice of the sampling parameter $s\ll n$ one has
\begin{equation}\label{eq:90h24hg932hg}
U_{[k]} {\Sigma}_{[k]}^{-2t} U_{[k]}^T\approx M^t S\cdot \widetilde{\Psi}\cdot S^T M^t,
\end{equation}
where 
\begin{equation}
\widetilde{\Psi}=\frac{n}{s}\cdot\widetilde{W}_{[k]} \widetilde{\Sigma}_{[k]}^{-4t} \widetilde{W}_{[k]}^T
\end{equation}
is an $s\times s$ matrix that can be computed explicitly. The corresponding primitive to compute $(M^tS)^T(M^tS)$ is presented as Algorithm~\ref{alg:gram-tt} (\textsc{EstimateCollisionProbabilities}) below. It basically estimates the Gram matrix of random walk distributions out of $I_S$ (denoted by $\G$) by counting collisions, and taking medians of estimates to reduce failure probability appropriately. After computing the approximate Gram matrix, we derive from it the matrix $\Psi=\frac{n}{s}\cdot\wh{W}_{[k]} \wh{\Sigma}_{[k]}^{-2} \wh{W}_{[k]}^T$, where $\G=\wh{W} \wh{\Sigma} \wh{W}^T$ is the eigendecomposition of $\G$ (see line \eqref{ln:QSVD} and line \eqref{ln:final-init} of Algorithm \ref{alg:LearnEmbedding-tt}; note that $G$ is a symmetric matrix, and hence an eigendecomposition exists).

\begin{algorithm}[H]
\caption{\textsc{EstimateCollisionProbabilities($G,I_S,R,t$)}  }\label{alg:gram-tt}
\begin{algorithmic}[1]
	\For{$i=1$ to $O(\log n)$}
		\State  ${\Q_i}:=\textsc{EstimateTransitionMatrix}(G,I_S,R,t)$ 
		\State  ${\widehat{P}_i}:=\textsc{EstimateTransitionMatrix}(G,I_S,R,t)$ 
		\State $\G_i:=\frac{1}{2}\left(\widehat{P}_i^T \Q_i+\Q_i^T \widehat{P}_i\right)$  \Comment{$\G_i$ is symmetric}
	\EndFor
	\State Let $\G$ be a matrix obtained by taking the entrywise median of $\mathcal{G}_i$'s \Comment{$\G$ is symmetric}
	\State \Return $\G$  \Comment $\G\in \R^{s\times s}$
\end{algorithmic}
\end{algorithm}
Algorithm~\ref{alg:gram-tt} uses an auxiliary primitive presented as
\begin{algorithm}[H]
\caption{\textsc{EstimateTransitionMatrix$(G,I_S,R,t)$}}
\begin{algorithmic}[1]
	\For{each sample $x\in I_S$}
		\State $\m_x:=\textsc{RunRandomWalks}(G,R,t,x)$  
	\EndFor
	\State Let ${\Q}$ be the matrix whose columns are $\m_x$  for $x\in I_S$ 
	\State \Return $\Q$ \Comment $\Q$ has at most $Rs$ non-zeros
\end{algorithmic}
\end{algorithm}

The proof of~\eqref{eq:90h24hg932hg} relies on matrix perturbation bounds (the Davis-Kahan $\sin \theta$ theorem) as well as spectral concentration inequalities, crucially coupled with our tail bounds on the spectral embedding (see Lemma~\ref{lem:tail_bound} and Lemma~\ref{lem:l-inf-bnd}). In particular Lemma~\ref{lem:tail_bound}  and it's consequence - Lemma~\ref{lem:l-inf-bnd}  can be used to bound the leverage scores of $U_{[k]}$ (i.e. $||f_x||^2_2$ for $x\in V$). This part of the analysis is presented in Section~\ref{subsubsec:topk}.

\begin{restatable}{lemma}{lemmatail}[Tail-bound]\label{lem:tail_bound}
Let $\varphi\in(0,1)$ and $\epsilon \leq \frac{\varphi^2}{100}$, and let $G=(V,E)$ be a $d$-regular graph that admits $(k,\varphi,\e)$-clustering $C_1, \ldots, C_k$.   Let $L$ be the normalized Laplacian of $G$. Let $u$ be a normalized eigenvector of $L$ with $||u||_2=1$ and with eigenvalue at most $2\epsilon$.  Then for any $\beta>1$ we have \[
\frac1{n}\cdot\biggl|
\left \{x\in V: |u(x)|\geq \beta\cdot \sqrt{\frac{10}{\min_{i\in [k]}|C_i|}} \right \}\biggl| \leq \left(\frac{\beta}{2}\right) ^{-{\varphi^2/20\cdot\epsilon}}\text{.}\]
\end{restatable}

\begin{restatable}{lemma}{lemmalinf}\label{lem:l-inf-bnd}
Let $\varphi\in(0,1)$ and $\epsilon \leq \frac{\varphi^2}{100}$, and let $G=(V,E)$ be a $d$-regular graph that admits $(k,\varphi,\e)$-clustering $C_1, \ldots, C_k$.  Let $u$ be a normalized eigenvector of $L$ with $||u||_2=1$ and with eigenvalue at most $2\epsilon$.     Then we have 
\[||u||_\infty \leq  n^{20\cdot\epsilon /\varphi^2}\cdot \sqrt{\frac{160}{\min_{i\in k}|C_i|}} \text{.} \]
\end{restatable}

We note that the number of samples $s$ is chosen as $s\approx k^{O(1)} n^{O(\e/\varphi^2)}$ (see Algorithm~\ref{alg:LearnEmbedding-tt}) , where the second factor is due to our upper bound on the $\ell_\infty$ norm of the bottom $k$ eigenvectors of the Laplacian of a $(k,\varphi,\epsilon)$-clusterable graph proved in Section~\ref{sebsec:moment-bounds}.

Once we establish~\eqref{eq:90h24hg932hg} in Section~\ref{subsubsec:cols} (see Lemma \ref{lem:bnd-e1}),  we get for every $x, y\in V$
\begin{equation}
\begin{split}
(M^t \mathds{1}_x)^T M^t S \cdot \widetilde{\Psi} \cdot S^T M^t (M^t\mathds{1}_y)\approx \mathds{1}_x^TU_{[k]}U_{[k]}^T\mathds{1}_y,
\end{split}
\end{equation}
which is what we would like to compute. One issue remains at this point, which is that we cannot compute $M^t\mathds{1}_x$ or $M^t \mathds{1}_y$ explicitly, and neither can we store and compute our approximation $M^t S\cdot \Psi\cdot S^T M^t$, since it is a dense, albeit low rank, matrix. We resolve this problem by running an appropriate number of random walks out of the sampled nodes $I_S$, as well as the queried nodes $x, y\in V$. Specifically, we run $\approx n^{1/2+O(\e)}$ random walks from every sampled node in $I_S$, defining an $n\times s$ matrix $Q$ whose $(a, b)$-th entry is the fraction of walks from $a$ that ended at $b$ and using the matrix $Q$ as a proxy for $M^tS$ (note that the expectation of $Q$ is exactly $M^tS$).  Such a matrix $Q$ is computed as per line \eqref{ln:Qi} and line \eqref{ln:Pi} of 
Algorithm~\ref{alg:gram-tt} (\textsc{EstimateCollisionProbabilities}). 
We note that Algorithm~\ref{alg:LearnEmbedding-tt} (\textsc{InitializeOracle}) performs $O(\log n)$ independent estimates that we ultimately use to boost confidence (by the median trick). The entire preprocessing is summarized in Algorithm~\ref{alg:LearnEmbedding-tt} (\textsc{InitializeOracle}) below:

\begin{algorithm}[H]
\caption{\textsc{InitializeOracle($G,\delta,\xi$)}  \Comment Need: $\epsilon/\varphi^2 \leq \frac{1}{10^5}$} \label{alg:LearnEmbedding-tt}
\begin{algorithmic}[1]
	\State $t:= \frac{20\cdot \log n}{\varphi^2}$	
	\State $R_{\text{init}}:=O{(n^{1-\delta +980 \cdot\epsilon / \varphi^2}  \cdot k^{17}/{\xi}^{2})}$	
	
	\State $s:= O(n^{480\cdot \epsilon / \varphi^2}\cdot \log n \cdot k^{8}/{\xi}^2)$  
	
	\State Let $I_S$ be the multiset of $s$ indices chosen independently and uniformly at random from $\{1,\ldots,n\}$
	
	\For{$i=1$ to $O(\log n)$}
		\State  ${\Q_i}:=\textsc{EstimateTransitionMatrix}(G,I_S,R_{\text{init}},t)$  \Comment $\Q_i$ has at most $R_{\text{init}}\cdot s$ non-zeros
	\EndFor	
	\State $\G:=$\textsc{EstimateCollisionProbabilities}$(G,I_S,R_{\text{init}},t)$

	\State Let $\frac{n}{s}\cdot\G:=\widehat{W}\widehat{\Sigma} \widehat{W}^T$ be the eigendecomposition of $\frac{n}{s}\cdot\G$ \Comment $\G\in \R^{s\times s}$
	\If{$\widehat{\Sigma}^{-1}$ exists}
	\State $\Psi:=\frac{n}{s}\cdot\widehat{W}_{[k]}\widehat{\Sigma}_{[k]}^{-2} \widehat{W}_{[k]}^T$ \Comment $\Psi\in\R^{s\times s}$ 
	\State \Return  $\mathcal{D}:=\{\Psi,\Q_1,\ldots ,\Q_{O(\log n)}\}$  
	\EndIf
	
\end{algorithmic}
\end{algorithm}

Equipped with the primitives presented above, we can now state our final dot product estimate:
\begin{equation}\label{eq:290hg92hg}
\wh{m}_x^T Q \Psi Q^T \wh{m}_y\approx \mathds{1}_x^T U_{[k]}U_{[k]}^T \mathds{1}_y=\langle f_x, f_y\rangle,
\end{equation}
where $\wh{m}_x$ and $\wh{m}_y$ are empirical distributions of $\approx n^{1/2+O(\e/\phi^2)}$ out of $x$ and $y$ respectively, $Q$ is an $n\times s$ matrix with $\approx n^{1/2+O(\e/\phi^2)}$ nonzeros per column, and $\Psi$ is a possibly dense $s\times s$ matrix, where the number of sampled vertices $s$ is ultimately chosen to be $k^{O(1)}n^{O(\e/\phi^2)}$.  The analysis of the error incurred in replacing~\eqref{eq:90h24hg932hg} with~\eqref{eq:290hg92hg} is presented in Section~\ref{subsubsec:rows}.  It relies on a birthday paradox style variance computation similar to previous sublinear time algorithms for testing graph cluster structure. The actual query procedure that implements~\eqref{eq:290hg92hg} is given by Algorithm~\ref{alg:dotProduct-tt} below. 
\begin{algorithm}[H]
\caption{\textsc{SpectralDotProductOracle}($G,x,y, \delta, \xi, \mathcal{D}$)  \Comment Need: $\epsilon/\varphi^2 \leq \frac{1}{10^5}$
\newline \text{ }\Comment $\mathcal{D}:=\{\Psi,\Q_1,\ldots ,\Q_{O(\log n)}\}$}
\label{alg:dotProduct-tt}
\begin{algorithmic}[1]
		\State $R_{\text{query}}:=O{(n^{\delta +500 \cdot\epsilon / \varphi^2}\cdot  k^{9}/{\xi}^{2})}$
				
	\For{$i=1$ to $O(\log n)$}
		\State ${\m^i_x:=\textsc{RunRandomWalks}(G,R_{\text{query}},t,x)}$
		\State ${\m^i_y:=\textsc{RunRandomWalks}(G,R_{\text{query}},t,y)}$ 
	\EndFor
	\State Let ${\alpha}_x$ be a vector obtained by taking the entrywise median of $(\Q_i)^T(\m^i_x)$ over all runs 
	\State Let ${\alpha}_y$ be a vector  obtained by taking the entrywise median of $(\Q_i)^T(\m^i_y)$ over all runs 
	\State  \Return $\adp{f_{x},f_y}:={\alpha}_x^T \Psi {\alpha}_y$  
\end{algorithmic}
\end{algorithm}
\paragraph{Trading off preprocessing time for query time.} Finally, we note that one can reduce query time (i.e., runtime of \textsc{SpectralDotProductOracle}) at the expense of increased preprocessing time and size of data structure. Specifically, one can run $\approx n^{\delta+O(\e/\phi^2)}$ random walks from nodes $x, y$ whose dot product is being estimated by \textsc{SpectralDotProductOracle} at the expense of increasing the number of random walks run to generate the matrix $Q$ in \textsc{InitializeOracle} to $\approx n^{1-\delta+O(\e/\phi^2)}$, for any $\delta\leq 1/2$. This in particular leads to a nearly linear time spectral clustering algorithm.

\subsection{Geometry of the spectral embedding}\label{sec:sp-overview}

We now describe our spectral clustering algorithm. Since we only have dot product access to the spectral embedding, the algorithm must be very simple. Indeed, our algorithm amounts to performing hyperplane partitioning in a sequence of carefully crafted subspaces of the embedding space, using (a good approximation to) cluster means $\mu_i$. 

We first present a simple hyperplane partitioning, then we give an example embedding to show why it might be hard to prove that this scheme works. After that we design a modification of the hyperplane partitioning scheme that, through the course of carving, carefully projects out some directions of the embedding. This modification is an idealized version of our final algorithm for which we can prove per cluster recovery guarantees.

First we assume that the cluster means~\eqref{eq:cluster-means} are known.  In that case we define, for every $i=1,\ldots, k$, the sets 
$$
\wt{C}_i:=\{x\in V: \langle f_x, \mu_i\rangle\geq 0.9 ||\mu_i||^2\}
$$
of points that are nontrivially correlated with the $i$-th cluster mean $\mu_i$. Note that $\wt{C}_i=C_{\mu_i, 0.9}$ in terms of Definition~\ref{def:thresholdsets}, but since $\mu_i$'s are fixed in this overview, we use the simpler notation.
We next define, for every $i=1,\ldots, k$,
\begin{equation}\label{eq:ball-carving}
\wh{C}_i:=\wt{C}_i\setminus \bigcup_{j=1}^{i-1} \wt{C}_j.
\end{equation}
In other words, this is a natural `hyperplane-carving' approach: points that belong to the first hyperplane $\wt{C}_1$ are taken as the first cluster, points in the second hyperplane $\wt{C}_2$ that were not captured by the first hyperplane are taken as the second cluster, etc. This is a natural high dimensional analog of the Cheeger cut that has been used in many results on spectral partitioning. The hope here would be to show that there exists a permutation $\pi$ on $[k]$ such that
\begin{equation}\label{eq:symm-small}
|\wh{C}_i\Delta C_{\pi(i)}|\leq O(\e) \cdot |C_{\pi(i)}|,
\end{equation}
for every $i=1,\ldots, k$, where we assume that the inner conductance $\phi$ of the clusters is constant. Here $\Delta$ stands for the symmetric difference operation. 

One natural approach to establishing~\eqref{eq:symm-small}  would be to prove that for every $i=1,\ldots, k$ vertices $x\in C_i$ concentrate well around cluster means $\mu_i$ (see Fig.~\ref{fig:example}). This would seem to suggest that $\wt{C}_i$'s are close to the $C_i$'s, and so are the $\wh{C}_i$'s. This property of the spectral embedding is quite natural to expect, and versions of this property have been used in the literature. For example, one can show that for every $\alpha\in \R^k, ||\alpha||_2=1$,
\begin{equation}\label{eq:223t}
\sum_{i=1}^k \sum_{x\in C_i} \langle f_x-\mu_i, \alpha\rangle^2\leq O(\e).
\end{equation}
The bound in~\eqref{eq:223t} follows using rather standard techniques -- see Section~\ref{sec:standard-variance-bounds} for this and related claims. One can check that~\eqref{eq:223t} suffices to show that $\wt{C}_i$'s are very close to $C_i$'s, namely that for every $i=1,\ldots, k$ there exists $j \in [k]$ such that
\begin{equation}\label{eq:923u5392}
|\wt{C}_i\Delta C_j|=O(\e)\cdot |C_j|.
\end{equation}
The formal proof is given in Section~\ref{sec:bound_int}.  The result in~\eqref{eq:923u5392} is encouraging and suggests that the clusters $\wh{C}_i$ defined by the simple hyperplane partitioning process approximate the $C_i$'s, but this is not the case! The problem lies in the fact that while $\wt{C}_i$'s approximate the $C_i$'s well as per~\eqref{eq:923u5392}, the bound in~\eqref{eq:923u5392} does not preclude nontrivial overlaps in the $\wt{C}_i$'s -- we give an example in below.

\subsubsection{Hard instance for natural hyperplane partitioning}\label{sec:counterexample}

We now give an example configuration of vertices in Euclidean space such that {\bf (a)} the configuration does not contradict~\eqref{eq:223t} and {\bf (b)} the natural hyperplane partitioning algorithm \eqref{eq:ball-carving} fails for this configuration. This shows why we develop a different algorithm that can deal with configurations like the one presented in this subsection.

Consider the following configuration of $C_i$'s and $\mu_i$'s. Suppose that all cluster sizes are equal $\frac{n}{k}$, and let $k=\frac1{\e}$. Let $\mu_i$'s form an orthogonal system and for each $i \in [k]$ let $||\mu_i||_2 = \sqrt{\frac{k}{n}}$. For all $i <k= 1/\e$ for all $x \in C_i$ we set $f_x = \mu_i$, that is points from all clusters except for $1/\e$'th one are tightly concentrated around cluster means -- see Fig.~\ref{fig:counterexample} for an illustration with $k=3$. Then for cluster $C_{1/\e}$ we distribute points as follows. For every $i = 1,\dots,1/\e - 1$ we move $\e/2$ fraction of its points to $\mu_{1/\e} + \mu_i$, and another $\e/2$ fraction of the points to $\mu_{1/\e}-\mu_i$. The remaining $\e$ fraction of $C_{1/\e}$ stays at $\mu_{1/\e}$. Now observe that all cluster means are where they should be, since we applied symmetric perturbations. Secondly notice that \eqref{eq:223t} is satisfied for every direction $\alpha$. Intuitively it is the case because we moved $1/\e - 1$ disjoint subsets of $C_{1/\e}$ of size $\e \frac{n}{k}$ in $1/\e - 1$ \textbf{orthogonal} directions. Lastly observe what happens to $\wt{C}_i$'s. For all $i = 1,\dots,1/\e - 1$ set $\wt{C}_i$ contains $C_i$ and $\e/2$ fraction of $C_{1/\e}$ that was moved in direction $\mu_i$. One can verify that this is perfectly consistent with~\eqref{eq:223t}, and in particular with \eqref{eq:923u5392}. The problem is that many clusters have large overlap with one particular cluster, namely $C_{1/\e}$. Indeed notice that the ball carving process returns $\wh{C}_{1/\e}$ such that $|\wh{C}_{1/\e} \cap C_{1/\e}| = (\frac{1+\e}{2})\frac{n}{k}$. That means that constant (almost $1/2$) fraction of cluster $C_{1/\e}$ is not recovered! 

\begin{figure}
\centering
\pgfmathsetseed{5}
\begin{tikzpicture}[scale=0.6] 
\draw [fill] (0,0) circle [radius=0.04];
\draw [very thick,->] (0,0,0) -- (10,0,0);
\draw [very thick,->] (0,0,0) -- (0,10,0);
\draw [very thick,->] (0,0,0) -- (0,0,10);
\draw [->] (0,7,0) -- (8,7,0);
\draw [->] (0,7,0) -- (-8,7,0);
\draw [->] (0,7,0) -- (0,7,9);
\draw [->] (0,7,0) -- (0,7,-9);




\foreach \i in {0,...,60}
        \shadedraw [ball color= gray] (7 - 0.5 + rnd,rnd - 0.5, rnd - 0.5) circle (0.1cm);
\foreach \i in {0,...,60}
        \shadedraw [ball color= gray] (rnd - 0.5, rnd - 0.5, 7 + rnd - 0.5) circle (0.1cm);
\foreach \i in {0,...,20}
        \shadedraw [ball color= gray] (rnd - 0.5, 7+  rnd - 0.5, rnd - 0.5) circle (0.1cm);
        
\foreach \i in {0,...,10}
        \shadedraw [ball color= gray] (7 + rnd/2 - 0.25, 7+  rnd/2 - 0.25, rnd/2 - 0.25) circle (0.1cm);
\foreach \i in {0,...,10}
        \shadedraw [ball color= gray] (-7 + rnd/2 - 0.25, 7+  rnd/2 - 0.25, rnd/2 - 0.25) circle (0.1cm);
\foreach \i in {0,...,10}
        \shadedraw [ball color= gray] (rnd/2 - 0.25, 7+  rnd/2 - 0.25, 7 +  rnd/2 - 0.25) circle (0.1cm);
\foreach \i in {0,...,10}
        \shadedraw [ball color= gray] (rnd/2 - 0.25, 7+  rnd/2 - 0.25, -7 + rnd/2 - 0.25) circle (0.1cm);
        
\node[] at (7,1.1,0) {$\mu_1$};
\node[] at (0,1.1,7) {$\mu_2$};
\node[] at (-0.8,7.8,0) (t) {$\mu_3$};


\node[] at (7,0,4) (n) {$C_1$};
\node[] at (0,3,12.5) (n) {$C_2$};
\node[] at (-1,11.5,0) (n) {$C_3$};


\draw [->] (7,0,3.3) -- (7,0,1.5); 
\draw [->] (0,2.4,11.4) -- (0,1,8.5); 
\draw [->] (-1,11,0) -- (-0.2,7.7,0); 
\draw [->] (-1,11,0) -- (6.6,7.2,0); 
\draw [->] (-1,11,0) -- (-6.6,7.2,0); 
\draw [->] (-1,11,0) -- (0,7.2,6.6); 
\draw [->] (-1,11,0) -- (0,7.2,-6.6); 

\end{tikzpicture}

\caption{\label{fig:counterexample} Example of a spectral embedding that is consistent with \eqref{eq:223t} and \eqref{eq:923u5392} but for which the natural hyperplane partitioning would not work.}
\end{figure}
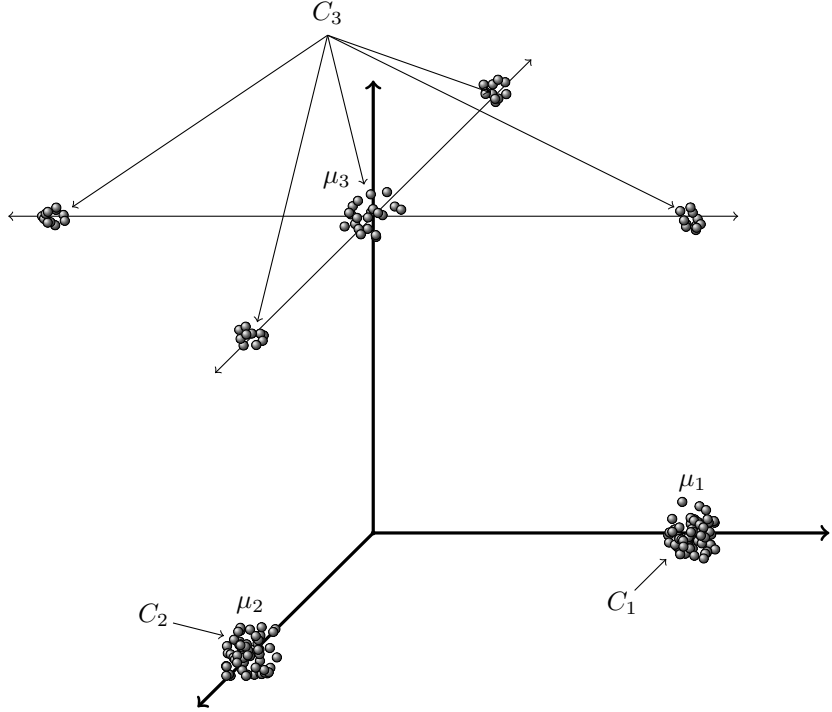

\subsubsection{Our hyperplane partitioning scheme}

The example in Section~\ref{sec:counterexample} suggests that we need to develop a diffferent algorithm. Our main contribution here is an algorithm that more carefully deals with the overlaps of $\widetilde{C}_i$'s. The high level idea for the algorithm is to recover clusters in stages and after every stage project out the directions corresponding to recovered clusters.

First we observe the following property of $(k,\varphi,\e)$-clusterable graphs (see Lemma~\ref{lem:howtocluster}). Any collection of pairwise disjoint sets with small outer-conductance matches the original clusters well. More precisely for every collection $\{ \wh{C}_1, \dots, \wh{C}_k \}$ of pairwise disjoint sets satisfying for every $i \in [k]$ $\phi(\wh{C}_i) \leq O(\e \log(k))$ there exists a permutation $\pi$ on $[k]$ such that
\begin{equation}\label{eq:symm-small2}
|\wh{C}_i\Delta C_{\pi(i)}|\leq O(\e \log(k)) \cdot |C_{\pi(i)}|,
\end{equation}
In the algorithm we will test many candidate clusters and the property above allows us to test if a particular candidate $\wh{C}$ is good by only computing its outer-conductance.

Now we describe our algorithm more formally. The algorithm proceeds in $O(\log(k))$ stages. In the first stage it considers $k$ candidate clusters $\wh{C}_i$, where $x \in \wh{C}_i$ if it has big correlation with $\mu_i$ but small correlation with all other $\mu_j$'s. More formally
\begin{equation}
\wh{C}_i := \widetilde{C}_i \setminus \bigcup_{j \neq i} \widetilde{C}_j \text{,}   
\end{equation}
which is equivalent to:
$$\rdp{f_x, \mu_i} \geq 0.9 \rn{\mu_i}^2 \text{ and for all } j \neq i \rdp{f_x, \mu_j} < 0.9 \rn{\mu_j}^2 \text{.}$$
Note that by definition all these clusters are disjoint. At this point we return all candidate clusters $\wh{C}_i$ for which $\phi(\wh{C}_i) \leq O(\e)$, remove the corresponding vertices from the graph, remove the corresponding $\mu$'s from the set $\{\mu_1, \dots, \mu_k \}$ of centers and proceed to the next stage.

In the next stage we restrict our attention to a lower dimensional subspace $\Pi$ of $\R^k$. Intuitively we want to project out all the directions corresponding to the removed cluster centers. Formally we define $\Pi$ to be the subspace orthogonal to all $\mu$'s removed up to this point (we overload notation by also using $\Pi$ for the orthogonal projection onto this subspace). We will see that $\mu$'s are close to being orthogonal (see Lemma~\ref{lem:dotmu}). This fact means that $\Pi \approx \text{span}(\{\mu_1, \dots, \mu_b \})$, where $\{\mu_1, \dots, \mu_b \}$ is the set of $\mu$'s that were not removed in the first step. Now the algorithm considers $b$ candidate clusters where the condition for $x$ being in a cluster $i$ changes to: 
$$\rdp{f_x, \Pi\mu_i} \geq 0.9 \rn{\Pi\mu_i}^2 \text{ and for all } j \in [b], j \neq i \rdp{f_x, \Pi\mu_j} < 0.9 \rn{\Pi\mu_j}^2 \text{.}$$
Now we return all candidate clusters that satisfy $\phi(\wh{C}_i) \leq O (\e ) $ but this time the constant hidden in the $O$ notation is bigger than in the first stage. In general at any stage $t$ we change the test to $O (\e \cdot t )$. At the end of the stage we proceed in a similar fashion by returning the clusters, removing the corresponding vertices and $\mu$'s and considering a lower dimensional subspace of $\Pi$ in the next stage.

The algorithm continues in such a fashion for $O(\log(k))$ stages. Thus for all returned clusters $\wh{C}_i$ it is true that there exists $j$ such that\footnote{Note that this algorithm may not return a partition of the graph but only a collection of disjoint clusters. Later, in Section~\ref{sec:combining} in Proposition~\ref{prop:partitioncollection}, we present a simple reduction that shows that an algorithm that guarantees \eqref{eq:symm-small2} is enough to construct a clustering oracle that, as required by Definition~\ref{def:oracle}, returns a partition. The high level idea is to assign the remaining vertices to clusters randomly.}:
$$|\wh{C}_{i} \triangle C_j| \leq O \left(\e \log(k) \right) \cdot |C_j| \text{.}$$

Let's analyze how this algorithm works for the configuration presented in Section~\ref{sec:counterexample}. In the first stage we have that, for all $i \neq \frac{1}{\e}$, $\wh{C}_i = C_i$ and moreover $|\wh{C}_{1/\e} \cap C_{1/\e}| = (\frac{1 + \e}{2}) \frac{n}{k}$. So all candidate cluster $\wh{C}_i$ for $i \neq 1/\e$ are returned but crucially this time (in contrast with the natural hyperplane partitioning) cluster $C_{1/\e}$ is left untouched. Then directions $\{ \mu_1, \dots, \mu_{1/\e -1 } \}$ are projected out. In the second stage the algorithm considers only vertices from $C_{1/\e}$ projected onto one dimensional subspace $\text{span}(\mu_{1/\e})$ and recovers this cluster up to $O(\e)$ error.

Because of the robustness property \eqref{eq:symm-small2}, to show that this algorithm works we only need to argue that at the end of $O(\log(k))$ stages $k$ sets are returned. We do that by showing that in every stage at least half of the remaining clusters is recovered. It is done in Lemma~\ref{lem:induction} and crucially relies on the following fact. 
When the algorithm considers a subspace $\Pi$ then the number of points in the union of sets:
$$\{ x \in V : \rdp{f_x, \Pi\mu_i} \geq 0.9 \rn{\Pi\mu_i}^2 \} \cap \{ x \in V : \rdp{f_x, \Pi\mu_j} \geq 0.9 \rn{\Pi\mu_j}^2 \} \text{,}$$
for all $i,j \in [b], i \neq j$
is bounded by $O(\e \cdot b \cdot \frac{n}{k})$ (see Lemma~\ref{lem:pointsoutside} and Remark~\ref{rem:twothresholds}). To prove that we observe that every point $x$ in this intersections has big projection onto some two $\mu_i, \mu_j$ from $\{\mu_1, \dots, \mu_b\}$.
Then using the fact that $\mu$'s are close to being orthogonal we deduce that $\Pi \approx \text{span}(\{\mu_1, \dots, \mu_b\})$  this in particular means that $\Pi\mu_i \approx \mu_i$, $\Pi\mu_j \approx \mu_j$. Because of that $f_x$ is abnormally far (further by a factor of $1/\e$ with respect to the average) from it's center $\mu_x$. Now applying \eqref{eq:223t} for an orthonormal basis of $\Pi$ and summing the inequalities we get that that the number of points in the intersections is bounded by $O(\e \cdot b \cdot \frac{n}{k})$. Having this bound we can argue that at least half of the remaining clusters is recovered as on average only  $O(\e \cdot \frac{n}{k})$ points from each cluster belong to the intersections. The formal argument is given in Section~\ref{sec:realcenterswork}.

The use of subspaces is crucial for our approach. If we relied solely on the bounds on norms (i.e. bounds on $\rn{f_x}$) we could only claim a recovery guarantee of $O(\e k)$ per cluster. One of the reasons is that there can be $\Theta(\e n)$ vertices of abnormally big norm and all of them can belong to one cluster (as it happens in the example from Section~\ref{sec:counterexample}). The use of carefully crafted sequence of subspaces solves this issue as it allows to derive better bounds for the number of abnormal vertices in each stage. It is possible as we can show that the "variance of the distribution" of $f_x$'s cannot concentrate on subspaces. This leads to an $O(\e \log(k))$ error guarantee per cluster.

What remains is to remove the assumption that the cluster means $\mu_i$ are known to the algorithm. We show, using our tail bounds from Lemma~\ref{lem:tail_bound}, that a random sample of $O( 1/\e \cdot k^3 \log k)$ points in every cluster is likely to concentrate around the mean. This allows us to take a $O(1/\e \cdot k^4 \log k)$ size sample of points, guess in exponential (in $1/\e \cdot k^4\log^2 k$) time which points belong to which cluster, and ultimately find surrogates $\widehat{\mu}_i$ that are sufficiently close to the actual $\mu_i$'s for the analysis to go through. This part of the analysis is presented in Section~\ref{sec:approx-muis}. We also need a mechanism for testing if a set of approximate $\wh{\mu}$'s induces (via our partitioning algorithm) a good clustering. We accomplish this goal by designing a simple sampling based tester that determines whether or not the clusters induced by a particular collection of candidate cluster means have the right size and outer conductance properties. See Section~\ref{sec:h_works} for this part of the analysis.

To design our spectral clustering algorithm we need to perform tests like $\rdp{f_x, \Pi\mu} \stackrel{?}{\geq} 0.9 \rn{\Pi\mu}_2^2$ for a given vertex $x$, a candidate cluster mean $\mu$, and the projection matrix $\Pi$. Hence, we need tools to approximate $\rdp{f_x, \Pi\mu}$ and $\rn{\Pi\mu}_2^2$. As explained above, instead of exact cluster means i.e. $\mu$ we will perform the test for approximate cluster means i.e, $\wh{\mu}=\frac{1}{|S|}\sum_{y\in S} f_y$, where $S$ is a small subset $S$ of sampled nodes. First observe that for any vertex $x$ one can estimate $\adp{f_x, \wh{\mu}}$ as follows:
$$\adp{f_x, \wh{\mu}} =\frac{1}{|S|}\sum_{y\in S}\adp{f_x, f_y} $$
where $\adp{f_x, f_y}$ can be computed using  (\textsc{SpectralDotProductOracle}) Algorithm \ref{alg:dotProduct-tt}. Next we will explain how to compute $\adp{f_x,\wh{\Pi} f_y}$ for $x,y\in V$. Recall that $\wh{\Pi}$ is the subspace orthogonal to all $\wh{\mu}$'s removed so far. Let $\{\wh{\mu}_1,\ldots, \wh{\mu}_r\}$ denote the set of removed cluster means, and let $X\in \R^{k\times r}$ denote a matrix whose columns are $\wh{\mu}_i$'s. Therefore the projection matrix onto the span of $\{\wh{\mu}_1,\ldots, \wh{\mu}_r\}$ is given by $X(X^TX)^{-1}X$. Hence, we have $\wh{\Pi}=I-X(X^TX)^{-1}X$ and we can compute  $\adp{f_x,\wh{\Pi} f_y}$ as follows:
$$
 \adp{f_x,\wh{\Pi} f_y} =  \adp{f_x, f_y} - (f_x^T X)(X^TX)^{-1}(X f_y) \text{.}
$$
Note that the $i$-th column of matrix $X$ is $\wh{\mu}_i$, thus $f_x^T X\in \R^{r}$ is a vector whose $i$-th entry can be computed by $\adp{f_x,\wh{\mu}_i}$. Moreover notice that $X^TX \in \R^{r\times r}$ is matrix such that its $(i,j)$-th entry can be computed by $\adp{\wh{\mu}_i, \wh{\mu}_j}$. Therefore $(f_x^T X)$, $(X f_y)$ and $(X^TX)^{-1}$ all can be computed explicitly which let us compute $\adp{f_x,\wh{\Pi} f_y}$. Given the primitive to compute $\adp{f_x,\wh{\Pi} f_y}$ we are able to estimate $\rdp{f_x, \Pi(\mu)}$ and $\rn{\Pi(\mu)}_2^2$ as follows:
 $$
\adp{f_x,\widehat{\Pi} \widehat{\mu}}:=  \frac{1}{|B|}\cdot\sum_{y\in B} \adp{f_x,\wh{\Pi} f_y} \text{,}
$$
$$
\an{\widehat{\Pi} \widehat{\mu}}^2:=  \frac{1}{|B|}\cdot\sum_{x\in B} \adp{f_x,\widehat{\Pi} \widehat{\mu}}  \text{.}
$$
This part of the analysis is presented in Section~\ref{sec:dotproductcomp}.

\section{Properties of the spectral embedding of $(k,\varphi,\epsilon)$-clusterable graphs} \label{sebsec:moment-bounds}

In this section we study the spectral embedding of $(k, \varphi,\epsilon)$-clusterable graphs. Recall that the spectral embedding maps every vertex $x\in V$ to a $k$-dimensional vector $f_x$.
We are interested in understanding the geometric properties of this embedding. We start by recalling some standard properties of the embedding: We show that the cluster means 
$$
\mu_i = \frac{1}{|C_i|} \sum_{x \in C_i} f_x
$$
 are almost orthogonal and of length roughly $1/\sqrt{|C_i|}$ (Lemma~\ref{lem:dotmu} below). Then we give a bound on the directional variance, by which we mean the sum of squared distances of points $f_x$ to their corresponding cluster centers when projected on direction $\alpha$.
We show in Lemma \ref{lem:dirvariance} below that the directional variance is bounded by $O(\epsilon/\varphi^2)$ for every direction $\alpha\in \R^k, \|\alpha\|=1$. This in particular implies (see Lemma~\ref{lem:spectraldistance} below) that `rounding' the spectral embedding by mapping each vertex to its corresponding cluster center results in a matrix $U$ that spectrally approximates the matrix of bottom $k$ eigenvectors of the Laplacian. These bounds are rather standard, and their proofs are provided  for completeness.  The main shortcoming of the standard bounds is that they can only allow us to apply averaging arguments, and are thus unable to rule out that some of the embedded points are quite far away from their corresponding cluster center. For example, they do not rule out the possibility of an $\Omega(1/k)$ fraction of the points being $\approx \sqrt{k}$ further away from their corresponding centers. Since we would like to recover every cluster to up an $O(\e)$ error, such bounds are not sufficient on their own.

For this reason we consider the distribution of the projection of the embedded points on the direction of any of the first $k$ eigenvectors and we give stronger tail bounds for these distributions (in Lemma \ref{lem:tail_bound}) than what follows from variance calculations only. Basically, we give a strong bound on the $O(\varphi^2/\e)$-th moment of the spectral embedding as opposed to just on the second moment, as above. These higher moment bounds are then crucially used to achieve sublinear time access to dot products in the embedded space in Section~\ref{sec:dot} (we need them to establish spectral concentration of a small number of random samples in Section~\ref{subsubsec:topk}) as well as to argue that a small sample of vertices contains a good approximation to the true cluster means $\mu_i, i=1,\ldots, k$ in its span in Section~\ref{sec:approx-muis}.

\subsection{Standard bounds on cluster means and directional variance}\label{sec:standard-variance-bounds}
The lemma below bounds the variance of the spectral embedding in any direction. 
\begin{restatable}{lemma}{lemmdirvariance}(Variance bounds)\label{lem:dirvariance}
 Let  $k \geq 2$ be an integer, $\varphi \in (0,1)$, and $\e \in (0,1)$. Let $G=(V,E)$ be a $d$-regular graph that admits $(k,\varphi,\e)$-clustering $C_1, \ldots, C_k$. Then for all $\alpha \in \R^k$, with $\|\alpha\|=1$ we have
$$
\sum_{i=1}^k \sum_{x \in C_i}  \rdp{f_x - \mu_i,  \alpha}^2 \leq \frac{4\epsilon}{\varphi^2}.
$$
\end{restatable}

\begin{proof}
For each $i \in [k]$, and any vertex $x\in C_i$, let $d_i(x)$ denote the degree of vertex $x$ in the subgraph $C_i$. Let $H_i$ be a graph obtained by adding $d-d_i(x)$ self-loops to each vertex $x\in C_i$. Let $L$ denote the normalized Laplacian of graph $G$. For each $i\in [k]$ and let $L_i$ denote the normalized Laplacian of $H_i$, and let $\lambda_2({H_i})$ be the second smallest eigenvalue of $L_i$.  

Let $z = U_{[k]} \alpha$. Note that $||z||_2=1$. By Lemma \ref{lem:bnd-lambda} we have $\lambda_1\leq\ldots \leq \lambda_k\leq 2\epsilon$, where $\lambda_i$ is the $i^{\text{th}}$ smallest eigenvalue of of $L$. Therefore we have
\begin{equation}\label{eq:zinsmallspace}
\langle z , Lz \rangle  \leq \lambda_k \leq 2\epsilon
\end{equation}
Fix some $i\in[k]$, let $z' \in \R^n$ be a vector such that $z'(x) := z(x) - \langle \mu_i, \alpha \rangle$. For any $S \subseteq V$, we define $z'_{S} \in \mathbb{R}^n$ to be a vector such that for all $x \in V$ $z'_S(x) = z'(x)$ if $x \in S$ and $z'_S(x) = 0$ otherwise. Note that $z(x)=\rdp{f_x,\alpha}$, thus we have
\[\sum_{x\in V} z'_{C_i}(x) = \sum_{x\in C_i} z'(x)=\sum_{x\in C_i} z(x)-\rdp{\mu_i,\alpha} =\sum_{x\in C_i} \rdp{f_x-\mu_i,\alpha} = 0\]
Thus we have  $z'_{C_i}\perp \mathds{1}$, so by  properties of Rayleigh quotient we get
\begin{equation}\label{eq:rayleighquotient}
\frac{\langle z'_{C_i}, L_i z'_{C_i} \rangle}{\langle z'_{C_i}, z'_{C_i} \rangle} =
\frac{1}{d} \frac{\sum_{x,y \in C_i, (x,y) \in E} (z'(x)-z'(y))^2}{\sum_{x \in C_i} (z'(x))^2 } =
\frac{1}{d} \frac{\sum_{x,y \in C_i, (x,y) \in E} (z(x)-z(y))^2}{\sum_{x \in C_i} (z(x) - \langle \mu_i, \alpha \rangle)^2 } \geq \lambda_2({H_i})
\end{equation}
Furthermore, by Cheeger's inequality for any $i\in [k]$ we have $\lambda_2(H_i)\geq \frac{\varphi^2}{2}$. Hence, for any $i\in [k]$ we have
\[\frac{\sum_{x,y \in C_i, (x,y) \in E} (z(x)-z(y))^2}{d\sum_{x \in C_i} (z(x) - \langle \mu_i, \alpha \rangle)^2 } \geq \lambda_2({H_i}) \geq \frac{\varphi^2}{2} \]
Now observe the following:
\begin{align*}
2\epsilon 
&\geq
\langle z, Lz \rangle &&\text{By (\ref{eq:zinsmallspace})}\\
&=
\frac{1}{d}\cdot \sum_{(x,y) \in E}(z(x)-z(y))^2 \\ 
&\geq \frac{1}{d}\cdot \sum_{i=1}^k \sum_{x,y \in C_i, (x,y) \in E} (z(x)-z(y))^2 \\ 
&\geq \frac{\varphi^2}{2} \cdot\sum_{i=1}^k \sum_{x \in C_i} (z(x) - \langle \mu_i, \alpha \rangle)^2    &&\text{By \eqref{eq:rayleighquotient}} 
\end{align*}
Recall that for all $x\in V$, $z(x)=\rdp{f_x,\alpha}$. Therefore for for any $\alpha \in \R^k$ with $\|\alpha\| = 1$ we have
\[
    \sum_{i=1}^k \sum_{x \in C_i} \langle f_x - \mu_i, \alpha \rangle^2 \leq \frac{4\epsilon}{\varphi^2}
\]
\end{proof}
The following lemma shows that the length of the cluster mean of cluster $C_i$ is roughly $1/\sqrt{|C_i|}$ and that cluster means are almost orthogonal.
\begin{restatable}{lemma}{lemmdotmu}(Cluster means)\label{lem:dotmu}
Let  $k \geq 2$ be an integer, $\varphi \in (0,1)$, and $\e \in (0,1)$. Let $G=(V,E)$ be a $d$-regular graph that admits $(k,\varphi,\e)$-clustering $C_1, \ldots, C_k$. Then we have
\begin{enumerate}
    \item $\text{ for all } i \in [k]$, $\left| ||\mu_i||_2^2 - \frac{1}{|C_i|} \right| \leq \frac{4\sqrt{\epsilon}}{\varphi} \frac{1}{|C_i|}$\label{itm1:lem-dotmu}
    \item $\text{ for all } i \neq j \in [k]$, $\left| \langle \mu_i, \mu_j \rangle \right| \leq \frac{8\sqrt{\epsilon}}{\varphi }\frac{1}{\sqrt{|C_i\|C_j|}}$ \label{itm2:lem-dotmu}
\end{enumerate}
\end{restatable}
To prove Lemma \ref{lem:dotmu} we need Lemma \ref{lem:spectraldistance} in which we will use the following result from \cite{HornJ} (Theorem 1.3.20 on page 53).
\begin{lemma}[\cite{HornJ}]\label{lem_commute}
Let $h,m,n$ be integers such that $1\leq h\leq m \leq n$. For any  matrix $A\in \R^{m\times n}$ and matrix $B\in \R^{n\times m}$, the multisets of nonzero eigenvalues of $AB$ and $BA$ are equal. In particular, if one of $AB$ and $BA$ is positive semidefinite, then $\nu_h(AB)=\nu_h(BA)$.
\end{lemma}
\begin{lemma}\label{lem:spectraldistance}
 Let  $k \geq 2$ be an integer, $\varphi \in (0,1)$, and $\e \in (0,1)$. Let $G=(V,E)$ be a $d$-regular graph that admits $(k,\varphi,\e)$-clustering $C_1, \ldots, C_k$.  Let $H\in \R^{k \times k}$ be a matrix whose $i$-th column is $\mu_i$. Let $W\in \R^{k\times k}$ be a diagonal matrix such that $W(i,i)=\sqrt{|C_i|}$. Then for any $\alpha \in \R^k$, $\|\alpha\|=1$, we have
 \begin{enumerate}
 \item $  |\alpha^T \left((HW)(HW)^T - I\right)\alpha| \leq \frac{4\sqrt{\epsilon}}{\varphi}$ \label{lem-spec-dist-itm1}
 \item $|\alpha^T \left((HW)^T(HW) - I\right)\alpha| \leq \frac{4\sqrt{\epsilon}}{\varphi}$ \label{lem-spec-dist-itm2}
 \end{enumerate}
\end{lemma}
\begin{proof}
\textbf{Proof of item \eqref{lem-spec-dist-itm1}:}
Let $Y\in \R^{k\times n}$ denote a matrix whose $x$-th column is $\mu_x$ for any $x\in V$. Note that 
\[YY^T=\sum_{i=1}^k |C_i|\mu_i\mu_i^T= (HW)(HW)^T\text{.}\] 
We define $\widetilde{z} := Y^T\alpha$, and $z := U_{[k]} \alpha$. Note that $U_{[k]}^TU_{[k]}=I$. Therefore we have
\begin{align}
|\alpha^T \left((HW)(HW)^T - I\right)\alpha| &= |\alpha^T (YY^T - U_{[k]}^TU_{[k]})\alpha| \nonumber\\ 
&= \left|\sum_{x \in V} \widetilde{z}(x)^2 - z(x)^2\right|  &&\text{From definition of $z(x)$ and  $\widetilde{z}(x)$}\nonumber\\
 &= \left| \sum_{x \in V} \left( z(x)-\widetilde{z}(x)\right)\left( z(x)+\widetilde{z}(x)\right) \right| \nonumber\\ 
 &\leq 
 \sqrt{\sum_{x \in V} ( z(x)- \widetilde{z}(x))^2 \sum_{x \in V} (\widetilde{z}(x) + z(x))^2}  && \text{By Cauchy-Schwarz inequality } \label{eq:hwI}
\end{align}
Note that for any $x\in V$, we have $z(x)=\rdp{f_x,\alpha}$ and $\widetilde{z}(x)=\rdp{\mu_x,\alpha}$. Therefore by Lemma \ref{lem:dirvariance} we have
\begin{equation}\label{eq:var-z}
\sqrt{\sum_{x \in V} ( z(x)- \widetilde{z}(x))^2} =  \sqrt{\sum_{x \in V} \rdp{f_x-\mu_x, \alpha}^2} \leq \frac{2\sqrt{\epsilon}}{\varphi}
\end{equation}
To complete the proof it suffices  to show that $\sum_{x \in V} (\tilde{z}(x) + z(x))^2 \leq 4$. Note that
\begin{align*}
\sum_{x \in V} \tilde{z}(x)^2 
&= \sum_{x \in V} \rdp{\alpha, \mu_x}^2 \\
&= \sum_{i} |C_i| \rdp{\alpha, \frac{\sum_{x \in C_i} f_x}{|C_i|}}^2 \\
&= \sum_{i} |C_i| \left( \frac{\sum_{x \in C_i} \rdp{\alpha,f_x}}{|C_i|} \right)^2 \\
&\leq \sum_{i} \sum_{x \in C_i} \rdp{\alpha,f_x}^2 && \text{By Jensen's inequality} \\
&= \sum_{x \in V} z(x)^2
\end{align*}
Thus we have
\begin{equation}\label{eq:means}
\sum_{x \in V} (\tilde{z}(x) + z(x))^2 \leq \sum_{x \in V} 2(\tilde{z}(x)^2 + z(x)^2) \leq 2 + 2\sum_{x \in V} \tilde{z}(x)^2 \leq 4
\end{equation}
In the first inequality we used the fact that $(\tilde{z}(x) - z(x))^2 \geq 0$ and for the second inequality we used the fact that $||z||_2^2=||U_{[k]}\alpha||_2^2=1$. Putting \eqref{eq:means}, \eqref{eq:var-z}, and \eqref{eq:hwI} together we get
\[|\alpha^T \left((HW)(HW)^T - I\right)\alpha| \leq \frac{4\sqrt{\epsilon}}{\varphi} \text{.}\]
\textbf{Proof of item \eqref{lem-spec-dist-itm2}:} 
Note that by item \eqref{lem-spec-dist-itm2} for any vector $
\alpha$ with $||\alpha||_2=1$ we have
\[ 1- \frac{4\sqrt{\epsilon}}{\varphi}\leq \alpha^T \left((HW)(HW)^T \right)\alpha \leq 1+\frac{4\sqrt{\epsilon}}{\varphi}\] 
Thus by Lemma \ref{lem_commute} we have that the set of eigenvalues of $(HW)(HW)^T$ and $(HW)^T(HW)$ are the same, and all of the eigenvalues lie in the interval $[1- \frac{4\sqrt{\epsilon}}{\varphi}, 1+\frac{4\sqrt{\epsilon}}{\varphi}]$. Thus for any vector $\alpha$ with $||\alpha||_2=1$ we have
\[ 1- \frac{4\sqrt{\epsilon}}{\varphi}\leq \alpha^T \left((HW)^T(HW) \right)\alpha \leq 1+\frac{4\sqrt{\epsilon}}{\varphi} \text{.}\] 
\end{proof}
Now we are able to prove Lemma \ref{lem:dotmu}.
\lemmdotmu*
\begin{proof}
\textbf{Proof of item \eqref{itm1:lem-dotmu}:}
Let $H\in \R^{k \times k}$ be a matrix whose $i$-th column is $\mu_i$. Let $W\in \R^{k\times k}$ be a diagonal matrix whose such that $W(i,i)=\sqrt{|C_i|}$. Thus by Lemma \ref{lem:spectraldistance} item \eqref{lem-spec-dist-itm2} for any $\alpha \in \R^k$ with $\|\alpha\|=1$, we have
\[|\alpha^T \left((HW)^T(HW) - I\right)\alpha| \leq \frac{4\sqrt{\epsilon}}{\varphi}\]
Let $\alpha = \mathds{1}_i$. Thus we have 
\begin{equation}\label{eq:smalldiag}
|((HW)^T(HW) )(i,i) - 1| \leq \frac{4\sqrt{\epsilon}}{\varphi}
\end{equation} 
Note that $((HW)^T(HW) )(i,i) = (WH^THW)(i,i)=||\mu_i||^2_2 |C_i|$.Therefore we get
\[\left|||\mu_i||^2_2 -\frac{1}{|C_i|}\right|\leq \frac{4\sqrt{\epsilon}}{\varphi} \cdot \frac{1}{|C_i|} \]
\textbf{Proof of item \eqref{itm2:lem-dotmu}:}
Let $\alpha = \frac{1}{\sqrt{2}}(\mathds{1}_i + \mathds{1}_j)$. Note that $||\alpha||_2=1$. Thus by Lemma \ref{lem:spectraldistance} item \eqref{lem-spec-dist-itm2} we have
\[|\alpha^T \left((HW)^T(HW) - I\right)\alpha| \leq \frac{4\sqrt{\epsilon}}{\varphi}\] 
Note that 
\[\left| \alpha^T \left((HW)^T(HW)-I\right) \alpha \right| = \left| \frac{1}{2}\left(||\mu_i||^2_2 |C_i|+||\mu_j||^2_2 |C_j|+2\rdp{\mu_i,\mu_j}\sqrt{|C_i||C_j|}-2 \right)\right| \]
Therefore we get
\[ \left|  ||\mu_i||^2_2 |C_i|+||\mu_j||^2_2 |C_j|+2\rdp{\mu_i,\mu_j}\sqrt{|C_i||C_j|}-2 \right| \leq \frac{8\sqrt{\epsilon}}{\varphi}\]
Thus
\begin{align*}
\left|   \rdp{\mu_i,\mu_j}\sqrt{|C_i||C_j|} \right| &\leq \left|  \frac{1}{2}\left(1-||\mu_i||^2_2 |C_i|\right)+\frac{1}{2}\left(1-||\mu_j||^2_2 |C_j|\right) \right| +\frac{4\sqrt{\epsilon}}{\varphi} \\
&\leq \frac{1}{2}\cdot \frac{4\sqrt{\epsilon}}{\varphi}  + \frac{1}{2}\cdot \frac{4\sqrt{\epsilon}}{\varphi} +  \frac{4\sqrt{\epsilon}}{\varphi} &&\text{By item \eqref{itm1:lem-dotmu}} \\
&\leq \frac{8\sqrt{\epsilon}}{\varphi}
\end{align*}
Therefore we get
\[\left|  \rdp{\mu_i,\mu_j} \right| \leq \frac{8\sqrt{\epsilon}}{\varphi}
\cdot\frac{1}{\sqrt{|C_i||C_j|}} \text{.}\]
\end{proof}

\subsection{Strong Tail Bounds on the Spectral Embedding}\label{sec:moment-bounds}
The main results of this section are the following two lemmas. The first lemma gives an upper bound on the length of the projection of any point $f_x$ on an arbitrary direction $\alpha \in \R^k$.
The second lemma considers the distribution of the lengths of projected $f_x$ and we get tail bounds that show that the fraction of points whose projected length exceeds the `expectation' (which is about $1/\sqrt{|C_i|}$ for the smallest cluster $C_i$) by a factor of $\beta$ is bounded by $\beta^{-\varphi^2/{10\epsilon}}$. In other words, we bound the $O(\varphi^2/\e)$-th moment as opposed to the second moment, which gives us tight control over the embedding when $\e/\varphi^2\ll1/\log k$.

\lemmalinf*

\lemmatail*

We are interested in deriving moment bounds for the distribution of the entries of the first $k$ eigenvectors $u$ of $L$ (i.e., eigenvectors with eigenvalue smaller than $2\epsilon$), and 
specifically in the distribution of the absolute values of the entries of $u$. In order to be able to analyze this distribution, we define the sets of all entries in $u$ that are bigger than a threshold $\theta$:

\begin{definition}[Threshold sets]
\label{def:th-sets}
Let $G=(V,E)$ be a graph with normalized Laplacian $L$.  Let $u$ be a normalized eigenvector of $L$ with $||u||_2=1$. Then for the vector $u$ and a threshold  $\theta\in \R^{+}$ we define the threshold set $S(\theta)$ with respect to the eigenvector $u$ and threshold $\theta$ as
$$
S(\theta) := \{x \in V : u(x) \geq \theta \}\text{.}
$$
\end{definition}

Our arguments will use that for every vertex $x$, we have  $u(x)\approx \frac{1}{d}\sum_{\{x,y\}\in E} u(y)$. So nodes neighboring other nodes with large $u(\cdot)$ values are likely to have large $u(\cdot)$ values as well. 
This motivates the following definition of the potential of a threshold set.
\begin{definition}[Potential of a threshold set]
\label{def:pot-th-set}
Let $G=(V,E)$ be a graph with normalized Laplacian $L$. Let $u$ be a normalized eigenvector of $L$ with $||u||_2=1$. Then for vector $u$ and a threshold  $\theta\in \R^{+}$ we define the potential of a threshold set $S(\theta)$ as
$$
p(\theta)= \sum_{x \in S(\theta)} u(x)\text{.}
$$
\end{definition}

We start by proving a core bound on the threshold sets  (Lemma~\ref{lem:maintool} below) that forms the basis of our approach: the main technical results of this section (Lemma \ref{lem:l-inf-bnd} and Lemma \ref{lem:tail_bound}) essentially follow by repeated application of Lemma~\ref{lem:maintool}. Specifically, we now argue that if a threshold set $S(\theta)$  expands in the graph $G$ and the relative potential of the set (i.e., $p(\theta)/|S(\theta)|$) is at most $2\theta$, then we can slightly decrease $\theta$ to obtain a new $\theta'$ such that the corresponding
threshold set is a constant factor larger that $S(\theta)$ and the relative potential is bounded by $2\theta'$. 

\begin{lemma}[Threshold shift for expanding threshold sets]\label{lem:maintool}
Let $G=(V,E)$ be a $d$-regular graph with normalized Laplacian $L$.   Let $u$ be a normalized eigenvector of $L$ with $||u||_2=1$ and with eigenvalue $\lambda\leq 2\epsilon$. Let $\theta \in \R^+$ be a threshold. Suppose that $S(\theta)$  is the  threshold set with respect to $u$ and $\theta$ such that $S(\theta)$ is non-empty, $\phi^G(S(\theta)) \geq \varphi$ and $\frac{p(\theta)}{|S(\theta)|} \leq 2\theta$. Then the following holds for $\theta'=\theta\left(1 -\frac{8\epsilon}{\varphi}\right)$:
\begin{enumerate}
\item $|S(\theta')| \geq (1+\varphi/2)|S(\theta)|
\text{, and}$ \label{item1:lem-th-shif}
\item $\frac{p(\theta')}{|S(\theta')|} \leq 2{\theta'}$. \label{item2:lem-th-shif}
\end{enumerate}
\end{lemma}

\begin{proof}
\textbf{Proof of item \eqref{item1:lem-th-shif}:}
Note that $\lambda u = Lu=(I-\frac{A}{d})u$. Thus for any $x\in V$ we have $\left(Lu\right)(x)=u(x)-\frac{1}{d}\sum_{\{x,y\}\in E} u(y)$. Thus we have, 
$$u(x)-\frac{1}{d}\sum_{\{x,y\}\in E} u(y)=\lambda \cdot u(x)\text{.}$$ 
We write the above as
\begin{equation} \label{eq:sum_dif_eigvec}
\sum_{y\in \mathcal{N}(x)} (u(x) - u(y)) = d \cdot \lambda \cdot u(x)\text{,}
\end{equation}
where $\mathcal{N}(x) = \{y\in V: \exists \{x,y\} \in E\}$. Summing \eqref{eq:sum_dif_eigvec} over all $x\in S(\theta)$ we get
\begin{equation} \label{eeq:sum_P}
\sum_{x\in S(\theta)} \sum_{y\in \mathcal{N}(x)} (u(x)-u(y)) = \sum_{x\in S(\theta)} \lambda \cdot d \cdot u(x) = \lambda\cdot d\cdot p(\theta) \text{,}
\end{equation}
and note that
\begin{equation} \label{eeq:sum_boundry}
\sum_{x\in S(\theta)} \sum_{y\in \mathcal{N}(x)} (u(x)-u(y)) = \sum_{\substack{\{x,y\}\in E\\
                  x \in S(\theta), y \not\in S(\theta)}} (u(x)-u(y)) \text{.}
\end{equation}
For any edge $e=\{x,y\}\in E$, we define $\Delta(e)=|u(x)-u(y)|$.
Note that for any $e=\{x,y\}$ such that $x\in S(\theta)$ and $y\not\in S(\theta)$ we have $u(x)\geq \theta > u(y)$, hence $\Delta(e)=u(x)-u(y)$. Therefore,
putting \eqref{eeq:sum_boundry} and \eqref{eeq:sum_P} together we get
\[
\sum_{e \in E(S(\theta), V\setminus S(\theta))} \Delta(e) = \lambda \cdot d \cdot p(\theta) \text{.}
\]
By an averaging argument there exists a set $E_L \subseteq E(S_\theta, V\setminus S_\theta)$ such that $|E_L|\geq \frac{|E(S(\theta), V\setminus S(\theta))|}{2}$ and all edges $e \in E_L$
satisfy $\Delta(e) \leq \frac{2\cdot\lambda \cdot d \cdot p(\theta)}{|E(S(\theta), V\setminus S(\theta))|} $.
We define $V_L$ as a subset of vertices of $V\setminus S(\theta)$ that are connected to vertices of $S(\theta)$ by edges in $E_L$, i.e.
\[V_L=\{y \in V\setminus S(\theta): \exists \ \{x,y\}\in E_L, x\in S(\theta)  \} \text{.}\]
Note that 
\begin{equation}
\label{L_i-size}
|V_L|\geq \frac{|E_L|}{d} \geq \frac{|E(S(\theta), V\setminus S(\theta))|}{2d}\text{.}
\end{equation}
Using the assumption of the lemma that $\phi^G(S(\theta))\geq\varphi$ we obtain 
\begin{equation}
\label{eeq:S_i-exp}
|E(S(\theta), V\setminus S(\theta))| \geq \varphi \cdot d \cdot |S(\theta)|\text{.}
\end{equation}
Putting~\eqref{eeq:S_i-exp} and~\eqref{L_i-size} together we get
\begin{equation}
\label{eeq:Li-size-final}
|V_L|\geq \frac{\varphi |S(\theta)|}{2} \text{.}
\end{equation}
Recall that for all $e \in E_L$ we have $\Delta(e) \leq \frac{2\cdot\lambda\cdot d\cdot p(\theta)}{|E(S(\theta), V\setminus S(\theta))|}$. We have $\lambda\leq 2\epsilon$, therefore for all $e \in E_L$ we have $\Delta(e) \leq \frac{4\cdot\epsilon\cdot d\cdot p(\theta)}{|E(S(\theta), V\setminus S(\theta))|}$. Thus for all $y\in V_L$ we get
\begin{equation}
\label{eeq:L_i-val}
u(y) \geq \theta- \frac{4\cdot\epsilon\cdot d\cdot p(\theta)}{|E(S(\theta), V\setminus S(\theta))|} \text{.}
\end{equation}
By the assumption of the lemma we have $
\frac{p(\theta)}{|S(\theta)|} \leq 2\theta$, hence, by inequality~\eqref{eeq:S_i-exp} we get
\begin{equation}
\label{eeq:theta_i+1-lw}
\theta-\frac{4\cdot\epsilon\cdot d\cdot p(\theta)}{|E(S(\theta), V\setminus S(\theta))|} \geq \theta-\frac{4\cdot\epsilon\cdot d\cdot p(\theta)}{\varphi \cdot d \cdot |S(\theta)|} = \theta-\frac{4\epsilon}{\varphi}\cdot\frac{p(\theta)}{ |S(\theta)|}
\geq \theta\left(1-\frac{8\epsilon}{\varphi}\right)\text{.}
\end{equation}
Putting~\eqref{eeq:theta_i+1-lw} and~\eqref{eeq:L_i-val} together we get for all $y\in V_L$, $ u(y) \geq \theta\left(1-\frac{8\epsilon}{\varphi}\right)$.  Let $\theta' := \theta(1- \frac{8\epsilon}{\varphi})$. Thus \[S(\theta) \cup V_L \subseteq S(\theta'){.}\]
By definition of $V_L$ we have $V_L\cap S(\theta)=\emptyset$. Therefore, $|S(\theta')|\geq |S(\theta)|+|V_L|$. Thus by inequality~\eqref{eeq:Li-size-final} we get
\begin{equation}
\label{eq:constantexpansion}
|S(\theta')|\geq |S(\theta)| \left(1+\frac{\varphi}{2}\right) \text{.}
\end{equation}

This concludes the proof of the first part of the lemma.

\textbf{Proof of item \eqref{item2:lem-th-shif}:}
Now using that for all $x\not\in S(\theta)$ we have $u(x) < \theta$ and that $p(\theta) \leq 2 \theta |S(\theta)|$ by assumption of the lemma we obtain

\begin{align*}
p(\theta')&=\sum_{u\in S(\theta')} u(x) \\
&= \sum_{x\in S(\theta)} u(x) + \sum_{x\in S(\theta')\setminus S(\theta)} u(x) \\
&\leq p(\theta) + \theta |S(\theta')\setminus S(\theta)|\\
&\leq   2\theta |S(\theta)| + \theta|S(\theta')\setminus S(\theta)|\text{.} && \text{Since }p(\theta) \leq 2 \theta |S(\theta)|\\
\end{align*}  
By \eqref{eq:constantexpansion} we have $|S(\theta')\setminus S(\theta)|\geq \frac{\varphi}{2}|S(\theta)|$. Therefore, using $\epsilon \leq \frac{\varphi^2}{100}$ we get
\begin{eqnarray*}
\frac{p(\theta')}{|S(\theta')|}&\leq& \frac{2\theta|S(\theta)|+\theta|S(\theta')\setminus S(\theta)|}{|S(\theta)|+|S(\theta')\setminus S(\theta)|} 
= \theta\cdot\frac{2+\frac{|S(\theta')\setminus S(\theta)|}{|S(\theta)|}}{1+\frac{|S(\theta')\setminus S(\theta)|}{|S(\theta)|}} 
\leq \theta\cdot\frac{2+\frac{\varphi}{2}}{1+\frac{\varphi}{2}}
\leq \theta \cdot 2\left(1-\frac{8\epsilon}{\varphi}\right)
\leq 2 \theta'
\end{eqnarray*}

\end{proof}

We would like to apply Lemma~\ref{lem:maintool} iteratively, but there is one hurdle: while the first condition on the threshold set $S(\theta)$ naturally follows as long as $S(\theta)$ is not too large (by Proposition \ref{lem:S_exp}), the second condition needs to be established at the beginning of the iterative process. Lemma~\ref{lem:concentrationnearboundary} accomplishes exactly that: we prove that for any value $\theta_1$ with threshold set $S(\theta_1)$  not empty or not too large, there exists a close value $\theta$ that meets the conditions of previous lemma.

\begin{proposition}
\label{lem:S_exp}
Let $G=(V,E)$ be a $d$-regular graph that admits a $(k,\varphi,\epsilon)$-clustering $C_1, \ldots, C_k$. For any set $S\subseteq V$ with size $|S|\leq \frac{1}{2}\cdot \min_{i\in k}|C_i|$ we have $\phi^G(S)\geq\varphi$.
\end{proposition}
\begin{proof}
For any $1\leq i \leq k$ we define $S_i=S\cap C_i$. Note that \[|S_i|\leq |S|\leq \frac{1}{2}\cdot \min_{i\in k}|C_i| \leq \frac{|C_i|}{2}\text{.}\] Therefore since $\phi^G(C_i)\geq\varphi$ we have $E(S_i, C_i\setminus S_i) \geq \varphi d |S_i|$. Thus we get 
\[E(S,V\setminus S) \geq \sum_{i=1}^k  E(S_i, C_i\setminus S_i) \geq    \varphi d  \sum_{i=1}^k |S_i| = \varphi d |S| \text{.}\]
Hence, $\phi^G(S)\geq\varphi$.
\end{proof}

\begin{lemma} \label{lem:concentrationnearboundary}
Let $\varphi\in(0,1)$ and $\epsilon \leq \frac{\varphi^2}{100}$, and let $G=(V,E)$ be a $d$-regular graph that admits $(k,\varphi,\e)$-clustering $C_1, \ldots, C_k$.  Let $L$ denote the normalized Laplacian of $G$.
  Let $u$ be a normalized eigenvector of $L$ with $||u||_2=1$ and with eigenvalue $\lambda\leq 2\epsilon$. Let $\theta_1 \in \R^+$ be a threshold. Let $S(\theta_1)$  be the  threshold set with respect to $u$ and $\theta_1$. Suppose that 
  $1 \leq |S(\theta_1)| \leq \frac{1}{2} \cdot \min_{i\in \{1,\dots, k\}}|C_i|$. Then there exists a threshold $\theta_2$ such that the following holds:
\begin{enumerate}
\item $\theta_1\left(1-\frac{8\epsilon}{\varphi}\right) \leq \theta_2 \leq \theta_1\text{, and}$
\item $ \frac{p(\theta_2)}{|S(\theta_2)|} \leq 2\theta_2 $
\end{enumerate}

\end{lemma}

\begin{proof}
Let 
$$
\theta^* := \min\left\{\theta \geq \theta_1 \; \biggl| \; S(\theta)\not=\emptyset \text{ and } \frac{p(\theta)}{|S(\theta)|} \leq 2\theta\right\} \text{.}
$$ 
We can conclude that $\theta^*$ exists, as by the assumption of the lemma we have $|S(\theta_1)|\ge 1$ and for $\theta_{\max} = \max_{x \in V} u(x)$ we have
$\frac{p(\theta_{\max})}{|S(\theta_{\max})|} = \theta_{\max}$. We also have $|S(\theta^*)| \leq \min_{i\in \{1,\dots,k\}} |C_i|/2$ as $\theta^* \geq \theta_1$ and by the assumption
of the lemma. So Proposition~\ref{lem:S_exp} implies
\begin{equation}\label{eq:bigexpansion}
    \phi^G(S(\theta^*)) \geq \varphi \text{.}
\end{equation}
Now Lemma~\ref{lem:maintool} implies 
$$
\frac{p(\theta^* (1 -\frac{8\epsilon}{\varphi}))}{|S(\theta^* (1 -\frac{8\epsilon}{\varphi}))|} \leq 2\theta^*\left(1 -\frac{8\epsilon}{\varphi}\right)$$
and by minimality of $\theta^*$ we have that:
$$
\theta_1\left(1-\frac{8\epsilon}{\varphi}\right) \leq \theta^*\left(1-\frac{8\epsilon}{\varphi}\right) \leq \theta_1 \text{.}
$$
So we can set $\theta_2 := \theta^*\left(1-\frac{8\epsilon}{\varphi}\right)$.
\end{proof}

We are now ready to prove our tail bound. The main idea behind the proof is to use Lemma~\ref{lem:maintool} and Lemma~\ref{lem:concentrationnearboundary} to show that if a vertex has a large entry along one of the  bottom $k$ eigenvectors this implies that many other vertices also have a relatively large value along the same eigenvector. Thus, not too many $f_x$ can have such a large value. 

\lemmatail*

\begin{proof}
Let $s_{\min}=\min_{i\in \{1,\dots,k\}}|C_i|$. We define 
\[S^{+}=\left\{x\in V: u(x) \geq \beta\cdot\sqrt{\frac{10}{s_{\min}}}\right\}\text{,}\]
and
\[S^{-}=\left\{x\in V: -u(x) \geq \beta\cdot\sqrt{\frac{10}{s_{\min}}}\right\}\]
Note that $-u$ is also an eigenvector of $L$ with the same eigenvalue as $u$, hence, without loss of generality suppose that $|S^{+}|\geq |S^{-}|$. Let $T=\left\{x\in V: u(x)^2 \geq \frac{10}{s_{\min}}\right\}\text{.}$ Since, $1=\|u\|^2_2 = \sum_{x\in V} u(x)^2$, an averaging argument implies $|T|\leq \frac{s_{\min}}{10}$.  Let 
\[T^{+}=\left\{x\in V: u(x) \geq \sqrt{\frac{10}{s_{\min}}}\right\}\text{.}\]
Note that $\beta>1$, hence, ${S^{+}} \subseteq {T^{+}} \subseteq T$, and so we have $|{S^{+}}| \leq |T^+|\leq|T|\leq \frac{s_{\min}}{10}$.
We may assume that $S^{+}$ is non-empty as otherwise the lemma follows immediately. Let $\theta_{0}=\beta\cdot\sqrt{\frac{10}{s_{\min}}}$. Note that $S^{+}=S(\theta_0)$. Hence, $1\leq |S(\theta_{0})|\leq \frac{s_{\min}}{10}$.  Therefore by Lemma \ref{lem:concentrationnearboundary} there exists a threshold $\theta_1$  such that
\begin{equation}
\label{1eq:thet0}
\left(1 - \frac{8\epsilon}{\varphi}\right)\beta\cdot\sqrt{\frac{10}{s_{\min}}} \leq \theta_1 \leq \beta\cdot\sqrt{\frac{10}{s_{\min}}}\text{, and}
\end{equation}
\[
\frac{p(\theta_1)}{|S(\theta_1)|} \leq 2\theta_1\text{.}\] 
For any $t\geq 1$ we define $\theta_{t+1}=\theta_t(1-\frac{8\epsilon}{\varphi})$. For some $t'\geq 0$ we must have $\theta_{t'+1}\leq \sqrt{\frac{10}{s_{\min}}} \leq \theta_{t'}$.
Thus by \eqref{1eq:thet0} we have
\begin{equation}
\label{1eq:theta-t-lwb}
\theta_{t'} = \left(1-\frac{8\epsilon}{\varphi}\right)^{t'-1}\theta_1 \geq \left(1-\frac{8\epsilon}{\varphi}\right)^{t'}\cdot\beta\cdot\sqrt{\frac{10}{s_{\min}}} \text{,}
\end{equation}
and
\begin{equation}
\label{1eq:theta-t-upb}
\theta_{t'}\leq \frac{\theta_{t'+1}}{\left(1-\frac{8\epsilon}{\varphi}\right)} \leq \frac{\sqrt{\frac{10}{s_{\min}}}}{\left(1-\frac{8\epsilon}{\varphi}\right)}
\end{equation}
Putting~\eqref{1eq:theta-t-lwb} and ~\eqref{1eq:theta-t-upb} together we get 
\begin{equation}
\label{e:t0up}
\beta \leq  \left(1-\frac{8\epsilon}{\varphi}\right)^{-
t'-1}
\end{equation}
Recall that for all $t\geq 1$ we have $\theta_{t+1}=\theta_t(1-\frac{8\epsilon}{\varphi})$, thus 
\[S^{+}={S(\theta_0)} \subseteq {S(\theta_1)} \subseteq {S(\theta_2)} \subseteq \ldots \subseteq S(\theta_{t'})\subseteq T^{+}\text{.} \]
Therefore for all $0\leq t\leq t'$ we have
\begin{equation}
\label{1eq:S_i-size}
|S^{+}|\leq |S(\theta_t)|\leq |T^{+}|\leq \frac{s_{\min}}{10} \text{.}
\end{equation}
Since $|S(\theta_t)|\leq \frac{\min_{i\in\{1,\dots,k\}}|C_i|}{10}=\frac{s_{\min}}{10}$, by Lemma~\ref{lem:maintool} for all $1\leq t\leq t'$ we have 
\begin{equation}
\label{1eeq:S_i+1-size}
|S(\theta_{t+1})|\geq |S(\theta_t)| \left(1+\frac{\varphi}{2}\right) \text{.}
\end{equation}
Therefore 
\begin{align}
t'&\leq \log_{1+\frac{\varphi}{2}}\left(\frac{|T^{+}|}{|S^{+}|}\right) &&\text{By \eqref{1eeq:S_i+1-size}} \nonumber\\ 
&\leq \log_{1+\frac{\varphi}{2}}\left(\frac{s_{\min}}{10\cdot|S^{+}|}\right)  &&\text{By \eqref{1eq:S_i-size}} \nonumber \\
&\leq \log_{1+\frac{\varphi}{2}}\left(\frac{s_{\min}}{5\cdot|S^+\cup S^-|}\right) &&\text{By the assumption }|S^{+}|\geq |S^{-}| \label{e:t0lw}
\end{align}
Putting \eqref{e:t0up} and \eqref{e:t0lw} together we get
\begin{align}
\beta &\leq  \left(1-\frac{8\epsilon}{\varphi}\right)^{-
t'-1} &&\text{By \eqref{e:t0up}} \nonumber \\
&\leq \left(1-\frac{8\epsilon}{\varphi}\right)^{-1-
\log_{1+\frac{\varphi}{2}}\left(\frac{s_{\min}}{5\cdot |S^+\cup S^-|}\right)}  &&\text{By \eqref{e:t0lw}} \nonumber \\
&\leq 2\cdot \left(\frac{s_{\min}}{5\cdot |S^+\cup S^-|} \right)^{-
\log_{1+\frac{\varphi}{2}}\left( 1-\frac{8\epsilon}{\varphi} \right)}  &&\text{Since }\frac{\epsilon}{\varphi^2}\leq \frac{1}{100} \label{eq-tail-beta-bnd}
\end{align}
Note that for any $x\in \R$ we have $1+x\leq e^x$, and for any $x<0.01$ we have $1-x\geq e^{-1.2 x}$, thus given $\frac{\epsilon}{\varphi}<0.01$ we have 
\begin{equation}
{\log_{1+\frac{\varphi}{2}}\left( 1-\frac{8\epsilon}{\varphi} \right)}=\frac{\ln \left( 1-\frac{8\epsilon}{\varphi} \right)}{\ln \left(1+\frac{\varphi}{2} \right)} \geq \frac{-\frac{10\epsilon}{\varphi}}{\frac{\varphi}{2}}\geq -\frac{20\cdot\epsilon}{\varphi^2} \label{eq-power-bnd}  
\end{equation}
Putting \eqref{eq-tail-beta-bnd} and \eqref{eq-power-bnd} together we get
\[\frac{\beta}{2}\leq \left(\frac{s_{\min}}{5\cdot|S^+\cup S^-|}\right)^{(20\cdot\epsilon / \varphi^2)} \]
Therefore we have
\[ |S^+\cup S^-| \leq s_{\min}\cdot \left(\frac{\beta}{2}\right)^{-\left(\varphi^2/20\cdot\epsilon\right)}\leq n\cdot \left(\frac{\beta}{2}\right)^{-\left(\varphi^2/20\cdot\epsilon\right)} \text{.}\]

\end{proof}

As a consequence of our tail bound we can prove a bound on $\ell_\infty$-norm on any unit vector in the eigenspace spanned by the bottom $k$ eigenvectors of $L$, i.e. $U_{[k]}$.

\lemmalinf*

\begin{proof}
We define \[S=\left \{x\in V: |u(x)|\geq n^{20\epsilon /\varphi^2}\cdot \sqrt{\frac{160}{\min_{i\in k}|C_i|}}\right \}\]
Let $\beta= 4\cdot n^{20\epsilon /\varphi^2}$.
By Lemma \ref{lem:tail_bound} we have
\begin{align*}
|S|&\leq n\cdot\left( \frac{\beta}{2}\right) ^{-{\varphi^2/20\cdot \epsilon}}  \leq  n\cdot \left(2\cdot 
n^{20\epsilon /\varphi^2}\right)^{-{\varphi^2/20\cdot\epsilon}} <  1
\end{align*}
Therefore $S=\emptyset$, hence 
\[\|u\|_\infty \leq n^{20\epsilon /\varphi^2}\cdot \sqrt{\frac{160}{\min_{i\in k}|C_i|}} \text{.} \]
\end{proof}

\subsection{Centers are strongly orthogonal}
The main result of this section is Lemma \ref{lem:dosubspace} which generalizes Lemma \ref{lem:dotmu} to the orthogonal projection of cluster centers into the subspace spanned by some of the centers. To prove Lemma \ref{lem:dosubspace} we first need to prove Lemma \ref{lem:neg-spectral-close}, Lemma \ref{lem:QQ-1} and Lemma \ref{lem:QQT}.

\begin{restatable}{lemma}{lemdosubspace}\label{lem:dosubspace}
Let $k \geq 2$, $\varphi \in (0,1)$ and $\frac{\e}{\varphi^2}$ be smaller than an absolute positive constant. Let $G=(V,E)$ be a $d$-regular graph that admits $(k,\varphi,\e)$-clustering $C_1, \ldots, C_k$. Let $S \subset \{\mu_1, \dots, \mu_k\}$ denote a subset of cluster means. Let $\Pi\in \R^{k\times k}$ denote the orthogonal projection matrix onto $span(S)^{\perp}$. Then the following holds:
\begin{enumerate}
\item For all ${\mu}_i \in \{{\mu}_1, \ldots, {\mu}_k\} \setminus S$  we have $\left| \|\Pi{\mu}_i \|_2^2 - ||\mu_i||_2^2 \right| \leq \frac{16\sqrt{\e}}{\varphi}\cdot ||\mu_i||_2^2 \text{.}$ \label{lem-dosubspace:itm1}
\item For all ${\mu}_i \neq \mu_j\in \{{\mu}_1, \ldots, {\mu}_k\} \setminus S$  we have $|\langle \Pi{\mu}_i, \Pi{\mu}_j \rangle | \leq \frac{40\sqrt{\e}}{\varphi}\cdot\frac{1}{\sqrt{|C_i|\cdot|C_j|}} \text{.}$ \label{lem-dosubspace:itm2}
\end{enumerate}
\end{restatable}

Matrix $A\in \R^{n}$ is poitive definite if  $x^TAx>0$ for all $x\neq 0$, and it is positive semidefinite if $x^TAx\geq0$ for all $x\in \R^n$. We write $A\succ 0$ to indicate that $A$ is positive definite, and $A\succcurlyeq0$ to indicate that it is positive
semidefinite. We use the semidefinite ordering on matrices, writing $A\succcurlyeq B$ if and only if $A -B\succcurlyeq0$.
\begin{theorem}[\cite{toda2011operator}]
\label{thm:operator-monotonicity}
Let $A,B\in \R^{n\times n}$ be invertible, positive definite matrices. Then $A\succcurlyeq B \implies B^{-1} \succcurlyeq A^{-1}$.
\end{theorem}
\begin{proof}
By symmetry, we only need to show $A\succcurlyeq B \implies B^{-1} \succcurlyeq A^{-1}$. Since $B \succ 0$  for any $x,y\in \R^{n}$ we obtain
\begin{align*}
0 &\leq \rdp{y-B^{-1}x, B (y-B^{-1}x)} \\
&=\rdp{y,By}-\rdp{y,x}-\rdp{B^{-1}x, By}+\rdp{x,B^{-1}x}\\
&=\rdp{y,By}-2\rdp{x,y}+\rdp{x,B^{-1}x}
\end{align*}
so
\begin{equation}
\label{eq:inv-op-1}
2\rdp{x,y}-\rdp{y,By}\leq \rdp{x,B^{-1}x}
\end{equation}
Since $A \succcurlyeq B$ it follows from \eqref{eq:inv-op-1} that 
\begin{equation}
\label{eq:inv-op-2}
2\rdp{x,y}-\rdp{y,Ay}\leq 2\rdp{x,y}-\rdp{y,Ay}\leq \rdp{x,B^{-1}x}
\end{equation}
Letting $y=A^{-1}x$ in the leftmost expression of \eqref{eq:inv-op-2} we obtain
\[
\rdp{x,A^{-1}x} \leq \rdp{x,B^{-1}x}
\]
Since $x\in \R^{n}$ is is arbitrary, we get $B^{-1} \succcurlyeq A^{-1}$.
\end{proof}
\begin{restatable}{lemma}{lemnegspectralclose}\label{lem:neg-spectral-close}
Let $H,\widetilde{H}\in \R^{n\times n}$ be invertible, positive definite matrices. Let $\delta<1$. Suppose that for any vector $x\in \mathbb{R}^n$ with $\|x\|_2=1$ we have
$ (1-\delta) x^T Hx   \leq x^T\widetilde{H} x  \leq (1+\delta) x^T H x  
   \text{.}$
Then for any vector $y\in \R^n$ with $\|y\|_2=1$ we have
$ \frac1{1+\delta} y^T H^{-1}y   \leq y^T \widetilde{H}^{-1} y  \leq \frac1{1-\delta} y^T H^{-1} y  
   \text{.}$ 
\end{restatable}

\begin{proof}
Note that we have $ (1-\delta) H\preceq  \widetilde{H} \preceq (1+\delta) H$ therefore, by Theorem \ref{thm:operator-monotonicity} we have \[\frac{1}{(1-\delta)}\cdot H^{-1} \succcurlyeq \widetilde{H}^{-1} \succcurlyeq \frac{1}{(1+\delta)}\cdot H^{-1} \]
\end{proof}

\begin{lemma}
\label{lem:QQ-1}
Let $k \geq 2$, $\varphi \in (0,1)$ and $\frac{\e}{\varphi^2}$ be smaller than an absolute positive constant. Let $G=(V,E)$ be a $d$-regular graph that admits $(k,\varphi,\e)$-clustering $C_1, \ldots, C_k$.  Let $S=\{\mu_1,\ldots,\mu_k\}\setminus \{\mu_i\}$. Let $H=[\mu_1,\mu_2, \ldots, \mu_{i-1}, \mu_{i+1}, \ldots, \mu_k]$ denote a matrix such that its columns are the vectors in $S$.  Let $W\in \R^{(k-1)\times (k-1)}$ denote a diagonal matrix such that for all $j<i$ we have $W(j,j)=\sqrt{|C_j|}$ and for all $j\geq i$ we have $W(j,j)=\sqrt{|C_{j+1}|}$. Let $Z=HW$.  Then $Z^TZ$ is invertible, and for any vector $x\in \R^{k-1}$ with $||x||_2=1$ we have 
\[|x^T ((Z^T Z)^{-1}- I)x| \leq \frac{5\sqrt{\epsilon}}{\varphi}  \text{.}\] 
\end{lemma}
\begin{proof}
Let $Y\in \R^{k\times k}$ be a matrix, whose $i$-th column is equal to $\sqrt{C_i}\cdot\mu_i$. By Lemma \ref{lem:spectraldistance} item \eqref{lem-spec-dist-itm2}  for any vector $z\in \R^k$ with $||\alpha||_2=1$ we have
\[|\alpha^T(Y^T Y-I)\alpha| \leq \frac{4\sqrt{\epsilon}}{\varphi}\]
Let $x\in \R^{k-1}$ be a vector with $||x||_2=1$, and let $\alpha\in \R^k$ be a vector defined as follows: 
\[\alpha_j= \begin{cases} 
      x_j & j<i\\
     0 &  j=i\\
      x_{j+1} & j>i
   \end{cases}
\]
Thus we have $||\alpha||_2=||x||_2=1$ and $Y\alpha=Zx$. Hence, we get  
\[
| x^T (Z^T Z- I) x|= |\alpha^T(Y^T Y-I)\alpha| \leq\frac{4\sqrt{\epsilon}}{\varphi}
\]
Thus for any vector $x\in \R^{k-1}$ with $||x||_2=1$ we have \[1-\frac{4\sqrt{\epsilon}}{\varphi}\leq x^T (Z^T Z) x \leq  1+\frac{4\sqrt{\epsilon}}{\varphi}\]
Note that $Z^T Z$ is symmetric and positive semidefinit. Also note that  $Z^T Z$ is spectrally close to $I$, hence, $Z^T Z$ is  invertible.
Thus by Lemma \ref{lem:neg-spectral-close} for any vector $x\in\R^{k-1}$ we have
\[1-\frac{5\sqrt{\epsilon}}{\varphi}\leq  x^T (Z^T Z)^{-1} x \leq  1+\frac{5\sqrt{\epsilon}}{\varphi}\]
Therefore we get
\[|x^T ((Z^T Z)^{-1}- I)x| \leq  \frac{5\sqrt{\epsilon}}{\varphi}  \text{.}\] 
\end{proof}

\begin{lemma}
\label{lem:QQT}
Let  $k \geq 2$ be an integer, $\varphi \in (0,1)$, and $\e \in (0,1)$. Let $G=(V,E)$ be a $d$-regular graph that admits $(k,\varphi,\e)$-clustering $C_1, \ldots, C_k$. Let $S=\{\mu_1,\ldots,\mu_k\}\setminus \{\mu_i\}$. Let $H=[\mu_1,\mu_2, \ldots, \mu_{i-1}, \mu_{i+1}, \ldots, \mu_k]$ denote a matrix such that its columns are the vectors in $S$.  Let $W\in \R^{(k-1)\times (k-1)}$ denote a diagonal matrix such that for all $j<i$ we have $W(j,j)=\sqrt{|C_j|}$ and for all $j\geq i$ we have $W(j,j)=\sqrt{|C_{j+1}|}$. Let $Z=HW$. Then we have \[\mu_i^T Z Z^{T} \mu_i \leq \frac{8\sqrt{\epsilon}}{\varphi} \cdot ||\mu_i||_2^2 \text{.}\] 
\end{lemma}
\begin{proof}
Note that $ZZ^T=(\sum_{j=1}^k |C_j|\mu_j\mu_j^T) -|C_i|\mu_i \mu_i^T$. Thus we have 
\begin{equation}
\label{eq:1qqt-val}
\mu_i^T Z Z^{T} \mu_i = \mu_i^T \left(\sum_{j=1}^k |C_j|\mu_j\mu_j^T \right)\mu_i- |C_i|\cdot||\mu_i||_2^4\text{.}
\end{equation}
By Lemma \ref{lem:spectraldistance} item \eqref{lem-spec-dist-itm1} for any vector $x$ with $||x||_2=1$ we have 
\[  x^T \left(\sum_{j=1}^k |C_j|\mu_j\mu_j^T - I\right)x \leq \frac{4\sqrt{\epsilon}}{\varphi}\]
Hence we can write 
\begin{align*}
\mu_i^T \left(\sum_{j=1}^k |C_j|\mu_j\mu_j^T \right) \mu_i &= \mu_i^T \left(\sum_{j=1}^k |C_j|\mu_j\mu_j^T -I \right)\mu_i+ \mu_i^T  \mu_i 
\leq    \left(1+ \frac{4\sqrt{\epsilon}}{\varphi} \right) ||\mu_i||_2^2  
\end{align*}
Therefore by \eqref{eq:1qqt-val} we get
\begin{align*}
\mu_i^T Z Z^{T} \mu_i &= \mu_i^T  \left(\sum_{j=1}^k |C_j|\mu_j\mu_j^T \right)\mu_i- |C_i|\cdot||\mu_i||_2^4 \\
&\leq \left(1+ \frac{4\sqrt{\epsilon}}{\varphi} -|C_i|\cdot||\mu_i||_2^2 \right) ||\mu_i||_2^2 
\end{align*}
By Lemma \ref{lem:dotmu} we have  $|C_i|\cdot||\mu_i||_2^2 \geq \left(1- \frac{4\sqrt{\epsilon}}{\varphi} \right)$. Thus we get
\begin{align*}
\mu_i^T Z Z^{T} \mu_i 
&\leq \left(1+ \frac{4\sqrt{\epsilon}}{\varphi} -|C_i|\cdot||\mu_i||_2^2 \right) ||\mu_i||_2^2 \\
&\leq \left(1+ \frac{4\sqrt{\epsilon}}{\varphi} -1+\frac{4\sqrt{\epsilon}}{\varphi} \right) ||\mu_i||_2^2 \\
&\leq \frac{8\sqrt{\epsilon}}{\varphi} \cdot ||\mu_i||_2^2
\end{align*}
\end{proof}

Now we prove the main result of the subsection (Lemma \ref{lem:dosubspace}).
\lemdosubspace*

\begin{proof}
\textbf{Proof of item \eqref{lem-dosubspace:itm1}:} Since $\Pi$ is a orthogonal projection matrix we have $||\Pi||_2=1$. Hence, we have $||\Pi\mu_i||_2^2\leq ||\mu_i||_2^2 \leq \left(1+\frac{16\sqrt{\epsilon}}{\varphi}\right)||\mu_i||_2^2$.
Thus it's left to prove $||\Pi\mu_i||_2^2 \geq \left(1-\frac{16\sqrt{\epsilon}}{\varphi}\right)||\mu_i||_2^2 $. Note that by Pythagoras' theorem $||\Pi\mu_i||^2_2 = ||\mu_i||^2_2 - ||(I-\Pi)\mu_i||^2_2$. We will prove $||(I-\Pi)\mu_i||^2_2 \leq \frac{16\sqrt{\epsilon}}{\varphi}||\mu_i||^2_2$ which implies \[||\Pi\mu_i||^2_2 \geq \left(1 - 16\frac{\sqrt{\epsilon}}{\varphi} \right) ||\mu_i||_2^2\text{.}\]

 Let $S'=\{\mu_1,\ldots,\mu_k\}\setminus \{\mu_i\}$. Let $\Pi'$ denote the orthogonal projection matrix onto $span(S')^{\perp}$. Note that $S\subseteq S'$, hence ${span}(S)$ is a subspace of ${span}(S')$, therefore we have $||(I-\Pi)\mu_i||^2_2\leq ||(I-\Pi')\mu_i||^2_2$. Thus it suffices to prove $||(I-\Pi')\mu_i||^2_2\leq \frac{16\sqrt{\epsilon}}{\varphi}||\mu_i||^2_2$. Let $H=[\mu_1,\mu_2, \ldots, \mu_{i-1}, \mu_{i+1}, \ldots, \mu_k]$ denote a matrix such that its columns are the vectors in $S'$.  Let $W\in \R^{(k-1)\times (k-1)}$ denote a diagonal matrix such that for all $j<i$ we have $W(j,j)=\sqrt{|C_j|}$ and for all $j\geq i$ we have $W(j,j)=\sqrt{|C_{j+1}|}$. Let $Z=HW$. The orthogonal projection matrix onto the span of $S'$ is defined as $(I-\Pi')=Z(Z^T Z)^{-1} Z^{T}$, and using Lemma \ref{lem:QQ-1} we get
\begin{align*}
||(I-\Pi')\mu_i||^2_2&= \mu_i^T Z(Z^T Z)^{-1} Z^{T} \mu_i \\
&= \mu_i^T Z((Z^T Z)^{-1}- I) Z^{T} \mu_i + \mu_i^T Z Z^{T} \mu_i 
\end{align*}
By Lemma \ref{lem:QQ-1} $(Z^T Z)^{-1}$ is spectrally close to $I$, therefore we have 
\[\left|\mu_i^T Z \left((Z^T Z)^{-1}-I \right)Z^{T} \mu_i  \right|\leq \frac{5\sqrt{\epsilon}}{\varphi}||Z^{T} \mu_i ||_2^2\]
Thus we get 
\[||(I-\Pi')\mu_i||^2_2 \leq \left(\frac{5\sqrt{\epsilon}}{\varphi}+1 \right)||Z^{T} \mu_i ||_2^2 \leq 2 ||Z^{T} \mu_i ||_2^2 \]
By Lemma \ref{lem:QQT} we have 
\[||Z^{T} \mu_i ||_2^2= \mu_i^T Z Z^{T} \mu_i \leq \frac{8\sqrt{\epsilon}}{\varphi} \cdot ||\mu_i||_2^2\]
Therefore  we get
\begin{equation}\label{eq:Pdotbnd-4}
||(I-\Pi)\mu_i||^2_2 \leq ||(I-\Pi')\mu_i||^2_2 \leq 2||Z^{T} \mu_i ||_2^2\leq \frac{16\sqrt{\epsilon}}{\varphi} ||\mu_i||_2^2 
\end{equation} 
Hence,
\[||\Pi\mu_i||^2_2 \geq \left(1 - 16\frac{\sqrt{\epsilon}}{\varphi}\right)  ||\mu_i||_2^2\text{.}\]
\textbf{Proof of item \eqref{lem-dosubspace:itm2}:} 
Note that
\begin{align*}
\langle \mu_i, \mu_j \rangle &= \langle (I-\Pi)\mu_i+\Pi\mu_i ,  (I-\Pi)\mu_j+\Pi\mu_j\rangle = \langle (I-\Pi)\mu_i ,  (I-\Pi)\mu_j\rangle + \langle \Pi\mu_i ,  \Pi\mu_j\rangle 
\end{align*}
Thus by triangle inequality we have
\begin{equation}\label{eq:prject-dot-i-mu1}
 |\langle \Pi\mu_i ,  \Pi\mu_j\rangle|  \leq |\langle \mu_i, \mu_j \rangle |+ |\langle (I-\Pi)\mu_i ,  (I-\Pi)\mu_j\rangle|
\end{equation}
By Cauchy Schwarz we have
\begin{align}
|\langle (I-\Pi)\mu_i ,  (I-\Pi)\mu_j\rangle| &\leq || (I-\Pi)\mu_i||_2|| (I-\Pi)\mu_i||_2 \nonumber\\
&\leq \frac{16\sqrt{\epsilon}}{\varphi}\cdot ||\mu_i||_2||\mu_j||_2 && \text{By \eqref{eq:Pdotbnd-4}}\nonumber\\ 
&\leq \frac{32\sqrt{\epsilon}}{\varphi}\cdot\frac{1}{\sqrt{|C_i||C_j|}} && \text{By Lemma \ref{lem:dotmu} for small enough }\frac{\epsilon}{\varphi^2} \label{eq:prject-dot-i-mu2}
\end{align}
Also by Lemma \ref{lem:dotmu} we have
\begin{equation}\label{eq:prject-dot-i-mu3}
|\langle \mu_i, \mu_j \rangle | \leq  \frac{8\sqrt{\epsilon}}{\varphi}\cdot\frac{1}{\sqrt{|C_i||C_j|}}
\end{equation}
Therefore by \eqref{eq:prject-dot-i-mu1}, \eqref{eq:prject-dot-i-mu2} and \eqref{eq:prject-dot-i-mu3} we get 
\[|\langle \Pi\mu_i ,  \Pi\mu_j\rangle|  \leq |\langle \mu_i, \mu_j \rangle |+ |\langle (I-\Pi)\mu_i ,  (I-\Pi)\mu_j\rangle| \leq \frac{40\sqrt{\epsilon}}{\varphi}\cdot \frac{1}{\sqrt{|C_i||C_j|}} \text{.}\]
\end{proof}

\subsection{Robustness property of $(k,\varphi, \e)$-clusterable graphs}

In this subsection we show a Lemma that establishes a robustness property of $(k,\varphi,\e)$-clusterable graphs. That is we show that any collection $\{S_1, S_2, \dots, S_k\}$ of pairwise disjoint subsets of vertices must match clusters $\{C_1, \dots, C_k\}$ well.

\begin{lemma}\label{lem:howtocluster}
Let $G=(V,E)$ be a $d$-regular graph that admits a $(k,\varphi,\epsilon)$-clustering $C_1,\dots,C_k$. Let $k \geq 2$, $\varphi \in (0,1)$ and $\frac{\e}{\varphi^3}$ be smaller than an absolute positive constant. If $S_1, S_2, \dots, S_k \subseteq V$ are $k$ disjoint sets such that for all $i \in [k]$
$$\phi(S_i) \leq O \left(\frac{\e}{\varphi^2}  \cdot \log(k)\right) $$
then there exists a permutation $\pi$ on $k$ elements so that for all $i \in [k]$:
$$|C_{\pi(i)} \triangle S_i| \leq O \left(\frac{\e}{\varphi^3}  \cdot \log(k) \right) |C_{\pi(i)}|$$
\end{lemma}

\begin{proof}
Fix $i \in [k]$ and let $J_i = \{j : |S_i \cap C_j | \leq |C_j|/2\}$. Then observe that because the inner conductance of every $C_i$ is at least $\varphi$ we get:
\begin{equation}\label{eq:smallintersectionsmall}
\varphi\sum_{j \in J_i} |S_i \cap C_j| \leq O \left(\frac{\e}{\varphi^2}  \cdot \log(k)\right) |S_i|
\end{equation}
Using \eqref{eq:smallintersectionsmall} and the assumption $\frac{\e}{\varphi^3}$ is sufficiently small we get that
\begin{equation}\label{eq:outsidecpiismall}
\sum_{j \in J_i} |S_i \cap C_j| \leq O \left(\frac{\e}{\varphi^3}  \cdot \log(k) \right) |S_i| < |S_i|
\end{equation}
\eqref{eq:outsidecpiismall} and $\sum_{j \in [k]} |S_i \cap C_j| = |S_i|$ gives us that 
\begin{equation}\label{eq:setsnonempty}
\text{For all } i \in [k], J_i \neq [k]
\end{equation}
We will show that for each $i$: $|[k] \setminus J_i| = 1$ and that a function $i \mapsto \pi(i) \in [k] \setminus J_i$ (that is $\pi(i)$ is the only element of $[k] \setminus J_i$) is a permutation and that it satisfies the claim of the Lemma.



Assume that there exist $i_1 \neq i_2 \in [k]$ and $j \in ([k] \setminus J_{i_1}) \cap ([k] \setminus J_{i_2})$. By definition of $J_i$'s we get that $|S_{i_1} \cap C_j|, |S_{i_2} \cap C_j| > |C_j|/2$ but $S_i$'s are disjoint so it's impossible that two of them intersect more than half of the same $C_j$. That means that sets $([k] \setminus J_i)$ are pairwise disjoint for all $i$'s. But we also know from \eqref{eq:setsnonempty} that for all $i$ $([k] \setminus J_i) \neq \emptyset$. So we have $k$ nonempty, pairwise disjoint subsets of $[k]$, which means that every set contains one element and all elements are different. That in turn means that we can define $\pi$ as a function $i \mapsto \pi(i) \in [k] \setminus J_i$ and $\pi$ is a permutation.

Now we show that $\pi$ satisfies the claim of the Lemma. Observe that because for all $i \in [k]$ the set $[k] \setminus J_i$ contains only one element we get for all $i \in [k]$.
\begin{equation}\label{eq:siminuscpi}
    \sum_{j \in J_i} |S_i \cap C_j| = |S_i \setminus C_{\pi(i)}|
\end{equation}
Note that because of \eqref{eq:smallintersectionsmall} and \eqref{eq:siminuscpi} for all $i \in [k]$:
\begin{equation}\label{eq:outsidecibound}
|S_i \setminus C_{\pi(i)}| \leq O \left(\frac{\e}{\varphi^3}  \cdot \log k\right)|S_i|.
\end{equation}
Moreover because inner conductance of every $C_i$ is at least $\varphi$ and $|C_{\pi(i)} \setminus S_i| < |C_{\pi(i)}|/2$ we get that for all $i \in [k]$
\begin{equation}\label{eq:outsidesibound}
\varphi \cdot |C_{\pi(i)} \setminus S_i| \leq O \left(\frac{\e}{\varphi^2}  \cdot \log(k) \right)  |S_i|   
\end{equation}
Finally combining \eqref{eq:outsidecibound} and \eqref{eq:outsidesibound} we get that:
$$|C_{\pi(i)} \triangle S_i| \leq  O \left(\frac{\e}{\varphi^3}  \cdot \log(k) \right) |C_{\pi(i)}|$$

\end{proof}

\newpage

\section{A spectral dot product oracle}
\label{sec:dot}
Our goal in this section is to develop what we call a \emph{spectral dot product oracle}. The oracle is a sublinear time and space data structure that 
has oracle access to a $(k, \varphi, \epsilon)$-clusterable graph $G$ and after a preprocessing step can answer dot products queries for the spectral embedding.
Specifically, if $L = U \Lambda U^T$ is the normalized Laplacian of $G$ and the $x$-th column of $F=U_{[k]}^T$ is called $f_x$ for $x\in V$ then our oracle gets as input 
two vertices $x,y$ and returns an approximation of $\langle f_x, f_y\rangle$. Both the preprocessing time and the time to evaluate an oracle 
query are $k^{O(1)} \cdot n^{1/2 + O(\epsilon/\varphi^2)}\cdot (\log n)^{O(1)}$, that is, sublinear in $n$ for $\epsilon \ll \varphi^2$.
We now state the main theorem that we prove in this section. The algorithms mentioned in Theorem \ref{thm:dot} can be found later in this section.

\thmdot*

\subsection{The spectral dot product oracle - overview}

In the following sections we provide the proof of the spectral dot product oracle. Recall from the technical overview that 
we are using the following algorithms (we restate them for convenience of the reader). Our main tool for accessing the spectral embedding of the graph is a primitive that runs a few short (logarithmic length) random walks from a given vertex. \setcounter{algorithm}{0}
\begin{algorithm}[H]
\caption{\textsc{RunRandomWalks}($G,R,t,x$)}
\label{alg:random-walk}
\begin{algorithmic}[1]
	\State Run $R$ random walks of length $t$ starting from $x$
	\State Let $\m_x(y)$ be the fraction of random walks that ends at $y$ \Comment vector $\m_x$ has support at most $R$ \label{ln:defm}
	\State \Return $\m_x$
\end{algorithmic}
\end{algorithm}
Another key primitive uses collision statistics to estimate the Gram matrix of random walk distributions started at vertices in a set $S$.  
\begin{algorithm}[H]
\caption{\textsc{EstimateCollisionProbabilities($G,I_S,R,t$)}  }\label{alg:gram}
\begin{algorithmic}[1]
	\For{$i=1$ to $O(\log n)$}
		\State  ${\Q_i}:=\textsc{EstimateTransitionMatrix}(G,I_S,R,t)$ \label{ln:Qi}
		\State  ${\widehat{P}_i}:=\textsc{EstimateTransitionMatrix}(G,I_S,R,t)$ \label{ln:Pi}
		\State $\G_i:=\frac{1}{2}\left(\widehat{P}_i^T \Q_i+\Q_i^T \widehat{P}_i\right)$  \label{ln:Gi}\Comment{$\G_i$ is symmetric}
	\EndFor
	\State Let $\G$ be a matrix obtained by taking the entrywise median of $\mathcal{G}_i$'s \Comment{$\G$ is symmetric} \label{ln:GG}
	\State \Return $\G$
\end{algorithmic}
\end{algorithm}
We also need the following procedure.
\begin{algorithm}[H]
\caption{\textsc{EstimateTransitionMatrix$(G,I_S,R,t)$}}\label{alg:compQ}
\begin{algorithmic}[1]
	\For{each sample $x\in I_S$}
		\State $\m_x:=\textsc{RunRandomWalks}(G,R,t,x)$  \label{ln:ranwalk22}
	\EndFor
	\State Let ${\Q}$ be the matrix whose columns are $\m_x$  for $x\in I_S$ \label{ln:defQ}
	\State \Return $\Q$ \Comment $\Q$ has at most $Rs$ non-zeros
\end{algorithmic}
\end{algorithm}

Then we can initialize the dot product oracle.
\begin{algorithm}[H]
\caption{\textsc{InitializeOracle($G,\delta,\xi$)}  \Comment Need: $\epsilon/\varphi^2 \leq \frac{1}{10^5}$}\label{alg:LearnEmbedding}
\begin{algorithmic}[1]
	\State $t:= \frac{20\cdot \log n}{\varphi^2}$	\label{ln:set-t}			
	\State $R_{\text{init}}:=O{(n^{1-\delta +980 \cdot\epsilon / \varphi^2}  \cdot k^{17}/{\xi}^{2})}$	\label{ln:setRR}	
	
	\State $s:= O(n^{480\cdot\epsilon / \varphi^2}\cdot \log n \cdot k^{8}/{\xi}^2)$  \label{ln:sets}
	
	\State Let $I_S$ be the multiset of $s$ indices chosen independently and uniformly at random from $\{1,\ldots,n\}$ \label{ln:sample}	
	
	\For{$i=1$ to $O(\log n)$}
		\State  ${\Q_i}:=\textsc{EstimateTransitionMatrix}(G,I_S,R_{\text{init}},t)$  \label{ln:setQi} \Comment $\Q_i$ has at most $R_{\text{init}}\cdot s$ non-zeros
	\EndFor	
	\State $\G:=$\textsc{EstimateCollisionProbabilities}$(G,I_S,R_{\text{init}},t)$
\label{ln:setG}	
	
	\State Let $\frac{n}{s}\cdot\G:=\widehat{W}\widehat{\Sigma} \widehat{W}^T$ be the eigendecomposition of $\frac{n}{s}\cdot\G$ \label{ln:QSVD} \Comment $\G\in \R^{s\times s}$
	\If{$\widehat{\Sigma}^{-1}$ exists}
	
	\State $\Psi:=\frac{n}{s}\cdot\widehat{W}_{[k]}\widehat{\Sigma}_{[k]}^{-2} \widehat{W}_{[k]}^T$ \Comment $\Psi\in\R^{s\times s}$ \label{ln:final-init} 
	\State \Return  $\mathcal{D}:=\{\Psi,\Q_1,\ldots ,\Q_{O(\log n)}\}$  \label{ln:retPiQ}
		\EndIf
\end{algorithmic}
\end{algorithm}
Finally, we have the query algorithm.
\begin{algorithm}[H]
\caption{\textsc{SpectralDotProductOracle}($G,x,y, \delta, \xi, \mathcal{D}$)  \Comment Need: $\epsilon/\varphi^2 \leq \frac{1}{10^5}$
\newline \text{ }\Comment $\mathcal{D}:=\{\Psi,\Q_1,\ldots ,\Q_{O(\log n)}\}$}
\label{alg:dotProduct}
\begin{algorithmic}[1]
		\State $R_{\text{query}}:=O{(n^{\delta +500\cdot\epsilon / \varphi^2}\cdot  k^{9}/{\xi}^{2})}$ \label{ln:set_r_q}
				
	\For{$i=1$ to $O(\log n)$}
		\State ${\m^i_x:=\textsc{RunRandomWalks}(G,R_{\text{query}},t,x)}$ \label{ln:mx}
		\State ${\m^i_y:=\textsc{RunRandomWalks}(G,R_{\text{query}},t,y)}$ \label{ln:my}
	\EndFor
	\State Let ${\alpha}_x$ be a vector obtained by taking the entrywise median of $(\Q_i)^T(\m^i_x)$ over all runs \label{ln:alx}
	\State Let ${\alpha}_y$ be a vector  obtained by taking the entrywise median of $(\Q_i)^T(\m^i_y)$ over all runs \label{ln:aly}
	\State  \Return $\adp{f_{x},f_y}:={\alpha}_x^T \Psi {\alpha}_y$  \label{ln:axay}
\end{algorithmic}
\end{algorithm}

Let $I_S=\{i_1,\ldots, i_s\}$ be a multiset of $s$ indices chosen independently and uniformly at random from $\{1,\dots,n\}$. Let $S$ be the $n\times s$ matrix whose $j$-th column equals $\mathds{1}_{i_j}$. As already explained in detail in the technical overview, we first prove stability bounds for the pseudoinverse.  Then we show that that $M^t$ is approximated by $M^tS$ and finally we show
that algorithm {\sc RunRandomWalks} approximates the $M^t\mathds{1}_x$ sufficiently well.  We conclude with the proof of Theorem \ref{thm:dot}.

\subsection{Stability bounds for the low rank approximation}
\label{subsubsec:topk}

The main result of this section is a bound on the stability of the pseudoinverse of the rank-$k$ approximation of two symmetric, positive semi-definite matrices $A, \widetilde A \in \R^{n \times n}$
that are spectrally close and that have an eigenvalue gap between the $k$-th and $(k+1)$-st eigenvalue.
In order to prove this result, we use Weyl's inequality, which gives bounds on the eigenvalues of the sum of a matrix $A$ and a perturbation matrix $P$. Recall that for a symmetric matrix $A$, we write $\nu_i(A)$ (resp. $\nu_{\max}(A), \nu_{\min}(A))$ to denote the $i^{\text{th}}$ largest (resp. maximum, minimum) eigenvalue of $A$. 
\begin{lemma}[Weyl's Inequality]\label{lem_Weyl}
Let $A,P \in \R^{n\times n}$ be two symmetric matrices. Then we have for all $i \in \{1,\dots,n\}$:
$$
\nu_i(A)+\nu_{\min}(P)\leq\nu_i(A+P)\leq\nu_i(A)+\nu_{\max}(P),
$$
where for a symmetric matrix $H \in \R^{n \times n}$ $\nu_i(H)$ denotes its $i^\text{th}$ largest eigenvalue and $\nu_{\min}(H)$ and $\nu_{\max}(H)$ refer to the smallest and largest eigenvalues of $H$.
\end{lemma}
We will use the Davis-Kahan $\sin(\theta)$ Theorem \cite{davis1970rotation} (the version given in the note \cite{daviskahan}).
\begin{theorem} [Davis-Kahan $sin(\theta)$-Theorem \cite{davis1970rotation}].
\label{lem:davis-kahan}
Let $H = E_0A_0E_0^T +E_1A_1E_1^T$ and $\widetilde{H} = F_0\Lambda_0 F_0^T + F_1\Lambda_1 F_1^T$ be symmetric real-valued matrices with $E_0, E_1$ and $F_0, F_1$ orthogonal. If the eigenvalues of $A_0$ are contained in an interval $(a,b)$,
and the eigenvalues of $\Lambda_1 $are excluded from the interval $(a-\eta, b+\eta)$for some $\eta > 0$, then for any unitarily invariant norm $\|.\|$
\[\|F_1^T E_0\| \leq \frac{\|F_1^T (\widetilde{H}-H) E_0\|}{\eta} \text{.}\]
\end{theorem}
Let $m\leq n$ be integers.  For any matrix $A\in \R^{n\times m}$ with singular value decomposition (SVD) $A=Y\Gamma Z^T$ we assume $Y\in \R^{n\times n}$, $\Gamma \in \R^{n\times n}$ is a diagonal matrix of singular values and $Z\in \R^{m\times n}$ (this is a slightly non-standard definition of the SVD, but having $\Gamma$ be a square matrix will be convenient).  $Y$ has orthonormal columns, the first $m$ columns of $Z$ are orthonormal, and the rest of the columns of $Z$ are zero. For any integer $q\in[m]$ we denote $Y_{[q]}\in \R^{n \times q}$ as the first $q$ columns of $Y$ and $Y_{-[q]}$ to denote the matrix of the remaining columns of $Y$.  We also denote by $Z_{[q]}\in \R^{m \times q}$ as the first $q$ columns of $Z$ and $Z_{-[q]}$ to denote the matrix of the remaining $n-q$ columns of $Z$. Finally we denote by $\Gamma_{[q]}\in \R^{q \times q}$ the submatrix of $\Gamma$ corresponding to the first $q$ rows and columns of $\Gamma$ and we use $\Gamma_{-[q]}$ to denote the submatrix corresponding to the last $n-q$ rows and $n-q$ columns of $\Gamma$. So for any $q\in[m]$ the span of $Y_{-[q]}$ is the orthogonal complement of the span of $Y_{[q]}$ in $\R^n$, also the span of the columns of $Z_{-[q]}$ is the orthogonal complement of the span of $Z_{[q]}$ in $\R^m$. Thus we can write $A=Y_{[q]}\Gamma_{[q]}Z^T_{[q]} + Y_{-[q]}\Gamma_{-[q]}Z^T_{-[q]}$. 
 
 \begin{claim}\label{cl:ptildep}
For every symmetric matrix $E$ and every pair of orthogonal projection matrices $P, \wt{P}$ one has
\begin{equation*}
\begin{split}
||P \cdot E\cdot P-\wt{P} \cdot E\cdot \wt{P}||_2\leq 2\| E \|_2\cdot (\|P\cdot (I-\wt{P})\|_2+\|\wt{P}\cdot (I-P)\|_2).
\end{split}
\end{equation*}
 \end{claim}
 \begin{proof}
Since $\wt{P}+(I-\wt{P})=I$ we can write
\begin{align}\label{eq:PEP}
P \cdot E\cdot P&=(\wt{P}+(I-\wt{P}))P \cdot E\cdot P\cdot (\wt{P}+(I-\wt{P}))\nonumber\\
  &= P \cdot E\cdot P\cdot (I-\wt{P})+ \wt{P}\cdot P \cdot E\cdot P\cdot \wt{P} +(I-\wt{P}) \cdot P \cdot E\cdot P\cdot \wt{P}
\end{align}
Since $P+(I-P)=I$ we have
\begin{align}
\label{eq:PEP-tild}
\wt{P} \cdot E\cdot \wt{P}&=\wt{P} (P+(I-P))  \cdot E\cdot \left(P+(I-P)\right)  \wt{P} ||_2 \nonumber\\
&= \wt{P} \cdot E\cdot (I-P)\wt{P} + \wt{P} \cdot P \cdot E\cdot P\cdot \wt{P}+\wt{P} \cdot (I-P) \cdot E\cdot P\cdot \wt{P}
\end{align}
Putting \eqref{eq:PEP} and \eqref{eq:PEP-tild} together and by triangle inequality we get
 \begin{align*}
&||P \cdot E\cdot P-\wt{P} \cdot E\cdot \wt{P}||_2\\
&\leq \|P \cdot E\cdot P\cdot (I-\wt{P})\|_2+\| (I-\wt{P}) \cdot P \cdot E\cdot P\cdot \wt{P}\|_2+\| \wt{P} \cdot E\cdot (I-P)\wt{P}\|_2+\|\wt{P} \cdot (I-P) \cdot E\cdot P\cdot \wt{P}\|_2
 \end{align*}
Thus by submultiplicativity of the operator norm we get
\begin{align*}
&||P \cdot E\cdot P-\wt{P} \cdot E\cdot \wt{P}||_2\\
&\leq  \|P \|_2 \|  E \|_2 \| P\cdot (I-\wt{P})\|_2 + \| (I-\wt{P}) \cdot P \|_2 \| E\|_2 \| P \|_2\| \wt{P}\|_2 
+ \| \wt{P} \|_2 \| E \|_2 \| (I-P)\wt{P}\|_2 + \| \wt{P} \|_2 \| E \|_2 \| (I-P)\wt{P}\|_2\\
&\leq \| E\|_2 \left(\| P\cdot (I-\wt{P})\|_2 + \|  (I-\wt{P})\cdot P\|_2 + \| (I-P)\wt{P}\|_2 + \| \wt{P}(I-P)\|_2 \right) \text{~Since $||P||=||\wt{P}||_2=1$}\\
&=2\cdot\| E\|_2\cdot (\|P\cdot (I-\wt{P})\|_2+\|\wt{P}\cdot (I-P)\|_2),
\end{align*}
where the last equality holds since $\|P\cdot (I-\wt{P})\|_2=\| (I-\wt{P})^T \cdot P^T\|_2=\| (I-\wt{P}) \cdot P\|_2$ and similarly since $\|\wt{P}\cdot (I-P)\|_2=\| (I-P)^T\cdot \wt{P}^T\|_2=\| (I-P)\cdot \wt{P}\|_2$.
 \end{proof}

Recall that for matrices $A,\widetilde{A}\in \R^{n\times n}$, we write $A\preccurlyeq \widetilde{A}$, if $\forall x \in \R^n$ we have $x^TAx \leq x^T \widetilde{A} x$ and we write $A\prec \widetilde{A}$, if $\forall x \in \R^n$ we have $x^TAx < x^T \widetilde{A} x$. 
Now we can state the main technical result of this section (Lemma \ref{lem:neg-l2-close}), whose proof 
relies on matrix perturbation bounds Davis-Kahan $\sin \theta$ theorem (Theorem \ref{lem:davis-kahan}).
\if 0\begin{restatable}{lemma}{lemnegclose}
\label{lem:neg-l2-close}
Let $A,\widetilde{A}\in \R^{n\times n}$ be symmetric matrices with eigendecompositions $A= Y \Gamma Y^T$ and $\widetilde{A} = \widetilde{Y} \widetilde{\Gamma} \widetilde{Y}^T$. Let the eigenvalues of $A$ be
$\gamma_1 \geq \dots \geq \gamma_n \geq 0$. Suppose that 
$\|A-\widetilde{A}\|_2\leq \frac{\gamma_k}{100}$ and $\gamma_{k+1} < \gamma_k/4$. Then we have
\[ \| Y_{[k]} \Gamma^{-1}_{[k]} Y^T_{[k]}- \widetilde{Y}_{[k]} \widetilde{\Gamma}^{-1}_{[k]} \widetilde{Y}^T_{[k]}\|_2 \leq \frac{12\|A-\widetilde{A}\|_2+ 4\gamma_{k+1}}{\gamma_k^2}
   \text{.}\] 
\end{restatable}
\fi

We will need the following claim, whose proof is inspired by the proof of the operator monotonicity of negative matrix inverse~\cite{toda2011operator}:
\begin{claim}\label{cl:384g8g834gDD}
Let $A, B\in \R^{n\times n}$ be symmetric positive semidefinite matrices. Let $\Pi_B$ denote orthogonal projection operator onto the range space of $B$. Then if $A\succeq B$, we have for every orthogonal projection $\Pi_A$ satisfying $\Pi_AA^+=A^+\Pi_A$ that
$$
(\Pi_AA\Pi_A)^+\preceq B^++ 2\|\Pi_AA^+\|_2 \|\Pi_A(I-\Pi_B)\|_2\cdot I.
$$

\end{claim}
\begin{proof}
For every $x\in \R^n$, and every $y\in \R^n$ (to be chosen as $y=A^+x$ later) since $B$ is positive semidefinite we have
\begin{equation*}
(y-B^+x)^TB(y-B^+x)\geq 0,
\end{equation*}
which in particular implies that 
\begin{equation*}
y^TBy-2x^TB^+B y+x^TB^+x\geq 0,
\end{equation*}
and since $A\succeq B$ by assumption, 
\begin{equation*}
y^TAy-2x^TB^+B y+x^TB^+x\geq 0.
\end{equation*}
We now chose $y=\Pi_A A^+x$ and rearrange, getting
\begin{equation}\label{eq:9034ghh294g0j}
2x^TB^+B\Pi_A A^+x-x^T\Pi_AA^+\Pi_Ax\leq x^TB^+x.
\end{equation}
Noting that $B^+B=\Pi_B$ and $\Pi_A\Pi_AA^+=\Pi_AA^+\Pi_A$, we write the lhs of~\eqref{eq:9034ghh294g0j} as
\begin{equation*}
\begin{split}
2x^T\Pi_B \Pi_AA^+x-x^T\Pi_AA^+\Pi_Ax&=2x^T\Pi_AA^+\Pi_Ax+2x^T(\Pi_B \Pi_A-\Pi_A) \Pi_AA^+x-x^T\Pi_AA^+\Pi_Ax\\
&=x^T\Pi_AA^+\Pi_Ax+2x^T((\Pi_B-I) \Pi_A) \Pi_AA^+x.
\end{split}
\end{equation*}
Substituting the above into~\eqref{eq:9034ghh294g0j}, and noting that 
$$
|x^T(\Pi_B \Pi_A-\Pi_A) \Pi_AA^+x|\leq \|\Pi_AA^+\|_2\cdot \|(\Pi_B-I)\Pi_A\|_2\cdot x^Tx,
$$
we get
\begin{equation*}
\begin{split}
x^T\Pi_A A^+ \Pi_Ax\leq x^TB^+x+2\|\Pi_AA^+\|_2\cdot \|(\Pi_B-I)\Pi_A\|_2\cdot x^Tx.
\end{split}
\end{equation*}
The above holds for all $x\in \R^n$. Also, $\|(\Pi_B-I)\Pi_A\|_2=\|\Pi_A(I-\Pi_B)\|_2$, since $\Pi_A, \Pi_B$ are projection
matrices. Therefore, for all $x\in \R^n$ we have
\begin{equation*}
\begin{split}
\Pi_A A^+ \Pi_A\preceq B^++2\|\Pi_AA^+\|_2\cdot \|\Pi_A(I-\Pi_B)\|_2\cdot I,
\end{split}
\end{equation*}
as required.
\end{proof}

\begin{restatable}{lemma}{lemnegclose}
\label{lem:neg-l2-close}
Let $A,\widetilde{A}\in \R^{n\times n}$ be symmetric matrices with eigendecompositions $A= Y \Gamma Y^T$ and $\widetilde{A} = \widetilde{Y} \widetilde{\Gamma} \widetilde{Y}^T$. Let the eigenvalues of $A$ be
$1\geq \gamma_1 \geq \dots \geq \gamma_n \geq 0$. Suppose that 
$\|A-\widetilde{A}\|_2\leq \frac{\gamma_k}{100}$ and $\gamma_{k+1} < \gamma_k/4$. Then we have
\[ \| Y_{[k]} \Gamma^{-1}_{[k]} Y^T_{[k]}- \widetilde{Y}_{[k]} \widetilde{\Gamma}^{-1}_{[k]} \widetilde{Y}^T_{[k]}\|_2 \leq  \frac{16 \|A-\wt{A}\|_2+4\gamma_{k+1}}{\gamma_k^2}\text{.}\] 
\end{restatable}

\begin{proof}
We define $P=Y_{[k]} Y^T_{[k]}$ and $\wt{P}=\wt{Y}_{[k]} \wt{Y}^T_{[k]}$, and let $M=PAP=Y_{[k]} \Gamma_{[k]} Y^T_{[k]}$ and 
$\wt{M}=\wt{P}\wt{A}\wt{P}=\wt{Y}_{[k]} \wt{\Gamma}_{[k]} \wt{Y}^T_{[k]}$.
 First note that
 \begin{align}
 \label{eq:823gt8g8GDSas}
\wt{M}&=\wt{P} \wt{A}\wt{P} \nonumber\\
&\preceq \wt{P} (\wt{A}+\|A-\wt{A}\|_2 \cdot I)\wt{P} \nonumber\\
&\preceq \wt{P} (\wt{A}+\|A-\wt{A}\|_2 \cdot I)\wt{P}+(I-\wt{P})(\wt{A}+\|A-\wt{A}\|_2 \cdot I)(I-\wt{P}) \nonumber\\
&= \wt{A}+\|A-\wt{A}\|_2\cdot I \nonumber\\
&\preceq A+2\|A-\wt{A}\|_2\cdot I \nonumber\\
&= P (A+2\|A-\wt{A}\|_2\cdot I)P+(I-P) (A+2\|A-\wt{A}\|_2\cdot I) (I-P) \nonumber\\
&\preceq M+ (2\|A-\wt{A}\|_2+\gamma_{k+1})I \nonumber\\
&=M+ \eta\cdot I,
 \end{align}
where we let $\eta=2\|A-\wt{A}\|_2+\gamma_{k+1}$. The transition from line~2 to line~3 is due to the fact that $\wt{A}+\|A-\wt{A}\|_2 \cdot I\succeq A\succeq 0$, and therefore $(I-\wt{P}) (\wt{A}+\|A-\wt{A}\|_2 \cdot I)(I-\wt{P})\succeq 0$. The transition from line~4 to line~5 is due to $\wt{A}\preceq A+\|A-\wt{A}\|_2 \cdot I$. The transition from line~6 to line~7 is due to the fact that $(I-P) A (I-P)\preceq \gamma_{k+1} I$.

Similarly,
\begin{equation}\label{eq:823gt8g8GDSihwf8hqw8gas}
\begin{split}
M&=P A P\\
&\preceq P A P+(I-P) A (I-P)\\
&= A\\
&\preceq \wt{A}+\|A-\wt{A}\|_2\cdot I\\
&= \wt{P} (\wt{A}+\|A-\wt{A}\|_2\cdot I)\wt{P}+(I-\wt{P}) (\wt{A}+\|A-\wt{A}\|_2\cdot I) (I-\wt{P})\\
&\preceq \wt{M}+ (2\|A-\wt{A}\|_2+\gamma_{k+1})I.
\end{split}
\end{equation}
The transition from line~1 to line~2 is due to the fact that $A\succeq 0$, and therefore $(I-P)A(I-P)^T\succeq 0$. The transition from line~3 to line~4 is due to $A\preceq \wt{A}+\|A-\wt{A}\|_2 \cdot I$. The transition from line~5 to line~6 is due to the fact that 
$$
(I-\wt{P}) \wt{A} (I-\wt{P})\preceq \nu_{k+1}(\wt{A})\cdot I\preceq (\|A-\wt{A}\|_2+\gamma_{k+1}) I.
$$

We now apply Claim~\ref{cl:384g8g834gDD} with $A=M+ (2\|A-\wt{A}\|_2+\gamma_{k+1})I$, $\Pi_A=P$, $B=\wt{M}$ and $\Pi_B=\wt{P}$. Note
that A is symmetric and positive semidefinite. Also, $B$ is symmetric and positive semidefinite because $\nu_{\min}(B)=\nu_{k}(\wt{A})\geq \nu_k(A)-||A-\wt{A}||_2 \geq \frac{99\cdot\gamma_k}{100}\geq 0$ by Weyl's inequality and the fact that $||A-\wt{A}||_2\leq \frac{\gamma_k}{100}$. Note that $\Pi_A A^+= A^+\Pi_A$, as required,  and $A\succeq B$ by~\eqref{eq:823gt8g8GDSas}. We get
\begin{align}
\label{eq:ig82eihSIINFS}
\wt{M}^+&\succeq (P(M+\eta I)P)^+-2\|P(M+\eta I)^+P\|_2\cdot \|P(I-\wt{P})\|_2\cdot I \nonumber\\
&\succeq Y_{[k]} (\Gamma_{[k]}+\eta I_k)^{-1} Y^T_{[k]}+\frac{2}{\gamma_k}\cdot \|P(I-\wt{P})\|_2\cdot I\text{~~~~~(since $\|P(M+\eta I)^+P\|_2\leq 1/\gamma_k$)} \nonumber\\
&\succeq M^+-\left(\frac{\eta}{\gamma_k^2}+\frac{2}{\gamma_k}\cdot \|P(I-\wt{P})\|_2\right)\cdot I \nonumber\\
&\succeq M^+-\left(\frac{\eta}{\gamma_k^2}+\frac{8 \|A-\wt{A}\|_2}{\gamma_k^2}\right)\cdot I.
\end{align}
The transition from line~2 to line~3 used the fact that 
\begin{equation}\label{eq:83hg8haugfugUGSUH}
\|Y_{[k]} (\Gamma_{[k]}+\eta I_k)^{-1} Y^T_{[k]}-M^+\|\leq \frac{\eta}{\gamma_k^2}.
\end{equation}
The transition from line~3 to line~4 used
\begin{equation}\label{eq:8gABIAHDxccxfsDSD}
\|P(I-\wt{P})\|_2\leq \frac{\|A-\wt{A}\|_2}{\gamma_k/4}.
\end{equation}
We verify both~\eqref{eq:83hg8haugfugUGSUH} and~\eqref{eq:8gABIAHDxccxfsDSD} below.

Similarly, to upper bound $\wt{M}^+$ in terms of $M^+$ we apply Claim~\ref{cl:384g8g834gDD} with $A=\wt{M}+ (2\|A-\wt{A}\|_2+\gamma_{k+1})I$, $\Pi_A=\wt{P}$, $B=M$ and $\Pi_B=P$. Note that $\Pi_A A= A\Pi_A$, as required, $A$ and $B$ are both symmetric and positive semidefinite, and $A\succeq B$ by~\eqref{eq:823gt8g8GDSihwf8hqw8gas}. We get
\begin{align}
\label{eq:iwehf823g8g8hASDwe23ds}
M^+&\succeq (\wt{P}(\wt{M}+\eta \cdot I)\wt{P})^+ +2\|\wt{P}(\wt{M}+\eta \cdot I)^+\|_2\cdot \|\wt{P}(I-P)\|_2\cdot I \nonumber\\
&\succeq \wt{Y}_{[k]} (\wt{\Gamma}_{[k]}+I_k)^{-1} \wt{Y}^T_{[k]} +\frac{2}{\gamma_k}\cdot \|\wt{P}(I-P)\|_2\cdot I \nonumber\\
&\succeq \wt{M}^+-\left(\frac{4\eta}{\gamma_k^2}+\frac{2}{\gamma_k}\cdot \|\wt{P}(I-P)\|_2\right)\cdot I \nonumber\\
&\succeq \wt{M}^+-\left(\frac{4\eta}{\gamma_k^2}+\frac{8 \|A-\wt{A}\|_2}{\gamma_k^2}\right)\cdot I.
\end{align}

The transition from line~1 to line~2 uses the fact that by Weil's inequality
$$
\|\wt{P}(\wt{M}+\eta \cdot I)^+\|_2=\frac1{\nu_k(\wt{A}+\eta \cdot I)}\leq \frac1{\nu_k(A)-\|A-\wt{A}\|_2+\eta}=\frac1{\nu_k(A)+\|A-\wt{A}\|_2+\gamma_{k+1}}\leq \frac1{\gamma_k},
$$ 
since $\eta=2\|A-\wt{A}\|_2+\gamma_{k+1}$. The transition from line~2 to line~3 used the fact that
\begin{equation}\label{eq:892g38gGiihDUZGD}
\|\wt{Y}_{[k]} (\wt{\Gamma}_{[k]}+\eta I_k)^{-1} Y^T_{[k]}-\wt{M}^+\|\leq \frac{4\eta}{\gamma_k}.
\end{equation}
The transition from line~3 to line~4 used
\begin{equation}\label{eq:23r8g8gdDBDN}
\|\wt{P}(I-P)\|_2\leq \frac{\|A-\wt{A}\|_2}{\gamma_k/4}.
\end{equation}
We verify both~\eqref{eq:892g38gGiihDUZGD} and~\eqref{eq:23r8g8gdDBDN} below.

Putting~\eqref{eq:ig82eihSIINFS} and~\eqref{eq:iwehf823g8g8hASDwe23ds} together, we get 
\begin{equation*}
\begin{split}
\|M^+-\wt{M}^+\|_2&\leq \frac{4\eta}{\gamma_k^2}+\frac{8 \|A-\wt{A}\|_2}{\gamma_k^2}\leq \frac{16 \|A-\wt{A}\|_2+4\gamma_{k+1}}{\gamma_k^2}
\end{split}
\end{equation*}
as required.

We now verify~\eqref{eq:83hg8haugfugUGSUH}, \eqref{eq:8gABIAHDxccxfsDSD}, \eqref{eq:892g38gGiihDUZGD}
and~\eqref{eq:23r8g8gdDBDN}. First, one has
\begin{equation*}
\begin{split}
\|Y_{[k]} (\Gamma_{[k]}^{-1}-(\Gamma_{[k]}+\eta\cdot I_k)^{-1})Y^T_{[k]}\|_2&\leq \max_{\xi\geq\gamma_k} \left(\frac1{\xi}-\frac1{\xi+\eta}\right)\\
&=\max_{\xi\geq\gamma_k} \frac{\eta}{\xi (\xi+\eta)}\\
&\leq \frac{\eta}{\gamma_k^2}\\
\end{split}
\end{equation*}
and similarly, since $\nu_k(\wt{A})\geq \nu_k(A)-\|A-\wt{A}\|_2$ by Weyl's inequality (Lemma \ref{lem_Weyl}),
\begin{align*}
\|\wt{Y}_{[k]} (\wt{\Gamma}_{[k]}^{-1}-(\wt{\Gamma}_{[k]}+\eta\cdot I_k)^{-1})\wt{Y}^T_{[k]}\|_2&\leq\max_{\xi\geq\gamma_k-\|A-\wt{A}\|_2} \left(\frac1{\xi}-\frac1{\xi+\eta}\right)\\
&=\max_{\xi\geq \gamma_k-\|A-\wt{A}\|_2} \frac{\eta}{\xi (\xi+\eta)}\\
&\leq \frac{4\eta}{\gamma_k^2}&&\text{Since $\|A-\wt{A}\|_2\leq \gamma_k/2$ by assumption}\\
\end{align*}
This verifies~\eqref{eq:83hg8haugfugUGSUH} and~\eqref{eq:892g38gGiihDUZGD}.

It remains to verify~\eqref{eq:8gABIAHDxccxfsDSD} and~\eqref{eq:23r8g8gdDBDN}. In order to bound $\|P\cdot (I-\wt{P})\|_2$ and $\|\wt{P}\cdot (I-P)\|_2$, we first note that by Weyl's inequality 
$$
\nu_{k+1}(\wt{A})\leq \nu_{k+1}(A)+||A-\wt{A}||_2\leq \gamma_k/4+\gamma_k/100< (3/4)\gamma_k
$$
and $\nu_{k}(A)=\gamma_k$ by assumption of the lemma. Hence we can apply  Theorem~\ref{lem:davis-kahan} by choice of $H=A$, $E_0=Y_{[k]}$, $E_1=Y_{-[k]}$, $A_0=\Gamma_{[k]}$, $A_1=\Gamma_{-[k]}$, and $\widetilde{H}=\widetilde{A}$, $F_0=\widetilde{Y}_{[k]}$, $F_1=\widetilde{Y}_{-[k]}$, $\Lambda_0=\widetilde{\Gamma}_{[k]}$, $\Lambda_1=\widetilde{\Gamma}_{-[k]}$.
Let $\eta=\frac{\gamma_k}{4}$. Note that the eigenvalues of $A_0=\Gamma_{[k]}$ are at least $\gamma_k$
and the eigenvalues of $\Lambda_1= \widetilde{\Gamma}_{-[k]}$ are at most $(3/4) \gamma_k=\gamma_k-\eta$. Therefore, by Theorem~\ref{lem:davis-kahan} we have
 \[\| \wt{Y}_{-[k]}^T Y_{[k]}  \|_2=\|F_1^T E_0\|_2 \leq \frac{\|F_1^T (\widetilde{A}-A) E_0\|_2}{\eta}\leq \frac{\|A-\widetilde{A}\|_2}{\gamma_k/4} \text{.}\]
Thus we have 
$\|Y_{[k]}^T \wt{Y}_{-[k]}\|_2 \leq \frac{\|A-\wt{A}\|_2}{\gamma_k/4}$. Similarly, we have 
$$
\nu_{k+1}(A)\leq \gamma_k/4
$$
and $\nu_k(\wt{A})\geq \nu_k(A)-\|A-\wt{A}\|_2\geq \gamma_k-\gamma_k/100$. Hence we can apply  Theorem~\ref{lem:davis-kahan} by choice of $H=A$, $E_0=Y_{-[k]}$, $E_1=Y_{[k]}$, $A_0=\Gamma_{-[k]}$, $A_1=\Gamma_{[k]}$, and $\widetilde{H}=\widetilde{A}$, $F_0=\widetilde{Y}_{-[k]}$, $F_1=\widetilde{Y}_{[k]}$, $\Lambda_0=\widetilde{\Gamma}_{-[k]}$, $\Lambda_1=\widetilde{\Gamma}_{[k]}$. Let $\eta=\frac{\gamma_k}{4}$. Note that the eigenvalues of $A_0=\Gamma_{-[k]}$ are at most $\gamma_{k+1}$
and the eigenvalues of $\Lambda_1 = \widetilde{\Gamma}_{[k]}$ are at least $\gamma_k-\gamma_k/100\geq \gamma_k-\eta$. Therefore, by Theorem~\ref{lem:davis-kahan} we have
\[\|\wt{Y}_{[k]}^T Y_{-[k]}\|_2=\|F_1^T E_0\| \leq \frac{\|F_1^T (\widetilde{A}-A) E_0\|}{\eta} \leq \frac{\|A-\widetilde{A}\|}{\gamma_k/4} \text{.}\]
Thus, we have $\|\wt{Y}_{[k]}^T Y_{-[k]}\|_2 \leq \frac{\|A-\wt{A}\|_2}{\gamma_k/4}$. Putting these two bounds together, we get
$$
\| P(I-\wt{P})\|_2=\| Y_{[k]} Y_{[k]}^T\wt{Y}_{-[k]}\wt{Y}_{-[k]}^T\|_2=\|Y_{[k]}^T\wt{Y}_{-[k]}\|_2\leq \frac{\|A-\wt{A}\|_2}{\gamma_k/4},
$$
and similarly 
$$
\| \wt{P}(I-P)\|_2\leq \frac{\|A-\wt{A}\|_2}{\gamma_k/4}.
$$
\end{proof}

\subsection{Stability bounds under sampling of vertices}
\label{subsubsec:cols}

The main result of this section is Lemma \ref{lem:bnd-e1}, in which we give bounds for the stability of the pseudoinverse of the rank-$k$-approximation when we are sampling columns of the $k$-step random walk matrix of a $(k,\varphi,\epsilon)$-clusterable graph.

\begin{restatable}{lemma}{lembndeone}
\label{lem:bnd-e1}
Let $k \geq 2$ be an integer, $\varphi \in (0,1)$ and $\epsilon\in (0,1)$. Let $G=(V,E)$ be a $d$-regular and $(k,\varphi,\epsilon)$-clusterable graph. Let $M$ be  the random walk transition matrix of $G$. Let  $1/n^6 < \xi < 1$, $t\geq  \frac{20\log n}{\varphi^2}$. Let $c>1$ be a large enough constant and let $s\geq c
\cdot n^{(480\cdot \epsilon / \varphi^2)}\cdot \log n \cdot k^{8}/{\xi}^2$. Let $I_S=\{i_1,\ldots, i_s\}$ be a multiset of $s$ indices chosen independently and uniformly at random from
$\{1,\dots,n\}$. Let $S$ be the $n\times s$ matrix whose $j$-th column equals $\mathds{1}_{i_j}$.  Let $M^t=U\Sigma^tU^T$ be an eigendecomposition of $M^t$. Let $\sqrt{\frac{n}{s}} \cdot M^tS=\widetilde{U}\widetilde{\Sigma}\widetilde{W}^T$  be an SVD of $\sqrt{\frac{n}{s}} \cdot M^tS$ where $\widetilde{U}\in \R^{n\times n}, \widetilde{\Sigma}\in \R^{n\times n}, \widetilde{W}\in \R^{s\times n}$. If $\frac{\epsilon}{\varphi^2}\leq \frac{1}{10^5}$ then with probability at least $1-n^{-100}$ matrix $\widetilde \Sigma_{[k]}^{-4}$ exists and we have
\[\left|\mathds{1}_x^T U_{[k]}{U}_{[k]}^T  \mathds{1}_y - (M^{t}\mathds{1}_{x})^T  (M^tS)\left(\frac{n}{s}\cdot\widetilde{W}_{[k]} \widetilde{\Sigma}^{-4}_{[k]} \widetilde{W}^T_{[k]}\right) (M^tS)^T (M^{t}\mathds{1}_{y}) \right| \leq  \frac{\xi}{n}\text{.}
\]
\end{restatable}

To prove Lemma \ref{lem:bnd-e1} we require the following matrix concentration bound, which is a generalization of Bernstein's inequality to matrices.

\begin{lemma}[Matrix Bernstein \cite{tropp2012user}] 
\label{lem:Bernstein}
Consider a finite sequence ${X_i}$ of independent, random matrices with dimensions $d_1 \times d_2$. Assume that each random matrix satisfies $\mathbb{E}[X_i] = 0$ and $\|X_i\|_2 \leq b$ almost surely. Define $\sigma^2=\max\{\|\sum_i \mathbb{E}[X_iX_i^T]\|_2, \|\sum_i
\mathbb{E}[X_i^T X_i]\|_2\}$. Then for all $t\geq 0$, \[\mathbb{P}\left[\|\sum_i X_i\|_2 \geq t\right] \leq (d_1+d_2)\cdot \exp\left(\frac{-t^2 / 2}{\sigma^2+bt / 3}\right) \text{.}\]
\end{lemma}

Equiped with the Matrix Bernstein bound, we can show that under certain spectral conditions we can approximate a matrix $A A^T$ by $(AS)(AS)^T$, i.e. by sampling rows of $M$. 
The idea is to write $AA^T = \sum_{i=1}^n (A\mathds{1}_i) (A\mathds{1}_i)^T$ as a sum over the outer products of its columns and make the sample size depend on the spectral norm
of the summands. 

\begin{restatable}{lemma}{lemGGhat}
\label{lem:G-Ghat}
Let $A \in \R^{n\times n}$ be a matrix. Let $B = \max_{\ell \in \{1,\dots,n\}} \|(A \mathds{1}_\ell) (A \mathds{1}_{\ell})^T\|_2$.
Let $1 > \xi > 0$. Let $s \ge \frac{40 n^2 B^2 \log n}{\xi^2}$. Let $I_S=\{i_1,\ldots, i_s\}$ be a multiset of $s$ indices chosen independently and uniformly at random from
$\{1,\dots,n\}$. Let $S$ be the $n\times s$ matrix whose $j$-th column equals $\mathds{1}_{i_j}$. Then we have
$$
 \mathbb{P}\left[ \left\|A A^T-\frac{n}{s}(AS)(AS)^T\right\|_2 \ge \xi  \right]  \leq n^{-100}.
$$
\end{restatable}
\begin{proof}
Observe that 
\begin{equation}
\label{eq:G-def}
A A^T=\sum_{\ell \in \{1,\dots,n\}} (A\mathds{1}_{\ell})(A\mathds{1}_{\ell})^T\text{.}
\end{equation}
and
\begin{equation}
\label{eq:G-hat-def}
\frac{n}{s}(AS)(AS)^T=\frac{n}{s}\cdot\sum_{i_j\in I_S} (A\mathds{1}_{i_j})(A\mathds{1}_{i_j})^T\text{.}
\end{equation}
For every $j=1,2,\ldots, s$ let $X_{j}=\frac{n}{s}\cdot (A\mathds{1}_{i_j})(A\mathds{1}_{i_j})^T$. Thus we have 
\begin{equation}
\label{eq:expXi}
\mathbb{E}[X_{j}]=\frac{n}{s}\cdot \mathbb{E}[(A\mathds{1}_{i_j})(A\mathds{1}_{i_j})^T] = \frac{n}{s}\cdot \frac{1}{n} \sum_{\ell \in \{1,\dots,n\} } (A\mathds{1}_{\ell})(A\mathds{1}_{\ell})^T = \frac{1}{s}\cdot A A^T
\end{equation}
By equality~\eqref{eq:G-hat-def} we have $\frac{n}{s}(AS)(AS)^T=\sum_{j=1}^s X_{j}$. Thus by equality~\eqref{eq:expXi} we get
\begin{equation}
\label{eq:G-Ghat}
\left\|\frac{n}{s}(AS)(AS)^T-A A^T\right\|_2=\left\|\sum_{j=1}^s (X_{j}- \mathbb{E}[X_{j}])\right\|_2 \text{.}
\end{equation}
Let $Z_{j}=X_{j}- \mathbb{E}[X_{j}]$.
We then have
$
\|Z_{j}\|_2 = \|X_{j}-\mathbb{E}[X_{j}]\|_2 \leq \|X_{j}\|_2 + \|\mathbb{E}[X_{j}]\|_2 
$
Now let $B = \max_{\ell \in \{1,\ldots,n\}} \|(A\mathds{1}_{\ell})(A\mathds{1}_\ell)^T\|_2$.
Furthermore, by our assumption we have 
\begin{equation}
\label{eq:XM}
\|X_{j}\|_2= \left\|\frac{n}{s}\cdot (A \mathds{1}_{j})(A\mathds{1}_{j})^T\right\|_2 \le \frac{n}{s} \cdot B
\end{equation}
By subadditivity of the spectral norm and \eqref{eq:expXi} we get
\begin{equation}
\label{eq:EXj}
\|\mathbb{E}[X_j]\|_2 \le \frac{n}{s} \cdot B
\end{equation}
Putting~\eqref{eq:XM} and~\eqref{eq:EXj} together we get
\begin{equation}
\label{eq:Z_i-bnd}
\|Z_{j}\|_2=\|X_{j}-\mathbb{E}[X_{j}]\|_2\leq \|X_{j}\|_2+\|\mathbb{E}[X_{j}]\|_2  \leq 2 \cdot \frac{n}{s} \cdot B
\end{equation}
We now bound for the variance. Since $Z_j$ is symmetric, we have $Z_j^T Z_j = Z_j Z_j^T = Z_j^2$.
\begin{align*}
\left\| \sum_{j=1}^s \mathbb{E}[Z_{j}^2]\right\|_2 = s\cdot \| \mathbb{E}[Z_{j}^2]\|_2 = s\cdot \| \mathbb{E}[X_{j}^2]-\mathbb{E}[X_{j}]^2 \|_2 
\leq s\cdot \| \mathbb{E}[X_{j}^2]\|_2 + s\cdot \|\mathbb{E}[X_{j}]^2 \|_2
\end{align*}
By submultiplicativity of the spectral norm we get
\begin{equation}
\label{eq:var1}
\| \mathbb{E}[X_{j}^2]\|_2=\left\|\frac{1}{n} \cdot \frac{n^2}{s^2} \sum_{\ell \in \{1,\dots, n\}} ((A\mathds{1}_{\ell})(A\mathds{1}_{\ell})^T)^2\right\|_2
\leq \frac{n^2}{s^2} \cdot B^2 
\end{equation}
Moreover by submultiplicativity of spectral norm we have $\|\mathbb{E}[X_{j}]^2 \|_2\le\|\mathbb{E}[X_{j}] \|^2_2 \le \frac{n^2}{s^2} \cdot B^2$.
Putting things together we obtain
$$
\| \sum_{j=1}^s \mathbb{E}[Z_{j}^2]\|_2 \leq \frac{2n^2 B^2}{s}
$$
Now we can apply Lemma ~\ref{lem:Bernstein} and we get with $b = 2 \frac{n}{s} B$ and $\sigma^2 \leq \frac{2n^2B^2}{s}$ using $s \ge \frac{40 n^2 B^2 \log n}{\xi^2}$
\begin{equation}
\label{eq:bernstein}
\mathbb{P}\left[ \|\sum_{j=1}^s Z_{j}\|_2 > \xi  \right]  \leq 2n \cdot \text{exp}\left(\frac{\frac{-\xi^2}{2}}{\sigma^2+\frac{b\xi}{3}} \right) \le n^{-100}
\end{equation}
\end{proof}
The following lemma upper bounds the collision probability from {\bf every} vertex in  a $(k, \varphi, \e)$-clusterable graph using our $\ell_\infty$ norm bounds on the bottom $k$ eigenvectors of the Laplacian of such graphs\footnote{It is interesting to note that a weaker average case version of this lemma was used in two prior works on testing graph cluster structure~\cite{DBLP:conf/stoc/CzumajPS15} and~\cite{chiplunkar2018testing}. The stronger version of the lemma presented here is important for spectral concentration bounds that we present, which are in turn crucial for sublinear time dot product access to the spectral embedding.}: 
\begin{restatable}{lemma}{lemMtbnd}\label{lem:Mt-bnd}
Let $k \geq 2$ be an integer, $\varphi \in (0,1)$ and $\epsilon\in (0,1)$. Let $G=(V,E)$ be a $d$-regular and that admits a $(k,\varphi,\epsilon)$-clustering  $C_1, \ldots , C_k$.  Let $M$ be the random walk transition matrix of $G$. For any $t\geq \frac{20\log n}{\varphi^2}$ and any $x\in V$ we have 
\[\|M^{t}\mathds{1}_{x}\|_2 \leq O(k\cdot n^{-1/2+(20\epsilon /\varphi^2)}) \text{.}\]
\end{restatable}
\begin{proof}
Let $L$ be the normalized Laplacian of $G$. Recall that $(u_1,\ldots,u_n)$ are an orthonormal basis of eigenvectors of $L$ with corresponding eigenvalues $0=\lambda_1\leq\ldots \leq\lambda_n$. Observe that each $u_i$ is also an eigenvector of $M$, with eigenvalue $1-\frac{\lambda_i}{2}$. We write $\mathds{1}_x$ in the eigenbasis of $L$ as $\mathds{1}_x = \sum_{j=1}^n \beta_j u_j$ and note that the $\beta_j$ correspond to the row of $x$ in the matrix $U$.
We have
\[
M^t \mathds{1}_x = M^t \left( \sum_{j=1}^n \beta_j u_j \right)= \sum_{j=1}^n \beta_j M^t u_j = \sum_{j=1}^n \beta_j \left(1-\frac{\lambda_j}{2}\right)^t u_j.
\]
Thus we get
\begin{equation}\label{eq:mmt}
\| M^t \mathds{1}_x\|_2^2 = \sum_{j=1}^n \beta_j^2 \left(1-\frac{\lambda_j}{2}\right)^{2t} \le \sum_{j=1}^k \beta_j^2 + \left(1-\frac{\lambda_{k+1}}{2}\right)^{2t}  \cdot \sum_{j=k+1}^n \beta_j^2.
\end{equation}

Note that $G$ is $(k,\varphi,\epsilon)$-clusterable, therefore by Lemma \ref{lem:bnd-lambda} we have $\lambda_{k+1} \geq \frac{\varphi^2}{2}$. Note that $t\geq \frac{20\log n}{\varphi^2}$. Hence, we have
\begin{equation}\label{eq:lambkk11}
\left(1-\frac{\lambda_{k+1}}{2}\right)^{2t} \leq n^{-10} \text{.}
\end{equation}
Moreover since $G$ is $(k,\varphi,\epsilon)$-clusterable and $\min_i{|C_i|}\geq \Omega(\frac{n}{k})$ by Lemma \ref{lem:l-inf-bnd} for all $j\in [k]$ we have 
\begin{equation}\label{eq:betaj}
\beta_j \le \|u_j\|_\infty \leq O(\sqrt{k}\cdot n^{-1/2+(20\epsilon /\varphi^2)}) \text{.}
\end{equation}
Thus by \eqref{eq:mmt}, \eqref{eq:lambkk11} and \eqref{eq:betaj} we get
$$
\|M^t\mathds{1}_x\|^2_2  \le O(k\cdot k \cdot \frac{1}{n} \cdot n^{40\epsilon/\varphi^2}) + n \cdot n^{-10}.  
$$
Therefore we have
\[\|M^{t}\mathds{1}_{x}\|_2 \leq O(k\cdot n^{-1/2+(20\epsilon /\varphi^2)}) \text{.}\]
\end{proof}
Combining the previous lemmas and Lemma \ref{lem:neg-l2-close} we obtain Lemma \ref{lem:v-close}. We show that for $(k,\varphi,\epsilon)$-clusterable graphs, the outer products of the columns of the $t$-step random walk
transition matrix have small spectral norm. This is because the matrix power is mostly determined by the first $k$ eigenvectors and by the fact that these eigenvectors
have bounded infinity norm.

\begin{restatable}{lemma}{lemvclose}\label{lem:v-close}
Let $k \geq 2$ be an integer, $\varphi \in (0,1)$ and $\epsilon\in (0,1)$. Let $G=(V,E)$ be a $d$-regular and $(k,\varphi,\epsilon)$-clusterable graph. Let $M$ be the random walk transition matrix of $G$. Let $1> \xi>1/n^{8}$,
$t\geq \frac{20 \log n}{\varphi^2}$. Let $c>1$ be a large enough constant and let $s\ge  c\cdot k^4\cdot  n^{(400\cdot  \epsilon/\varphi^2)}\log n / \xi^2  $.  Let $I_S=\{i_1,\ldots, i_s\}$ be a multiset of $s$ indices chosen independently and uniformly at random from
$\{1,\dots,n\}$. Let $S$ be the $n\times s$ matrix whose $j$-th column equals $\mathds{1}_{i_j}$. Let $M^t=U\Sigma^tU^T$ be an eigendecomposition of $M^t$. Let $\sqrt{\frac{n}{s}} \cdot M^tS=\widetilde{U}\widetilde{\Sigma}\widetilde{W}^T$  be an SVD of $\sqrt{\frac{n}{s}} \cdot M^tS$ where $\widetilde{U}\in \R^{n\times n}, \widetilde{\Sigma}\in \R^{n\times n}, \widetilde{W}\in \R^{s\times n}$.
If $\frac{\epsilon}{\varphi^2}\leq \frac{1}{10^5}$ then with probability at least $1-n^{-100}$ matrix $\widetilde \Sigma_{[k]}^{-2}$ exists and we have
$$
\left|\left|U_{[k]} {\Sigma}_{[k]}^{-2t} U_{[k]}^T - \widetilde{U}_{[k]} \widetilde{\Sigma}_{[k]}^{-2} \widetilde{U}_{[k]}^T \right|\right|_2 < \xi
$$
\end{restatable}

\begin{proof}
Let  \[A=(M^{t})(M^{t})^T= U \Sigma^{2t} U^T ,\]and \[\widetilde{A}=\frac{n}{s} \left(M^{t}S \right)\left(M^{t}S\right)^T = \widetilde{U} \widetilde{\Sigma}^{2} \widetilde{U}^T\text{.}\] 
Let $\gamma_k$ and $\gamma_{k+1}$ denote the $k$-th and $(k+1)$-th largest eigenvalues of $A$.
Let $U$ be an orthonormal basis of eigenvectors of $L$ with corresponding eigenvalues $\lambda_1\leq\ldots \leq\lambda_n$. Observe that each $u_i$ is also an eigenvector of $M$, with eigenvalue $1-\frac{\lambda_i}{2}$. 
Note that $G$ is $(k,\varphi,\epsilon)$-clusterable, therefore by Lemma \ref{lem:bnd-lambda} we have $\lambda_k\leq 2\epsilon$ and $\lambda_{k+1} \geq \frac{\varphi^2}{2}$. Note that $t\geq \frac{20\log n}{\varphi^2}$. Hence, we have
\begin{equation}
\label{2mu_k+1-M}
\gamma_{k+1} = \left(1-\frac{\lambda_{k+1}}{2} \right)^{2t}\leq n^{-10}
\end{equation}
and
\begin{equation}
\label{2mu_k-M}
\gamma_k =  \left(1-\frac{\lambda_{k}}{2} \right)^{2t} \ge n^{(-80 \epsilon / \varphi^2)}
\text{.}
\end{equation}
In order to apply Lemma \ref{lem:G-Ghat} we need to derive an upper bound on the spectral norm of $(M^{t}\mathds{1}_{x})(M^{t}\mathds{1}_{x})^T$ for any column of $A$ corresponding to vertex $x$. By Lemma \ref{lem:Mt-bnd} we have
\[B=\|(M^{t}\mathds{1}_{x})(M^{t}\mathds{1}_{x})^T\|_2 = \|M^{t}\mathds{1}_{x}\|^2_2\leq O(k^2\cdot n^{-1+(40\epsilon /\varphi^2)}) \text{.}\]
Thus, with $1\ge \xi>1/n^{8}$ and for large enough $c$ we have $s \ge c\cdot k^4 n^{(400\cdot  \epsilon/\varphi^2)}\log n / \xi^2
\ge \frac{40n^2 B^2 \log n}{1/32^2\cdot\xi^2 n^{-320{\e/\varphi}}}$. Thus by Lemma \ref{lem:G-Ghat} we obtain that with probability at least $1-{n^{-100}}$ that 
\begin{equation}\label{2eq:gap_g}
\|A-\widetilde{A}\|_2 \leq  \frac{1}{32}\cdot \xi \cdot n^{-160\e/\varphi^2}  \text{.}
\end{equation}
We observe that equation \ref{2eq:gap_g} together with our bound on $\gamma_k$ \eqref{2mu_k-M} and the positive semi-definiteness of $\widetilde A$ imply 
that the $k$ largest eigenvalues of $\widetilde A$ are non-zero and so $\widetilde \Sigma_{[k]}^{-2}$ is exists with high probability.

Now observe that $A$ is positive semi-definite, we have $\gamma_k/4 > \gamma_{k+1}$  and $\|A-\widetilde A\| \leq \gamma_k/100$,
so the preconditions of Lemma \ref{lem:neg-l2-close} are met and we have with probability $1-n^{-100}$
\begin{align*}
\left|\left| U_{[k]} {\Sigma}_{[k]}^{-2t} U_{[k]}^T - \widetilde{U}_{[k]} \widetilde{\Sigma}_{[k]}^{-2} \widetilde{U}_{[k]}^T \right|\right|_2 &\leq  \frac{16 \|A-\wt{A}\|_2+4\gamma_{k+1}}{\gamma_k^2} \leq \frac{16\cdot \frac{1}{32}\cdot\xi\cdot n^{(-160\cdot\epsilon/\varphi^2)} + 4\cdot n^{-10}}{n^{(-160\cdot\epsilon/\varphi^2)} }\leq \frac{\xi}{2}+\frac{\xi}{2}
= \xi \text{.}
\end{align*}
\end{proof}

Now we are ready to prove Lemma \ref{lem:bnd-e1}.
\lembndeone*
\begin{proof}
Let $m_x=M^{t}\mathds{1}_{x}$ and $m_y=M^{t}\mathds{1}_{y}$.  We first prove $m_{x}^T (U_{[k]} {\Sigma}_{[k]}^{-2t} U_{[k]}^T)m_{y} = \mathds{1}_x^T U_{[k]} U_{[k]}^T \mathds{1}_y$ and $m_x^T  (M^tS)(\widetilde{W}_{[k]} \widetilde{\Sigma}^{-4}_{[k]} \widetilde{W}^T_{[k]}) (M^tS)^T m_y = m_x^T  \widetilde{U}_{[k]}\widetilde{\Sigma}^{-2}_{[k]} \widetilde{U}_{[k]}^T m_y$. Then we upper bound \[\left|m_x^T U_{[k]} {\Sigma}_{[k]}^{-2t} U_{[k]}^Tm_y - m_x^T\widetilde{U}_{[k]} \widetilde{\Sigma}_{[k]}^{-2} \widetilde{U}_{[k]}^T m_y \right| \text{.}\]
\paragraph{Step $1$:}
Note that $M^t=U\Sigma^tU^T$. Therefore we get $M^t\mathds{1}_x= U\Sigma^tU^T  \mathds{1}_x$, and $M^t\mathds{1}_y= U\Sigma^tU^T  \mathds{1}_y$. Thus we have
\begin{equation}
\label{eq:pxto1x}
m_x^T U_{[k]} {\Sigma}_{[k]}^{-2t} U_{[k]}^Tm_y = \mathds{1}_x^T \left(\left(U\Sigma^tU^T \right) \left(U_{[k]} {\Sigma}_{[k]}^{-2t} U_{[k]}^T\right) \left(U\Sigma^tU^T \right) \right)\mathds{1}_y
\end{equation}
Note that $U^T U_{[k]}$ is an $n\times k$ matrix such that the top $k\times k$ matrix is $I_{k\times k}$ and the rest is zero. Also $U_{[k]}^T U$ is a $k\times n$ matrix such that the left $k\times k$ matrix is $I_{k\times k}$ and the rest is zero. Therefore we have 
\[U{\Sigma}^{t}\left(U^T U_{[k]}\right) {\Sigma}^{-2t}_{[k]} \left(U^T_{[k]} U\right){\Sigma}^{t}U^T = UHU^T \text{,}\]
where $H$ is an $n\times n$ matrix such that the top left $k\times k$ matrix is $I_{k\times k}$ and the rest is zero. Hence, we have
\[ U H U^T =U_{[k]} U_{[k]}^T\text{.}\]
Thus we have
\begin{equation}
\label{eq:p-vk}
m_{x}^T (U_{[k]} {\Sigma}_{[k]}^{-2t} U_{[k]}^T)m_{y} = \mathds{1}_x^T U_{[k]} U_{[k]}^T \mathds{1}_y
\end{equation}
\paragraph{Step $2$:}
We have $\sqrt{\frac{n}{s}} \cdot M^tS=\widetilde{U}\widetilde{\Sigma}\widetilde{W}^T$ where $\widetilde{U}\in \R^{n\times n}$, $\widetilde{\Sigma}\in \R^{n\times n}$ and $\widetilde{W}\in \R^{s\times n}$. Therefore, 
\begin{align}
\label{eq:pxpy}
 &(m_x)^T  (M^tS)\left(\frac{n}{s}\cdot\widetilde{W}_{[k]} \widetilde{\Sigma}^{-4}_{[k]} \widetilde{W}^T_{[k]}\right) (M^tS)^T (m_{y}) \nonumber\\
 &= m_x^T  \left(\sqrt{\frac{s}{n}}\cdot \widetilde{U}\widetilde{\Sigma}\widetilde{W}^T\right)\left(\frac{n}{s}\cdot\widetilde{W}_{[k]} \widetilde{\Sigma}^{-4}_{[k]} \widetilde{W}^T_{[k]}\right) \left(\sqrt{\frac{s}{n}}\cdot\widetilde{W}\widetilde{\Sigma}\widetilde{U}^T\right) m_y \nonumber\\
 &= m_x^T  \left( \widetilde{U}\widetilde{\Sigma}\widetilde{W}^T\right)\left(\widetilde{W}_{[k]} \widetilde{\Sigma}^{-4}_{[k]} \widetilde{W}^T_{[k]}\right) \left(\widetilde{W}\widetilde{\Sigma}\widetilde{U}^T\right) m_y
\end{align}
Note that $\widetilde{W}^T\widetilde{W}_{[k]}$ is an $n\times k$ matrix such that the top $k\times k$ matrix is $I_{k\times k}$ and the rest is zero. Also $\widetilde{W}_{[k]}^T\widetilde{W}$ is a $k\times n$ matrix such that the left $k\times k$ matrix is $I_{k\times k}$ and the rest is zero. Therefore we have 
\[\widetilde{\Sigma}\left(\widetilde{W}^T\widetilde{W}_{[k]}\right) \widetilde{\Sigma}^{-4}_{[k]} \left(\widetilde{W}^T_{[k]} \widetilde{W}\right)\widetilde{\Sigma} = \widetilde{H} \text{,}\]
where $\widetilde{H}$ is an $n\times n$ matrix such that the top left $k\times k$ matrix is $\widetilde{\Sigma}^{-2}_{[k]}$ and the rest is zero. Hence, we have
\begin{equation}
\label{eq:vsigv}
(\widetilde{U}\widetilde{\Sigma}\widetilde{W}^T)\left(\frac{n}{s}\cdot\widetilde{W}_{[k]} \widetilde{\Sigma}^{-4}_{[k]} \widetilde{W}^T_{[k]}\right) (\widetilde{W}\widetilde{\Sigma}\widetilde{U}^T) = \widetilde{U}\widetilde{H}\widetilde{U}^T = \widetilde{U}_{[k]}\widetilde{\Sigma}^{-2}_{[k]} \widetilde{U}_{[k]}^T
\end{equation}
Putting \eqref{eq:vsigv} and \eqref{eq:pxpy} together we get
\begin{equation}
\label{eq:vtild}
 m_x^T  (M^tS)(\widetilde{W}_{[k]} \widetilde{\Sigma}^{-4}_{[k]} \widetilde{W}^T_{[k]}) (M^tS)^T m_y = m_x^T  \widetilde{U}_{[k]}\widetilde{\Sigma}^{-2}_{[k]} \widetilde{U}_{[k]}^T m_y
\end{equation}
\paragraph{Put together:}
Let $c'>1$ be a large enough constant we will set later. Let $\xi' = \frac{\xi}{c'\cdot k^2 \cdot n^{40 \e/\varphi^2}}$. Let $c_1$ be a constant in front of $s$ in Lemma  \ref{lem:v-close}. Thus for large enough $c$ we have $s\geq c
\cdot n^{(480\cdot \epsilon / \varphi^2)}\cdot \log n \cdot k^{8}/{\xi}^2 \geq c_1\cdot k^4\cdot  n^{(400\cdot  \epsilon/\varphi^2)}\log n / \xi'^2 $, hence, by Lemma \ref{lem:v-close} applied with 
$\xi'$,  with probability at least $1-n^{-100}$ we have
\[\left|\left|U_{[k]} {\Sigma}_{[k]}^{-2t} U_{[k]}^T - \widetilde{U}_{[k]} \widetilde{\Sigma}_{[k]}^{-2} \widetilde{U}_{[k]}^T \right|\right|_2\leq \xi' \]
Therefore by submultiplicativity of norm we have
\begin{align}
\left|m_x^T U_{[k]} {\Sigma}_{[k]}^{-2t} U_{[k]}^Tm_y - m_x^T\widetilde{U}_{[k]} \widetilde{\Sigma}_{[k]}^{-2} \widetilde{U}_{[k]}^T m_y \right| 
&\leq  \left|\left|U_{[k]} {\Sigma}_{[k]}^{-2t} U_{[k]}^T - \widetilde{U}_{[k]} \widetilde{\Sigma}_{[k]}^{-2} \widetilde{U}_{[k]}^T \right| \right|_2\|m_x\|_2\|m_y\|_2 \nonumber\\
&\leq \xi' \|m_x\|_2\|m_y\|_2\label{eq:pxpyv-til-close}
\end{align}
Therefore we have
\begin{align}
&\left|m_x^T  (M^tS)\left(\frac{n}{s}\cdot\widetilde{W}_{[k]} \widetilde{\Sigma}^{-4}_{[k]} \widetilde{W}^T_{[k]}\right) (M^tS)^T m_y - \mathds{1}_x^T U_{[k]}{U}_{[k]}^T  \mathds{1}_y \right| \nonumber \\
&= \left| m_x^T\widetilde{U}_{[k]} \widetilde{\Sigma}_{[k]}^{-2} \widetilde{U}_{[k]}^T m_y - m_x^T U_{[k]} {\Sigma}_{[k]}^{-2t} U_{[k]}^Tm_y\right| && \text{By \eqref{eq:p-vk} and \eqref{eq:vtild}}\nonumber \\
& \leq  \xi' \cdot\|m_x\|_2\|m_y\|_2 && \text{By \eqref{eq:pxpyv-til-close}} 
\end{align}
By Lemma \ref{lem:Mt-bnd} for any vertex $x\in V$ we have 
\begin{equation}
\|m_x\|^2_2=\|M^{t}\mathds{1}_{x}\|^2_2 \leq O(k^2\cdot n^{-1+(40\epsilon /\varphi^2)})  \text{.}
\end{equation}
Therefore by choice of $c'$ as a large enough constant and choosing $\xi' = \frac{\xi}{c'\cdot k^2 \cdot n^{40 \e/\varphi^2}}$ we have
\begin{equation}
\label{eq:e2-done}
\left|m_x^T  (M^tS)\left(\frac{n}{s}\cdot\widetilde{W}_{[k]} \widetilde{\Sigma}^{-4}_{[k]} \widetilde{W}^T_{[k]}\right) (M^tS)^T m_y - \mathds{1}_x^T U_{[k]}{U}_{[k]}^T  \mathds{1}_y \right|\leq O\left(\xi'\cdot k^2\cdot n^{-1+(40\epsilon /\varphi^2)}\right) \leq  \frac{\xi}{n} \text{.}
\end{equation}
\end{proof}
\subsection{Stability bounds under approximations of columns by random walks}
\label{subsubsec:rows}
The main result of this section is Lemma \ref{lem-u-close}, which shows that if a graph is $(k,\varphi,\epsilon)$-clusterable, then 
the pseudoinverseve of the low rank approximation of a random walk matrix are stable when it is empirically approximated by running random walks from sample vertices. 

\begin{restatable}{lemma}{lemuclose}
\label{lem-u-close}
Let $k \geq 2$ be an integer, $\varphi \in (0,1)$ and $\epsilon\in (0,1)$. Let $G=(V,E)$ be a $d$-regular and $(k,\varphi,\epsilon)$-clusterable graph. Let  $1/n^8 < \xi < 1$ and $t\geq \frac{20\log n}{\varphi^2}$. Let $c_1>1$ and $c_2>1$ be a large enough constants. Let $s\geq c_1\cdot n^{240\epsilon / \varphi^2}\cdot \log n \cdot k^{4} $ and $R\geq \frac{c_2\cdot k^{9}\cdot n^{(1/2+820\cdot\epsilon/\varphi^2)}}{\xi^2} $. Let $I_S=\{i_1,\ldots, i_s\}$ be a multiset of $s$ indices chosen independently and uniformly at random from
$\{1,\dots,n\}$. Let $S$ be the $n\times s$ matrix whose $j$-th column equals $\mathds{1}_{i_j}$.   Let $\G \in \R^{s\times s}$  be the output of $\textsc{EstimateCollisionProbabilities($G,I_S,R,t$)}$(Algorithm \ref{alg:gram}). Let $M$ be  the random walk transition matrix of $G$. Let $\sqrt{\frac{n}{s}} \cdot M^tS=\widetilde{U}\widetilde{\Sigma}\widetilde{W}^T$ be an SVD of $\sqrt{\frac{n}{s}} \cdot M^tS$ where $\widetilde{U}\in \R^{n\times n}, \widetilde{\Sigma}\in \R^{n\times n}, \widetilde{W}\in \R^{s\times n}$.
Let $\frac{n}{s}\cdot\G=\widehat{W} \widehat{\Sigma} \widehat{W}^T$ be an eigendecomposition of $\frac{n}{s}\cdot\G$.
If $\frac{\epsilon}{\varphi^2}\leq \frac{1}{10^5}$ then with probability at least $1-2\cdot n^{-100}$
matrices $\widehat \Sigma_{[k]}^{-2}$ and $\widetilde \Sigma_{[k]}^{-4}$ exist and we have
\[
\left|\left|\widehat{W}_{[k]}\widehat{\Sigma}_{[k]}^{-2}\widehat{W}_{[k]}^T - \widetilde{W}_{[k]}
\widetilde{\Sigma}^{-4}_{[k]} \widetilde{W}^T_{[k]} \right|\right|_2< \xi
\]
\end{restatable}
To prove Lemma \ref{lem-u-close} we need the following lemma.

\begin{restatable}{lemma}{lemMtbnd2}\label{lem:Mt-bndr}
Let $k \geq 2$ be an integer, $\varphi \in (0,1)$ and $\epsilon\in (0,1)$. Let $G=(V,E)$ be a $d$-regular and $(k,\varphi,\epsilon)$-clusterable graph.  Let $L$ and $M$ be the normalized Laplacian and transition matrix of $G$ respectively. For any $t\geq \frac{10\log n}{\varphi^2}$ and any $r$ and any $x \in V$ we have 
\[\|M^{t}\mathds{1}_{x}\|_r \leq  O\left(k^2\cdot n^{-1+1/r+(40\epsilon /\varphi^2)}\right) \text{.}\]
\end{restatable}
\begin{proof}
Let $L$ be the normalized Laplacian of $G$ with eigenvectors $u_1,\dots, u_n$ and corresponding eigenvalues $\lambda_1 \le \ldots \leq \lambda_n$.
Observe that each $u_i$ is also an eigenvector of $M$, with eigenvalue $1-\frac{\lambda_i}{2}$. Note that $G$ is $(k,\varphi,\epsilon)$-clusterable. 
Therefore by Lemma \ref{lem:bnd-lambda} we have 
\begin{equation}
\label{eq:lm-k+1}
\lambda_{k+1} \geq \frac{\varphi^2}{2}  \text{.}
\end{equation}
We write $\mathds{1}_{x}$ in the eigenbasis of $L$ as $\mathds{1}_{x}=\sum_{j=1}^{n} \beta_j u_j$ where $\beta_j=u_j\cdot \mathds{1}_{x}=u_j(x)$. Thus for any vertex $u$ we have
\begin{align*}
M^{t}\mathds{1}_{x}=M^{t} \left(\sum_{j=1}^{n} \beta_j u_j \right) 
= \sum_{j=1}^{n} \beta_j M^{t} u_j  
= \sum_{j=1}^{n} \beta_j \left(1-\frac{\lambda_j}{2}\right)^{t} u_j \text{.}
\end{align*}
Let $m_x=M^{t}\mathds{1}_{x}$. Therefore for any vertex $y\in V$ we have
\begin{align*}
m_x(y)&=\sum_{j=1}^{n} \beta_j \left(1-\frac{\lambda_j}{2}\right)^{t} u_j(y) \\
&= \sum_{j=1}^{k} \beta_j \left(1-\frac{\lambda_j}{2}\right)^{t} u_j(y) + \sum_{j=k+1}^{n} \beta_j \left(1-\frac{\lambda_j}{2}\right)^{t} u_j(y) 
\end{align*}
Therefore,
\begin{equation}
\label{eqpuabs}
 |m_x(y)| \leq  \left(1-\frac{\lambda_1}{2}\right)^{t} \sum_{j=1}^{k} |\beta_j| \cdot |u_j(y)| + \left(1-\frac{\lambda_{k+1}}{2}\right)^{t} \sum_{j=k+1}^{n} |\beta_j| \cdot| u_j(y)| 
\end{equation}
By \eqref{eq:lm-k+1} we have $\lambda_{k+1}\geq \frac{\varphi^2}{2}$, and $t\geq \frac{8\log n}{\varphi^2}$. Thus we have
\[\left(1-\frac{\lambda_{k+1}}{2}\right)^{t} \leq n^{-2}\]
Note that for any $j\in [n]$
\begin{equation}
\label{eq:bju}
|\beta_j| \leq \sqrt{\sum_{j=1}^n \beta_j ^2 } = \|\mathds{1}_{x}\|_2=1 \text{.}
\end{equation}
Morover for any $j\in [n]$ and any $y \in V$
\begin{equation}
\label{eq:vju}
|u_j(y)| \leq \|u_j\|_2=1
\end{equation}
Putting \eqref{eq:bju}, \eqref{eq:vju} and \eqref{eqpuabs} together we get 
\begin{align}
\label{eq:mt-1a}
 |m_x(y)|  &\leq   \sum_{j=1}^{k} |\beta_j| \cdot |u_j(y)| + \left(1-\frac{\lambda_{k+1}}{2}\right)^{t} \sum_{j=k+1}^{n} |\beta_j| \cdot | u_j(y)| \nonumber \\
 &\leq \sum_{j=1}^{k} |\beta_j| \cdot |u_j(y)| + n^{-2}\cdot n
\end{align}
Note that $G$ is $(k,\varphi,\epsilon)$-clusterable and $\min_{i}|C_i|\geq \Omega (\frac{n}{k})$. Therefore by Lemma \ref{lem:l-inf-bnd} for all $j\leq k$  we have 
\[\beta_j= u_j(x) \leq \|u_j\|_\infty  \leq O\left(\sqrt{k}\cdot n^{-1/2+(20\epsilon /\varphi^2)}\right)\text{.}\] 
Moreover
\[ u_j(y) \leq \|u_j\|_\infty  \leq O\left(\sqrt{k}\cdot n^{-1/2+(20\epsilon /\varphi^2)}\right) \]
Thus, we get
\begin{equation}
\label{eq:b_j-small}
\sum_{j=1}^k |\beta_j| \cdot |u_j(y)| \leq  O\left(k\cdot k\cdot n^{-1+(40\epsilon /\varphi^2)}\right) \text{.}
\end{equation}
Therefore by \eqref{eq:mt-1a} and \eqref{eq:b_j-small} we get 
\begin{align}
\label{eq:M-bnd}
 |m_x(y)| &\leq  O\left(k^2\cdot n^{-1+(40\epsilon /\varphi^2)}\right)  + n^{-1} \nonumber \\
 &\leq O\left(k^2\cdot n^{-1+(40\epsilon /\varphi^2)}\right) \text{.}  
\end{align}
Therefore we have
\begin{align*}
\|m_x\|_r \leq 
\left(n\cdot  O\left(k^2\cdot n^{-1+(40\epsilon /\varphi^2)}\right)^r \right)^{1/r}
&= O\left(k^2\cdot n^{-1+1/r+(40\epsilon /\varphi^2)}\right)\text{.}
\end{align*}
\end{proof}


\begin{restatable}{lemma}{lempairwisecollision}\label{lem:pairwise-collision}
Let $k \geq 2$ be an integer, $\varphi \in (0,1)$ and $\epsilon\in (0,1)$. Let $G=(V,E)$ be a $d$-regular and $(k,\varphi,\epsilon)$-clusterable graph. Let $M$ be  the random walk transition matrix of $G$. Let  $\sigma_{\text{err}}>0$. Let $t$, $R_1$ and $R_2$ be integers.
Let $a,b\in V$. Suppose that we run $R_1$ random walks of length $t$ from vertex $a$ and $R_2$ random walks of length $t$ from vertex $b$.
For any $x\in V$, let $\m_a(x)$ (resp. $\m_b(x)$) be a random variable which denotes the fraction out of the $R_1$ (resp. $R_2$) random walks starting from $a$ (resp. $b$), which end in $x$.  Let $c>1$ be a large enough constant.
If 
\begin{align*}
\min(R_1, R_2) \geq  \frac{c\cdot  k^5\cdot n^{-2+(100\epsilon /\varphi^2)}}{ \sigma_\text{err}^2} 
\text{,  and } &
 R_1 R_2 \geq     \frac{ c\cdot  k^2\cdot n^{-1+(40\epsilon /\varphi^2)}}{ \sigma_\text{err}^2}
 \end{align*}
then with probability at least $0.99$ we have
\[|\m_a^T \m_b-(M^{t}\mathds{1}_{a})^{T}(M^{t}\mathds{1}_{b})|\leq {\sigma_{\text{err}}} \text{.}\]
\end{restatable}
\begin{remark}
The success probability of Lemma \ref{lem:pairwise-collision} can be boosted up to $1-n^{-100}$ using
standard techniques (taking the median of $O(\log n)$  independent runs).
\end{remark}

\begin{proof}
Let $m_a=M^{t}\mathds{1}_{a}$ and $m_b=M^{t}\mathds{1}_{b}$. Let $X_{a,r}^i$ be a random variable which is $1$ if the $r^{\text{th}}$ random walk starting from $a$, ends at vertex $i$, and $0$ otherwise. Let $Y_{b,r}^i$ be a random variable which is $1$ if the $r^{\text{th}}$ random walk starting from $b$, ends at vertex $i$, and $0$ otherwise. Thus, $\mathbb{E}[X_{a,r}^i]={m_a(i)}$ and $\mathbb{E}[Y_{b,r}^i]={m_b(i)}$. 
For any two vertices $a,b\in S$, let  $Z_{a,b}=\m_a^T\m_b$ be a random variable given by
\[Z_{a,b}=\frac{1}{R_1 R_2} \sum_{i\in V} (\sum_{r_1=1}^{R_1} X_{a,r_1}^i)(\sum_{r_2=1}^{R_2} Y_{b,r_2}^i).\]
Thus, 
\begin{align} 
\mathbb{E}[Z_{a,b}] &=\frac{1}{R_1 R_2} \sum_{i\in V} (\sum_{r_1=1}^{R_1} \mathbb{E}[X_{a,r_1}^i])(\sum_{r_2=1}^{R_2} \mathbb{E}[Y_{b,r_2}^i])  \nonumber \\
&=\frac{1}{R_1 R_2} \sum_{i\in V} \left(R_1\cdot {m_a(i)} \right)\left(  R_2\cdot {m_b(i)}   \right)\nonumber \\
&= \sum_{i\in V} {m_a(i)}\cdot {m_b(i)} = (m_a)^T (m_b)\text{.}  \label{eq:expz}
\end{align} 
We know that $\text{Var}(Z_{a,b}) = \mathbb{E}[Z^2_{a,b}]-\mathbb{E}[Z_{a,b}]^2$. Let us first compute $\mathbb{E}[Z^2_{a,b}]$. 
\begin{align*}
 \mathbb{E}[Z^2_{a,b}] &= \mathbb{E}\left[\frac{1}{(R_1 R_2)^2} \sum_{i\in V} \sum_{j\in V} \sum_{r_1=1}^{R_1} \sum_{r_2=1}^{R_2} \sum_{r^{\prime}_1=1}^{R_1} \sum_{r^{\prime}_2=1}^{R_2} X_{a,r_1}^i  Y_{b,r_2}^i X_{a,r^{\prime}_1}^j  Y_{b,r^{\prime}_2}^j\right] \\
 &= \frac{1}{(R_1 R_2)^2} \sum_{i\in V} \sum_{j\in V} \sum_{r_1=1}^{R_1} \sum_{r_2=1}^{R_2} \sum_{r^{\prime}_1=1}^{R_1} \sum_{r^{\prime}_2=1}^{R_2} \mathbb{E}[X_{a,r_1}^i  Y_{b,r_2}^i X_{a,r^{\prime}_1}^j  Y_{b,r^{\prime}_2}^j]  
\end{align*}
To compute $\mathbb{E}[X_{a,r_1}^i  Y_{b,r_2}^i X_{a,r^{\prime}_1}^j  Y_{b,r^{\prime}_2}^j] $, we need to consider the following cases.
\begin{enumerate}
\item $i\neq j $: 
$\mathbb{E}[X_{a,r_1}^i  Y_{b,r_2}^i X_{a,r^{\prime}_1}^j  Y_{b,r^{\prime}_2}^j] \leq 
{m_a(i)} \cdot 
{m_b(i)} \cdot 
{m_a(j)} \cdot 
{m_b(j)} $. 
(This is an equality if $r_1\neq r^{\prime}_1$ and $r_2\neq r^{\prime}_2$. Otherwise, the expectation is zero.)
\item $i=j, \quad r_1=r_1^\prime, \quad r_2=r_2^\prime$: 
$\mathbb{E}[X_{a,r_1}^i  Y_{b,r_2}^i X_{a,r^{\prime}_1}^j  Y_{b,r^{\prime}_2}^j] = 
{m_a(i)} \cdot 
{m_b(i)} 
 $.
\item $i=j, \quad r_1=r_1^\prime, \quad r_2 \neq r_2^\prime$: 
$\mathbb{E}[X_{a,r_1}^i  Y_{b,r_2}^i X_{a,r^{\prime}_1}^j  Y_{b,r^{\prime}_2}^j] = 
{m_a(i)} \cdot 
{m_b(i)} \cdot  \cdot 
{m_b(i)}$.
\item $i=j, \quad r_1 \neq r_1^\prime, \quad r_2 = r_2^\prime$: 
$\mathbb{E}[X_{a,r_1}^i  Y_{b,r_2}^i X_{a,r^{\prime}_1}^j  Y_{b,r^{\prime}_2}^j] = 
{m_a(i)} \cdot 
{m_b(i)} \cdot 
{m_a(i)} 
$.

\item $i=j, \quad r_1 \neq r_1^\prime, \quad r_2 \neq r_2^\prime$: 
$\mathbb{E}[X_{a,r_1}^i  Y_{b,r_2}^i X_{a,r^{\prime}_1}^j  Y_{b,r^{\prime}_2}^j] = 
{m_a(i)} \cdot 
{m_b(i)} \cdot 
{m_a(i)} \cdot 
{m_b(i)}$.

\end{enumerate}
Thus we have,
\begin{align*}
 \mathbb{E}[Z^2_{a,b}] &= \frac{1}{(R_1 R_2)^2} \sum_{i\in V} \sum_{j\in V} \sum_{r_1=1}^{R_1} \sum_{r_2=1}^{R_2} \sum_{r^{\prime}_1=1}^{R_1} \sum_{r^{\prime}_2=1}^{R_2} \mathbb{E}[X_{a,r_1}^i  Y_{b,r_2}^i X_{a,r^{\prime}_1}^j  Y_{b,r^{\prime}_2}^j]   \\
 &\leq \sum_{i\in V} \sum_{j\in V\setminus\{i\}} 
{m_a(i) \cdot m_a(j) \cdot m_b(i) \cdot m_b(j)} + \sum_{i\in V} {{m_a(i)}^2\cdot {m_b(i)}^2} \\
&+ \frac{1}{R_1 R_2} \sum_{i\in V} {m_a(i)\cdot m_b(i)} + \frac{1}{R_1} \sum_{i\in V} {m_a(i)\cdot {m_b(i)}^2} + \frac{1}{R_2} \sum_{i\in V} { m_a(i)^2\cdot m_b(i)}  \\
&= \sum_{i,j\in V} {m_a(i) \cdot m_a(j) \cdot m_b(i) \cdot m_b(j)} + \frac{1}{R_1 R_2} \sum_{i\in V} {m_a(i)\cdot m_b(i)} \\
 &+ \frac{1}{R_1} \sum_{i\in V} {m_a(i)\cdot {m_b(i)}^2} + \frac{1}{R_2} \sum_{i\in V} { m_a(i)^2\cdot m_b(i)}.
\end{align*}
Therefore we get, 
\begin{align}
\text{Var}(Z_{a,b}) & = \mathbb{E}[Z^2_{a,b}]-\mathbb{E}[Z_{a,b}]^2 \nonumber \\
& \leq \sum_{i,j\in V} {m_a(i) \cdot m_a(j) \cdot m_b(i) \cdot m_b(j)} + \frac{1}{R_1 R_2} \sum_{i\in V} {m_a(i)\cdot m_b(i)} \\
 &+ \frac{1}{R_1} \sum_{i\in V} {m_a(i)\cdot {m_b(i)}^2} + \frac{1}{R_2} \sum_{i\in V} { m_a(i)^2\cdot m_b(i)} - \left(\sum_{i\in V} {m_a(i)\cdot m_b(i)}\right)^2 \nonumber\\
 &=   \frac{1}{R_1 R_2} \sum_{i\in V} {m_a(i)\cdot m_b(i)} + \frac{1}{R_1} \sum_{i\in V} {m_a(i)\cdot {m_b(i)}^2} + \frac{1}{R_2} \sum_{i\in V} { m_a(i)^2\cdot m_b(i)} \nonumber\\
&\leq  \frac{1}{R_1 R_2} \|m_a\|_2 \|m_b\|_2 + \frac{1}{R_1 } \|m_a\|_2 \|m_b\|_4^2 + \frac{1}{R_2} \|m_a\|_4^2 \|m_b\|_2 && \text{By Cauchy-Schwarz} \nonumber
 \label{eq:varz}
 \end{align}
Since $G=(V,E)$ is $(k,\varphi,\epsilon)$ clusterable by Lemma \ref{lem:Mt-bndr} we have
\[\|m_a\|_4 \leq O\left(k^2\cdot n^{-3/4+(40\epsilon /\varphi^2)}\right)  \text{.}\]
and by Lemma \ref{lem:Mt-bnd} we have
\[\|m_a\|_2 \leq O(k\cdot n^{-1/2+(20\epsilon /\varphi^2)}) \text{.}\]

Thus we get
\begin{align}
\text{Var}(Z_{a,b}) \leq O\left(\frac{k^2\cdot n^{-1+(40\epsilon /\varphi^2)}}{R_1 R_2}  + \left(\frac{1}{R_1 }+\frac{1}{R_2 } \right) \cdot k^5\cdot n^{-2+(100\epsilon /\varphi^2)}\right) 
\end{align}
Then by Chebyshev's inequality, we get,
\begin{align}
\Pr\left[|Z_{a,b}-\mathbb{E}[Z_{a,b}]|>{{\sigma_\text{err}}}\right]&\leq 
\frac{\text{Var}[Z_{a,b}]} {{{\sigma_\text{err}}}^2}\nonumber \\
&\leq O\left( \frac{1}{{\sigma_\text{err}}^2}\cdot \left(\frac{k^2\cdot n^{-1+(40\epsilon /\varphi^2)}}{R_1 R_2}  + \left(\frac{1}{R_1 }+\frac{1}{R_2 } \right) \cdot k^5\cdot n^{-2+(100\epsilon /\varphi^2)}\right)\right)  \label{eq:const-o}\\
&\leq  \frac{1}{100} \text{.} \nonumber
\end{align}
The last inequality holds by our choice of $R_1$ and $R_2$ as follows where $c$ is a large enough constant that cancels the constant hidden in $O\left(\cdot\right)$ in \eqref{eq:const-o}.
\[
\min(R_1, R_2) \geq  \frac{c\cdot  k^5\cdot n^{-2+(100\epsilon /\varphi^2)}}{ \sigma_\text{err}^2} \]
and
\[
 R_1 R_2 \geq     \frac{ c\cdot  k^2\cdot n^{-1+(40\epsilon /\varphi^2)}}{ \sigma_\text{err}^2}
\]
\end{proof}


\begin{restatable}{lemma}{lemcol}\label{lem:collision}
Let $k \geq 2$ be an integer, $\varphi \in (0,1)$ and $\epsilon\in (0,1)$. Let $G=(V,E)$ be a $d$-regular and $(k,\varphi,\epsilon)$-clusterable graph. Let $\sigma_{\text{err}}>0$ and let $s>0$, $R>0$, $t>0$ be integers. Let $I_S=\{i_1,\ldots, i_s\}$ be a multiset of $s$ indices chosen from
$\{1,\dots,n\}$. Let $S$ be the $n\times s$ matrix whose $j$-th column equals $\mathds{1}_{i_j}$. Let $c>1$ be a large enough constant. Let $R\geq \max\left\{\frac{c\cdot k^{5}\cdot n^{-2+100\epsilon/\varphi^2}}{\sigma_{\text{err}}^2} , \frac{c\cdot  k\cdot n^{-1/2+20\epsilon/\varphi^2}}{\sigma_{\text{err}}}\right\}$
  Let $\G\in \R^{s\times s}$ be the output of Algorithm \textsc{EstimateCollisionProbabilities($G,I_S,R,t$)} (Algorithm \ref{alg:gram}). Let $M$ be  the random walk transition matrix of $G$. 
then with probability at least $1-n^{-100}$ we have
\[\|\G-(M^tS)^{T}(M^tS)\|_2\leq s\cdot \sigma_{\text{err}}\text{.}\]
\end{restatable}
\begin{proof}
Note that as per line \eqref{ln:Qi} and \eqref{ln:Pi} of Algorithm \ref{alg:gram} we first construct matrices $\widehat P_i \in \R^{n\times s}$ and $\Q_i \in \R^{n\times s}$ using Algorithm \ref{alg:compQ}. as per line \eqref{ln:defQ} of Algorithm \ref{alg:compQ} matrix $\widehat P_i$ (or $\Q_i$) has $s$ columns each corresponds to a vertex $x\in S$. The column corresponding to vertex $x$ is $\m_x$. as per line \ref{ln:ranwalk22} of Algorithm \ref{alg:compQ},  $\m_x$ is defined as the empirical probability distribution of running $R$ random walks of length $t$ starting from vertex $x$. Thus for any $x,y\in S$ we have the entry corresponding to the $x^{\text{th}}$ row and $y^{\text{th}}$ column of $\Q_i^T \widehat P_i$ (or $\widehat P_i^T\Q_i$) is $\langle \m_x , \m_y \rangle$. Since 
\[R\geq \max\left\{\frac{c\cdot k^{5}\cdot n^{-2+100\epsilon/\varphi^2}}{\sigma_{\text{err}}^2} , \frac{c\cdot  k\cdot n^{-1/2+20\epsilon/\varphi^2}}{\sigma_{\text{err}}}\right\}\]
 then by Lemma \ref{lem:pairwise-collision} with probability at least $0.99$ we have 
\[|\m_x^T \m_y-(M^{t}\mathds{1}_{x})^{T}(M^{t}\mathds{1}_{y})|\leq \sigma_{\text{err}} \text{.}\] 
Note that as per line \ref{ln:Gi} of Algorithm \ref{alg:gram} we define
$\G_i:=\frac{1}{2}\left(\widehat{P}_i^T \Q_i+\Q_i^T \widehat{P}_i\right)$. Thus for any $x,y\in I_S$ we have the entry corresponding to the $x^{\text{th}}$ row and $y^{\text{th}}$ column of $\G_i$ (i.e., $\G_i(x,y)$) with probability $0.99$ satisfies the following:
\[|\G_i(x,y)-(M^{t}\mathds{1}_{x})^{T}(M^{t}\mathds{1}_{y})|\leq \sigma_{\text{err}} \text{.}\] 
Note that as Line \ref{ln:GG} of Algorithm \ref{alg:gram} we define $\G$ as a matrix obtained by taking the entrywises median of $\mathcal{G}_i$'s over $O(\log n)$ runs. Thus with probability at least $1-n^{-100}$ we have for all $x,y\in I_S$
\[|\G(x,y)-(M^{t}\mathds{1}_{x})^{T}(M^{t}\mathds{1}_{y})|\leq \sigma_{\text{err}} \text{.}\] 
which implies
\[\|\G-(M^tS)^{T}(M^tS)^T\|_F\leq s\cdot {\sigma_\text{err}} \text{.}\]
Since the Frobenius norm of a matrix bounds its maximum eigenvalue from above we get
 \[\|\G-(M^tS)^{T}(M^tS)^T\|_2 \leq s\cdot{\sigma_\text{err}} \text{.}\]
\end{proof}
Recall that for a symmetric matrix $A$, we write $\nu_i(A)$ (resp. $\nu_{\max}(A), \nu_{\min}(A))$ to denote the $i^{\text{th}}$ largest (resp. maximum, minimum) eigenvalue of $A$.
\begin{lemma}
\label{lem:eigMS}
Let $k \geq 2$ be an integer, $\varphi \in (0,1)$ and $\epsilon\in (0,1)$. Let $G=(V,E)$ be a $d$-regular and $(k,\varphi,\epsilon)$-clusterable graph. Let  $t\geq  \frac{20\log n}{\varphi^2}$. Let $c>1$ be a large enough constant and $s\geq
c\cdot n^{240\cdot\epsilon / \varphi^2} \cdot \log n \cdot k^{4}$. Let $I_S=\{i_1,\ldots, i_s\}$ be a multiset of $s$ indices chosen independently and uniformly at random from
$\{1,\dots,n\}$. Let $S$ be the $n\times s$ matrix whose $j$-th column equals $\mathds{1}_{i_j}$.   Let $M$ be  the random walk transition matrix of $G$.
If $\frac{\epsilon}{\varphi^2}\leq \frac{1}{10^5}$ then with probability at least $1-n^{-100}$ we have
\begin{enumerate}
\item $\nu_{k}\left(\frac{n}{s} \cdot (M^tS)(M^tS)^T \right)\geq  \frac{ n^{-80\epsilon/\varphi^2}}{2}$ \label{itm1:nu_k_mS}
\item $\nu_{k+1}\left(\frac{n}{s} \cdot (M^tS)(M^tS)^T \right) \leq n^{-9} \text{.}$ \label{itm2:nu_k_mS}
\end{enumerate}
\end{lemma}
\begin{proof}
Let $(u_1,\ldots,u_n)$ be an orthonormal basis of eigenvectors of $L$ with corresponding eigenvalues $0 \leq \lambda_1\leq \ldots \leq \lambda_n$. 
Observe that each $u_i$ is also an eigenvector of $M$, with eigenvalue $1-\frac{\lambda_i}{2}$.
Note that $G$ is $(k,\varphi,\epsilon)$-clusterable, therefore by Lemma \ref{lem:bnd-lambda} we have $\lambda_k\leq 2\epsilon$ and $\lambda_{k+1} \geq \frac{\varphi^2}{2}$. We have
\begin{equation}
\label{mu_k+1-M}
\nu_{k+1}(M^{2t})=\left(1-\frac{\lambda_{k+1}}{2} \right)^{2t}\leq n^{-10} \text{, and}
\end{equation}
\begin{equation}
\label{mu_k-M}
\nu_{k}(M^{2t})=\left(1-\frac{\lambda_{k}}{2} \right)^{2t}\geq n^{-80\epsilon/\varphi^2}
\end{equation}
\textbf{Proof of item \eqref{itm1:nu_k_mS}:}
Let $A= (M^{t}) \left(M^{t}\right)^T$, and $\widetilde{A}=\frac{n}{s} \cdot (M^{t}S) \left(M^{t}S\right)^T$. By Lemma \ref{lem:Mt-bnd} we have 
\[B=\|(M^{t}\mathds{1}_x)(M^{t}\mathds{1}_x)^T\|_2 \le \|M^{t}\mathds{x} \mathds{1}_x\|_2^2 \le O\left(k^2 \cdot n^{-1+40 \e/\varphi^2)}\right)\text{.}\] 
Let $\xi = n^{-80\e/\varphi^2}/2$.
Therefore for large enough constant $c$ and by choice of $s = c \cdot k^4 n^{240\e/\varphi^2}\log n$ we have $s \ge \frac{ 40n^2 B^2 \log n}{(\xi)^2}$. Thus  Lemma \ref{lem:G-Ghat}
yields that with probability at least $1-\frac{1}{n^{100}}$ we have
\begin{equation}
\|A-\widetilde{A}\|_2\leq  \frac{n^{-80\e/\varphi^2}}{2}  \text{.}
\end{equation}
Hence, by Weyl's Inequality (see Lemma~\ref{lem_Weyl}) we have
$$
\nu_{k}(\widetilde{A}) \geq \nu_{k}(A)+\nu_{\min}(\widetilde{A}-A) =\nu_{k}(A)-\nu_{\max}(A-\widetilde{A}) =\nu_{k}(A)-\|A-\widetilde{A}\|_2
$$
By \eqref{mu_k-M} we have $\nu_{k}(A)=\nu_{k}(M^{2t})\geq n^{-10\epsilon/\varphi^2}$ and so
\[
\nu_{k}(\widetilde{A})\geq \nu_{k}(A)-\|\widetilde{A}-A\|_2 \geq  n^{-80\epsilon/\varphi^2}-  \frac{n^{-80\e/\varphi^2}}{2}  \geq \frac{n^{-80\e/\varphi^2}}{2}  \text{.}
\]
\textbf{Proof of item \eqref{itm2:nu_k_mS}:} By Lemma \ref{lem_commute} we have
\begin{align*}
\nu_{k+1}(\widetilde{A})=\frac{n}{s}\cdot \nu_{k+1}((M^tS)(M^tS)^T) = \frac{n}{s}\cdot \nu_{k+1}((M^tS)^T(M^tS)) 
= \frac{n}{s}\cdot \nu_{k+1}(S^{T}{M}^{2t}S).
\end{align*}

Recall that $1-\frac{\lambda_1}{2}\geq \cdots\geq 1-\frac{\lambda_n}{2}$ are the eigenvalues of $M$, and $\Sigma$ is the diagonal matrix of these eigenvalues in descending order, and $U$ is the matrix whose columns are orthonormal eigenvectors of $M$ arranged in descending order of their eigenvalues. We have $M^{2t}=U\Sigma^{2t}U^{T}$. Recall that $\Sigma_{[k]}$ is $k \times k$ diagonal matrix with entries  $1-\frac{\lambda_1}{2}\geq \cdots\geq 1-\frac{\lambda_k}{2}$, and $\Sigma_{-[k]}$ is a  $(n-k) \times (n-k)$ diagonal matrix with  entries  $1-\frac{\lambda_{k+1}}{2}\geq \cdots\geq 1-\frac{\lambda_n}{2}$. We can write $U\Sigma^{2t}U=U_{[k]}\Sigma_{[k]}^{2t}U_{[k]}^T+U_{-[k]}\Sigma_{-[k]}^{2t}U^T_{-[k]}$, thus we get
\begin{align*}
\nu_{k+1}(\widetilde{A})&=\frac{n}{s}\cdot \nu_{k+1}\left(S^{T}M^{2t} S\right) \\
&=\frac{n}{s}\cdot\nu_{k+1}\left(S^{T}(U\Sigma^{2t}U^{T})S\right) \\
&=\frac{n}{s}\cdot\nu_{k+1}\left(S^{T}\left(U_{[k]}\Sigma_{[k]}^{2t}U_{[k]}^T+U_{-[k]}\Sigma_{-[k]}^{2t}U^T_{-[k]}\right)S\right) \\
&\leq \frac{n}{s}\cdot\nu_{k+1}\left(S^{T}U_{[k]}\Sigma_{[k]}^{2t}U_{[k]}^TS\right) + \frac{n}{s}\cdot\nu_{\max}\left(S^{T}U_{-[k]}\Sigma_{-[k]}^{2t}U^T_{-[k]}S\right)  && {\text{By Weyl's inequality (Lemma \ref{lem_Weyl})}}
\end{align*}
Here $\nu_{k+1}(S^{T}U_{[k]}\Sigma_{[k]}^{2t}U_{[k]}^{T}S) =0$, because the rank of $\Sigma_{[k]}^{2t}$ is $k$. We then need to bound $\nu_{\max}(S^{T}U_{-[k]}\Sigma_{-[k]}^{2t}U^T_{-[k]} S) $. We have,
\begin{align*}
&\nu_{\max}\left(S^{T}U_{-[k]}\Sigma_{-[k]}^{2t}U^T_{-[k]}S\right)  =\nu_{\max}\left(U_{-[k]}\Sigma_{-[k]}^{2t}U^T_{-[k]}SS^\top\right) && {\text{By Lemma \ref{lem_commute}}}\\
&\leq \nu_{\max}\left(U_{-[k]}\Sigma_{-[k]}^{2t}U^T_{-[k]}\right) \cdot \nu_{\max}\left(SS^\top\right) && {\text{By submultiplicativity of norm}} \\
&=\nu_{\max}\left(\Sigma_{-[k]}^{2t}U^T_{-[k]}U_{-[k]}\right) \cdot \nu_{\max}\left(SS^\top\right)  && {\text{By Lemma \ref{lem_commute}}} \\
&=\nu_{\max}\left(\Sigma_{-[k]}^{2t}\right) \cdot \nu_{\max}\left(SS^\top\right)  && {\text{Since } U^T_{-[k]}U_{-[k]}=I} \\
\end{align*}
Next, observe that $SS^{\top}\in \R^{n\times n}$ is a diagonal matrix whose $(a,a)^{\text{\tiny{th}}}$ entry is the multiplicity of vertex $a$ is sampled in $S$. Thus, $\nu_{\max}(SS^{\top})$ is the maximum multiplicity over all vertices, which is at most $s$. Also note that $\nu_{\max}(\Sigma_{-[k]}^{2t}) = \left(1-\frac{\lambda_{k+1}}{2}\right)^{2t}$. Thus by \eqref{mu_k+1-M} we get,
\[\nu_{k+1}(\widetilde{A}) \leq \frac{n}{s}\cdot \nu_{\max}\left(S^{T}U_{-[k]}\Sigma_{-[k]}^{2t}U^T_{-[k]}S\right) \leq \frac{n}{s} \cdot s\cdot \left(1-\frac{\lambda_{k+1}}{2}\right)^{2t} \leq n\cdot n^{-10}=n^{-9}\text{.}\]

\end{proof}
Now we are ready to prove the main result of this section (Lemma \ref{lem-u-close}).
\lemuclose*
\begin{proof}
Let $\widetilde{A}=\frac{n}{s} \cdot (M^{t}S)^T \left(M^{t}S\right)=\widetilde{W}\widetilde{\Sigma^2}\widetilde{W}^T$ and $\widehat{A}=\frac{n}{s}\cdot\G$. Thus we have
\[\widetilde{A}^2= \left(\frac{n}{s} \cdot (M^{t}S)^T \left(M^{t}S\right) \right)^2 =\widetilde{W}\widetilde{\Sigma^4}\widetilde{W}^T \]
and
\[\widehat{A}^2=\left(\frac{n}{s}\cdot\G\right)^2= \widehat{W} \widehat{\Sigma^2} \widehat{W}^T \text{.}\]
Recall that for a symmetric matrix $A$, we write $\nu_i(A)$  to denote the $i^{\text{th}}$ largest  eigenvalue of $A$. We want to apply Lemma \ref{lem:neg-l2-close} to get
\begin{align*}
\left|\left|\widetilde{W}_{[k]} \widetilde{\Sigma}^{-4}_{[k]} \widetilde{W}^T_{[k]} - \widehat{W}_{[k]}\widehat{\Sigma}_{[k]}^{-2}\widehat{W}_{[k]}^T\right|\right|_2 &\leq  \frac{16 \cdot \|\widetilde{A}^2-\widehat{A}^2\|_2+4\cdot\nu_{k+1}(\widetilde{A}^2)}{\nu_k(\widetilde{A}^2)^2} 
\end{align*}
Hence, we first need to verify the prerequisites of Lemma \ref{lem:neg-l2-close}. Let $c_3>1$ be a large enough constant that we will define soon, and let $\sigma_{\text{err}}=  \frac{ \xi \cdot n^{(-1-360\cdot \epsilon/\varphi^2)}}{c_3\cdot k^2}$. Let $c$ be a constant from Lemma \ref{lem:collision}. By the assumption of the lemma for  large enough constant $c_2>1$ we have
\[R\geq \frac{ c_2 \cdot k^{9}\cdot n^{1/2+820\cdot\epsilon/\varphi^2}}{\xi^2}  \geq \max\left\{ \frac{c 
\cdot k^{5}\cdot n^{-2+100\epsilon/\varphi^2}}{\sigma_{\text{err}}^2} ,  \frac{c\cdot k\cdot n^{-1/2+20\epsilon/\varphi^2}}{\sigma_{\text{err}}} \right\}\text{.}\] 
Thus we can apply Lemma \ref{lem:collision}. Hence, with probability at least $1-n^{-100}$ we have 
\begin{equation}
\label{eq:q-mts}
\|\G-(M^{t}S)^{T}(M^{t}S)\|_2\leq s\cdot\sigma_{\text{err}}\text{.}
\end{equation}
Therefore we have
\begin{align}
\|\G^2-\left((M^{t}S)^{T}(M^{t}S)\right)^2\|_2 &= \|\G \left(\G-(M^{t}S)^{T}(M^{t}S)\right) + \left(\G-(M^{t}S)^{T}(M^{t}S)\right) (M^{t}S)^{T}(M^{t}S)\|_2 \nonumber\\
&\leq \|\G-(M^{t}S)^{T}(M^{t}S)\|_2 \left( \|\G\|_2 +  \|(M^{t}S)^{T}(M^{t}S)\|_2 \right) \nonumber \\
&\leq s\cdot\sigma_{\text{err}} \left( (s\cdot\sigma_{\text{err}} + \|(M^{t}S)^{T}(M^{t}S)\|_2) + \|(M^{t}S)^{T}(M^{t}S)\|_2 \right) \nonumber\\
&=(s\cdot\sigma_{\text{err}})^2 + 2
\cdot s\cdot\sigma_{\text{err}} \|(M^{t}S)^{T}(M^{t}S)\|_2 \label{eq:gap-gram}
\end{align}
Note that
\begin{align}
\|(M^{t}S)^{T}(M^{t}S)\|_2&\leq \|(M^{t}S)^{T}(M^{t}S)\|_F \nonumber\\
&= \sqrt{\sum_{x,y\in S} \left((M^{t}\mathds{1}_x)^T (M^{t}\mathds{1}_y) \right)^2} \nonumber\\
&\leq \sqrt{\sum_{x,y\in S} \|M^{t}\mathds{1}_x\|_2^2 \|M^{t}\mathds{1}_y\|_2^2 } && \text{By Cauchy Schwarz} \nonumber\\
&\leq O\left(\sqrt{s^2\cdot \left(k^2\cdot n^{-1+(40\epsilon /\varphi^2)}\right)^2}\right) && \text{By Lemma \ref{lem:Mt-bnd}} \nonumber\\
&=O\left(s\cdot k^2\cdot n^{-1+(40\epsilon /\varphi^2)}\right) && \text{.} \label{eq:MtsForb}
\end{align}
Puuting \eqref{eq:MtsForb} and \eqref{eq:gap-gram} and by choice of $\sigma_{\text{err}}= \frac{ \xi \cdot n^{(-1-360\cdot \epsilon/\varphi^2)}}{c_3\cdot k^2}$ we get
\begin{equation}
\label{eq:A2aTIL2}
\|\widetilde{A}^2-\widehat{A}^2 \|_2=  \left(\frac{n}{s}\right)^2\|\G^2-\left((M^{t}S)^{T}(M^{t}S)\right)^2\|_2 \leq O\left(\frac{\xi^2\cdot n^{-720\cdot\epsilon/\varphi^2}}{(c_3)^2\cdot k^4}+ \frac{ \xi\cdot n^{-320\epsilon/\varphi^2}}{c_3}\right) = O\left(\frac{ \xi\cdot n^{-320\epsilon/\varphi^2}}{c_3}\right)
\end{equation}
By Lemma \ref{lem_commute} for any $i\in [s]$ we have
\[\nu_i (\widetilde{A})=\nu_i\left(\frac{n}{s} \cdot (M^{t}S) \left(M^{t}S\right)^T \right) = \nu_i\left(\frac{n}{s} \cdot (M^{t}S)^T \left(M^{t}S\right) \right)\]
Let $c_1$ be the constant from Lemma \ref{lem:eigMS}. Since $s\geq c_1 \cdot n^{240\epsilon / \varphi^2}\cdot \log n \cdot k^{4}$ therefore by Lemma \ref{lem:eigMS}  with probability at least $1-n^{-100}$ we have 
\begin{equation}
\label{eq:mu-k-gtild}
\nu_{k}{\left(\widetilde{A}^2\right)}=\nu_{k}{\left(\left(\frac{n}{s} \cdot (M^{t}S)^T \left(M^{t}S\right) \right)^2\right)} \geq \left(\frac{n^{-80\epsilon/\varphi^2}}{2} \right)^2 \geq \frac{n^{-160\epsilon/\varphi^2}}{4}
\end{equation}
and
\begin{equation}
\label{eq:mu-k+1-gtild}
\nu_{k+1}{\left(\widetilde{A}^2\right)}=\nu_{k+1}{\left(\left(\frac{n}{s} \cdot (M^{t}S)^T \left(M^{t}S\right) \right)^2\right)}\leq  n^{-18 }
\end{equation}
By the bound on the $\nu_{k}{(\widetilde{A}^2)}$ and the  inequality on $\|\widetilde{A}^2-\widehat{A}^2 \|_2$, we know that $\nu_{k}{(\widehat{A}^2)}$ is non-zero and so $\widehat \Sigma_{[k]}^{-2}$ exist.
Recall that $\widetilde{A}= \widetilde{W}\widetilde{\Sigma}^{2} \widetilde{W}^T \text{.}$ Observing that $\widetilde{A}$ is positive semi-definite, $\nu_{k+1}(\widetilde{A}^2)< \nu_{k}(\widetilde{A}^2)/4$,
and $\|\widetilde{A}^2-\widehat{A}^2 \|_2  \le \frac{1}{100}\cdot \nu_k(\widetilde{A}^2) $ we can apply Lemma \ref{lem:neg-l2-close} and we get
\begin{align*}
\left|\left|\widetilde{W}_{[k]} \widetilde{\Sigma}^{-4}_{[k]} \widetilde{W}^T_{[k]} - \widehat{W}_{[k]}\widehat{\Sigma}_{[k]}^{-2}\widehat{W}_{[k]}^T\right|\right|_2 &\leq  \frac{16 \cdot \|\widetilde{A}^2-\widehat{A}^2\|_2+4\cdot\nu_{k+1}(\widetilde{A}^2)}{\nu_k(\widetilde{A}^2)^2} \\
&\leq
\frac{ O\left(\frac{\xi\cdot  n^{(-320\epsilon/\varphi^2)}}{c_3}  \right)+ 4\cdot  n^{-18}}{\frac{1}{16}\cdot n^{(-320\epsilon/\varphi^2)}} && \text{By \eqref{eq:A2aTIL2} and \eqref{eq:mu-k-gtild} }\\
&\leq O\left(\frac{\xi}{c_3}\right)+64\cdot n^{-17}\\
&\leq \xi
\end{align*}
The last inequality holds since $\xi\geq n^{-8}$ and by setting $c_3$ to a large enough constant to cancel the constant hidden in $O\left(\frac{\xi}{{c_3}} \right)$.
\end{proof}

\subsection{Proof of Theorem \ref{thm:dot}}
\thmdot*
To prove Theorem \ref{thm:dot} we need to combine Lemma \ref{lem:bnd-e1} from Section \ref{subsubsec:cols} with the following lemma.


\begin{restatable}{lemma}{lemuabsclose}\label{lem:u-abs-close}
Let $G=(V,E)$ be a $d$-regular and $(k,\varphi,\epsilon)$-clusterable graph. Let $0<\delta<1/2$, and $1/n^6 < \xi < 1$.
Let $\mathcal{D}$ denote the  data structure constructed by  Algorithm \textsc{InitializeOracle($G,\delta,\xi$)} (Algorithm \ref{alg:LearnEmbedding}). 
Let $x,y\in V$. Let $\adp{f_x,f_y}\in \R$ denote the value returned by 
$\textsc{SpectralDotProductOracle}(G,x,y, \delta, \xi, \mathcal{D})$ (Algorithm \ref{alg:dotProduct}).
Let  $t\geq  \frac{20\log n}{\varphi^2}$. Let $c>1$ be a large enough constant and let $s\geq c\cdot n^{240\cdot\epsilon / \varphi^2}\cdot \log n \cdot k^{4}$. Let $I_S=\{i_1,\ldots, i_s\}$ be a multiset of $s$ indices chosen independently and uniformly at random from
$\{1,\dots,n\}$. Let $S$ be the $n\times s$ matrix whose $j$-th column equals $\mathds{1}_{i_j}$.   Let $M$ be  the random walk transition matrix of $G$. Let $\sqrt{\frac{n}{s}} \cdot M^tS=\widetilde{U}\widetilde{\Sigma}\widetilde{W}^T$  be  an SVD of $\sqrt{\frac{n}{s}} \cdot M^tS$ where $\widetilde{U}\in \R^{n\times n}, \widetilde{\Sigma}\in \R^{n\times n}, \widetilde{W}\in \R^{s\times n}$. If $\frac{\epsilon}{\varphi^2}\leq \frac{1}{10^5}$, and Algorithm \ref{alg:LearnEmbedding} succeeds,  then with probability at least $1-n^{-100}$ matrix $\widetilde \Sigma_{[k]}^{-4}$ exists and we have
\[\left|\adp{f_{x},f_y} - (M^{t}\mathds{1}_{x})^T  (M^tS)\left(\frac{n}{s}\cdot\widetilde{W}_{[k]} \widetilde{\Sigma}^{-4}_{[k]} \widetilde{W}^T_{[k]}\right) (M^tS)^T (M^{t}\mathds{1}_{y}) \right|<
\frac{\xi }{n} \text{.}
\]
\end{restatable}

\begin{proof}
Note that as per line \ref{ln:axay} of Algorithm \ref{alg:dotProduct}  $\adp{f_{x},f_y}$ is defined as 
\[\adp{f_{x},f_y}=\alpha_x^T\Psi\alpha_y \text{.}\] 
where as per line \ref{ln:Pi} of Algorithm \ref{alg:LearnEmbedding} we define matrix $\Psi\in\R^{s\times s}$ as 
\[\Psi=\frac{n}{s}\cdot\widehat{W}_{[k]}\widehat{\Sigma}_{[k]}^{-2} \widehat{W}_{[k]}^T\text{,}\] 
and $\alpha_x,\alpha_y\in\R^{s}$ are vectors obtained by taking entrywise median over all $(\Q_i)^T(\m^i_x)$ and $(\Q_i)^T(\m^i_y)$. (See line \ref{ln:alx} and \ref{ln:aly} of Algorithm \ref{alg:dotProduct}). 
For any vertex $a\in V$ recall that $m_a$ denote $m_a=M^{t}\mathds{1}_{a}$. We then define  
\[\mathbf{a}_x=m_x^T  (M^tS), \quad A=\frac{n}{s}\cdot\widetilde{W}_{[k]} \widetilde{\Sigma}^{-4}_{[k]} \widetilde{W}^T_{[k]}, \quad  \mathbf{a}_y=(M^tS)^T m_y \text{, and}\]
\[\mathbf{e}_x=\alpha_x^T -\mathbf{a}_x, \quad  E= \Psi-A, \quad  \mathbf{e}_y=\alpha_y-\mathbf{a}_y\]
Thus by triangle inequality we have
\begin{align*}
&\left|\left| \alpha_x^T\Psi\alpha_y -  m_x^T  (M^tS)\left(\frac{n}{s}\cdot\widetilde{W}_{[k]} \widetilde{\Sigma}^{-4}_{[k]} \widetilde{W}^T_{[k]}\right) (M^tS)^T m_y \right|\right|_2  \\
&=\| \left(\mathbf{a}_x +\mathbf{e}_x\right) \left(A+E\right) \left(\mathbf{a}_y +\mathbf{e}_y\right) -\mathbf{a}_xA\mathbf{a}_y\|_2 \\
&\leq \|\mathbf{e}_x\|_2\|A\|_2\|\mathbf{a}_y\|_2+\|\mathbf{a}_x\|_2\|E\|_2\|\mathbf{a}_y\|_2+\|\mathbf{a}_x\|_2\|A\|_2\|\mathbf{e}_y\|_2  \\
&+\|\mathbf{e}_x\|_2\|E\|_2\|\mathbf{a}_y\|_2+\|\mathbf{a}_x\|_2\|E\|_2\|\mathbf{e}_y\|_2+\|\mathbf{e}_x\|_2\|A\|_2\|\mathbf{e}_y\|_2+\|\mathbf{e}_x\|_2\|E\|_2\|\mathbf{e}_y\|_2
\end{align*}
Therefore we need to bound $\|\mathbf{e}_x\|_2$, $\|\mathbf{e}_y\|_2$, $\|E\|_2$, $\|\mathbf{a}_x\|_2$, $\|\mathbf{a}_y\|_2$ and $\|A\|_2$. 
Let $c'>1$ be a constant we will define soon, and let $\xi' = \frac{\xi}{c'\cdot k^4 \cdot n^{80 \e/\varphi^2}}$. Let $c_1$ be a constant in front of $s$ and let $c_2$ be a constant in front of $R$ in Lemma \ref{lem-u-close}. Thus for large enough $c$ we have
$s  \geq c_1\cdot n^{240\epsilon / \varphi^2}\cdot \log n \cdot k^{4}$ and $R_\text{init}= \varTheta{(n^{1-\delta +980 \cdot\epsilon / \varphi^2} \cdot k^{17}/{\xi}^{2})}\geq \frac{c_2\cdot k^{9}\cdot n^{1/2+820\cdot\epsilon/\varphi^2}}{\xi'^2}$ as per line \ref{ln:setRR} of Algorithm \ref{alg:LearnEmbedding}, hence,  by Lemma \ref{lem-u-close} applied with $\xi'$ we have with probability at least $1-n^{-100}$, $\widehat{W}_{[k]}^T-$ and $\widetilde{\Sigma}^{-4}_{[k]}$ exist and we have
\begin{equation}
\label{eq:Eu}
\|E\|_2=\frac{n}{s}\cdot\left|\left| \widehat{W}_{[k]}\widehat{\Sigma}_{[k]}^{-2}\widehat{W}_{[k]}^T- \widetilde{W}_{[k]} \widetilde{\Sigma}^{-4}_{[k]} \widetilde{W}^T_{[k]} \right|\right|_2 \leq \frac{n}{s} \cdot  \xi'  = \frac{\xi\cdot n}{c'\cdot k^4 \cdot n^{80 \e/\varphi^2}\cdot s} \text{.}
\end{equation}
Recall that for a symmetric matrix $A$, we write $\nu_i(A)$ (resp. $\nu_{\max}(A), \nu_{\min}(A))$ to denote the $i^{\text{th}}$ largest (resp. maximum, minimum) eigenvalue of $A$. We have
\[\|A\|_2=\frac{n}{s}\cdot\|\widetilde{W}_{[k]} \widetilde{\Sigma}^{-4}_{[k]} \widetilde{W}^T_{[k]} \|_2 = \frac{n}{s}\cdot\nu_{\max}\left(\widetilde{W}_{[k]} \widetilde{\Sigma}^{-4}_{[k]} \widetilde{W}^T_{[k]} \right) = \frac{n}{s}\cdot\frac{1}{\nu_{k}\left(\widetilde{W}_{[k]} \widetilde{\Sigma}^{4}_{[k]} \widetilde{W}^T_{[k]} \right)}\]
Note that $\frac{n}{s}\cdot(M^{t}S)^{T}(M^{t}S)=\widetilde{W} \widetilde{\Sigma}^{2} \widetilde{W}^T$. Thus by Lemma \ref{lem:eigMS} item \eqref{itm1:nu_k_mS} we have
\[
\nu_{k}\left(\widetilde{W}_{[k]} \widetilde{\Sigma}^{4}_{[k]} \widetilde{W}^T_{[k]} \right)
=\nu_k\left(\widetilde{W} \widetilde{\Sigma}^{4} \widetilde{W}^T \right)  
= \nu_{k}\left(\left( \frac{n}{s}\cdot (M^{t}S)^{T}(M^{t}S) \right)^2\right)  
\geq \frac{ n^{-160\epsilon/\varphi^2}}{4}  \text{.}
\]
Therefore we have
\begin{equation}
\label{eq:Au}
\|A\|_2 \leq  4\cdot \frac{n}{s}\cdot  n^{160\epsilon/\varphi^2} = \frac{4\cdot n^{1+160\epsilon/\varphi^2}}{ s}\text{.}
\end{equation}
Since $G$ is $(k,\varphi,\epsilon)$-clusterable by Lemma \ref{lem:Mt-bnd} for any vertex $x\in V$ we have 
\begin{equation}
\label{eq:Mt1aa}
\|m_{x}\|^2_2 \leq O\left(k^2\cdot n^{-1+(40\epsilon /\varphi^2)}\right)  \text{.}
\end{equation}
Then we get
\begin{align}
\label{eq:Ax}
\|\mathbf{a}_x\|_2&=\|(m_{x})^T(M^tS)\|_2 \nonumber\\
&=\sqrt{\sum_{a\in I_S} \left((m_{x})^T(m_a) \right)^2} \nonumber\\
&\leq \sqrt{\sum_{a\in I_S} \|m_{x}\|^2_2\|m_a\|_2^2} && \text{By Cauchy Schwarz}\nonumber\\
&\leq O\left(\sqrt{s\cdot \left( k^2\cdot n^{-1+(40\epsilon /\varphi^2)} \right)^2} \right)&&\text{By \eqref{eq:Mt1aa}} \nonumber\\
&=O\left(\sqrt{s}\cdot  k^2\cdot n^{-1+(40\epsilon /\varphi^2)}  \right)
\end{align}
By the same analysis we get 
\begin{equation}
\label{eq:Ay}
\|\mathbf{a}_y\|_2\leq O\left(\sqrt{s}\cdot  k^2\cdot n^{-1+(40\epsilon /\varphi^2)}  \right)
\end{equation}
Now we left to bound $\|\mathbf{e}_x\|_2$ and $\|\mathbf{e}_y\|_2$. Recall that $\mathbf{e}_x=\alpha_x- (M^{t}\mathds{1}_{x})^T(M^tS)$ where $\alpha_x,\alpha_y\in\R^{s}$ are vectors obtained by taking entrywises median over all $(\Q_i)^T(\m^i_x)$ and $(\Q_i)^T(\m^i_y)$. (See line \ref{ln:alx} and \ref{ln:aly} of Algorithm \ref{alg:dotProduct}). Also note that as per line \ref{ln:mx} and line \ref{ln:my} of Algorithm \ref{alg:dotProduct}, $\m^i_x$ and $\m^i_y$ are defined as the empirical probability distribution of running $R_{\text{query}}$ random walks of length $t$ starting from vertex $x$ and $y$. Also note that $\Q_i$s are generated by Algorithm \ref{alg:compQ} which runs $R_{\text{init}}$ random walks from vertices in $I_S$. For any $z\in I_S$ any $i\in \{1,\ldots, O(\log n)\}$ let $\mathbf{q}_z^i$ denote the column corresponding to vertex $z$ in $\Q_i$. 

Let $c_3$ be a constant in front of $R_1$ and  $R_2$ in Lemma \ref{lem:pairwise-collision}. Let $\sigma_{\text{err}}=\frac{\xi}{c'\cdot k^2\cdot  n^{(1+200\epsilon/\varphi^2)}}$.  Thus by choice of $R_\text{init}= \varTheta{(n^{1-\delta +980 \cdot\epsilon / \varphi^2} \cdot k^{17}/{\xi}^{2})}$ as per line \ref{ln:setRR} of Algorithm \ref{alg:LearnEmbedding} and $R_\text{query}=\varTheta{(n^{\delta +500\cdot \epsilon / \varphi^2}\cdot k^{9}/{\xi}^{2})} $ as per line \ref{ln:set_r_q} of Algorithm \ref{alg:dotProduct},  the prerequisites of Lemma \ref{lem:pairwise-collision} are satisfied:
\[\min(R_\text{init}, R_\text{query}) \geq  \frac{c_3\cdot  k^5\cdot n^{-2+(100\epsilon /\varphi^2)}}{ \sigma_\text{err}^2}
\text{,  and, }
 R_\text{init} \cdot R_\text{query} \geq     \frac{ c_3\cdot k^2\cdot n^{-1+(40\epsilon /\varphi^2)}}{ \sigma_\text{err}^2}\]
Thus we can apply Lemma \ref{lem:pairwise-collision}. Hence, for any $z\in I_S$ with probability at least $0.99$ we have 
\[|{(\m^i_x)}^T {\mathbf{q}^i_z} - ({m_{x}})^T (m_{z}) |\leq {\sigma_{\text{err}}}\] 
 Note that as per line \ref{ln:alx} and line \ref{ln:aly} of Algorithm \ref{alg:dotProduct} we take entrywise median over all $(\Q_i)^T(\m^i_x)$ and $(\Q_i)^T(\m^i_y)$. Since we are running $O(\log n)$ copies of the same algorithm with success probability at least $0.99$, thus by simple Chernoff bound  with probability at least $1-n^{-100}$ for all $z\in I_S$ we have
\[|\alpha_x(z) - {(m_{x})}^T (m_{z}) |\leq {\sigma_{\text{err}}}\] 
Therefore by choice of $\sigma_{\text{err}}=\frac{\xi}{c'\cdot k^2\cdot  n^{1+200\cdot\epsilon/\varphi^2}}$ we get
\begin{equation}
\label{eq:Ex}
\|\mathbf{e}_x\|_2=\|\alpha_x-(m_{x})^T(M^tS)\|_2\leq {\sqrt{s}}\cdot\sigma_{\text{err}}=\frac{\sqrt{s}\cdot\xi}{c'\cdot k^2\cdot  n^{(1+200\epsilon/\varphi^2)}}\text{.}
\end{equation}
By the same analysis we get 
\begin{equation}
\label{eq:Ey}
\|\mathbf{e}_y\|_2\leq \frac{\sqrt{s}\cdot\xi}{c'\cdot k^2\cdot  n^{(1+200\epsilon/\varphi^2)}}\text{.}
\end{equation}
Putting \eqref{eq:Eu},  \eqref{eq:Au}, \eqref{eq:Mt1aa}, \eqref{eq:Ax}, \eqref{eq:Ay}, \eqref{eq:Ex}, and \eqref{eq:Ey} and for large enough $n$ we get:      

\begin{align*}
&\left|\left| \adp{f_x,f_y} - \cdot m_x^T  (M^tS)\left(\frac{n}{s}\cdot\widetilde{W}_{[k]} \widetilde{\Sigma}^{-4}_{[k]} \widetilde{W}^T_{[k]}\right) (M^tS)^T m_y \right|\right|  \leq \\
& \|\mathbf{e}_x\|_2\|A\|_2\|\mathbf{a}_y\|_2+\|\mathbf{a}_x\|_2\|E\|_2\|\mathbf{a}_y\|_2+\|\mathbf{a}_x\|_2\|A\|_2\|\mathbf{e}_y\|_2 + \\
&\|\mathbf{e}_x\|_2\|E\|_2\|\mathbf{a}_y\|_2+\|\mathbf{a}_x\|_2\|E\|_2\|\mathbf{e}_y\|_2+\|\mathbf{e}_x\|_2\|A\|_2\|\mathbf{e}_y\|_2+\|\mathbf{e}_x\|_2\|E\|_2\|\mathbf{e}_y\|_2 \\ 
&\leq  2\cdot \left( \frac{\sqrt{s}\cdot\xi}{c'\cdot k^2\cdot   n^{(1+200\epsilon/\varphi^2)}}  \right)\left( \frac{4\cdot n^{1+160\epsilon/\varphi^2}}{ s} \right) \cdot O\left(\sqrt{s}\cdot k^2\cdot n^{-1+(40\epsilon /\varphi^2)}\right)  \\
&+ 2\cdot\left( \frac{\sqrt{s}\cdot n^{(80\epsilon/\varphi^2)} k^2}{n}\right)\left( \frac{\xi\cdot n}{c'\cdot k^4 \cdot n^{80 \e/\varphi^2}\cdot s} \right)\left( \frac{\xi}{c'\cdot \sqrt{s}\cdot n^{(1+20\epsilon/\varphi^2)}}\right) \\
&+ \left( \frac{ \sqrt{s}\cdot\xi}{c'\cdot k^2\cdot  n^{(1+200\epsilon/\varphi^2)}} \right)^2\left( \frac{4\cdot n^{1+160\epsilon/\varphi^2}}{ s} \right)  \\
&+ O\left(\sqrt{s}\cdot k^2\cdot n^{-1+(40\epsilon /\varphi^2)}  \right)^2 \left(  \frac{\xi \cdot n}{c'\cdot k^4 \cdot n^{80 \e/\varphi^2}\cdot s} \right)  \\
& + \left(  \frac{\sqrt{s}\cdot\xi}{c'\cdot k^2\cdot  n^{(1+200\epsilon/\varphi^2)}} \right)^2\left(  \frac{\xi\cdot n}{c'\cdot k^4 \cdot n^{80 \e/\varphi^2}\cdot s} \right)\\
&\leq O\left(\frac{\xi}{c'\cdot n}\right) \\
&\leq \frac{\xi}{ n} \text{.}
\end{align*}
The last inequality holds by setting $c'$ to a large enough constant to cancel the hidden constant of $O\left(\frac{\xi}{c'\cdot n} \right)$.
\end{proof}

Now we are able to complete the proof of Theorem \ref{thm:dot}.
\thmdot*
\begin{proof}[Proof of Theorem \ref{thm:dot}]
\textbf{Correctness:}
Note that as per line \ref{ln:sets} of Algorithm \ref{alg:LearnEmbedding} we set $s= \varTheta{(n^{480\cdot \epsilon / \varphi^2}\cdot \log n \cdot k^{8}/{\xi}^2)}$. Recall that $I_S=\{i_1,\ldots, i_s\}$ is the multiset of $s$ vertices each sampled uniformly at random (see line \ref{ln:sample} of Algorithm \ref{alg:LearnEmbedding}).  Let $S$ be the $n\times s$ matrix whose $j$-th column equals $\mathds{1}_{i_j}$. Recall that $M$ is  the random walk transition matrix of $G$. Let $\sqrt{\frac{n}{s}} \cdot M^tS=\widetilde{U}\widetilde{\Sigma}\widetilde{W}^T$  be  the eigendecomposition of $\sqrt{\frac{n}{s}} \cdot M^tS$. We define
\[e_1=\left|(M^{t}\mathds{1}_{x})^T  (M^tS)\left(\frac{n}{s}\cdot\widetilde{W}_{[k]} \widetilde{\Sigma}^{-4}_{[k]} \widetilde{W}^T_{[k]}\right) (M^tS)^T (M^{t}\mathds{1}_{y}) -  \mathds{1}_x^T U_{[k]}{U}_{[k]}^T  \mathds{1}_y \right| \]
and 
\[e_2=\left| \adp{f_{x},f_y} - (M^{t}\mathds{1}_{x})^T  (M^tS)\left(\frac{n}{s}\cdot\widetilde{W}_{[k]} \widetilde{\Sigma}^{-4}_{[k]} \widetilde{W}^T_{[k]}\right) (M^tS)^T (M^{t}\mathds{1}_{y}) \right| \]
By triangle inequality we have 
\[\left|\adp{f_{x},f_y}  - \langle f_x, f_y \rangle \right|  = \left|\adp{f_{x},f_y} - \mathds{1}_x^T U_{[k]}{U}_{[k]}^T  \mathds{1}_y \right| \leq e_1 + e_2 \text{.}\]
Let $\xi'=\xi/2$.  Let $c$ be a constant in front of $s$ in Lemma \ref{lem:bnd-e1} and  $c'$ be a constant in front of $s$ in Lemma \ref{lem:u-abs-close}. Note that as per line \ref{ln:sets} of Algorithm \ref{alg:LearnEmbedding} we set $s= \varTheta{(n^{480\cdot \epsilon / \varphi^2}\cdot \log n \cdot k^{8}/{\xi}^2)}$. Since $\frac{\epsilon}{\varphi^2}\leq \frac{1}{10^5}$ and $s\geq c\cdot n^{480\epsilon / \varphi^2}\cdot \log n \cdot k^{8}/{\xi'}^2$ by Lemma \ref{lem:bnd-e1} with probability at least $1-n^{-100}$ we have $e_1\leq \frac{\xi'}{n}=\frac{\xi}{2\cdot n} \text{.}$ Since  $s\geq c'\cdot n^{240\epsilon / \varphi^2}\cdot \log n \cdot k^{4}$, by Lemma \ref{lem:u-abs-close} with probability at least $1-2\cdot n^{-100}$ we have $e_2\leq \frac{\xi}{2\cdot n}\text{.}$ Thus with probability at least $1-3\cdot n^{-100}$ we have \[\left| \adp{f_{x},f_y}  - \langle f_x, f_y \rangle \right|  \leq e_1 + e_2 \leq \frac{\xi}{2\cdot n}+\frac{\xi}{2\cdot n} \leq \frac{\xi}{n}\text{.}\]

\textbf{Space and runtime of \textsc{InitializeOracle}:} 
Algorithm  \textsc{InitializeOracle}($G,\delta,\xi$) (Algorithm~\ref{alg:LearnEmbedding}) samples a set $I_S$. Then as per line \ref{ln:setQi} of Algorithm~\ref{alg:LearnEmbedding} it estimates the empirical probability distribution of random walks starting from any vertex $x\in I_S$ for $O(\log n)$ times. To that end as per line \ref{ln:ranwalk22} of Algorithm \ref{alg:compQ} it runs $R_{\text{init}}$ random walks of length $t$ from each vertex $x\in I_S$. So it takes $O(\log n\cdot s\cdot R_{\text{init}} \cdot t)$ time and requires $O(\log n\cdot s\cdot R_{\text{init}})$ space to store endpoints of random walks. Then as per line \ref{ln:setG} of Algorithm~\ref{alg:LearnEmbedding} it estimates matrix $\G$ such that the entry corresponding to the $x^\text{th}$ row and $y^{\text{th}}$ column of $\G$ is an estimation of pairwise collision probability of random walks starting from $x,y \in I_S$. To compute $\G$ we call Algorithm \textsc{EstimateCollisionProbabilities}($G,I_S,R_{\text{init}},t$) (Algorithm \ref{alg:gram}) for $O(\log n)$ times. Algorithm \ref{alg:gram}  runs $R_{\text{init}}$ random walks of length $t$ from each vertex $x\in I_S$, hence, It takes $O(s\cdot R_{\text{init}}\cdot t \cdot \log n)$ time and it requires $O(s^2\cdot \log n)$ space to store matrix $\G$. Then as per line \ref{ln:QSVD} of Algorithm \ref{alg:LearnEmbedding} we compute the SVD of matrix $\G$ in time $O(s^3)$. Thus overall Algorithm \ref{alg:LearnEmbedding} runs in time $O \left(\log n\cdot s\cdot R_{\text{init}} \cdot t+ s^3\right)$. Thus, by choice of $t= \varTheta\left(\frac{\log n}{\varphi^2}\right)$	, $R_{\text{init}}=\varTheta{(n^{1-\delta +980 \cdot\epsilon / \varphi^2}  \cdot k^{17}/{\xi}^{2})}$ and $s= \varTheta(n^{480\cdot \epsilon / \varphi^2}\cdot \log n \cdot k^{8}/{\xi}^2)$  as in Algorithm \ref{alg:LearnEmbedding} we get that Algorithm \ref{alg:LearnEmbedding} runs in time $O \left(\log n\cdot s\cdot R_{\text{init}} \cdot t+ s^3\right)= (\frac{k}{\xi})^{O(1)}\cdot n^{1-\delta+O(\epsilon/\varphi^2)}\cdot \log^3 n\cdot\frac{1}{\varphi^2}$ and returns a data structure of size $O \left(s^2+\log n\cdot s \cdot R_{\text{init}}\right)=(\frac{k}{\xi})^{O(1)}\cdot n^{1-\delta+O(\epsilon/\varphi^2)}\cdot \log^2 n\text{.}$

\textbf{Space and runtime of \textsc{SpectralDotProductOracle}:} 
Algorithm \textsc{SpectralDotProductOracle}($G,x,y, \delta, \xi, \mathcal{D})$(Algorithm \ref{alg:dotProduct}) repeats $O(\log n)$ copies of the following procedure: it runs $R_{\text{query}}$ random walks of lenght $t$ from vertex $x$ and vertex $y$, then it computes $\m_x\cdot \Q_i$ and $\m_y\cdot \Q_i$. Since $\Q_i\in \R^{n\times s}$ has $s$ columns and since $\m_x$ has at most $R_{\text{query}}$ non-zero entries, thus one can compute $\m_x\cdot \Q_i$ in time $R_{\text{query}}\cdot s$. Finally Algorithm \ref{alg:dotProduct} take entrywises median of computed vectors (see line \ref{ln:alx} and line \ref{ln:aly} of Algorithm \ref{alg:dotProduct}), and returns value $\alpha_x \Psi \alpha_y$ (see line \ref{ln:axay} of Algorithm \ref{alg:dotProduct}). Since $\alpha_x,\alpha_y\in \R^{s}$ and $\Psi\in \R^{s\times s}$ one can compute  $\alpha_x \Psi \alpha_y$ in time $O(s^2)$. Thus overall Algorithm \ref{alg:dotProduct} takes $O\left(t\cdot R_{\text{query}}\cdot \log n + s\cdot R_{\text{query}} \cdot \log n + s^2\right)$ time and $O\left( R_{\text{query}}\cdot \log n + s\cdot R_{\text{query}} \cdot \log n + s^2\right)$ space. Thus, by choice of $t= \varTheta\left(\frac{\log n}{\varphi^2}\right)$	, $R_{\text{query}}=\varTheta{(n^{\delta +500 \cdot\epsilon / \varphi^2}  \cdot k^{9}/{\xi}^{2})}$ and $s= \varTheta(n^{480\cdot\epsilon / \varphi^2}\cdot \log n \cdot k^{8}/{\xi}^2)$  as in Algorithm \ref{alg:LearnEmbedding} and Algorithm \ref{alg:dotProduct} we get that the  Algorithm \ref{alg:dotProduct} runs in  time $ (\frac{k}{\xi})^{O(1)}\cdot n^{\delta+O(\epsilon/\varphi^2)}\cdot \frac{\log^2 n}{\varphi^2}$ and returns a data structure of size $ (\frac{k}{\xi})^{O(1)}\cdot n^{\delta+O(\epsilon/\varphi^2)}\cdot \log^2 n \text{.}$

\end{proof}

\subsection{Computing approximate norms and spectral dot products (Proof of Theorem \ref{thm:dotpi})}\label{sec:dotproductcomp}
To design the clustering algorithm in Section \ref{sec:algo}, since we cannot evaluate the dot-product of the spectral embedding exactly in sublinear time, we prove that it is enough to have access to approximate dot-product of the spectral embedding. 
In Algorithm~\ref{alg:ballcarving}, Algorithm~\ref{alg:inside}  and throughout the analysis of in Section \ref{sec:algo} we will use $\adp{\cdot,\cdot}$ to denote approximate spectral dot products and $ \an{\cdot}$ to denote the approximate norm of a vector. Let $r\in [k]$ and $B,B_1, \ldots, B_r \subseteq V$. Let $\widehat{\mu}, \widehat{\mu}_1,\ldots, \widehat{\mu}_r \in \R^k$ where $\wh{\mu} = \frac{\sum_{z \in B} f_z}{|B|}$ and $\wh{\mu}_i = \frac{\sum_{z \in B_i} f_z}{|B_i|}$. All dot products we will try to approximate in Section \ref{sec:algo} will be of the form $\rdp{f_x, \wh{\Pi}(\wh{\mu})} $ and all the norms that we approximate are of the form $\an{\wh{\Pi}(\wh{\mu})}$, where $x \in V$ and $\wh{\Pi}$ is defined as a orthogonal projection onto $span(\{\wh{\mu}_1,\dots, \wh{\mu}_r \})^{\perp}$. To compute such dot products we call Algorithm~\ref{alg:dot-apx-pi} in the following way (see Corollary \ref{corr:dotpi}):
\begin{equation}\label{eq:fx-pi-mu}
\adp{f_x,\widehat{\Pi} \widehat{\mu}}:=  \frac{1}{|B|}\cdot\sum_{y\in B} \adp{f_x,\widehat{\Pi} f_y} \text{,}
\end{equation}
\begin{equation}\label{eq:norm-pi-mu}
\an{\widehat{\Pi} \widehat{\mu}}^2:=  \frac{1}{|B|}\cdot\sum_{x\in B} \adp{f_x,\widehat{\Pi} \widehat{\mu}}  \text{.}
\end{equation}

\begin{algorithm}[H]
\caption{\textsc{DotProductOracleOnSubspace}($G,x,y, \delta, \xi, \mathcal{D}, B_1,\ldots, B_r$)  \Comment Need: $\epsilon/\varphi^2 \leq \frac{1}{10^5}$
\newline \text{ }\Comment $\mathcal{D}:=\{\Psi,\Q_1,\ldots ,\Q_{O(\log n)}\}$}
\label{alg:dot-apx-pi}
\begin{algorithmic}[1]
		\State Let $X\in \R^{r\times r}, h_x\in \R^r, h_y\in \R^r$.
		\State Let $\xi':=\Theta(\xi\cdot n^{(-80\epsilon/\varphi^2)}\cdot k^{-6})$ \label{ln:alg-dot-pi-xi}
	\For{$i,j$ in $[r]$}
		\State $X(i,j):= \frac{1}{|B_i||B_j|} \cdot \sum_{z_i\in B_i}\sum_{z_j \in B_j}\textsc{SpectralDotProduct}(G,z_i, z_j, \delta,\xi', \mathcal{D})$  \label{ln:Yij} 
		\State \Comment{$X(i,j)=\adp{\widehat{\mu}_i, \widehat{\mu}_j}$}
	\EndFor
	\For{$i$ in $[r]$}
		\State $h_x(i):= \frac{1}{|B_i|} \cdot \sum_{z_i\in B_i}\textsc{SpectralDotProduct}(G,z_i, x, \delta, \xi',\mathcal{D})$ \label{ln:hx} \Comment{$h_x(i)=\adp{\widehat{\mu}_i, f_x}$}
		\State $h_y(i):= \frac{1}{|B_i|} \cdot \sum_{z_i\in B_i}\textsc{SpectralDotProduct}(G,z_i, y, \delta, \xi',\mathcal{D})$ \label{ln:hy}
		\Comment{$h_y(i)=\adp{\widehat{\mu}_i, f_y}$}
	\EndFor
	\State  \Return $\adp{f_x,\widehat{\Pi} f_y}:= \textsc{SpectralDotProduct}(G,x, y, \delta,\xi', \mathcal{D}) - h_x^T X^{-1} h_y$ \label{ln:fin-apx}
\end{algorithmic}
\end{algorithm}

The following Lemma is a generalization of Lemma \ref{lem:QQ-1} to the  approximation of the cluster means (i.e, $\wh{\mu}_1,\ldots, \wh{\mu}_k$), where  $\wh{\mu}_i\in \R^k$ is a vector that approximates the center of cluster $C_i$ (i.e., $\mu_i$) such that $||\wh{\mu}_i-\mu_i||_2$ is small. 
\begin{restatable}{lemma}{lemapxQQ}
\label{lem:apxQQ-1}
Let  $k \geq 2$ be an integer, $\varphi \in (0,1)$, and $\e \in (0,1)$. Let $G=(V,E)$ be a $d$-regular graph that admits a $(k,\varphi,\e)$-clustering $C_1, \ldots, C_k$.  Let $\mu_1,\ldots,\mu_k$ denote the cluster means of $C_1, \ldots, C_k$. Let $0<\zeta<\frac{\sqrt{\epsilon}}{20\cdot k\cdot \varphi}$. Let $\wh{\mu}_1, \dots, \wh{\mu}_k \in \R^k$ denote an approximation of the cluster means such that for each $i\in[k]$, $||\mu_i-\wh{\mu}_i||_2\leq \zeta ||\mu_i||_2$. Let $S \subseteq \{\wh{\mu}_1,\ldots,\wh{\mu}_k\}$. Let $|S|=r$ and $\wh{H}\in \R^{k\times r}$ denote a matrix whose columns are the vectors in $S$.  Let $\sigma:[r]\rightarrow[k]$ denote a mapping from the the columns of $H$ to the corresponding cluster.  Let $\wh{W}\in \R^{r\times r}$ be a diagonal matrix such that $\wh{W}(i,i)=\sqrt{|C_{\sigma(i)}|}$. Let $\wh{Z}=\wh{H}\wh{W}$.Then for any vector $x\in \R^{r}$ with $||x||_2=1$ we have 
\begin{enumerate}
\item $
| x^T (\wh{Z}^T \wh{Z}- I) x| \leq\frac{5\sqrt{\epsilon}}{\varphi} \label{itm1:apxZZ}
$
\item $|x^T ((\wh{Z}^T \wh{Z})^{-1}- I)x| \leq  \frac{5\sqrt{\epsilon}}{\varphi}  \text{.}$ \label{itm2:apxZZ}
\end{enumerate}
\end{restatable}
\begin{proof}
\textbf{Proof of item \eqref{itm1:apxZZ}}
Let $Y\in \R^{k\times k}$ be a matrix, whose $i$-th column is equal to $\sqrt{C_i}\cdot{\mu}_i$. By Lemma \ref{lem:spectraldistance} item \eqref{lem-spec-dist-itm2}  for any vector $\alpha\in \R^k$ with $||\alpha||_2=1$ we have
\begin{equation}\label{eq:zYI}
|\alpha^T(Y^T Y-I)\alpha| \leq \frac{4\sqrt{\epsilon}}{\varphi}
\end{equation}
Let $\wh{Y}\in \R^{k\times k}$ be a matrix, whose $i$-th column is equal to $\sqrt{C_i}\cdot\wh{\mu}_i$. Note that for any $i,j \in [k]$ we have $(Y^TY)(i,j)=\sqrt{|C_i||C_j|}\rdp{\mu_i, \mu_j}$ and  $(\wh{Y}^T\wh{Y})(i,j)=\sqrt{|C_i||C_j|}\rdp{\wh{\mu}_i, \wh{\mu}_j}$.  Therefore for any $i \in [k]$ we have
\begin{align*}
\left| (Y^TY)(i,i)-(\wh{Y}^T\wh{Y})(i,i)\right| &= |C_i| \left| ||\mu_i||^2_2 - ||\wh{\mu_i}||^2_2 \right| \\
&\leq |C_i| \cdot |(||\mu_i||_2 - ||\wh{\mu_i}||_2) (||\mu_i||_2 + ||\wh{\mu_i}||_2) |  \\
&\leq |C_i| \cdot  \left|(\zeta||\mu_i||_2) (||\mu_i||_2 + (1+\zeta)||{\mu_i}||_2) \right| &&\text{Since }   ||\wh{\mu_i}||_2 \leq (1+\zeta)||\mu_i||_2    \\
&\leq 3\cdot \zeta |C_i| \cdot   ||\mu_i||^2_2  &&\text{Since }  \zeta <1\\
&\leq 6\cdot \zeta &&\text{By Lemma \ref{lem:dotmu} }  ||\mu_i||^2_2\leq \frac{2}{|C_i|}
\end{align*}
Also for any $i\neq j\in [k]$ we have
\begin{align}
&\left| (Y^TY)(i,j)-(\wh{Y}^T\wh{Y})(i,j)\right| \\
&= \sqrt{|C_i||C_j|} \cdot \left|  \rdp{\wh{\mu}_i, \wh{\mu}_j} - \rdp{\mu_i, \mu_j}\right| \nonumber\\
&= \sqrt{|C_i||C_j|} \cdot\left| \rdp{ {\mu}_i+(\wh{\mu}_i- {\mu}_i),  {\mu}_j+(\wh{\mu}_j- {\mu}_j)} - \rdp{{\mu}_i, {\mu}_j} \right|   \nonumber\\
&\leq \sqrt{|C_i||C_j|}\cdot\left( | \rdp{ \wh{\mu}_i- {\mu}_i , \wh{\mu}_j- {\mu}_j}| + |\rdp{ \wh{\mu}_i- {\mu}_i,  \mu_j}| + |\rdp{ \wh{\mu}_j- {\mu}_j,  \mu_i} | \right) &&\text{By triangle inequality}  \nonumber \\
&\leq \sqrt{|C_i||C_j|} \cdot \left( || \wh{\mu}_i- {\mu}_i ||_2|| \wh{\mu}_j- {\mu}_j||_2 + || \wh{\mu}_i- {\mu}_i||_2||\mu_j||_2 + || \wh{\mu}_j- {\mu}_j||_2||  \mu_i||_2 \right) &&\text{By Cauchy-Schwarz}  \nonumber \\
&\leq \sqrt{|C_i||C_j|} \cdot  (\zeta^2+2\zeta)\left( ||\mu_i||_2 ||\mu_j||_2   \right) &&\text{Since }   ||\wh{\mu_i}-\mu_i||_2 \leq \zeta||\mu_i||_2 \text{ for all }i \nonumber \\
&\leq \sqrt{|C_i||C_j|} \cdot  6\cdot \zeta \cdot \frac{1}{\sqrt{|C_i||C_j|}} && \text{By Lemma \ref{lem:dotmu} }  ||\mu_i||^2_2\leq \frac{2}{|C_i|} \text{ for all }i  \nonumber\\
&\leq 6\cdot\zeta \label{eq:muijhats}
\end{align}
Therefore we have
\begin{align*} 
||(Y^TY)-(\wh{Y}^T\wh{Y})||_2 &\leq ||(Y^TY)-(\wh{Y}^T\wh{Y})||_F \\
&\leq \sqrt{\sum_{i=1}^k\sum_{j=1}^k \left((Y^TY)(i,j)-(\wh{Y}^T\wh{Y})(i,j)\right)^2} \\
&\leq  6\cdot k\cdot  \zeta \\
&\leq \frac{\sqrt{\epsilon}}{2\varphi} && \text{Since $\zeta\leq \frac{\sqrt{\epsilon}}{20\cdot k\cdot \varphi}$}
\end{align*}
Thus for any $\alpha\in \R^k$ with $||\alpha||_2=1$ we have
\begin{equation}\label{eq:YYzeta}
\left| \alpha^T \left((Y^TY)-(\wh{Y}^T\wh{Y})\right) \alpha \right| \leq \frac{\sqrt{\epsilon}}{2\varphi}
\end{equation}
Putting \eqref{eq:YYzeta} and \eqref{eq:zYI} together we get
\[\left| \alpha^T \left(\wh{Y}^T\wh{Y} -I \right) \alpha \right|\leq 4.5 \frac{\sqrt{\epsilon}}{\varphi}\]
Let $x\in \R^{r}$ be a vector with $||x||_2=1$, and let $\alpha\in \R^k$ be a vector that is $x_j=\alpha_j$ if $\wh{\mu}_j\in S$ and otherwise $x_j=0$. Thus we have $||\alpha||_2=||x||_2=1$ and $\wh{Y}z=\wh{Z}x$. Hence, we get  
\[
| x^T (\wh{Z}^T \wh{Z}- I) x|= |\alpha^T(\wh{Y}^T \wh{Y}-I)\alpha| \leq\frac{4.5\sqrt{\epsilon}}{\varphi}
\]
\textbf{Proof of item \eqref{itm2:apxZZ}}
For any vector $x\in \R^{r}$ with $||x||_2=1$ we have 
\begin{equation}
\label{eq:zhatzhat^t-inv}
1-\frac{4.5\sqrt{\epsilon}}{\varphi}\leq x^T (\wh{Z}^T \wh{Z}) x \leq  1+\frac{4.5\sqrt{\epsilon}}{\varphi}
\end{equation}
Note that $\wh{Z}^T \wh{Z}$ is symmetric and positive semidefinit. Also note that  $\wh{Z}^T \wh{Z}$ is spectrally close to $I$, hence, $\wh{Z}^T \wh{Z}$ is  invertible.  Thus by \eqref{eq:zhatzhat^t-inv} and Lemma \ref{lem:neg-spectral-close} for any vector $x\in\R^{r}$ we have
\[1-\frac{5\sqrt{\epsilon}}{\varphi}\leq  x^T (\wh{Z}^T \wh{Z})^{-1} x \leq  1+\frac{5\sqrt{\epsilon}}{\varphi}\]
Therefore we get
\[|x^T ((\wh{Z}^T \wh{Z})^{-1}- I)x| \leq  \frac{5\sqrt{\epsilon}}{\varphi}  \text{.}\] 
\end{proof}

\begin{restatable}{theorem}{thmdotpi}
\label{thm:dotpi}
Let $G=(V,E)$ be a $d$-regular graph that admits a $(k,\varphi,\epsilon)$-clustering $C_1,\dots,C_k$. Let $k \geq 2$ be an integer, $\varphi \in (0,1)$, $\frac{1}{n^5}<\xi<1$, and $\frac{\e }{\varphi^2} $ be smaller than a positive absolute constant. Then there exists an event $\mathcal{E}$ such that $\mathcal{E}$ happens with probability $1 - n^{-48}$ and conditioned on $\mathcal{E}$ the following holds.

Let $r\in[k]$. Let $\delta\in (0,1)$. Let $B_1,\ldots, B_r$ denote multisets of points. Let $b=\max_{i\in r}|B_i|$. Let $\sigma:[r]\rightarrow[k]$ denote a mapping from the set $B$ to the cluster $C=\sigma(B)$. Suppose that for all $i\in[r]$, $B_i\subseteq \sigma(B_i)$ and for all $i\neq j\in [r]$, $\sigma(B_i)\neq \sigma(B_j)$. Let $\widehat{\mu}_i=\frac{1}{|B_i|}\cdot \sum_{z\in B_i} f_z$. Suppose that for each $i\in[r]$, $||\wh{\mu}_i-\mu_{\sigma(i)}||_2\leq \frac{\sqrt{\epsilon}}{20\cdot k\cdot\varphi}||\mu_i||_2$. Let $\wh{\Pi}$ is defined as a orthogonal projection onto then $span(\{\wh{\mu}_1,\dots, \wh{\mu}_r \})^{\perp}$. Then for all $x,y \in V$ we have   
\[\left|\adp{f_x,\widehat{\Pi} f_y}-\rdp{f_x,\widehat{\Pi} f_y}\right|\leq \frac{\xi}{n} ,\]
where $\adp{f_x,\widehat{\Pi} f_y}:=\textsc{DotProductOracleOnSubspace}(G,x,y, \delta, \xi, \mathcal{D}, B_1,\ldots, B_r)$.
Algorithm \ref{alg:dot-apx-pi} runs in time $b^2\cdot (\frac{k}{\xi})^{O(1)} \cdot n^{\delta+O(\epsilon/\varphi^2)}\cdot  \frac{(\log n)^2}{\varphi^2}$. 
\end{restatable}

\begin{proof}

\textbf{Runtime:} Note that Algorithm \ref{alg:dot-apx-pi}, first computes matrix $X\in \R^{r\times r}$, and vectors $h_x,h_y\in \R^k$. To compute $X(i,j)$ for any $i,j\in [r]$, as per line \ref{ln:Yij} of Algorithm \ref{alg:dot-apx-pi}, we run $\textsc{SpectralDotProduct}(G,z_i, z_j, \delta,\xi', \mathcal{D})$ for all $z_i\in B_i$ and $z_j\in B_j$, where $|B_i|\leq b$ and $|B_j|\leq b$. 

Note that by Theorem \ref{thm:dot}, Algorithm  $\textsc{SpectralDotProduct}(G,z_i, z_j, \delta,\xi', \mathcal{D})$ runs in time $(\frac{k}{\xi'})^{O(1)} \cdot n^{\delta+O(\epsilon/\varphi^2)}\cdot \frac{ (\log n)^2}{\varphi^2} $. Thus one can compute the matrix $X^{-1}$ in time $O(k^3 + k^2\cdot b^2\cdot (\frac{k}{\xi'})^{O(1)}  n^{\delta+O(\epsilon/\varphi^2)} \cdot \frac{(\log n)^2}{\varphi^2}  )$. Also, to compute $h_x(i)$ (respectively, $h_y(i)$) for any $i\in[r]$, as per line \ref{ln:hx} and line \ref{ln:hy} of Algorithm \ref{alg:dot-apx-pi}, we run $\textsc{SpectralDotProduct}(G,x, z, \delta,\xi', \mathcal{D})$ for all $z \in B_i$ (respectively, $z\in B_j$). Thus one can compute $h_x$ and $h_y$ in time $k\cdot b\cdot (\frac{k}{\xi'})^{O(1)} \cdot n^{\delta+O(\epsilon/\varphi^2)} \cdot \frac{(\log n)^2}{\varphi^2} $. As per line \eqref{ln:alg-dot-pi-xi} of Algorithm \ref{alg:dot-apx-pi} we set $\xi':=\Theta(\xi\cdot n^{(-80\epsilon/\varphi^2)}\cdot k^{-6})$.
Therefore the runtime of the algoritm is $b^2\cdot (\frac{k}{\xi})^{O(1)} \cdot n^{\delta+O(\epsilon/\varphi^2)}\cdot  \frac{(\log n)^2}{\varphi^2}$.

\textbf{Correctness:} 
Let $x,y \in V$. Let $H \in \R^{k \times r}$ be a matrix whose columns are $\wh{\mu}_1, \dots, \wh{\mu}_r$.  Then we have $H\left(H^T H\right)^{-1}H^T$ is the orthogonal projection matrix onto $span(\{\wh{\mu}_1,\dots, \wh{\mu}_r \})$. Let $W\in \R^{r\times r}$ denote a matrix such that for any $i\in [r]$, $W(i,i)=\sqrt{|C_{\sigma(i)}|}$. Note that
\[(HW)\left((HW)^T (HW)\right)^{-1}(HW)^T = HW \left(W^{-1}\left(H^T H\right)^{-1}W^{-1}\right) WH^T =  H\left(H^T H\right)^{-1}H^T\]
Thus we have $(HW)\left((HW)^T (HW)\right)^{-1}(HW)^T$ is the orthogonal projection matrix onto $span(\{\wh{\mu}_1,\dots, \wh{\mu}_r \})$ and we get 
$$\wh{\Pi} = I - HW\left(WH^T HW\right)^{-1}WH^T $$
Therefore, we have
\begin{equation}\label{eq:expressionfordots}
\rdp{f_x, \wh{\Pi}f_y} = \rdp{f_x,f_y} - f_x^THW\left(WH^T HW\right)^{-1}WH^T f_y
\end{equation}
Let $\adp{f_x,f_y}:=\textsc{SpectralDotProduct}(G,x, y, \delta,\xi', \mathcal{D})$. Then as per line \ref{ln:fin-apx} of Algorithm \ref{alg:dot-apx-pi} we have
\begin{equation}\label{eq:fpihat}
\adp{f_x,\widehat{\Pi} f_y}:= \adp{f_{x},f_y} - h_x^T X^{-1} h_y \text{,}
\end{equation}
where as per line \eqref{ln:Yij} of Algorithm \ref{alg:dot-apx-pi} for any $i,j\in [r]$ we have $X(i,j)=\adp{\widehat{\mu}_i, \widehat{\mu}_j}$, and as per line \eqref{ln:hx} and line \eqref{ln:hy} of Algorithm \ref{alg:dot-apx-pi} for any $i\in [r]$ we have $h_x(i)=\adp{\widehat{\mu}_i, f_x} $ and $h_y(i)=\adp{\widehat{\mu}_i, f_y}$. Note that 
\[h_x^T X^{-1} h_y = h_x^T W W^{-1}X^{-1}W^{-1} W h_y =h_x^T W (WX W)^{-1} W h_y\]
Therefore by \eqref{eq:fpihat}, \eqref{eq:expressionfordots} and triangle inequality we have
\begin{align*}
\label{eq:th7-main}
\left|\adp{f_x,\widehat{\Pi} f_y}-\rdp{f_x,\widehat{\Pi} f_y}\right| 
\leq |\adp{f_{x},f_y} - \rdp{f_x,f_y}| + \left|    h_x^T W (WXW)^{-1} W h_y -   f_x^THW\left(WH^T HW\right)^{-1}WH^T f_y \right| 
\end{align*}
Note that by Theorem \ref{thm:dot} and by union bound over all pair of vertices with probability at least $1-n^{-100}\cdot n^2 $ for all $a,b\in V$ we have
\begin{equation}
\label{eq:fxhatfy}
|\adp{f_{a},f_b} - \rdp{f_a,f_b}|\leq \frac{\xi'}{n}
\end{equation}
We define  
\[\mathbf{a}_x= f_x^THW, \quad A=(WH^T HW)^{-1}, \quad  \mathbf{a}_y=WH^T f_y \text{, and}\]
\[\mathbf{e}_x=h_x^TW  -\mathbf{a}_x, \quad  E= (WXW)^{-1}-A, \quad  \mathbf{e}_y= Wh_y-\mathbf{a}_y\]
Thus by triangle inequality we have
\begin{align}
& \left|   h_x^T W(WXW)^{-1} Wh_y - f_x^T H W(WH^THW)^{-1}W H^T f_y\right|   = \nonumber \\
&\| \left(\mathbf{a}_x +\mathbf{e}_x\right) \left(A+E\right) \left(\mathbf{a}_y +\mathbf{e}_y\right) -\mathbf{a}_xA\mathbf{a}_y\|_2 \leq \nonumber \\
& \|\mathbf{e}_x\|_2\|A\|_2\|\mathbf{a}_y\|_2+\|\mathbf{a}_x\|_2\|E\|_2\|\mathbf{a}_y\|_2+\|\mathbf{a}_x\|_2\|A\|_2\|\mathbf{e}_y\|_2 + \nonumber \\
&\|\mathbf{e}_x\|_2\|E\|_2\|\mathbf{a}_y\|_2+\|\mathbf{a}_x\|_2\|E\|_2\|\mathbf{e}_y\|_2+\|\mathbf{e}_x\|_2\|A\|_2\|\mathbf{e}_y\|_2+\|\mathbf{e}_x\|_2\|E\|_2\|\mathbf{e}_y\|_2 \label{eq:fin-7}
\end{align}
Thus we need to bound $||\mathbf{a}_x||_2,||\mathbf{a}_y||_2,||\mathbf{e}_x||_2,||\mathbf{e}_y||_2,||A||_2,||E||_2$.  Note that $||\mathbf{a}_x||_2=||f_x^THW||_2$, Thus we have $||\mathbf{a}_x||_2\leq ||f_x^TH||_2||W||_2$. Note that 
\begin{equation}\label{eq:l2-Z}
||W||_2\leq \max_i W(i,i) = \max_i{\sqrt{|C_i|}}\leq \sqrt{n}
\end{equation}
Then we bound $||f_x^TH||_2$. 
Note that  $||f_x^TH||_2=\sqrt{\sum_{i=1}^r{\rdp{f_x,\wh{\mu}_i}}^2}$.  We first bound $\rdp{f_x,\wh{\mu}_i}$.
\begin{align*}
\rdp{f_x,\wh{\mu}_i}&=\frac{1}{|B_i|}\cdot \sum_{z\in B_i}\rdp{f_x, f_z} \\
&\leq \frac{1}{|B_i|} \sum_{z\in B_i} ||f_x||_2 ||f_z||_2\\
&\leq \frac{1}{|B_i|}\cdot \sum_{z\in B_i} \sqrt{k^2\cdot ||f_x||^2_\infty ||f_z||^2_\infty} \\
&\leq  \frac{1}{|B_i|}\cdot |B_i|\cdot k\cdot O\left(\frac{k\cdot n^{40\epsilon/\varphi^2}}{n}\right) && \text{By Lemma \ref{lem:l-inf-bnd} and since }\min_{i\in k}|C_i|\geq \Omega\left(\frac{n}{k}\right)\\
&\leq O(k^{2}\cdot n^{-1+40\epsilon/\varphi^2})
\end{align*} 
Since, $r<k$, we get
\begin{equation}\label{eq:fxTY}
||f_x^TH||_2=\sqrt{\sum_{i=1}^r{\rdp{f_x,\wh{\mu}_i}}^2}\leq \sqrt{k}\cdot O(k^{2}\cdot n^{-1+40\epsilon/\varphi^2}) \leq O(k^{2.5}\cdot n^{-1+40\epsilon/\varphi^2})
\end{equation}
Thus we get
\begin{equation} \label{eq:ax-7}
||\mathbf{a}_x||_2= ||f_x^THW||_2 \leq ||f_x^TH||_2||W||_2\leq  O\left(k^{2.5}\cdot n^{-1/2+40\epsilon/\varphi^2}\right)
\end{equation}
By the same computation we also have
\begin{equation} \label{eq:ay-7}
||\mathbf{a}_y||_2\leq  O\left(k^{2.5}\cdot n^{-1/2+40\epsilon/\varphi^2}\right)
\end{equation}
Next we bound $||\mathbf{e}_x||_2$. We have $\mathbf{e}_x=h_x^T W-f_x^THW$. Thus we get $||\mathbf{e}_x||_2\leq ||h_x^T -f_x^TH||_2||W||_2$. By \eqref{eq:l2-Z} we have a bound on $||W||_2$. 
Note that for any $i\in r$, we have $h_x(i)=\frac{1}{|B_i|}\sum_{z\in B_i} \adp{f_x, f_z}$ and $(f_x^TH)(i)=\frac{1}{|B_i|}\sum_{z\in B_i} \rdp{f_x, f_z}$. Therefore with probability at least $1-n^{-98}$ we have
\begin{align*}
|h_x(i)-(f_x^TH)(i)|&= \left|\frac{1}{b}\sum_{z\in B_i} (\adp{f_x, f_z} - \rdp{f_x, f_z})\right| \nonumber \\
&\leq  \frac{1}{|B_i|}\sum_{z\in B_i} | \adp{f_x, f_z}- \rdp{f_x, f_z}| &&\text{By triangle inequality} \nonumber \\
&\leq   \frac{1}{|B_i|}\cdot |B_i| \cdot \frac{\xi'}{n} &&\text{By \eqref{eq:fxhatfy}}
\end{align*}
Since $r\leq k$,  we have
\[||h_x^T -f_x^TH||_2=\sqrt{\sum_{i=1}^r (h_x(i)-\mathbf{a}_x(i))^2} \leq \sqrt{k}\cdot \frac{\xi'}{n}\]
Therefore by \eqref{eq:l2-Z} we have
\begin{equation}\label{eq:ex-bnd}
||\mathbf{e}_x||_2\leq ||h_x^T -f_x^TH||_2||W||_2 \leq \frac{\xi'\sqrt{k}}{\sqrt{n}}
\end{equation}
By the same computation we also have
\begin{equation} \label{eq:ey-bnd}
||\mathbf{e}_y||_2\leq \frac{\xi'\sqrt{k}}{\sqrt{n}}
\end{equation}
Next we bound $||A||_2$. Note that $A=((HW)^T(HW))^{-1}$. By Lemma \ref{lem:apxQQ-1} item \eqref{itm2:apxZZ}  for any vector $x\in \R^r$ with $||x||_2=1$ we have
\[
\left| x^T \left( \left((HW)^T(HW)\right)^{-1}- I\right) x\right| \leq\frac{5\sqrt{\epsilon}}{\varphi} 
\]
Therefore
\begin{equation}
\label{eq:l2A-bnd}
||A||_2=||((HW)^T(HW))^{-1}||_2\leq 1+\frac{5\sqrt{\epsilon}}{\varphi} \leq 2
\end{equation}
Now we bound $||E||_2=||(WXW)^{-1}-(WH^THW)^{-1}||_2$. For any $i,j\in [r]$ we have  
\[(WXW)(i,j)= \sqrt{|C_{\sigma(B_i)}||C_{\sigma(B_j)}|}\cdot \frac{1}{|B_i|\cdot |B_j|} \cdot\sum_{z_i\in B_i, z_j\in B_j} \adp{f_{z_i}, f_{z_j}}\]
and 
\[(WH^THW)(i,j)= \sqrt{|C_{\sigma(B_i)}||C_{\sigma(B_j)}|} \cdot \frac{1}{|B_i|\cdot |B_j|} \cdot \sum_{z_i\in B_i, z_j\in B_j} \rdp{f_{z_i}, f_{z_j}}\] 
Therefore with probability at least $1-n^{-98}$ we have
\begin{align}
&|(WXW)(i,j)-(WH^THW)(i,j)| \nonumber \\
&=\left|\sqrt{|C_{\sigma(B_i)}||C_{\sigma(B_j)}|}\cdot \frac{1}{|B_i|\cdot |B_j|}\sum_{z_i\in B_i, z_j\in B_j} (\hat{f}_{z_i z_j} -  \rdp{f_{z_i}, f_{z_j}}) \right| \nonumber\\
&\leq \sqrt{|C_{\sigma(B_i)}||C_{\sigma(B_j)}|}\cdot \frac{1}{|B_i|\cdot |B_j|}\sum_{z_i\in B_i, z_j\in B_j} |\hat{f}_{z_i z_j} -  \rdp{f_{z_i}, f_{z_j}}|   &&\text{By triangle inequality} \nonumber \\
&\leq n\cdot \frac{1}{|B_i|\cdot |B_j|}\cdot |B_i|\cdot |B_j|\cdot\frac{\xi'}{n}   &&\text{By \eqref{eq:fxhatfy} and since }|C|\leq n \label{eq:Eij}
\end{align}
Since $r\leq k$ and by \eqref{eq:Eij} we get
\begin{align*}
 \left|||WXW-WH^THW||_2 \right| &\leq ||WXW-WH^THW||_F \\
 &\leq \sqrt{\sum_{i=1}^r\sum_{j=1}^r \left((WXW)(i,j)-(WH^THW)(i,j)\right)^2} \\
 &\leq k\cdot \xi'
\end{align*}
Thus for any vector $x\in \R^{r}$ with $||x||_2=1$ we have
\begin{equation}\label{eq:xixiA}
x^T(WH^THW)x -  k\cdot \xi' \leq  x^T(WXW)x\leq  x^T(WH^THW)x +  k\cdot\xi' 
\end{equation}
By Lemma \ref{lem:apxQQ-1} item \eqref{itm1:apxZZ}  for any vector $x\in \R^r$ with $||x||_2=1$ we have
\[
| x^T \left((HW)^T(HW)- I\right)x| \leq\frac{5\sqrt{\epsilon}}{\varphi} 
\]
Hence we have 
\begin{equation}\label{eq:A-1xi}
x^T (HW)^T(HW) x \geq 1- \frac{5\sqrt{\epsilon}}{\varphi}  \geq \frac{1}{2}
\end{equation}
Therfore by \eqref{eq:xixiA} and \eqref{eq:A-1xi} we get for any vector $x\in \R^r$ with $||x||_2=1$  we have
\begin{equation}
\label{eq:useinv}
(1-2\cdot k \cdot \xi') \cdot x^T(WH^THW)x  \leq  x^T(WXW)x\leq  (1+  2\cdot k\cdot\xi') \cdot x^T(WH^THW)x  
\end{equation}
Note that $WH^THW$ is a symmetric matrix. Also note that by definition of $X$ in line \ref{ln:Yij} of Algorithm \ref{alg:dot-apx-pi}, $X$ is a symmetric matrix, hence, $WXW$ is symmetric and positive semidefinit. Also note that  $WXW$ is spectrally close to $WH^THW$ and $I$, hence, $WXW$ is  invertible. Thus by \eqref{eq:useinv} and Lemma \ref{lem:neg-spectral-close} we have
\[(1-4\cdot k\cdot \xi') \cdot x^T(WH^THW)^{-1}x \leq  x^T(WXW)^{-1} x\leq (1+  4\cdot k\cdot \xi') \cdot x^T(WH^THW)^{-1} x  \]
Therefore by \eqref{eq:l2A-bnd} we have
\begin{equation}\label{eq:Ebnd}
||E||_2=||(WH^THW)^{-1} - (WXW)^{-1} ||_2 \leq 4\cdot k\cdot \xi' \cdot ||(WH^THW)^{-1}||_2 = 8\cdot k\cdot \xi'
\end{equation}
Putting \eqref{eq:Ebnd}, \eqref{eq:l2A-bnd}, \eqref{eq:ex-bnd}, \eqref{eq:ey-bnd}, \eqref{eq:ax-7}, \eqref{eq:ay-7} and\eqref{eq:fin-7} together, with probability at least $1-n^{-50}$  we have
\begin{align} 
& \left|   h_x^T W(WXW)^{-1} Wh_y - f_x^T H W(WH^THW)^{-1}W H^T f_y\right|   = \nonumber\\
&\| \left(\mathbf{a}_x +\mathbf{e}_x\right) \left(A+E\right) \left(\mathbf{a}_y +\mathbf{e}_y\right) -\mathbf{a}_xA\mathbf{a}_y\|_2 \leq \nonumber\\
& \|\mathbf{e}_x\|_2\|A\|_2\|\mathbf{a}_y\|_2+\|\mathbf{a}_x\|_2\|E\|_2\|\mathbf{a}_y\|_2+\|\mathbf{a}_x\|_2\|A\|_2\|\mathbf{e}_y\|_2 + \nonumber\\
&\|\mathbf{e}_x\|_2\|E\|_2\|\mathbf{a}_y\|_2+\|\mathbf{a}_x\|_2\|E\|_2\|\mathbf{e}_y\|_2+\|\mathbf{e}_x\|_2\|A\|_2\|\mathbf{e}_y\|_2+\|\mathbf{e}_x\|_2\|E\|_2\|\mathbf{e}_y\|_2
\nonumber \\
&\leq O\left( \xi'\cdot\frac{\sqrt{k}}{\sqrt{n}}\cdot k^{2.5}\cdot n^{-1/2 +40\epsilon/\varphi^2} \right) 
+O\left( k\cdot\xi'\cdot k^{5}\cdot n^{-1+80\epsilon/\varphi^2} \right) \nonumber \\
&+ O\left(\xi'\cdot\frac{\sqrt{k}}{\sqrt{n}}\cdot k\cdot\xi'\cdot k^{2.5}\cdot n^{-1/2+40\epsilon/\varphi^2} \right) 
+ O\left( \xi'^2\cdot\frac{k}{n} \right) 
+ O \left(\xi'^2\cdot\frac{k}{n} \cdot k\cdot \xi' \right) \nonumber\\ 
&\leq O\left( \frac{\xi'\cdot k^6\cdot n^{80\epsilon/\varphi^2}}{n}\right) \nonumber \\
&\leq \frac{1}{2}\cdot \frac{\xi}{n} \label{eq:xi2-finn}
\end{align}
The last inequality holds by setting $\xi'=\frac{\xi\cdot n^{(-80\epsilon/\varphi^2)}\cdot k^{-6}}{c}$ as per line \label{ln:xi''} of Algorithm \ref{alg:dot-apx-pi} where $c$ is a large enough constant to cancel the constant hidden in $O\left( \frac{\xi'\cdot k^6\cdot n^{80\epsilon/\varphi^2}}{n}\right)$.

Therefore with probability at least $1-n^{-98}\geq 1-n^{-50}$ we have
\begin{align}
&\left|\adp{f_x,\widehat{\Pi} f_y}-\rdp{f_x,\widehat{\Pi} f_y}\right| \nonumber \\
&\leq |\wh{f}_{xy} - \rdp{f_x,f_y}| + \left|    h_x^T W (WXW)^{-1} W h_y -   f_x^THW\left(WH^T HW\right)^{-1}WH^T f_y \right| \nonumber \\
&\leq \frac{\xi'}{n}+\frac{1}{2}\cdot \frac{\xi}{n} \nonumber\\
&\leq \frac{\xi}{n} &&\text{By \eqref{eq:fxhatfy}, \eqref{eq:xi2-finn}, and since }\xi'<\xi/2 \label{eq:lastguarantee}
\end{align}

Now let $\mathcal{E}$ be the event that for all $x,y \in V$ we have $|\adp{f_x, \wh{\Pi}f_y} - \rdp{f_x, \wh{\Pi}f_y} | \leq \frac{\xi}{n}$. Then by \eqref{eq:lastguarantee} and the union bound we get that $\mathcal{E}$ happens with probability at least $1 - n^{-48}$ and it is the claimed high probability event from the statement.

\end{proof}


\begin{corollary}
\label{corr:dotpi}
Let $G=(V,E)$ be a $d$-regular graph that admits a $(k,\varphi,\epsilon)$-clustering $C_1,\dots,C_k$. Let $k \geq 2$ be an integer, $\varphi \in (0,1)$, $\delta\in (0,1)$, $\frac{1}{n^5}<\xi<1$, $\frac{\e }{\varphi^2}$ be smaller than a positive absolute constant. Let $\mathcal{E}$ be the event that happens with probability $1 - n^{-48}$ that is guaranteed by Theorem~\ref{thm:dotpi}. Then conditioned on $\mathcal{E}$ the following conditions hold. 

Let $r\in[k]$.  Let $B_1,\ldots, B_r,B'$ denote multisets of points. Let $b=\max \{ |B_1|,\ldots, |B_r|, |B'|\} $. Let $\sigma:[r]\rightarrow[k]$ denote a mapping from the set $B$ to the cluster $C=\sigma(B)$. Suppose that for all $i\in[r]$, $B_i\subseteq \sigma(B_i)$ and for all $i\neq j\in [r]$, $\sigma(B_i)\neq \sigma(B_j)$. Let $\widehat{\mu}_i=\frac{1}{|B_i|}\cdot \sum_{z\in B} f_z$ for all $i\in[r]$, and let $\widehat{\mu}=\frac{1}{|B'|}\cdot \sum_{z\in B_i} f_z$. Suppose that for each $i\in[r]$, $||\wh{\mu}_i-\mu_{\sigma(i)}||_2\leq \frac{\sqrt{\epsilon}}{20\cdot k\cdot\varphi}||\mu_i||_2$. Let $\wh{\Pi}$ is defined as a orthogonal projection onto then $span(\{\wh{\mu}_1,\dots, \wh{\mu}_r \})^{\perp}$.  Then  the following hold:
\begin{enumerate}
\item There exits an algorithm that runs in time $b^3\cdot(\frac{k}{\xi})^{O(1)} \cdot n^{\delta+O(\epsilon/\varphi^2)}\cdot \frac{(\log n)^2}{\varphi^2}$ and for any $x\in V$ returns a value $\adp{f_x,\widehat{\Pi} \wh{\mu}}$ such that \[\left|\adp{f_x,\widehat{\Pi} \wh{\mu}}-\rdp{f_x,\widehat{\Pi} \wh{\mu}}\right|\leq \frac{\xi}{n}\text{.}\] \label{itm:dtpi1}
\item There exits an algorithm that runs in time $b^4\cdot (\frac{k}{\xi})^{O(1)} \cdot n^{\delta+O(\epsilon/\varphi^2)}\cdot \frac{(\log n)^2}{\varphi^2}$ and returns a value $\an{\widehat{\Pi} \widehat{\mu}}^2$ such that $\left| \an{\widehat{\Pi} \widehat{\mu}}^2 -  ||\widehat{\Pi} \widehat{\mu}||_2^2\right|\leq \frac{\xi}{n}$. \label{itm:dtpi2}
\end{enumerate}

\end{corollary}

\begin{proof}
\textbf{Proof of \ref{itm:dtpi1}:}
To compute $\adp{f_x,\widehat{\Pi} \widehat{\mu}}$ we call Algorithm \ref{alg:dot-apx-pi}, $b$ times in the following way:
\begin{equation}
\adp{f_x,\widehat{\Pi} \widehat{\mu}}:=  \frac{1}{|B|}\cdot\sum_{y\in B} \textsc{DotProductOracleOnSubspace}(G,x,y, \delta,  \mathcal{D}, \xi, B_1,\ldots, B_r) 
\end{equation}
The runtime of Algorithm \ref{alg:dot-apx-pi} is $b^2\cdot (\frac{k}{\xi})^{O(1)} \cdot n^{\delta+O(\epsilon/\varphi^2)}\cdot (\log n)^2\cdot \frac{1}{\varphi^2}$, thus the runtime of computation of  $\adp{f_x,\widehat{\Pi} \widehat{\mu}}$ is $b^3\cdot (\frac{k}{\xi})^{O(1)} \cdot n^{\delta+O(\epsilon/\varphi^2)}\cdot (\log n)^2\cdot \frac{1}{\varphi^2}$. Moreover by Theorem \ref{thm:dotpi} and the assumption that $\mathcal{E}$ holds we have
\begin{align*}
\left|\adp{f_x,\widehat{\Pi} \wh{\mu}}-\rdp{f_x,\widehat{\Pi} \wh{\mu}}\right| &= \left| \frac{1}{|B'|}\sum_{y\in B'} \adp{f_x,\widehat{\Pi} y}-\rdp{f_x,\widehat{\Pi} \wh{\mu}}\right| \\
&\leq \frac{1}{|B'|}\sum_{y\in B'}  \left|  \adp{f_x,\widehat{\Pi} y}-\rdp{f_x,\widehat{\Pi} \wh{\mu}}\right| &&\text{By triangle inequality} \\
&\leq \frac{1}{|B'|}\cdot |B'|\cdot \frac{\xi}{n} &&\text{By Theorem \ref{thm:dotpi}}\\
&\leq \frac{\xi}{n}
\end{align*}
\textbf{Proof of \ref{itm:dtpi2}:}
To compute $\an{\widehat{\Pi} \widehat{\mu}}^2$ we call the procedure from item \eqref{itm:dtpi1} $b$ times in the following way:
\begin{equation}
\an{\widehat{\Pi} \widehat{\mu}}^2:=  \frac{1}{|B|}\cdot\sum_{x\in B} \adp{f_x,\widehat{\Pi} \widehat{\mu}}  \text{.}
\end{equation}
The runtime of the procedure from item \eqref{itm:dtpi1} is $b^3\cdot (\frac{k}{\xi})^{O(1)} \cdot n^{\delta+O(\epsilon/\varphi^2)}\cdot (\log n)^2\cdot \frac{1}{\varphi^2}$, thus the runtime of computation of  $\adp{f_x,\widehat{\Pi} \widehat{\mu}}$ is $b^4\cdot (\frac{k}{\xi})^{O(1)} \cdot n^{\delta+O(\epsilon/\varphi^2)}\cdot (\log n)^2\cdot \frac{1}{\varphi^2}$. Moreover by  item \eqref{itm:dtpi1} we have
\begin{align*}
\left| \an{\widehat{\Pi} \widehat{\mu}}^2 -  ||\widehat{\Pi} \widehat{\mu}||_2^2\right| &= \left|  \adp{\wh{\mu},\widehat{\Pi} \widehat{\mu}} - \rdp{\wh{\mu},\widehat{\Pi} \widehat{\mu}} \right| 
\\&= \left| \frac{1}{|B'|}\cdot\sum_{x\in B'} \adp{f_x,\widehat{\Pi} \widehat{\mu}} - \frac{1}{|B'|}\cdot\sum_{x\in B'} \rdp{f_x,\widehat{\Pi} \widehat{\mu}}
\right|  \\
&\leq \frac{1}{|B'|}\cdot  \sum_{x\in B'} \left| \adp{f_x,\widehat{\Pi} \widehat{\mu}} - \sum_{x\in B'} \rdp{f_x,\widehat{\Pi} \widehat{\mu}}  \right|  &&\text{By triangle inequality} \\
&\leq \frac{1}{|B'|}\cdot |B'|\cdot \frac{\xi}{n} &&\text{By item \eqref{itm:dtpi1}} \\
&\leq \frac{\xi}{n} \text{.}
\end{align*}
\end{proof}

\newpage
\section{The main algorithm and its analysis}\label{sec:algo}

In this section we show that, by having access to approximate spectral dot-products for a $(k,\varphi,\epsilon)$-clusterable graph $G$, we can assign each vertex in $G$ to a cluster in sublinear time so that the resulting collection of clusters is, with high probability, a good approximation of a $(k,\varphi,\epsilon)$-clustering of $G$. In particular, we can show that the fraction of wrong assignments per cluster is at most $C \cdot \frac{\e}{\varphi^3} \cdot \log(k)$, for some constant $C>0$. In the next subsection we describe our algorithm then in the  remaining part of the section we present its analysis.

\subsection{The Algorithm (Partitioning Scheme, Algorithm~\ref{alg:ballcarving})}\label{sec:algorithmhighlevel}

We first present an idealized version of the sublinear clustering scheme defined by Algorithm~\ref{alg:ballcarving} and Algorithm~\ref{alg:findrepresentatives}. In this section to simplify presentation we assume $\varphi$ to be constant.


The algorithm can be thought of as consisting of $3$ parts. The first part, described in paragraph \textbf{Idealized Clustering Algorithm}, is a procedure that explicitly, in iterative fashion, produces a $k$-clustering of $G$. More precisely it recovers clusters in $O(\log(k))$ stages, where for every $i$ after the $i$-th stage at most $k/2^i$ clusters are left unrecovered. The algorithm can be thought of as a version of carving of halfspaces in $\R^k$ and it relies on the knowledge of cluster means $\mu_1, \dots, \mu_k$ (recall that $\mu_i = \frac{1}{|C_i|}\sum_{x \in C_i} f_x$). That is why in paragraph \textbf{Finding approximate centers} we show how to compute approximations of $\mu_i$'s. To find good approximation to $\mu_i$'s we need to test many candidate sets $\{ \wh{\mu}_1, \dots, \wh{\mu}_k \}$, which also means considering many candidate clusterings. This is a problem as we want our procedure to run in sublinear time but the idealized partitioning algorithm constructs clusterings explicitly! To solve this we explain in paragraph \textbf{Verifying a clustering} how to emulate the partitioning algorithm to test that, for a set of $\{ \wh{\mu}_1, \dots, \wh{\mu}_k \}$, it indeed induces a good clustering.  

\noindent\paragraph{Idealized Clustering Algorithm.} Assume that the we have access to cluster means $\{\mu_1, \dots, \mu_k \}$ and dot product evaluations. The algorithm proceeds in $O(\log(k))$ stages, in the first stage it considers $k$ candidate sets $\wh{C}_i$, where $x \in \wh{C}_i$ iff $f_x$ has big correlation with $\mu_i$ but small correlation with all other $\mu_j$'s. More precisely $x \in \wh{C}_i$ iff:
$$\rdp{f_x, \mu_i} \geq 0.93 \rn{\mu_i}^2 \text{ and for all } j \neq i \rdp{f_x, \mu_j} < 0.93 \rn{\mu_j}^2 \text{.}$$
Note that by definition all these clusters are disjoint. Moreover we are able to show (see Lemma~\ref{lem:induction}) that at least $k/2$ out of $\wh{C}_i$'s are good approximate clusters, that is for each one of them there exists $j$ such that $|\wh{C}_i \triangle C_j| \leq O (\e ) \cdot |C_j|$ . At this point we return these good clusters, remove the corresponding vertices from the graph, remove the corresponding $\mu$'s from the set $\{\mu_1, \dots, \mu_k \}$ of still alive centers and proceed to the next stage.

In the next stage we restrict our attention to a lower dimensional subspace $\Pi$ of $\R^k$. Intuitively we want to project out all the directions corresponding to the removed cluster centers. Recall that $\mu_i$'s are close to being orthogonal (see Lemma~\ref{lem:dosubspace} and \ref{lem:dotmu}) so projecting the returned directions out is almost equivalent to considering the subspace $\Pi := \text{span}(\{\mu_1, \dots, \mu_b \})$, where $\{\mu_1, \dots, \mu_b \}$ is the set of still alive $\mu$'s. Now the algorithm considers $b$ candidate clusters where the condition for $x$ being in a cluster $i$ changes to: 
$$\rdp{f_x, \Pi\mu_i} \geq 0.93 \rn{\Pi\mu_i}^2 \text{ and for all } j \in [b], j \neq i \rdp{f_x, \Pi\mu_j} < 0.93 \rn{\Pi\mu_j}^2 \text{.}$$
We are still able to show (also Lemma~\ref{lem:induction}) that at least $b/2$ out of them are good approximate clusters. That is for each $i$ there exists $j$ such that $|\wh{C}_i \triangle C_j| \leq O (\e ) \cdot |C_j|$ but this time the constant hidden in the $O$ notation is bigger than in the first stage. In general at any stage $t$ the bound degrades to $O (\e \cdot t )$. At the end of the stage we proceed in a similar fashion by returning the clusters, removing the corresponding vertices and $\mu$'s and considering a lower dimensional subspace of $\Pi$ in the next stage.

The algorithm continues in such a fashion for $O(\log(k))$ steps, as we guarantee that in each stage at least half of the remaining cluster means is removed. Thus the final guarantee is: there exists a permutation $\pi$ on $k$ elements such that for every $i$:
$$|\wh{C}_{\pi(i)} \triangle C_i| \leq O \left(\e \log(k) \right) \cdot |C_i| \text{.}$$
The decreasing (in the inclusion sense) sequence of subspaces $(\Pi_1, \dots, \Pi_{\log(k)})$ corresponds to the subspaces constructed in Algorithm~\ref{alg:ballcarving}, while this offline algorithm as a whole corresponds to the sublinear Algorithm~\ref{alg:findrepresentatives} that implicitly tries to construct a sequence of subspaces that (with respect to Algorithm~\ref{alg:ballcarving}) defines a good clustering.


\noindent\paragraph{Finding approximate centers.} Note that cluster means are defined by the clustering, so it may seem that finding approximate means is a difficult operation. However, there is a relatively simple solution to this. In 
Algorithm~\ref{alg:findrepresentatives} we find approximate cluster means by sampling $O(\frac{\varphi^2}{\e} k^4 \log(k))$ points, guessing cluster memberships and considering the means of the samples as cluster centers. We use that
the mean of a random sample of a cluster is typically close to the true mean of its cluster and so our sample means will provide a good estimation of the true means. We also remark that sampling a single vertex
from each cluster does not seem to provide a sufficiently good estimate, i.e. we require to take the mean of a sample \emph{set}.

\noindent\paragraph{Verifying a clustering.} We also need a procedure that given an implicit sequence of  subspaces $(\Pi_1, \dots, \Pi_{\log(k)})$ checks whether they indeed define (via Algorithm~\ref{alg:ballcarving}) a good clustering. In fact, for every guess of cluster centers and the corresponding (as implicitly created by Algorithm~\ref{alg:findrepresentatives}) sequence of $\Pi$'s we need to be able to check efficiently if the resulting clustering is a good approximation of a $(k,\varphi,\epsilon)$-clustering. Since we would like to do this in sublinear time as well,
we need to do this verification by random sampling. Then we design a procedure that consists of two steps. In a first step, we check if the cluster sizes are not too small. This is only a technical step, which is needed to
make sure that the later steps work. The main step is to test whether every cluster has small outer conductance (Algorithm~\ref{alg:estimateconductande}). In order to do so, we sample vertices uniformly at random and check
whether they are contained in the cluster that is currently checked. If this is the case, we sample a random edge incident to the sample vertex. This way, we obtain a random edge incident to a random vertex from the current
cluster (this follows since the conditional distribution is uniform over the cluster). We use standard concentration bounds to prove that we get a good approximation.

\paragraph{}In the partitioning scheme and in the analysis a useful definition are subsets of vertices called threshold sets. A threshold set of a point $y$ is the set of vertices with dot products (or approximate dot product) with $y$ being above a specific threshold, more formally:
\begin{definition}[\textbf{Threshold sets}]\label{def:thresholdsets}
Let $G=(V,E)$ be a $(k,\varphi,\epsilon)$-clusterable graph (\textit{as in Definition~\ref{def:clusterable}}). Recall that $f_x=F  \mathds{1}_x$. For $y \in \mathbb{R}^k, \theta \in \mathbb{R}^{+}$ we define:
$$\Cr{y,\theta} := \{x \in V : \rdp{f_x,y} \geq \theta \rn{y}^2 \}$$
\end{definition}

\begin{definition}[\textbf{Approximate threshold sets}]\label{def:apxthreshold sets}
Let $G=(V,E)$ be a $(k,\varphi,\epsilon)$-clusterable graph (\textit{as in Definition~\ref{def:clusterable}}). Recall that $f_x=F  \mathds{1}_x$. For $\theta \in \mathbb{R}^{+}$ and $y \in \mathbb{R}^k$ such that $y = \wh{\Pi}(\wh{\mu})$, where $\wh{\Pi}$ is the orthogonal projection onto $span(\{\wh{\mu}_1,\dots, \wh{\mu}_b \})^{\perp}$ and each $\wh{\mu}, \wh{\mu}_1, \dots, \wh{\mu}_b$ is an average of a set of embedded vertices:
\begin{equation}\label{eq:capxdef}
\Ca{y,\theta} := \{x \in V : \adp{f_x,y} \geq \theta \an{y}^2\}\text{.}
\end{equation}
Recall that a discussion of how $\langle \cdot, \cdot \rangle_{apx}$ and $\|\cdot\|_{apx}$ are  computed is presented in Section~\ref{sec:dotproductcomp}. 
\end{definition}


\begin{algorithm}
\caption{\textsc{HyperplanePartitioning}($x, (T_1,T_2,\ldots, T_{b}))$ \newline \text{ }
\Comment $T_i$'s are sets of $\widehat{\mu}_j$ where $\widehat{\mu}_j$'s are given as sets of points \newline \text{ } \Comment  see Section~\ref{sec:dotproductcomp} for the reason of such representation}\label{alg:ballcarving}
\begin{algorithmic}[1]
    \For{$i=1$ to $b$}
    	\State Let $\Pi$ be the projection onto the $\text{span} (\bigcup_{j<i} T_j)^{\perp}$.
    	\State Let $S_i=\bigcup_{j\geq i} T_j$  
    	\For{$\hat{\mu}\in T_i$}
    		
        	\If{$x\in   \Ca{\Pi \widehat{\mu},0.93} \setminus \bigcup_{\widehat{\mu}'\in S_i\setminus \{\widehat{\mu}\}} \Ca{\Pi \widehat{\mu}',0.93}  $}  \label{ln:dot-x-mu1}\Comment see \eqref{eq:capxdef} for definition of $\Ca{y,\theta}$
            \State \Return $\widehat{\mu}$
        	\EndIf
        \EndFor
    \EndFor
\end{algorithmic}
\end{algorithm}

\textsc{HyperplanePartitioning} is the algorithm that, after preprocessing, is used to assign vertices to clusters. In the preprocessing step (see \textsc{ComputeOrderedPartition} in Section~\ref{sec:realcenterswork}) an ordered partition $(T_1, \dots, T_b)$ of approximate cluster means $\{\wh{\mu}_1, \dots, \wh{\mu}_k \}$ is computed. \textsc{HyperplanePartitioning} invoked with this ordered partition as a parameter induces a collection of clusters as follows:

\begin{definition}[\textbf{Implicit clustering}]\label{def:implicitclustering}
For an ordered partition $(T_1, \dots, T_b)$ of approximate cluster means $\{\wh{\mu}_1, \dots, \wh{\mu}_k \}$ we say that \textbf{$(T_1, \dots, T_b)$ induces a collection of clusters $\{\wh{C}_{\wh{\mu}_1}, \dots, \wh{C}_{\wh{\mu}_k}\}$} if for all $i \in [k]$:
$$\wh{C}_{\wh{\mu}_i} = \left\{ x \in V : \textsc{HyperplanePartitioning}(x,(T_1, \dots, T_b)) = \wh{\mu}_i \right\} \text{.}$$
\end{definition}

\begin{remark} 
Ordered partition $(T_1, \dots, T_b)$, precomputed in the preprocessing step (assuming access to $\{\mu_1, \dots, \mu_k\}$), will correspond to the \textbf{Idealized Clustering Algorithm} in the following sense. Number of sets in the partition (i.e. $b$) corresponds to the number of stages of \textbf{Idealized Clustering Algorithm} and for every $i \in [b]$ $T_i$ contains exactly the $\mu$'s returned in stage $i$.
\end{remark}


In the rest of this section we explain how to compute an ordered partition $(T_1, \dots, T_b)$ of a set of approximate centers $(\widehat{\mu}_1,\widehat{\mu}_2, \dots, \widehat{\mu}_k )$ such that the induced clustering $\{\wh{C}_{\wh{\mu}_1}, \dots, \wh{C}_{\wh{\mu}_k} \}$ satisfies that there exists a permutation $\pi$ on $k$ elements such that for all $i \in [k]$:
$$\left|\wh{C}_{\wh{\mu}_{i}} \triangle C_{\pi(i)}\right| \leq O \left(\frac{\e}{\varphi^3}  \cdot \log(k) \right)|C_{\pi(i)}| \text{.}$$  
We start, in Subsection~\ref{sec:bound_int}, by studying geometric properties of our clustering instance. Recall, that we denote with $\mu_i$ the center of cluster $C_i$ in the spectral embedding. We show that, for specific choices of $\theta$, the threshold sets of $\mu_i$ have large intersection with the cluster $C_i$ and small intersections with all other cluster $C_j$. This fact intuitively suggests that our partitioning algorithm works. Unfortunately, as discussed in the technical overview, this is not enough to prove a per cluster guarantee. For this reason in Subsection~\ref{sec:realcenterswork} we analyze the overlap structure of $\{\Cr{\mu_1, \theta}, \dots, \Cr{\mu_k, \theta}\}$ more carefully and we give an algorithm (see \textsc{ComputeOrderedPartition}) that given real centers $\{\mu_1, \dots, \mu_k\}$ and access to exact dot product evaluations computes an ordered partition of $\{\mu_1, \dots, \mu_k\}$ that induces a valid clustering. In Subsection~\ref{sec:findthecenters} we present an algorithm that guesses the cluster memberships for a set of randomly selected nodes and, using those guesses, approximates cluster centers. Interestingly, we can show, in Subsection~\ref{sec:approx-muis},  that for the set of correct guesses the algorithm returns a good approximation of the cluster centers. Finally in Subsection~\ref{sec:h_works} we show that we can find an ordered partition that induces a good clustering even if we have access only to approximate quantities. That is we show that even if we have access only to approximate means $\{\wh{\mu}_1, \dots, \wh{\mu}_k \}$ and the dot product evaluations are only approximately correct then we can find an ordered partition $(T_1, \dots, T_b)$ that induces a good collection of clusters. The last ingredient is to show that we are able to check if the clustering induced by a specific ordered partition is good. To solve this problem, we design an efficient and simple sampling algorithm which is also analyzed in Subsection~\ref{sec:h_works}.

\subsection{Bounding intersections of $C_{\mu_i,\theta}$ with true clusters $C_i$}\label{sec:bound_int}

In this subsection we show that, for specific choices of $\theta$, the threshold sets of $\mu_i$ (recall that $\mu_i$'s are cluster means in the spectral embedding) have large intersection with $C_i$ and small intersections with other clusters. The main idea behind the proof is to use the bounds on dot product of cluster centers presented in Lemma~\ref{lem:dotmu}. In particular, we use Lemma~\ref{lem:dirvariance} to relate $\frac{\epsilon}{\varphi^2}$ with the directional variance of the spectral embedding in the direction of $\mu_i$ (i.e. $\sum_{x \in C_i} \langle f_x - \mu_i, \alpha \rangle^2$). Then we use the definition of threshold set to upper and lower bound $\langle f_x, \frac{\mu_i}{\|\mu_i\|} \rangle$ and Lemma~\ref{lem:dotmu} to upper and lower bound the dot product between cluster centers. By combining the bounds we obtain the following result:

\begin{lemma}\label{lem:mostinset}
Let $k \geq 2$, $\varphi \in (0,1)$ and $\frac{\e}{\varphi^2}$ be smaller than a sufficiently small constant.
Let $G=(V,E)$ be a $d$-regular graph that admits a $(k,\varphi,\epsilon)$-clustering $\{ C_1, \dots, C_k \}$. If $\mu_i$'s are cluster means then the following conditions hold. Let $S \subset \{{\mu}_1, \dots, {\mu}_k\}$. Let $\Pi$ denote the orthogonal projection matrix on to the $ span(S)^\perp$. Let $\mu \in \{{\mu}_1, \dots, {\mu}_k\} \setminus S$. Let $C$ denote the cluster corresponding to the center ${\mu}$. Let $$\widehat{C} := \{ x \in V :  \rdp{\Pi f_x, \Pi{\mu}} \geq 0.96 \| \Pi\mu \|_2^2 \}$$ then we have:
$$\left| C \setminus \widehat{C} \right|  \leq  \frac{10^4 \e}{\varphi^2} |C| \text{.}$$
\end{lemma}

\begin{proof}
Let $x \in C \setminus \widehat{C}$. Then:
\begin{align*}
\left| \rdp{\mu-f_x, \frac{\Pi{\mu}}{\|\Pi {\mu}\|_2}} \right|
&=
\left|\rdp{\Pi( {\mu}-f_x), \frac{\Pi{\mu}}{\|\Pi{\mu}\|_2}} \right|    \\
&\geq 0.04 \cdot \|\Pi {\mu} \|_2  && \text{Since $\rdp{\Pi f_x, \Pi{\mu}} <0.96 \| \Pi\mu \|_2^2$} \\
&\geq 0.04 \cdot \left(1 - 24\frac{\sqrt{\e}}{\varphi}\right) ||\mu||_2 && \text{By Lemma~\ref{lem:dosubspace}} \\
&\geq 
0.04 \cdot \left(1 - 40\frac{\sqrt{\e}}{\varphi}\right) \sqrt{\frac{1}{|C|}} && \text{By Lemma~\ref{lem:dotmu}} \\
&\geq 0.02 \cdot \sqrt{\frac{1}{|C|}} && \text{Since } \frac{\e}{\varphi^2} \text{ is  sufficiently small}
\end{align*}
Then by Lemma~\ref{lem:dirvariance} applied to direction $\alpha = \frac{\Pi {\mu}}{\|\Pi {\mu}\|_2}$ we have
$
\sum_{i=1}^k \sum_{x \in C_i}  \rdp{f_x - \mu_i,  \alpha}^2 \leq \frac{4\epsilon}{\varphi^2}.
$
On the other hand
\[
\frac{4\epsilon}{\varphi^2}\geq \sum_{i=1}^k \sum_{x \in C_i}  \rdp{f_x - \mu_i,  \alpha}^2 \geq \sum_{x\in C\setminus  \widehat{C}} \rdp{f_x-\mu, \frac{\Pi{\mu}}{\|\Pi {\mu}\|_2}}^2 \geq 0.0004\cdot \frac{| C \setminus \widehat{C}|}{|C|}  .
\]
Using the above we conclude with $| C \setminus \widehat{C}| \leq 10^4 \frac{\e}{\varphi^2} |C|$.
\end{proof}


\begin{remark}\label{rem:twothresholds}
Notice that the constants in Lemma~\ref{lem:notalostfromoutside} are different, they are equal $0.96$ and $0.9$. The reason is that the real tests for membership in Algorithm~\ref{alg:ballcarving} are performed with constant $0.93$ and the slacks are needed as we have access only to approximate dot products. See \eqref{eq:goodintersection} for the formal reason.
\end{remark}

\begin{lemma}\label{lem:notalostfromoutside}
Let $k \geq 2$, $\varphi \in (0,1)$ and $\frac{\e}{\varphi^2}$ be smaller than a sufficiently small constant.
Let $G=(V,E)$ be a $d$-regular graph that admits a $(k,\varphi,\epsilon)$-clustering $\{ C_1, \dots, C_k \}$. If $\mu_i$'s are cluster means then the following conditions hold. Let $S \subset \{{\mu}_1, \dots, {\mu}_k\}$. Let $\Pi$ denote the projection matrix on to $ span(S)^\perp$. Let $\mu \in \{{\mu}_1, \dots, {\mu}_k\} \setminus S$. Let $C$ denote the cluster corresponding to the center ${\mu}$. Let $$\widehat{C} := \{ x \in V :  \rdp{\Pi f_x, \Pi{\mu}} \geq 0.9 \| \Pi\mu \|_2^2 \}$$ then we have:
$$\left| \widehat{C} \cap (V \setminus C) \right|  \leq 100 \frac{\e}{\varphi^2} |C| \text{.}$$
\end{lemma}

\begin{proof}
Let $x \in \widehat{C} \cap (V \setminus C)$. Then there exists cluster $C' \neq C$ such that $x \in C'$. Let $\mu'$ be the cluster mean of $C'$. Then:
\begin{align}
\left| \rdp{f_x -\mu', \frac{\Pi{\mu}}{\|\Pi{\mu}\|_2}} \right|
&\geq \left| \rdp{\Pi f_x, \frac{\Pi\mu}{\|\Pi\mu\|_2}} \right| - \left| \rdp{\Pi\mu', \frac{\Pi\mu}{\|\Pi\mu\|_2}} \right| && \text{By triangle inequality} \nonumber \\
&\geq 0.9 \| \Pi\mu \|_2 - \left| \rdp{\Pi \mu', \frac{\Pi\mu}{\| \Pi\mu\|_2}} \right| && \text{As $ x \in \widehat{C}$} \nonumber
\end{align}
Note that either $\mu' \in S$ and then $\Pi\mu' = 0$ and in turn $|\rdp{\Pi\mu', \Pi\mu} | = 0$ or $\mu' \not\in S$ and then $|\rdp{\Pi\mu', \Pi\mu} | \leq \frac{60\sqrt{\e}}{\varphi^2}\frac{1}{\sqrt{|C| \cdot |C'|}}$ by Lemma~\ref{lem:dosubspace}. Thus we have
\begin{align}
\left| \rdp{f_x -\mu', \frac{\Pi{\mu}}{\|\Pi{\mu}\|_2}} \right|
&\geq 0.9 \| \Pi\mu \|_2 - \frac{60\sqrt{\e}}{\varphi^2}\frac{1}{\sqrt{|C| \cdot |C'|}} \frac{1}{\|\Pi\mu\|_2} &&  \nonumber \\
&\geq 0.8  \frac{1}{\sqrt{|C| }} - \frac{120\sqrt{\e}}{\varphi^2}\frac{1}{\sqrt{|C| \cdot |C'|}}\cdot \sqrt{|C|} && \text{by Lemma~\ref{lem:dosubspace} and Lemma \ref{lem:dotmu},  $\|\Pi\mu\|_2\geq \frac{1}{2\cdot \sqrt{|C|}}$} \nonumber \\
&\geq 0.2 \sqrt{\frac{1}{|C|}} && \text{Since } \frac{\e}{\varphi^2} \text{ sufficiently small and } \frac{|C|}{|C'|} \text{ constant} \label{eq:lowerboundinproj}
\end{align}
Then by Lemma~\ref{lem:dirvariance} applied to direction $\alpha = \frac{\Pi {\mu}}{\|\Pi {\mu}\|_2}$ we have
$
\sum_{i=1}^k \sum_{x \in C_i}  \rdp{f_x - \mu_i,  \alpha}^2 \leq \frac{4\epsilon}{\varphi^2}.
$
On the other hand using \eqref{eq:lowerboundinproj} we get
\[
\frac{4\epsilon}{\varphi^2}\geq \sum_{i=1}^k \sum_{x \in C_i}  \rdp{f_x - \mu_i,  \alpha}^2 \geq \sum_{x\in \widehat{C} \cap (V \setminus C)} \rdp{f_x-\mu_x, \frac{\Pi{\mu_x}}{\|\Pi {\mu_x}\|_2}}^2 \geq 0.04\cdot \frac{|  \widehat{C} \cap (V \setminus C)|}{|C|}.
\]
Therefore we have  $\left| \widehat{C} \cap (V \setminus C) \right|  \leq 100 \frac{\e}{\varphi^2} |C|$.
\end{proof}

\subsection{Partitioning scheme works with \textit{exact} cluster means \& dot products}\label{sec:realcenterswork}

The goal of this section is to present the main ideas behind the algorithms and the analysis. In this section we make a couple of simplifying assumptions. We assume that:
\begin{itemize}
    \item We have access to real centers $\{\mu_1, \dots, \mu_k\},$
    \item Dot products computed by the algorithm are exact,
    \item A test, that relies on computing outer-conductance of candidate sets, for assessing the quality of clusters is perfect.
\end{itemize}
Whenever we use one (or more) of these assumptions we state them explicitly in the Lemmas. Later in Section~\ref{sec:h_works} we show that we can get rid of all of these assumptions.

In the previous section we showed geometric properties of the threshold sets. Recall that threshold sets are defined as follows:
$$\Cr{y,\theta} := \{x \in V : \rdp{f_x,y} \geq \theta \rn{y}^2 \}\text{.} $$
In this section, using these properties of threshold sets, we show an algorithm that given exact centers, access to real dot products and a perfect primitive for computing outer-conductance computes an ordered partition $(T_1, \dots, T_b)$ of $\{\mu_1, \dots, \mu_k \}$ such that $(T_1, \dots, T_b)$ induces a good collection of clusters.

\begin{algorithm}[H]
\caption{\textsc{ComputeOrderedPartition}($G,\widehat{\mu}_1,\widehat{\mu}_2, \dots, \widehat{\mu}_k,s_1,s_2)$  \text{ } \Comment $\widehat{\mu}_i$'s given as sets of points \newline \text{ } \Comment $s_1$ is \# sampled points for size estimation \newline \text{ } \Comment $s_2$ is \# of sampled points for conductance estimation}\label{alg:testmus}
\label{alg:cluster}
\begin{algorithmic}[1]
		\State $S := \{\hat{\mu}_1, \dots, \hat{\mu}_k \}$
		\For{$i = 1$ to $\lceil \log(k) \rceil$} \label{ln:testcentersmainloop}
		    \State $T_i := \emptyset$
		    \For{$\widehat{\mu} \in S$} \label{testcentersforloop}

				\State $\psi := \textsc{OuterConductance}\left(G,\widehat{\mu}, (T_1,T_2,\ldots, T_{i-1}), S,s_1,s_2\right)$\label{ln:estout} \Comment{Algorithm \ref{alg:estimateconductande}}
				\If{$ \psi \leq O( \frac{\epsilon}{\varphi^2}\cdot \log(k) )$} \label{testcenterscondactancetest}
		            \State $T_i := T_i\cup \{\widehat{\mu}\}$
		          
		       \EndIf
		    \EndFor
		      \State $S := S \setminus T_i$
		    \If{$S = \emptyset$} 
             \State \Return $(\textsc{True},(T_1, \dots, T_i) )$	\label{ln:testcenterstruereturn}
			\EndIf
		\EndFor
		\State \Return $(\textsc{False}, \bot)$
\end{algorithmic}
\end{algorithm}

To explain and analyze \textsc{ComputeOrderedPartition} we first need to introduce another algorithm and some definitions.

\begin{definition}
For a set $\{a_1, \dots, a_i \}$ we say a sequence $(S_1, \dots, S_p)$ is an ordered partial partition of $\{a_1, \dots, a_i \}$ if:
\begin{itemize}
    \item $\bigcup_{j \in [p]} S_j \subseteq \{a_1, \dots, a_i \}$,
    \item $S_i$'s are pairwise disjoint.
\end{itemize}
\end{definition}

\begin{algorithm}
\caption{\textsc{IsInside}($x,\widehat{\mu},(T_1,T_2,\ldots, T_{b}), S$) \text{ } \newline \text{ } \Comment $T_i$'s are sets of $\widehat{\mu}_j$ where $\widehat{\mu}_j$'s are given as sets of points \newline \text{ } \Comment see Section~\ref{sec:dotproductcomp} for the reason of such representation \newline \text{ } \Comment $S = $ set of not yet processed centers, $\wh{\mu} \in S$}\label{alg:inside}
\begin{algorithmic}[1]
		\For{$i=1$ to $b$}
    	\State Let $\Pi$ be the projection onto the $\text{span} (\bigcup_{j<i} T_j)^{\perp}$.  
    	\State Let $S_i=\left(\bigcup_{j\geq i} T_j \right)\cup S$
    	\For{$\hat{\mu}_i\in T_i$}
    		
        	\If{$x\in   \Ca{\Pi \widehat{\mu}_i,0.93} \setminus \bigcup_{\widehat{\mu}'\in S_i\setminus \{\widehat{\mu}_i\}} \Ca{\Pi \widehat{\mu}',0.93}  $}  \label{ln:dot-x-mu2}\Comment see \eqref{eq:capxdef} for definition of $\Ca{y,\theta}$
            \State \Return \textsc{False}
        	\EndIf
        \EndFor
    \EndFor
    \State Let $\Pi$ be the projection onto the $\text{span} (\bigcup_{j\leq b} T_j)^{\perp}$.  
     \If{$x\in   \Ca{\Pi \widehat{\mu},0.93} \setminus \bigcup_{\widehat{\mu}'\in  S\setminus \{\widehat{\mu}\}} \Ca{\Pi \widehat{\mu}',0.93}  $}  \label{ln:dot-x-mu}\Comment see \eqref{eq:capxdef} for definition of $\Ca{y,\theta}$
            \State \Return \textsc{True}
        	\EndIf
	\State \Return \textsc{False}
\end{algorithmic}
\end{algorithm}

Intuitively Algorithm \textsc{IsInside} emulates \textsc{ClassifyByHyperplanePartitioning} on ordered partial partition $(T_1, \dots, T_b)$. This intuition is made formal, after introducing Definition~\ref{def:candidateclusters}, in Remark~\ref{rem:equivalenceofdef}. For this we need additional notation for clusters that are implicitly created by \textsc{IsInside}. We define:

\begin{definition}[\textbf{Candidate cluster}]\label{def:candidateclusters}
For an ordered partial partition $P = (T_1, \dots, T_p)$ of approximate cluster means $\{\wh{\mu}_1, \dots, \wh{\mu}_k \}$ and $\wh{\mu} \in \{\wh{\mu}_1, \dots, \wh{\mu}_k \} \setminus \bigcup_{i \in [p]} T_i$ we say that $\wh{C}_{\wh{\mu}}^P$ is a \textbf{candidate cluster corresponding to $\wh{\mu}$ with respect to} \textbf{$P$} if:
$$\wh{C}_{\wh{\mu}}^{P} = \left\{ x \in V : \textsc{IsInside} \left(x,\wh{\mu},P, \{\wh{\mu}_1, \dots, \wh{\mu}_k \} \setminus \bigcup_{i \in [p]} T_i \right) = \textsc{True} \right\} \text{.}$$
Furthermore we define:
$V^{P} := V \setminus \bigcup_{j < p} \bigcup_{\wh{\mu} \in T_j} \wh{C}_{\wh{\mu}}^{(T_1, \dots, T_{j-1})}$.
\end{definition}

Algorithm \textsc{IsInside} receives a vertex $x$, the centre of a cluster $\widehat{\mu}$, and an ordered partial partition, then it tests if vertex $x$ is not recovered by any of the previous stages (see line \eqref{ln:dot-x-mu2} of Algorithm \ref{alg:inside}) and can be recovered at the current stage using $\widehat{\mu}$. More formally, it can be recovered at the current stage if it only belongs to the candidate cluster corresponding to the center $\widehat{\mu}$ (see line \eqref{ln:dot-x-mu} of Algorithm \ref{alg:inside}).

\begin{remark}\label{rem:equivalenceofdef}
Note that Definitions~\ref{def:implicitclustering} and \ref{def:candidateclusters} are compatible in the following sense. For an ordered partition $(T_1, \dots, T_b)$ of approximate cluster means $\{\wh{\mu}_1, \dots, \wh{\mu}_k \}$ that induces a collection of clusters $\{\wh{C}_{\wh{\mu}_1}, \dots, \wh{C}_{\wh{\mu}_k}\}$ it is true that:
$$ \{\wh{C}_{\wh{\mu}_1}, \dots, \wh{C}_{\wh{\mu}_k}\} = \bigcup_{i \in [b]} \bigcup_{\wh{\mu} \in T_i} \{\wh{C}_{\wh{\mu}}^{(T_1, \dots, T_{i-1})} \}\text{,}$$
\end{remark}

\noindent Equipped with Definition~\ref{def:candidateclusters} we are ready to explain Algorithm \textsc{ComputeOrderedPartition}. The Algorithm proceeds in $O(\log(k))$ stages. It maintains a set $S$ of approximate cluster means, that initially is equal to $\{ \wh{\mu}_1, \dots, \wh{\mu}_k\}$, from which $\wh{\mu}$'s are removed after every stage. At every stage $i$ a collection of sets 
$$\mathcal{C}_i := \bigcup_{\wh{\mu} \in S} \{\wh{C}_{\wh{\mu}}^{(T_1, \dots, T_{i-1})} \} \text{,}$$ 
is implicitly considered. In fact sets in this collection are, by definition, pairwise disjoint (see Defnition~\ref{def:candidateclusters} and line:~\ref{ln:dot-x-mu} of \textsc{IsInside}). $\wh{C}_{\wh{\mu}}^{(T_1, \dots, T_{i-1})}$'s are defined as threshold sets (see Definition~\ref{def:thresholdsets}) that are made disjoint by removing intersections. The main idea behind the Algorithm is to use properties from Section~\ref{sec:bound_int} so that we can show that $\wh{C}_{\wh{\mu}}^{(T_1, \dots, T_{i-1})}$'s match some $C_j$'s well. Unfortunately after removing the intersections the above property might not hold for every cluster in $\mathcal{C}_i$. In the rest of this section we show however that it is true for a constant fraction of sets from $\mathcal{C}_i$. The Algorithm \textsc{ComputeOrderedPartition} proceeds by discarding, from set $S$, the $\wh{\mu}$'s for which $\wh{C}_{\wh{\mu}}^{(T_1, \dots, T_{i-1})}$ matches some $C_j$'s well and implicitly removes the vertices of $\wh{C}_{\wh{\mu}}^{(T_1, \dots, T_{i-1})}$ from consideration. Moreover it projects out the directions corresponding to the removed $\wh{\mu}$'s and restricts its attention to a lower dimensional subspace $\Pi$ of $\R^k$ (see \textbf{Idealized Clustering Algorithm} from Section~\ref{sec:algorithmhighlevel} for comparison). The Algorithm doesn't know which sets from $\mathcal{C}_i$ are good as it runs in sublinear time. That is why we develop a simple sampling procedure that computes outer-conductance of candidate clusters (see Algorithm~\ref{alg:estimateconductande}). Then the Algorithm removes the $\wh{\mu}$'s for which the corresponding $\wh{C}_{\wh{\mu}}^{(T_1, \dots, T_{i-1})}$ have small outer-conductance. We conclude using the robustness property of $(k,\varphi,\e)$-clusterable graphs (Lemma~\ref{lem:howtocluster}) that these tests are enough.

The rest of this subsection is devoted to showing that if \textsc{ComputeOrderedPartition} is called with $(\hat{\mu}_1, \dots, \hat{\mu}_k)$ equal to $(\mu_1, \dots, \mu_k)$ and the algorithm has access to real dot products then \textsc{ComputeOrderedPartition} returns \textsc{True} and an ordered partition $(T_1, \dots, T_b)$ (of $\{\mu_1, \dots, \mu_k\}$) that induces a collection of pairwise disjoint clusters $\{\wh{C}_{\mu_1}, \dots, \wh{C}_{\mu_k}\}$ such that for every $i$:
\begin{equation}\label{eq:sec3outerconductancegoal}
\phi \left( \wh{C}_{\mu_{i}} \right) \leq O \left(\frac{\e}{\varphi^2}  \cdot \log(k) \right) \text{.}
\end{equation}
Then using Lemma~\ref{lem:howtocluster} we get that there exists a permutation $\pi$ such that for all $i \in [k]$:
\begin{equation}\label{eq:sec3maingoal}
\left|\wh{C}_{\mu_{i}} \triangle C_{\pi(i)}\right| \leq O \left(\frac{\e}{\varphi^3}  \cdot \log(k) \right)|C_{\pi(i)}| \text{.}
\end{equation}


The core of the argument is an averaging argument that, for every linear subspace of $\R^k$, bounds the average distance of embedded points to their centers in this subspace. What is important is that the bound depends linearly on the dimensionality of the subspace.  

\begin{lemma}\label{lem:smallinsubspace}
Let $k \geq 2$, $\varphi \in (0,1)$ and $\frac{\e}{\varphi^2}$ be smaller than a sufficiently small constant.
Let $G=(V,E)$ be a $d$-regular graph that admits a $(k,\varphi,\epsilon)$-clustering $\{ C_1, \dots, C_k \}$.
Then for all $L \subseteq \R^k$ - a linear subspace of $\R^k$, $\Pi$ the orthogonal projection onto $L$ we have:
$$\sum_{x \in V} \|\Pi f_x - \Pi \mu_x \|_2^2 \leq O \left( \text{dim}(L) \cdot \frac{\e}{\varphi^2} \right)$$
\end{lemma}

\begin{proof}
Let $b := dim(L)$ and $\{w_1, \dots, w_b \}$ be any orthonormal basis of $L$ and recall that for $x \in V$ $\mu_x$ is the cluster mean of the cluster which $x$ belongs to. Then 
\begin{align}
\sum_{x \in V} \|\Pi f_x - \Pi \mu_x \|_2^2 
&= \sum_{x \in V} \sum_{i=1}^b \rdp{f_x - \mu_x, w_i}^2 \nonumber \\
&= \sum_{i=1}^b \sum_{x \in V} \rdp{f_x - \mu_x, w_i}^2 \nonumber \\
&\leq b \cdot \frac{4\e}{\varphi^2} && \text{By Lemma~\ref{lem:dirvariance}} \nonumber \end{align}
\end{proof}

In order to show \eqref{eq:sec3outerconductancegoal} we need to show that a constant fraction of candidate sets $\wh{C}_{\mu}^{(T_1, \dots, T_{i-1})}$'s match some $C_j$'s well. To do that we argue that that sets of the form $\Cr{\Pi\wh{\mu}, 0.9}$ (where $\Pi$ is the orthogonal projection onto the $\text{span}(\bigcup_{j < i} T_j)^{\perp}$) don't overlap too much. We do this in two steps. First in Lemma~\ref{clm:technicalaboutalpha2} and Lemma~\ref{lem:farfromcenter} we show that points from the intersections are far from their centers. Then in Lemma~\ref{lem:pointsoutside} below we show that having too many such vertices would contradict Lemma~\ref{lem:smallinsubspace}.

\begin{restatable}{lemma}{claimtechnicalaboutalpha}\label{clm:technicalaboutalpha2}
Let $k \geq 2$, $\varphi \in (0,1)$ and $\frac{\e}{\varphi^2} $ be smaller than a sufficiently small constant.
Let $G=(V,E)$ be a $d$-regular graph that admits a $(k,\varphi,\epsilon)$-clustering $\{ C_1, \dots, C_k \}$. 
Let $\{v_1, \dots, v_k \} \in \R^k$ be a set of vectors satisfying:
\begin{itemize}
    \item $|\rdp{v_i, v_j}| \leq O \left(\frac{\sqrt{\e}}{\varphi} \right) \frac{1}{\sqrt{|C_i| |C_j|}}$
    \item $\left| \rn{v_i}^2 - \frac{1}{|C_i|} \right| \leq O \left(\frac{\sqrt{\e}}{\varphi} \right) \frac{1}{|C_i|}$
\end{itemize}

Then for every pair $i \neq j \in [k]$ for every $\theta \in (0,1)$ if $\alpha := \frac{v_i\frac{\|v_j\|}{\|v_i\|}+v_j\frac{\|v_i\|}{\|v_j\|}}{\sqrt{\|v_i\|^2+\|v_j\|^2}}$ and 
$I := C_{v_i,\theta} \cap C_{v_j,\theta} = \{ x \in V : \langle f_x, v_i \rangle  \geq \theta\|v_i\|^2 \wedge \langle f_x, v_j \rangle  \geq \theta\|v_j\|^2\}$
then the following conditions hold: 

\begin{enumerate}
\item Correlation of vector $v_p$ with the direction $\alpha$ is as follows: 
\begin{itemize}
\item $\text{for all } p \in [k] \setminus \{i,j\},  \rdp{\frac{\alpha}{\|\alpha\|}, v_p} \leq O\left(\frac{\sqrt{\e}}{\varphi} \right) \cdot \frac{\|v_i\| \cdot \|v_j\|}{\sqrt{\|v_i\|^2+\|v_j\|^2}} $, for all $i \neq j \in [k]$
\item $\text{for all } p \in \{i,j\},  \rdp{\frac{\alpha}{\|\alpha\|}, v_p} \leq \left(1+ O\left(\frac{\sqrt{\e}}{\varphi} \right) \right) \cdot \frac{\|v_i\| \cdot \|v_j\|}{\sqrt{\|v_i\|^2+\|v_j\|^2}} $ for all $i \in [k]$
\end{itemize}
\item Spectral embeddings of vertices from set $I$ have big correlation with direction $\alpha$.
$$\min_{x \in I} \rdp{\frac{\alpha}{\|\alpha\|},f_x} \geq \left(2\theta - O\left(\frac{\sqrt{\e}}{\varphi} \right)\right) \cdot \frac{\|v_i\| \cdot \|v_j\|}{\sqrt{\|v_i\|^2+\|v_j\|^2}} $$
\end{enumerate}
\end{restatable}

\begin{proof}
For all $p \in [k]$ let $\wt{v}_p:=v_p/||v_p||$. Let $\gamma:=\frac{||v_j||}{\sqrt{||v_i||^2+||v_j||^2}}$, $\alpha := \gamma \wt{v}_i + \sqrt{1-\gamma^2}\wt{v}_j,$ and $\wt{\alpha} := \alpha/||\alpha||$. Fix $i \neq j \in [1, \dots, k]$. First we show that since $v_i$'s are close to orthogonal we have $||\alpha||^2 \approx 1$. More precisely we will upper bound $|\rn{\alpha}^2 - 1|$
\begin{align}
\left|\rn{\alpha}^2 - 1\right| 
&= 
\left|\gamma^2 \rn{\wt{v}_i}^2 + (1-\gamma^2)\rn{\wt{v}_j}^2 + 2\gamma \sqrt{1-\gamma^2} \rdp{\wt{v}_i, \wt{v}_j} - 1\right| \nonumber \\
&= \frac{2\rdp{v_i,v_j}}{\rn{v_i}^2 + ||v_j||^2} &&\text{as } ||\wt{v}_i|| = ||\wt{v}_j|| = 1 \nonumber \\
&\leq \frac{2 \cdot O\left(\frac{\sqrt{\epsilon}}{\varphi}\right)\frac{1}{\sqrt{|C_i||C_j|}}}{\left(1-O\left(\frac{\sqrt{\epsilon}}{\varphi}\right)\right)(\frac{1}{|C_i|} + \frac{1}{|C_j|})} && \text{By assumptions} \nonumber \\
&\leq O\left(\frac{\sqrt{\epsilon}}{\varphi} \right) \frac{\sqrt{|C_i||C_j|}}{|C_i| + |C_j|} \nonumber \\
&\leq O\left(\frac{\sqrt{\epsilon}}{\varphi}\right) && \text{as } \frac{\sqrt{|C_i||C_j|} }{ \max(|C_i|,|C_j|)} \leq 1\label{eq:normofalpha}
\end{align}

Observe the following fact:
\begin{equation}\label{eq:gammaproperty}
\sqrt{1-\gamma^2} \cdot \rn{v_j} = \gamma \cdot \rn{v_i}    
\end{equation}

Next notice the following:
\begin{equation}\label{eq:dotwithmui}
\langle \alpha, v_i \rangle = \gamma ||v_i||+\langle \wt{v}_i, \wt{v}_j \rangle\cdot \sqrt{1-\gamma^2}  ||v_i||
\end{equation}

\begin{equation}\label{eq:dotwithmuj}
\langle \alpha, v_j \rangle = \langle \wt{v}_i, \wt{v}_j \rangle\cdot\gamma ||v_j||+\sqrt{1-\gamma^2}  ||v_j||
\end{equation}

For all $p\in \{1, 2,\ldots, k\}\setminus \{i,j\}$
\begin{equation}\label{eq:dotwithmup}
\langle \alpha, v_p \rangle = \langle \wt{v}_i, \wt{v}_p \rangle\cdot\gamma ||v_p||+ \rdp{\wt{v}_j, \wt{v}_p}\sqrt{1-\gamma^2}  ||v_p||
\end{equation}

Moreover for all $p \neq q \in [1,\dots,k]$ we have 
\begin{align}
\left|\frac{1}{\rn{v_p}^2} \cdot \rdp{v_q, v_p} \right|
&\leq 
O\left(\frac{\sqrt{\epsilon}}{\varphi}\right) \frac{1}{\sqrt{|C_q||C_p|}} |C_p|\frac{1}{\left(1 - O\left(\frac{\sqrt{\epsilon}}{\varphi}\right)\right)} && \text{By assumptions} \nonumber \\
&\leq 
O\left(\frac{\sqrt{\epsilon}}{\varphi}\right) \sqrt{\frac{|C_p|}{|C_q|}} && \text{for small enough } \frac{\epsilon}{\varphi^2}\nonumber \\
&\leq 
O\left(\frac{\sqrt{\e}}{\varphi} \right) && \text{as } \frac{|C_p|}{|C_q|}  = O(1)  \label{eq:dotbylensquared}
\end{align}

Using the above we can prove:
\begin{align}
\left| \langle \wt{v}_i, \wt{v}_j \rangle\cdot \sqrt{1-\gamma^2}  ||v_i|| \right| 
&= \left|\sqrt{1-\gamma^2} \cdot \rn{v_j} \cdot \frac{1}{\rn{v_j}^2} \cdot \rdp{v_i, v_j} \right| \nonumber \\
&\leq
\sqrt{1-\gamma^2} \cdot \rn{v_j} \cdot O\left(\frac{\sqrt{\e}}{\varphi} \right) && \text{By \eqref{eq:dotbylensquared}} \nonumber \\
&= O\left(\frac{\sqrt{\e}}{\varphi} \right) \cdot \gamma \cdot \rn{v_i} && \text{By \eqref{eq:gammaproperty}} \label{eq:bounderrormui}
\end{align}

And similarly we show:
\begin{align}
\left| \langle \wt{v}_i, \wt{v}_j \rangle\cdot \gamma  ||v_j|| \right| 
&= \left|\gamma \cdot \rn{v_i} \cdot \frac{1}{\rn{v_i}^2} \cdot \rdp{v_i, v_j} \right| \nonumber \\
&\leq
O\left(\frac{\sqrt{\e}}{\varphi}  \right) \cdot \gamma \cdot \rn{v_i} && \text{By \eqref{eq:dotbylensquared}} \label{eq:bounderrormuj}
\end{align}
For all $p\in \{1, 2,\ldots, k\}\setminus \{i,j\}$ we get
\begin{align}
\left| \rdp{\alpha, 
v_p} \right| 
&\leq \Big| \langle \wt{v}_i, \wt{v}_p \rangle\cdot\gamma ||v_p|| \Big|+ \left| \rdp{\wt{v}_j, \wt{v}_p}\sqrt{1-\gamma^2}  ||v_p|| \right| &&\text{By \eqref{eq:dotwithmup}}\nonumber \\
&= \left| \langle v_i, v_p \rangle\cdot \frac{1}{\rn{v_i}^2} \rn{v_i}\gamma \right|+ \left| \rdp{v_j, v_p} \frac{1}{\rn{v_j}^2} \rn{v_j} \sqrt{1-\gamma^2}   \right| \nonumber \\
&\leq O\left(\frac{\sqrt{\e}}{\varphi} \right) \cdot \gamma \cdot \rn{v_i} && \text{By \eqref{eq:dotbylensquared} and \eqref{eq:gammaproperty}} \label{eq:bounderrormup}
\end{align}
Combining \eqref{eq:dotwithmui}, \eqref{eq:dotwithmuj}, \eqref{eq:bounderrormui}, \eqref{eq:bounderrormuj} and \eqref{eq:bounderrormup} we get that for all $p \in \{i,j\}$ we have
\begin{equation}
\rdp{\alpha, v_p} \leq \left(1 + O\left(\frac{\sqrt{\e}}{\varphi}  \right)\right) \cdot \gamma \cdot \rn{v_i}
\end{equation}
and
for all $p \in \{1, \dots, k\} \setminus \{i,j\}$
\begin{equation}
\rdp{\alpha, v_p} \leq  O\left(\frac{\sqrt{\e}}{\varphi}  \right) \cdot \gamma \cdot \rn{v_i}
\end{equation}
Now using \eqref{eq:normofalpha} we get that for all $p \in \{i,j\}$
$$\rdp{\wt{\alpha}, v_p} \leq \frac{1}{\sqrt{1 - O\left(\frac{\sqrt{\e}}{\varphi} \right)}} \left(1 + O\left(\frac{\sqrt{\e}}{\varphi}  \right)\right) \cdot \gamma \cdot \rn{v_i} \leq \left(1 + O\left(\frac{\sqrt{\e}}{\varphi}  \right)\right) \cdot \frac{||v_i|| ||v_j||}{\sqrt{||v_i||^2+||v_j||^2}} $$
and for all $p \in \{1, \dots, k\} \setminus \{ i,j\}$
$$\rdp{\wt{\alpha}, v_p} \leq \frac{1}{\sqrt{1 - O\left(\frac{\sqrt{\e}}{\varphi} \right)}} O\left(\frac{\sqrt{\e}}{\varphi}  \right) \cdot \gamma \cdot \rn{v_i} \leq  O\left(\frac{\sqrt{\e}}{\varphi}  \right) \cdot \frac{||v_i|| ||v_j||}{\sqrt{||v_i||^2+||v_j||^2}} $$
These two inequalities establish the first statement of the Claim.

Recall that
$$I = \{ x \in V : \langle f_x, v_i \rangle  \geq \theta||v_i||^2 \wedge \langle f_x, v_j \rangle  \geq \theta||v_j||^2\}$$
Now let $x \in I$. Then observe
\begin{align*}
\rdp{\alpha, f_x} 
&= \rdp{\gamma \cdot \wt{v}_i,f_x} + \rdp{\sqrt{1-\gamma^2} \cdot \wt{v}_j, f_x} \\
&\geq \gamma \cdot \theta \cdot \rn{v_i} + \sqrt{1-\gamma^2} \cdot \theta \cdot \rn{v_j} && \text{because } x \in I\\
&= 2\theta \cdot \gamma \cdot \rn{v_i} && \text{by \eqref{eq:gammaproperty} }
\end{align*}
Hence
\begin{align*}
\rdp{\wt{\alpha}, f_x} 
&\geq \frac{1}{\sqrt{1+O\left(\frac{\sqrt{\e}}{\varphi} \right)}} 2\theta \cdot \gamma \cdot \rn{v_i} &&\text{By \eqref{eq:normofalpha}} \\
&\geq \left(2\theta - O\left(\frac{\sqrt{\e}}{\varphi} \right)\right) \cdot \gamma \cdot \rn{v_i}
\end{align*}

\end{proof}

Now we use technical Lemma~\ref{clm:technicalaboutalpha2} to show that vertices from the intersections of $\Cr{\Pi \mu, 0.9}$'s are far from their centers.

\begin{lemma}\label{lem:farfromcenter}
Let $k \geq 2$, $\varphi \in (0,1)$ and $\frac{\e}{\varphi^2} $ be smaller than a sufficiently small constant.
Let $G=(V,E)$ be a $d$-regular graph that admits a $(k,\varphi,\epsilon)$-clustering $\{ C_1, \dots, C_k \}$.
If $\mu_i$'s are cluster means then the following conditions hold. For all $S \subset \{\mu_1, \dots, \mu_k\}$ if $L := span(S)^\perp$ and $\Pi$ is the projection on $L$ then if $x \in V$ is such that 
$$ \rdp{\Pi f_x, \Pi\mu_i} \geq 0.9 \| \Pi\mu_i \|_2^2 \wedge \rdp{\Pi f_x, \Pi\mu_j} \geq 0.9 \| \Pi\mu_j \|_2^2$$
for some $\mu_i, \mu_j \in \{\mu_1, \dots, \mu_k\} \setminus S, \mu_i \neq \mu_j$. Then:
$$\|\Pi f_x - \Pi \mu_x \| \geq 0.3 \sqrt{\frac{1}{\max_{p \in [k]} |C_p|}} $$

\end{lemma}

\begin{proof}
Let $x \in V$ be such that $ \rdp{\Pi f_x, \Pi\mu_i} \geq 0.9 \| \Pi\mu_i \|_2^2$ and $ \rdp{\Pi f_x, \Pi\mu_j} \geq 0.9 \| \Pi\mu_j \|_2^2$. 
Note that by Lemma~\ref{lem:dosubspace} set $\{\Pi \mu_1, \dots, \Pi \mu_k \}$ satisfies assumptions of Lemma~\ref{clm:technicalaboutalpha2}. So applying Lemma~\ref{clm:technicalaboutalpha2} for $\theta = 0.9$ we get that there exists $\alpha \in \text{span}\{\Pi\mu_i, \Pi\mu_j \}, \rn{\alpha} = 1$ such that:
\begin{itemize}
    \item $\rdp{\alpha, f_x} = \rdp{\alpha, \Pi f_x} \geq (1.8 - O(\frac{\sqrt{\e}}{\varphi})) \cdot \frac{\|\Pi\mu_i\| \cdot \|\Pi\mu_j\|}{\sqrt{\|\Pi\mu_i\|^2+\|\Pi\mu_j\|^2}} $
    \item $\rdp{\alpha, \Pi \mu_p} \leq (1 + O(\frac{\sqrt{\e}}{\varphi} )) \cdot \frac{\|\Pi\mu_i\| \cdot \|\Pi\mu_j\|}{\sqrt{\|\Pi\mu_i\|^2+\|\Pi\mu_j\|^2}} $, for all $p \in [k]$
\end{itemize}
Thus we get 
\begin{align}
\rn{\Pi f_x - \Pi \mu_x}
&\geq |\rdp{\alpha, \Pi f_x} - \rdp{\alpha, \Pi \mu_x} |  \nonumber \\
&\geq \left(0.8 - O \left(\frac{\sqrt{\e}}{\varphi} \right) \right) \cdot \frac{\|\Pi\mu_i\| \cdot \|\Pi\mu_j\|}{\sqrt{\|\Pi\mu_i\|^2+\|\Pi\mu_j\|^2}} \nonumber \\
&\geq 0.75 \cdot \frac{\|\Pi\mu_i\| \cdot \|\Pi\mu_j\|}{\sqrt{\|\Pi\mu_i\|^2+\|\Pi\mu_j\|^2}} && \text{By assumption that }   \frac{\e}{\varphi^2} \text{ small} \label{eq:lemma34step1} 
\end{align}

without loss of generality we can assume $\rn{\Pi\mu_i} \geq \rn{\Pi\mu_j}$. Then we get:

\begin{align}
\frac{\|\Pi\mu_i\| \cdot \|\Pi\mu_j\|}{\sqrt{\|\Pi\mu_i\|^2+\|\Pi\mu_j\|^2}}
&= \frac{ \|\Pi\mu_j\|}{\sqrt{1+\|\Pi\mu_j\|^2/\rn{\Pi\mu_i}^2}} \nonumber \\
&\geq \frac{1}{\sqrt{2}} \rn{\Pi\mu_j} \nonumber \\
&\geq \frac{1}{2 \sqrt{\max_{p \in [k]} |C_p|}} && \text{Lemma~\ref{lem:dosubspace}, assumption that }   \frac{\e}{\varphi^2} \text{ small} \label{eq:lemma34step2}
\end{align}
Combining \eqref{eq:lemma34step1} and \eqref{eq:lemma34step2} we get:
$$\rn{\Pi f_x - \Pi \mu_x} \geq 0.3 \cdot  \frac{1}{\sqrt{\max_{p \in [k]} |C_p|}} $$

\end{proof}

Combining Lemma~\ref{lem:smallinsubspace} and Lemma~\ref{lem:farfromcenter} we show that sets $\Cr{\Pi\mu,0.9}$'s don't overlap much.

\begin{lemma}\label{lem:pointsoutside}
Let $k \geq 2$, $\varphi \in (0,1)$ and $\frac{\e}{\varphi^2}$ be smaller than a sufficiently small constant.
Let $G=(V,E)$ be a $d$-regular graph that admits a $(k,\varphi,\epsilon)$-clustering $\{ C_1, \dots, C_k \}$.
If $\mu_i$'s are cluster means then the following conditions hold. For all $S \subset \{\mu_1, \dots, \mu_k\}$ if $L := span(S)^\perp$, $dim(L) = b$ and $\Pi$ is projection on $L$ then:
$$\left| \bigcup_{\substack{\mu, \mu' \in \{\mu_1, \dots, \mu_k\} \setminus S\\\mu\neq \mu'}}
\Cr{\Pi \mu, 0.9} \cap \Cr{\Pi \mu', 0.9} \right|  \leq O \left( b \cdot  \frac{\e}{\varphi^2} \right) \cdot \frac{n}{k} \text{.}$$
\end{lemma}

\begin{proof}
Let $x \in V$ be such that $\rdp{\Pi f_x, \Pi\mu} \geq 0.9 \| \Pi\mu \|_2^2$ and $ \rdp{\Pi f_x, \Pi\mu'} \geq 0.9 \| \Pi\mu' \|_2^2$ for some $\mu, \mu' \in \{\mu_1, \dots, \mu_k\} \setminus S$. Then by Lemma~\ref{lem:farfromcenter} we get that 
\begin{equation}\label{eq:farfromcenter}
\|\Pi f_x - \Pi \mu_x \| \geq 0.3 \sqrt{\frac{1}{\max_{p \in [k]} |C_p|}} \text{.}
\end{equation}
On the other hand Lemma~\ref{lem:smallinsubspace} guarantees:
\begin{equation}\label{eq:varianceinsubspace}
\sum_{x \in V} \|\Pi f_x - \Pi \mu_x \|_2^2 \leq O \left( dim(L) \cdot \frac{\e}{\varphi^2} \right)
\end{equation}
Combining \eqref{eq:farfromcenter}, \eqref{eq:varianceinsubspace} and the fact that $\frac{\max_{p \in [k]} |C_p|}{\min_{p \in [k]} |C_p|} = O(1)$ we get
$$\left| \bigcup_{\mu, \mu' \in \{\mu_1, \dots, \mu_k\} \setminus S}
\Cr{\Pi \mu, 0.93} \cap \Cr{\Pi \mu', 0.93} \right|  \leq O\left(  b \cdot \frac{\e}{\varphi^2} \right) \cdot \frac{n}{k} $$ 

\end{proof}

Our bounds above enable the following analysis. At every stage of the for loop from line~\ref{testcentersforloop} of Algorithm~\ref{alg:testmus} at least half of the candidate clusters:
$$\mathcal{C}_i := \bigcup_{\wh{\mu} \in S} \{ \wh{C}_{\wh{\mu}}^{(T_1, \dots, T_{i-1})} \} \text{,}$$
passes the test from line~\ref{testcenterscondactancetest} of Algorithm~\ref{alg:testmus}, which means that they have small outer-conductance and satisfy condition \eqref{eq:sec3outerconductancegoal}.

\begin{lemma}\label{lem:induction}
Let $k \geq 2$, $\varphi \in (0,1)$ and $\frac{\e}{\varphi^2}  \cdot \log(k)  $ be smaller than a sufficiently small constant.
Let $G=(V,E)$ be a $d$-regular graph that admits a $(k,\varphi,\epsilon)$-clustering $\{ C_1, \dots, C_k \}$.

If \textsc{ComputeOrderedPartition}($G,\widehat{\mu}_1,\widehat{\mu}_2, \dots, \widehat{\mu}_k,s_1,s_2)$ is invoked with $(\hat{\mu}_1, \dots, \hat{\mu}_k) = (\mu_1, \dots, \mu_k)$ and we assume that all tests Algorithm~\ref{alg:testmus} performs $\left( \textit{i.e.} \adp{f_x, \wh{\Pi}\wh{\mu}} \stackrel{?}{\geq} 0.93 \an{\wh{\Pi}\wh{\mu}}^2 \right)$ are exact and \textsc{OuterConductance} computes outer-conductance precisely then there exists an absolute constant $\Upsilon$ such that the following conditions hold.

For any $i \in [0..\log(k)]$ assume that at the beginning of the $i$-th iteration of the for loop from line~\ref{testcentersforloop} of Algorithm~\ref{alg:testmus} $|S| = b$ and, up to renaming of $\mu$'s, $S = \{\mu_1,\dots, \mu_b\}$, the corresponding clusters are $\mathcal{C} = \{C_1, \dots, C_b \}$ respectively and the ordered partial partition of $\mu$'s is equal to $(T_1, \dots, T_{i-1})$. 
Then if for every $C \in \mathcal{C}$ we have that $|V^{(T_1, \dots, T_{i-1})} \cap C| \geq \left(1 - \Upsilon \cdot i \cdot \frac{\e}{\varphi^2} \right) |C|$ then at the beginning of $(i+1)$-th iteration:
\begin{enumerate}
    \item $|S| \leq b/2$ (that is at least half of the remaining cluster means were removed in $i$-th iteration), 
    \item\label{cond:second1} for every $\mu \in S$ the corresponding cluster $C$ satisfies $|V^{(T_1, \dots, T_{i})} \cap C| \geq \left(1- \Upsilon \cdot (i+1) \cdot \frac{\e}{\varphi^2} \right) |C| $, where $(T_1, \dots, T_{i})$ is the ordered partial partition of $\mu$'s created in the first $i$ iterations. 
\end{enumerate}
\end{lemma}

\begin{proof}
Let $i \in [0..\log(k)]$, without loss of generality we can assume that $S = \{\mu_1,\dots, \mu_b\}$ (if not we can rename the $\mu$'s) at the beginning of the $i$-th iteration and the corresponding clusters be $\mathcal{C} = \{C_1, \dots, C_b\}$ respectively. Assume that for every $C \in \mathcal{C}$ we have that $|V^{(T_1, \dots, T_{i-1})} \cap C| \geq \left(1 - \Upsilon \cdot i \cdot \frac{\e}{\varphi^2} \right) |C|$. We start by showing the first part of the Lemma.

\textbf{At least half of the cluster means is removed from $S$:}

Let $\mu \in S$, $\Pi_i$ be the orthogonal projection onto the $\text{span} (\bigcup_{j<i} T_j)^{\perp}$, where $(T_1, \dots, T_{i-1})$ is the ordered partial partition of $\{\mu_1, \dots, \mu_k \}$ created before iteration $i$ by \textsc{ComputeOrderedPartition}. For brevity we will refer to $(T_1, \dots, T_{i-1})$ as $P$ in this proof. Let 
$$ I :=  \bigcup_{\mu', \mu'' \in \{\mu_1, \dots, \mu_b\} } \Cr{\Pi_i \mu', 0.93} \cap \Cr{\Pi_i \mu'', 0.93} \text{.}$$

\noindent By Lemma~\ref{lem:pointsoutside} we have that 
$$\left| I \right|  \leq O \left( b \cdot \frac{\e}{\varphi^2} \right)  \cdot \frac{n}{k} $$
So by Markov inequality we get that there exists a subset of clusters $\mathcal{R}  \subseteq \mathcal{C}$ such that $|\mathcal{R}| \geq b/2$ and for every $C \in \mathcal{R}$ we have that 
\begin{equation}\label{eq:smallnumboutsideeasy}
|C \cap I| \leq 2 \cdot O \left( \frac{\e}{\varphi^2} \right) \cdot \frac{n}{k}
\end{equation}
We will argue that for any order of the for loop from line~\ref{testcentersforloop} of Algorithm~\ref{alg:testmus} it is true that for every $C \in \mathcal{R}$ with corresponding mean $\mu$ the 
candidate cluster $\wh{C}_{\mu}^P$ satisfies the if statement from line \ref{testcenterscondactancetest}.

First note that behavior of the algorithm is independent of the order of the for loop from line~\ref{testcentersforloop} of Algorithm~\ref{alg:testmus}  as by definition $\wh{C}_{\mu}^P$'s for $\mu \in S$ are pairwise disjoint. Now let $C \in \mathcal{R}$, $\mu$ be the corresponding mean to $C$ and $\wh{C}_{\mu}^P$ be the candidate cluster corresponding to $\mu$ with respect to $P = (T_1, \dots, T_{i-1})$. By inductive assumption $|V^P \cap C| \geq \left(1 - \Upsilon \cdot i \cdot \frac{\e}{\varphi^2} \right) |C|$ so by \eqref{eq:smallnumboutsideeasy}, Lemma~\ref{lem:mostinset} and the fact that $\frac{\max_{p \in [k]} |C_p|}{\min_{p \in [k]} |C_p|} = O(1)$ we get that:  
\begin{align}
|\wh{C}_{\mu}^P \cap C| 
&\geq 
\left(1 - \Upsilon \cdot i \cdot \frac{\e}{\varphi^2} \right) |C| - O \left( \frac{\e}{\varphi^2}  \right) \frac{n}{k} - O \left( \frac{\e}{\varphi^2}  \right) |C|  \nonumber \\
&\geq 
\left(1 - O \left(\frac{\e}{\varphi^2}  \cdot \log(k)  \right)\right)|C| \label{eq:regioncontainslotofclustereasy}
\end{align}
To prove that $\widehat{C}_{\mu}^P$ passes the outer-conductance test we also  need to show that $\widehat{C}_{\mu}^P$ doesn't contain a lot of points from $V^{P} \setminus C$. By Lemma~\ref{lem:notalostfromoutside} we get that:
\begin{equation}\label{eq:regioncontainsfewfromoutside}
|\widehat{C}_{\mu}^P \cap (V^P \setminus C) | \leq  |\widehat{C}_{\mu}^P \cap (V \setminus C) | \leq O \left( \frac{\e}{\varphi^2} \right)  |C|  \text{.}
\end{equation}
Combining \eqref{eq:regioncontainsfewfromoutside} and \eqref{eq:regioncontainslotofclustereasy} we get that:
\begin{equation}\label{eq:smallsymdiffeasy}
    |\widehat{C}_{\mu}^P \triangle C| \leq O \left(\frac{\e}{\varphi^2}  \cdot \log(k) \right) |C|
\end{equation}
Now we want to argue that $\widehat{C}_{\mu}^P$ passes the outerconductance test from line \ref{testcenterscondactancetest} of Algorithm~\ref{alg:testmus}. From the definition of outer conductance:
\begin{align*}
\phi(\widehat{C}_{\mu}^P) 
&\leq 
\frac{E(C,V \setminus C) + d|\widehat{C}_{\mu}^P \triangle C|}{d(|C| - |\widehat{C}_{\mu}^P \triangle C|)} \\
&\leq 
\frac{E(C,V \setminus C) + d \cdot O \left(\frac{\e}{\varphi^2}  \cdot \log(k) \right)|C|}{d(|C| - O \left(\frac{\e}{\varphi^2}  \cdot \log(k) \right) |C|)} && \text{from \eqref{eq:smallsymdiffeasy}} \\
&\leq 
\frac{O \left(\frac{\e}{\varphi^2} \right) +O \left(\frac{\e}{\varphi^2}  \cdot \log(k) \right)}{1-O \left(\frac{\e}{\varphi^2}  \cdot \log(k) \right)} && \text{because } \frac{E(C,V \setminus C)}{d|C|} \leq O \left(\frac{\e}{\varphi^2} \right)\\
&\leq 
O \left(\frac{\e}{\varphi^2}  \cdot \log(k) \right) && \text{for sufficiently small } \frac{\e}{\varphi^2}  \cdot \log(k)  
\end{align*}

 and it follows that
  $$\phi(\widehat{C}_{\mu}^P) \leq O \left(\frac{\e}{\varphi^2}  \cdot \log(k) \right) \text{,}$$ 
which means that $\widehat{C}_{\mu}^P$ passes the test as we assumed that \textsc{OuterConductance} computes outer-conductance precisely.

\textbf{Clusters corresponding to unremoved $\mu$'s satisfy condition \ref{cond:second1}:}

\noindent
Now we prove that for every $\mu$ that was not removed from set $S$ only small fraction of its corresponding cluster is removed. 

Let $\mu \in S$ be such that it is not removed in the $i$-th step. Let $\Pi_i$ be the orthogonal projection onto the $\text{span} (\bigcup_{j<i} T_j)^{\perp}$. Let $C \in \mathcal{C}$ be the cluster corresponding to $\mu$. By assumption $|V^P \cap C| \geq \left(1- \Upsilon \cdot i \cdot  \frac{\e}{\varphi^2}  \right) |C| $. Now let $x \in V^{(T_1, \dots, T_{i-1})} \setminus V^{(T_1, \dots, T_{i})}$, where $(T_1, \dots, T_{i})$ is the partial partition of $\mu$'s created in the first $i$-th steps of the for loop. We get that there exists $\mu' \in \{\mu_1, \dots, \mu_b \}$ such that $x \in \wh{C}_{\mu'}^P$ (recall that $\wh{C}_{\mu'}^P$ is the candidate cluster corresponding to $\mu'$ with respect to $P = (T_1, \dots, T_{i-1})$). Recall (Definition~\ref{def:candidateclusters}) that $\wh{C}_{\mu'}^P$ is defined as:
$$\wh{C}_{\mu'}^P = \left\{ x \in V : \textsc{IsInside} \left(x,\mu',P, \{\mu_1, \dots, \mu_k \} \setminus \bigcup_{j \in [i-1]} T_j \right) = \textsc{True} \right\} \text{.}$$
This in particular means (see line~\ref{ln:dot-x-mu}: of Algorithm \textsc{IsInside}) that:
$$\wh{C}_{\mu'}^P \subseteq \Cr{\Pi_i \mu',0.93} \setminus \bigcup_{\mu''\in  S\setminus \{\mu'\}} \Cr{\Pi_i \mu'',0.93}  ,$$
which, as $\mu \in S \setminus \{ \mu' \}$, gives us that:
$$\wh{C}_{\mu'}^P \cap  \Cr{\Pi_i \mu,0.93} = \emptyset \text{,}$$
and finally, using Definition~\ref{def:thresholdsets}, we have:
\begin{equation}\label{eq:propertyofreturnedeasy1}
\rdp{f_x, \Pi_i \mu} < 0.93 \rn{\Pi_i \mu}^2\text{.}
\end{equation}

\noindent
But by Lemma~\ref{lem:mostinset}:
\begin{equation}\label{eq:smalloutsideeasy1}
|\{x \in C :\rdp{\Pi_i f_x, \Pi_i \mu} < 0.93 \| \Pi_i \mu \|_2^2 \}| \leq O \left( \frac{\e}{\varphi^2} \right) \cdot |C| 
\end{equation}
Combining \eqref{eq:propertyofreturnedeasy1} and \eqref{eq:smalloutsideeasy1} we get that: 
\begin{equation}\label{eq:indsecondtolast}
|C \cap (V^{(T_1, \dots, T_{i-1})} \setminus V^{(T_1, \dots, T_{i})})| \leq O \left( \frac{\e}{\varphi^2} \right) |C| \text{.}
\end{equation}
By assumption that $|V^{(T_1, \dots, T_{i-1})} \cap C| \geq \left(1 - \Upsilon \cdot i \cdot  \frac{\e}{\varphi^2} \right) |C| $ and \eqref{eq:indsecondtolast} we get that: 
$$|V^{(T_1, \dots, T_{i})} \cap C| \geq \left(1 - \Upsilon \cdot  (i+1) \cdot \frac{\e}{\varphi^2 } \right) |C| \text{,}$$
provided that $\Upsilon$ is bigger than the constant from  $O$ notation in \eqref{eq:indsecondtolast}, which is the same constant as the one in the statement of Lemma~\ref{lem:mostinset}.

\end{proof}

\begin{remark}\label{rem:firstphithensym}
Note that in this section we assume that the Algorithm has access to real centers $\{\mu_1, \dots, \mu_k \}$. If it was the case in the final algorithm we could in fact prove a stronger guarantee, i.e. "Algorithm~\ref{alg:testmus} returns \textsc{True} and an ordered partition $(T_1, \dots, T_b)$ (of $\{\mu_1, \dots, \mu_k\}$) that induces a collection of pairwise disjoint clusters $\{\wh{C}_{\mu_1}, \dots, \wh{C}_{\mu_k}\}$ such that there exists a permutation $\pi$ such that for all $i \in [k]$:
$$ \left|\wh{C}_{\mu_{i}} \triangle C_{\pi(i)}\right| \leq O \left(\frac{\e}{\varphi^2}  \cdot \log(k) \right)|C_{\pi(i)}| \text{\text{".}}$$
Compare the above statement with with \eqref{eq:sec3maingoal} and the main theorem of this section, Theorem~\ref{lem:realcenterswork}. The reason we present it this way is the following. 

The final algorithm doesn't have access to $\mu$'s but instead tests many candidate sets $\{\wh{\mu}_1, \dots, \wh{\mu}_k\}$. Moreover Algorithm~\ref{alg:testmus} returns an ordered partition $(T_1, \dots, T_b)$ that induces a collection of clusters $\{ \wh{C}_1, \dots, \wh{C}_k \}$ whenever every set from this collection passes the test from line~\ref{testcenterscondactancetest} of \textsc{ComputeOrderedPartition}, that is when for every $\wh{C} \in \{ \wh{C}_1, \dots, \wh{C}_k \}$:
$$\phi \left( \wh{C} \right) \leq O \left(\frac{\e}{\varphi^2}  \cdot \log(k) \right) \text{.}$$
This in particular means that Algorithm~\ref{alg:testmus} may return \textsc{True} even for a set $ \{\wh{\mu}_1, \dots, \wh{\mu}_k\}$ that is \textbf{not} a good approximation to $\{\mu_1, \dots, \mu_k \}$. 

Because of that, once we know that \textsc{ComputeOrderedPartition} invoked with $\{\mu_1, \dots, \mu_k \}$ returns an ordered partition $(T_1, \dots, T_b)$ that induces a collection of clusters $\{ \wh{C}_1, \dots, \wh{C}_k \}$, when proving the final result of this section (Theorem~\ref{lem:realcenterswork}) the only thing we assume about $\wh{C}$'s is that they passed the outer-conductance test. And that is why we use Lemma~\ref{lem:howtocluster} and we "loose" a factor $\frac{1}{\varphi}$ in the final guarantee.

Moreover structuring the argument in this way helps the presentation as later, in Section~\ref{sec:h_works}, the proof will follow a similar structure.
\end{remark}

The following Theorem concludes this subsection by showing \eqref{eq:sec3maingoal}. It does so by induction using Lemma~\ref{lem:induction} as an inductive step. At the end it uses Lemma~\ref{lem:howtocluster} to go from the guarantees for outer-conductance to guarantees for recovery.

\begin{theorem}\label{lem:realcenterswork}
Let $k \geq 2$, $\varphi \in (0,1)$ and $\frac{\e}{\varphi^2}  \log(k) $ be smaller than a sufficiently small constant.
Let $G=(V,E)$ be a $d$-regular graph that admits a $(k,\varphi,\epsilon)$-clustering $\{C_1, \dots, C_k \}$.

If \textsc{ComputeOrderedPartition}($G,\widehat{\mu}_1,\widehat{\mu}_2, \dots, \widehat{\mu}_k,s_1,s_2)$ is invoked with $(\hat{\mu}_1, \dots, \hat{\mu}_k) = (\mu_1, \dots, \mu_k)$ and we assume that all tests Algorithm~\ref{alg:testmus} performs $\left( \textit{i.e.} \adp{f_x, \wh{\Pi}\wh{\mu}} \stackrel{?}{\geq} 0.93 \an{\wh{\Pi}\wh{\mu}}^2 \right)$ are exact and \textsc{OuterConductance} computes outer-conductance precisely then the following conditions hold. 

\textsc{ComputeOrderedPartition} returns $(\textsc{True}, (T_1 ,\dots, T_b))$ such that $(T_1, \dots, T_b)$ induces a collection of clusters $\{\wh{C}_{\mu_1}, \dots, \wh{C}_{\mu_k}\}$ such that there exists a permutation $\pi$ on $k$ elements such that for all $i \in [k]$:
$$\left|\wh{C}_{\mu_{i}} \triangle C_{\pi(i)}\right| \leq O \left(\frac{\e}{\varphi^3}  \cdot \log(k) \right)|C_{\pi(i)}|$$ and 
 $$\phi(\wh{C}_{\mu_i}) \leq O \left(\frac{\e}{\varphi^2}  \cdot \log(k) \right) \text{.}$$
\end{theorem}

\begin{proof}
Note that for $i=0$ in the for loop in line~\ref{ln:testcentersmainloop} of \textsc{ComputeOrderedPartition} $S$ and clusters $\{C_1, \dots, C_k \}$ trivially satisfy assumptions of Lemma~\ref{lem:induction}. So using Lemma~\ref{lem:induction} and induction we get that for every $i \in [0..\lceil \log(k) \rceil]$ at the beginning of the $i$-th iteration:
\begin{itemize}
    \item $|S| \leq k / 2^i$, 
    \item for every $\mu \in S$ and the corresponding cluster $C$ we have $ |V^{(T_1, \dots, T_{i-1})} \cap C| \geq \left(1 - \Upsilon \cdot i \cdot \frac{\e}{\varphi^2} \right)|C|$ (where $\Upsilon$ is the constant from the statement of Lemma~\ref{lem:induction}).
\end{itemize} 
In particular this means that after at most $\lceil \log(k) \rceil$ iterations set $S$ becomes empty. This also means that \textsc{ComputeOrderedPartition} returns in line~\ref{ln:testcenterstruereturn}, so it returns \textsc{True} and the ordered partial partition $(T_1, \dots, T_b)$ is in fact an ordered partition of $\{\mu_1, \dots, \mu_k\}$. 

Note that by definition (see Definition~\ref{def:implicitclustering}) all the approximate clusters $\{\wh{C}_{\mu_1}, \dots, \wh{C}_{\mu_k}\}$ are pairwise disjoint and moreover for every constructed cluster $\wh{C} \in \{\wh{C}_{\mu_1}, \dots, \wh{C}_{\mu_k}\}$ we have:
$$\phi(\widehat{C}) \leq O \left(\frac{\e}{\varphi^2}  \cdot \log(k) \right) \text{,}$$
as it passed the test in line~\ref{testcenterscondactancetest} of \textsc{ComputeOrderedPartition}. So by Lemma~\ref{lem:howtocluster} it means that there exists a permutation $\pi$ on $k$ elements such that for all $i \in [k]$:
$$\left|\wh{C}_{\mu_{i}} \triangle C_{\pi(i)}\right| \leq O \left(\frac{\e}{\varphi^3}  \cdot \log(k) \right)|C_{\pi(i)}| \text{.}$$
\end{proof}

\subsection{Finding the cluster means}\label{sec:findthecenters}

In the previous subsection we showed that \textsc{ComputeOrderedPartition} succeeds if we have access to real cluster centers (i.e. $\mu_i$'s). In this section we present a search procedure for finding the centers.

The main idea behind our algorithm is to guess the clustering assignment of few random nodes and use this assignment to compute the approximate cluster means. More precisely, the first step of our algorithm is to learn the spectral embedding as described in Section~\ref{sec:dot}. Then we sample $s=\Omega(\frac{\varphi^2}{\e} \cdot k^4 \log(k))$ random nodes and we consider all the possible clustering assignments for them. For each assignment, we implicitly define the cluster center for a specific cluster as $\widehat{\mu}_i := \frac{1}{|P_i|}\sum_{x \in P_i} f_x$.

\begin{remark}
We note that in \textsc{FindCenters} we don't necessarily find $\mu_1, \dots, \mu_k$ exactly but we are able to show (see Section~\ref{sec:approx-muis}) that it finds a good approximation to $\mu_i$'s. Then in Section~\ref{sec:h_works} we show that such approximation is sufficient for the partitioning scheme to work.
\end{remark}

\begin{algorithm}
\caption{\textsc{FindCenters}($G,\eta, \delta)$}\label{alg:findrepresentatives}
\begin{algorithmic}[1]
    \State $\textsc{InitializeOracle}(G,\delta)$
    \For {$t \in [1 \dots {\log(2/\eta)} ]$}	
    	\State $S := $ Random sample of vertices of $V$ of size $s=\varTheta(\frac{\varphi^2}{\e} k^4 \log(k))$ \label{ln:sampleS}
    	\For  {$(P_1,P_2,\dots,P_k) \in \textsc{Partitions}(S)$}
    	    \For {$i =1$ to $k$}
                \State $\widehat{\mu}_i := \frac{1}{|P_i|}\sum_{x \in P_i} f_x$\label{ln:emrical-mean} \Comment{Note that we compute the centers only implicitly.} \label{ln:emp-cent}
    	    \EndFor
    	    \State $(r, C) :=  $
    	    \State \textsc{ComputeOrderedPartition$\left(G,(\widehat{\mu}_1,\widehat{\mu}_2,\dots,\widehat{\mu}_k),\Theta\left( \frac{\varphi^2}{\e} k^5 \log^2(k) \log(1/\eta)\right), \Theta\left( \frac{\varphi^4}{\e^2} k^5 \log^2(k) \log(1/\eta)\right)\right)$} 
    	    \label{ln:testcent}
    	    \If{$r = $ \textsc{True}}
                \State \Return $C$	    
    	    \EndIf
    	\EndFor 
    \EndFor
\end{algorithmic}
\end{algorithm}

\subsubsection{Quality of cluster means approximation}\label{sec:approx-muis}

In the previous Section~\ref{sec:realcenterswork} we showed that the partitioning scheme works if we can find $\mu_1, \dots, \mu_k$ exactly. In this section we show that it is possible to estimate the cluster means with a small error factor (i.e $\mu_i \approx \widehat{\mu}_i$). Later in Section~\ref{sec:h_works} we show that such an approximation to $\mu_i$'s is enough for the partitioning scheme to work.

In the rest of this section we show that if $\textsc{Partitions}(S)$ (see Algorithm~\ref{alg:findrepresentatives}) computes a correct guess of cluster assignments then the cluster means computed in line \eqref{ln:emrical-mean} are close to the real cluster means with constant probability. Then we repeat the procedure $O(\log(1/\eta))$ times to achieve success probability of at least $1 - \eta$.

In particular, in Lemma~\ref{lem:app-mu-norm} we show using Matrix Bernstein that if we have enough samples in a cluster $i$ then $\|\mu_i-\widehat{\mu}_i\|_2\leq \zeta \cdot \|\mu_i\|_2$ . Then we prove that if we sample enough random nodes we have enough samples in every cluster.

Before proving Lemma~\ref{lem:app-mu-norm} we show a tail bound for the spectral projection of a node that will be useful to apply Matrix Bernstein.

\begin{lemma}
\label{lem:ap-mu-s}
Let $k \geq 2$, $\varphi \in (0,1)$ and $\frac{\e}{\varphi^2}  \log(k) $ be smaller than a sufficiently small constant. Let $G=(V,E)$ be a $d$-regular and a $(k,\varphi,\epsilon)$-clusterable graph. Let $\beta>1$ .Let  
\[T=\left\{x\in V: ||f_x||_\infty \geq \beta\cdot \sqrt{\frac{10}{\min_{i\in [k]} |C_i| }  } \right\}\text{.}\]
Then we have $|T|\leq k\cdot \left(\frac{\beta}{2}\right)^{- \varphi^2/20\cdot\epsilon}\cdot (\min_{i\in [k]}|C_i|)$.
\end{lemma}
\begin{proof}
Recall that $f_x=U_{[k]}^T \mathds{1}_x$, and $u_i$ denote the $i^{\text{th}}$ column of $U_{[k]}$. Thus we have $\|f_x\|_\infty =\max_{i\in [k]}\{u_i(x)\}$. Let $s_{\min}=\min_{i\in k}|C_i|$. We define  
\[T_i=\left \{x\in V: |u_i(x)|\geq \beta\cdot \sqrt{\frac{10}{s_{\min}}} \right \}\]
Therefore, by Lemma \ref{lem:tail_bound} we have $|T_i|\leq \left(\frac{\beta}{2}\right)^{-{\varphi^2/20\cdot\epsilon}} \cdot s_{\min} \text{.} $ Note that $T=\bigcup_{i=1}^k T_i $. Therefore we have
\[|T|\leq k\cdot \left(\frac{\beta}{2}\right)^{-{\varphi^2/20\cdot\epsilon}} \cdot s_{\min} \]
\end{proof}

Now we are ready to derive a bound on the difference between $\mu_i$ and $\widehat{\mu}_i$.
\begin{lemma}
\label{lem:app-mu-norm}
 Let $\zeta,\delta \in (0,1)$, $k \geq 2$, $\varphi \in (0,1)$, $\frac{\e \log k}{\varphi^2}$ be smaller than a positive sufficiently small constant. Let $G=(V,E)$ be a $d$-regular graph that admits a $(k,\varphi,\epsilon)$-clustering $C_1, \ldots, C_k$. Let $s\geq c\cdot\left(k\cdot \log \left(\frac{k}{\delta}\right) \cdot \left(\frac{1}{\delta}\right)^{(80\cdot\epsilon/\varphi^2)}\cdot\left(\frac{1}{\zeta}\right)^2 \right)^{1/(1-(80\cdot\epsilon/\varphi^2))}$ for large enough constant $c$. Let $S=\{x_1,x_2,\ldots, x_s\}$ be the multiset with $s$ vertices sampled uniformly at random from cluster $C$. Let $\mu=\frac{1}{|C|} \sum_{x\in C} f_x$ denote the cluster mean, and let $\widehat{\mu}=\frac{1}{s} \sum_{i=1}^s f_{x}$ denote the empirical cluster mean. Then with probability at least $1-\delta$ we have
\[\|\mu-\widehat{\mu}\|_2\leq \zeta \cdot \|\mu\|_2\]
\end{lemma}

\begin{proof}
Let $s_{\min} := \min_{i \in [k]} |C_i|$. We define 
\[C'=\left\{x\in C: ||f_x||_\infty \leq 2\cdot \left(\frac{s\cdot k}{\delta}\right)^{(40\cdot\epsilon/\varphi^2)}\cdot\sqrt{\frac{10}{s_{\min}}} \right\} \]
Note that by Lemma \ref{lem:ap-mu-s} and by choice of $\beta=2\cdot \left(\frac{s\cdot k}{\delta}\right)^{(40\cdot\epsilon/\varphi^2)}$ we have
\[|C\setminus C'| \leq k\cdot \left(\frac{\beta}{2}\right)^{-\varphi^2/(20\cdot\epsilon)}\cdot s_{\min} \leq k\cdot \left(\frac{s\cdot k}{\delta}\right)^{-2}\cdot |C| = (k^{-1}\cdot s^{-2}\cdot \delta^{2})\cdot |C| \]
Thus we have
\begin{equation}
\label{eq:bnd-C'}
|C'|\geq \left(1-(k^{-1}\cdot s^{-2}\cdot \delta^{2})\right)|C|
\end{equation}
Let $\mu'=\frac{1}{|C'|}\sum_{x\in C'} f_x$. By triangle inequality we have
\begin{equation}
\label{eq:mu-tri}
||\wh{\mu}-\mu||_2 \leq  ||\wh{\mu}-\mu'||_2 + ||\mu'-\mu||_2
\end{equation}
In the rest of the proof we will upper bound both of these terms by $\frac{\zeta}{2}\cdot ||\mu||_2$.

\textbf{Step $1$:} We first prove $||\wh{\mu}-\mu'||_2\leq \frac{\zeta}{2}\cdot ||\mu||_2$.     By the assumption of the lemma for sufficiently small $\frac{\epsilon \log k}{\varphi^2}$ we have $k^{(40\cdot\epsilon/\varphi^2)} \leq 2$. Thus for any $x\in C'$ we have $||f_x ||_\infty \leq \left(\frac{s}{\delta}\right)^{(40\cdot\epsilon/\varphi^2)}\cdot\sqrt{\frac{160}{s_{\min}}}$. Therefore by triangle inequality we have 
\begin{equation}
\label{eq:mu'bnd}
||\mu'||_2=\left|\left|  \frac{1}{|C'|}\cdot \sum_{x\in C'} f_x \right|\right| \leq \frac{1}{|C'|} \cdot \sum_{x\in C'} ||f_x ||_2 \leq  \frac{\sqrt{k}}{|C'|} \cdot \sum_{x\in C'} ||f_x ||_\infty \leq  \left(\frac{s}{\delta}\right)^{(40\cdot\epsilon/\varphi^2)}\cdot \sqrt{\frac{160\cdot k}{s_{\min}}} \text{.}
\end{equation}
By \eqref{eq:bnd-C'} and by union bound over all samples in $S$ with probability at least $1-s\cdot (k^{-1}\cdot s^{-2}\cdot \delta^{2})=1-s^{-1}\cdot k^{-1}\cdot \delta^2\geq 1-\frac{\delta}{2}$ for all $x_i\in S$ we have $x_i\in C'$, hence, $||f_x ||_\infty \leq \left(\frac{s}{\delta}\right)^{(40\cdot\epsilon/\varphi^2)}\cdot\sqrt{\frac{160}{s_{\min}}}$. Thus with probability at least $1-\frac{\delta}{2}$, $S$ is chosen uniformly at random from $C' $ so for all $x_i\in S$ we have
\begin{equation}
\label{eq:f_xi-bnd}
|f_x ||_\infty \leq \left(\frac{s}{\delta}\right)^{(40\cdot\epsilon/\varphi^2)}\cdot\sqrt{\frac{160}{s_{\min}}}
\end{equation}
In the rest of the proof of step $1$ we assume $S\subseteq C'$  which holds with probability at least $1-\frac{\delta}{2}$. Therefore conditioned on $S\subseteq C'$ we have $\mathbb{E}[f_{x_i}]=\mu'$.
\[
\|\widehat{\mu}- \mu'\|_2=\left|\left|\sum_{i=1}^s\left(\frac{f_{x_i}}{s}- \mu'\right)\right|\right|_2 \text{.}
\]
We define $\mathbf{z}_{i}=\frac{f_{x_i}}{s}- \frac{\mu'}{s}$, so $\|\widehat{\mu}-\mu'\|_2=\|\sum_{i=1}^s \mathbf{z}_{i}\|_2$. Observe that $\mathbb{E}\left[ \mathbf{z}_{i}  \right]=\mathbb{E}\left[ \frac{f_{x_i}}{s}  \right]- \frac{\mu'}{s}=0$, thus we can apply Lemma~\ref{lem:Bernstein}. Therefore we get
\begin{equation}
\label{1eq:GZ}
\mathbb{P}\left[\left|\left| \widehat{\mu}- {\mu'}\right|\right|_2 > q \right]
= \mathbb{P}\left[ \|\sum_{i=1}^s \mathbf{z}_{i}\|_2 > q  \right] \leq (k+1) \cdot \text{exp}\left(\frac{\frac{-q^2}{2}}{\sigma^2+\frac{bq}{3}} \right) \text{,}
\end{equation}
where $\sigma^2= \max\{\|\sum_{i=1}^s \mathbb{E}[\mathbf{z}_{i} \mathbf{z}_{i}^T]\|_2, \|\sum_{i=1}^s
\mathbb{E}[\mathbf{z}_{i}^T \mathbf{z}_{i}]\|_2\} $ and $b$ is an upper bound on $\|\mathbf{z}_{i}\|_2$ for all random variables $\mathbf{z}_{i}$. Therefore we need to upperbound $\|\mathbf{z}_{i}\|_2$ and $\sigma^2$. Note that
\begin{equation}\label{eq:zibnd}
\|\mathbf{z}_{i}\|_2 = \left|\left|  \frac{f_{x_i}}{s}-\frac{\mu'}{s} \right|\right|_2  \leq \left|\left|\frac{f_{x_i}}{s}\right|\right|_2 + \left|\left|  \frac{\mu'}{s} \right|\right|_2 \leq \frac{\sqrt{k}}{s} \cdot \|f_{x_i}\|_\infty + \frac{1}{s}\cdot ||\mu'||_2
\end{equation}
Therefore by \eqref{eq:mu'bnd}, \eqref{eq:f_xi-bnd} and \eqref{eq:zibnd} we have
\begin{equation}
\label{eq:BB}
\|\mathbf{z}_{i}\|_2  \leq \frac{\sqrt{k}}{s} \cdot \|f_{x_i}\|_\infty + \frac{1}{s}\cdot ||\mu'||_2  \leq \frac{2}{s}\cdot \left(\frac{s}{\delta}\right)^{(40\cdot\epsilon/\varphi^2)}\cdot\sqrt{\frac{160\cdot k}{s_{\min}}}\text{,} 
\end{equation}
Thus $b\leq \frac{2}{s}\cdot \left(\frac{s}{\delta}\right)^{(40\cdot\epsilon/\varphi^2)}\cdot\sqrt{\frac{160\cdot k}{s_{\min}}}$. We also need to upper bound $\sigma^2$. By \eqref{eq:BB} we get
\begin{align}
\label{eq:sigsig}
\sigma^2&= \max\{\|\sum_{i=1}^s \mathbb{E}[\mathbf{z}_{i} \mathbf{z}_{i}^T]\|_2, \|\sum_{i=1}^s
\mathbb{E}[\mathbf{z}_{i}^T \mathbf{z}_{i}]\|_2\}  
= s\cdot   \mathbb{E}\left[\|\mathbf{z}_{i}\|^2_2\right] 
\leq s\cdot \frac{4}{s^2}\cdot  \left(\frac{s}{\delta}\right)^{(80\cdot\epsilon/\varphi^2)}\cdot\frac{160 \cdot k}{s_{\min}}\text{.}
\end{align}
We set $q=\frac{\zeta}{2}\cdot ||\mu||_2$. 
Having upper bound for $\sigma^2$ by \eqref{eq:sigsig} and on $b$ by \eqref{eq:BB} we can apply Lemma~\ref{lem:Bernstein} and we get
\begin{align}
\label{1eq:bernstein}
\mathbb{P}\left[ \left|\left| \widehat{\mu}- {\mu'}\right|\right|_2 > \frac{\zeta}{2}\cdot ||\mu||_2  \right]  &\leq (k+1) \cdot \text{exp}\left(\frac{\frac{-q^2}{2}}{\sigma^2+\frac{bq}{3}} \right)  \nonumber\\
 &\leq (k+1) \cdot \text{exp}\left(\frac{\frac{-\zeta^2\cdot  ||\mu||^2_2}{8}}{ \frac{640\cdot k\cdot  \left(\frac{s}{\delta}\right)^{(80\cdot\epsilon/\varphi^2)}}{s\cdot s_{\min}} +\frac{\zeta}{2}\cdot ||\mu||_2\cdot \frac{2\cdot  \left(\frac{s}{\delta}\right)^{(40\cdot\epsilon/\varphi^2)}}{3\cdot s}\sqrt{\frac{160\cdot k}{s_{\min}}}} \right)
\end{align}
By Lemma \ref{lem:dotmu} for small enough $\frac{\epsilon}{\varphi^2}$ we have $\|\mu\|^2_2 \geq \frac{1}{2\cdot |C|}$ and since $\min_{i,j}\frac{|C_i|}{|C_j|}\geq \Omega(1)$. Thus for a small enough constant $c'$ we have 
\begin{equation}
\label{eq:bnd-r}
s_{\min}\cdot ||\mu||^2_2 \geq \frac{s_{\min}}{2\cdot |C| }\geq c'\text{,}
\end{equation}
Thus by \eqref{eq:bnd-r} and by choice of $s^{(1-80\cdot\epsilon/\varphi^2)} \geq \frac{10^6}{c'}\cdot k\cdot \log \left(\frac{k}{\delta}\right) \cdot \left(\frac{1}{\delta}\right)^{(80\cdot\epsilon/\varphi^2)} \cdot \left(\frac{1}{\zeta}\right)^2  \geq\frac{10^6 \cdot k\cdot \log \left(\frac{k}{\delta}\right) \cdot \left(\frac{1}{\delta}\right)^{(80\cdot\epsilon/\varphi^2)} \cdot \left(\frac{1}{\zeta}\right)^2}{s_{\min}\cdot ||\mu||^2_2}$ we get
\begin{equation}\label{eq:pw1-bn}
\frac{\zeta^2\cdot  ||\mu||^2_2}{8} \geq 400\cdot \log \left(\frac{k}{\delta}\right)\cdot \left( \frac{640\cdot k\cdot \left(\frac{s}{\delta}\right)^{(80\cdot\epsilon/\varphi^2)}}{s\cdot s_{\min}} \right)
\end{equation}
and
\begin{equation}\label{eq:pw2-bn}
\frac{\zeta^2\cdot  ||\mu||^2_2}{8} \geq 400\cdot  \log \left(\frac{k}{\delta}\right)\left(  \frac{\zeta}{2}\cdot  ||\mu||_2\cdot \frac{2\cdot \left(\frac{s}{\delta}\right)^{(80\cdot\epsilon/\varphi^2)}}{3\cdot s}\sqrt{\frac{160\cdot k}{s_{\min}}} \right)
\end{equation}
Therefore since $s\geq c\cdot\left(k\cdot \log \left(\frac{k}{\delta}\right) \cdot \left(\frac{1}{\delta}\right)^{(80\cdot\epsilon/\varphi^2)}\cdot\left(\frac{1}{\zeta}\right)^2 \right)^{1/(1-(80\cdot\epsilon/\varphi^2))}$ for large enough constant $c$, and putting \eqref{1eq:bernstein}, \eqref{eq:pw1-bn} and \eqref{eq:pw2-bn} together we get
\[
\mathbb{P}\left[\left|\left| \widehat{\mu}- {\mu'}\right|\right|_2 > \frac{\zeta}{2}\cdot ||\mu||_2  \right] \leq (k+1)\cdot e^{-200\cdot \log \left(\frac{k}{\delta}\right)}\leq \left(\frac{\delta}{k}\right)^{100}
\]
Thus with probability at least $1-\frac{\delta}{2}- \left(\frac{\delta}{k}\right)^{100}\geq 1-\delta$ we have
\begin{align}
\label{eq:finmu1}
\|\widehat{\mu}-\mu'\|_2 &\leq  \frac{\zeta}{2}\cdot \|\mu\|_2  \text{.}
\end{align}
\textbf{Step $2$:}
Next we want to bound $\|\mu-\mu'\|_2$. We have
\begin{align}
\|\mu'-\mu\|_2 &= \left|\left| \frac{1}{|C'|}\sum_{x\in C'} f_x  - \frac{1}{|C|}\sum_{x\in C} f_x \right|\right|_2 \nonumber\\
&\leq \left|\left| \frac{1}{|C'|}\sum_{x\in C} f_x  - \frac{1}{|C|}\sum_{x\in C} f_x \right|\right|_2+\left|\left|\frac{1}{|C'|}\sum_{x\in C\setminus C'} f_x \right|\right|_2 &&\text{By triangle inequality}\nonumber\\
&\leq \left(\frac{1}{1-(k^{-1}\cdot s^{-2}\cdot \delta^{2})}-1\right) \left|\left| \mu \right|\right|_2 +\left|\left|  \frac{1}{|C'|}\sum_{x\in C\setminus C'} f_x \right|\right|_2 && \text{Since $|C'|\geq \left(1-(k^{-1}\cdot s^{-2}\cdot \delta^{2})\right)|C|$ by \eqref{eq:bnd-C'}} \nonumber\\
&\leq 2\cdot (k^{-1}\cdot s^{-2}\cdot \delta^{2})\cdot \|\mu\|_2+\left|\left| \frac{1}{|C'|}\sum_{x\in C\setminus C'} f_x\right|\right|_2 \label{eq:fin4}
\end{align}
It thus remains to upper bound the second term. We now note that 
\begin{equation}\label{eq:2nd-term}
\left|\left| \frac{1}{|C'|}\sum_{x\in C\setminus C'} f_x \right|\right|_2 \leq \frac{1}{|C'|}\sum_{x\in C\setminus C'} ||f_x ||_2 \leq \frac{\sqrt{k}}{|C'|}\sum_{x\in C\setminus C'} \|f_x\|_\infty
\end{equation}
For any $y\geq 1$ we define 
\[T(y)=\left \{x\in V: ||f_x||_\infty \geq 2\cdot y \cdot \left(\frac{s\cdot k}{\delta}\right)^{(40\cdot \epsilon/\varphi^2)} \cdot \sqrt{\frac{10}{s_{\min}}} \right \}\]
Therefore, by Lemma \ref{lem:ap-mu-s} we have 
\begin{equation}\label{eq:T'-bn}
|T(y)|\leq  k\cdot \left(\frac{2\cdot y \cdot \left(\frac{s\cdot k}{\delta}\right)^{(40\cdot \epsilon/\varphi^2)} }{2}\right)^{-\varphi^2/(20\cdot\epsilon)}\cdot s_{\min} = \left(\frac{s\cdot k}{\delta}\right)^{-2}\cdot y^{-\varphi^2/(20\cdot\epsilon)} \cdot s_{\min} \text{.}
\end{equation}
Using the bound on $|T(y)|$ above, we now get 
\begin{align}
&\sum_{x\in C \setminus C'} \|f_x\|_\infty \\
& \leq \int_1^\infty \left( y \cdot \left(\frac{s\cdot k}{\delta}\right)^{(40\cdot \epsilon/\varphi^2)}\cdot \sqrt{\frac{40}{s_{\min}}}  \right)\cdot |T(y)|\cdot dy &&\text{By definition of $T(y)$ and $C'$} \nonumber \\
&\leq \sqrt{\frac{160}{s_{\min}}}   \cdot \left(\frac{s}{\delta}\right)^{(40\cdot \epsilon/\varphi^2)}\cdot \int_1^\infty y\cdot| T(y)| \cdot dy &&\text{Since $k^{(40\cdot \epsilon/\varphi^2)}\leq 2$ for small enough $\frac{\epsilon \cdot \log k}{\varphi^2}$} \nonumber \\
&\leq \sqrt{\frac{160}{s_{\min}}}   \cdot  \left(\frac{s}{\delta}\right)^{(40\cdot \epsilon/\varphi^2)}\cdot  \int_1^\infty \left(\frac{s\cdot k}{\delta}\right)^{-2}\cdot y^{\left( 1-\varphi^2/(20\cdot\epsilon)\right)} \cdot s_{\min} \cdot dy &&\text{By \eqref{eq:T'-bn}} \nonumber \\
&\leq \sqrt{\frac{160}{s_{\min}}}   \cdot s_{\min}\cdot \left(\frac{s}{\delta}\right)^{(40\cdot \epsilon/\varphi^2)}\cdot \left(\frac{s\cdot k}{\delta}\right)^{-2} \frac{1}{\varphi^2/(20\cdot \epsilon)-2} &&\text{Since for any $c<0, \int_1^\infty y^{c} dy=\frac{-1}{c+1}$ } \nonumber \\
&\leq k^{-2}\cdot s^{-1}\cdot\sqrt{s_{\min }} && \text{For small enough $\frac{\epsilon }{\varphi^2}$} \label{eq:anteg}
\end{align}
Therefore we get
\begin{align*}
\left|\left| \frac{1}{|C'|}\sum_{x\in C\setminus C'} f_x \right|\right|_2 
&\leq  \frac{\sqrt{k}}{|C'|}\sum_{x\in C\setminus C'} \|f_x\|_\infty  && \text{By \eqref{eq:2nd-term}} \\
&\leq \frac{\sqrt{k}\cdot k^{-2}\cdot s^{-1}\cdot \sqrt{s_{\min}}}{|C'|}&& \text{By \eqref{eq:anteg}} \\
&\leq   \frac{2\cdot  k^{-1}\cdot s^{-1}}{\sqrt{|C|}} \cdot \frac{\sqrt{s_{\min}}}{\sqrt{|C|}} && \text{By \eqref{eq:bnd-C'} } \\
&\leq   \frac{k^{-1}\cdot s^{-1}}{\sqrt{|C|}}  && \text{Since $|C|\geq s_{\min}$} \\
&\leq 2\cdot k^{-1}\cdot s^{-1}\cdot ||\mu||_2 && \text{By Lemma \ref{lem:dotmu} $||\mu||_2\geq \frac{1}{2\cdot\sqrt{|C|}}$ }
\end{align*}
Therefore by \eqref{eq:fin4} we have
\begin{equation}\label{eq:step2}
\|\mu'-\mu\|_2 \leq 2\cdot (k^{-1}\cdot s^{-2}\cdot\delta^2) \|\mu\|_2+\left|\left| \frac{1}{|C'|}\sum_{x\in C\setminus C'} f_x\right|\right|_2  \leq 2\left(k^{-1}\cdot s^{-2}\cdot\delta^2+ \cdot k^{-1}\cdot s^{-1}\right) ||\mu||_2 \leq \frac{\zeta}{2}\cdot ||\mu||_2
\end{equation}
The last inequality holds since $s\geq 8\cdot \left(\frac{1}{\zeta}\right)^2 $, hence, $2\left(k^{-1}\cdot s^{-2}\cdot\delta^2+ \cdot k^{-1}\cdot s^{-1}\right)\leq \frac{\zeta}{2}$.
Putting \eqref{eq:mu-tri}, \eqref{eq:finmu1} and \eqref{eq:step2} together with probability at least $1-\delta$  we get
\[ ||\wh{\mu}-\mu||_2 \leq  ||\wh{\mu}-\mu'||_2 + ||\mu'-\mu||_2\leq \frac{\zeta}{2}\cdot ||\mu||_2+ \frac{\zeta}{2}\cdot ||\mu||_2 \leq \zeta \cdot ||\mu||_2 \]
\end{proof}

To conclude our argument we show that if we sample enough nodes, we have a large number of samples in each cluster. 

\begin{restatable}{lemma}{lemseachCi}
\label{lem:s-eachCi}
Let $k \geq 2$, $\varphi \in (0,1)$, $\frac{\e \log k}{\varphi^2}$ be smaller than a positive sufficiently small constant. Let $G=(V,E)$ be a $d$-regular graph that admits a $(k,\varphi,\epsilon)$-clustering $C_1, \ldots, C_k$.   Let $S$ be the multiset of  $s\in \Omega(k \log k )$ vertices each sampled independently at random from  $V$. Then with probability at least $\frac{9}{10}$, for every $i\in [k]$,
\[|S\cap C_i| \geq  \frac{0.9 \cdot s }{k} \cdot \min_{p,q \in [k]}\frac{|C_p|}{|C_q|}
\text{.}\]
\end{restatable}

\begin{proof}
For $i\in [k]$, and $1\leq r \leq s$, let $X_{i}^r$ be a random variable which is $1$ if the $r$-th sampled vertex is in $C_i$, and $0$ otherwise. Thus $\mathbb{E}[X_{i}^r]=\frac{|C_i|}{n}$. Observe that $|S\cap C_i|$ is a random variable defined as $\sum_{r=1}^{s} X_i^r$, where its expectation is given by
\[\mathbb{E}[|S\cap C_i|]=\sum_{r=1}^{s} \mathbb{E}[X_i^r]= s\cdot\frac{|C_i|}{n} \geq \frac{ s \cdot s_{\min}}{k \cdot s_{\max}}\text{.}\]
Notice that random variables $X_{i}^r$ are independent,  Therefore, by Chernoff bound,
\[\Pr\left[|S\cap C_i|<\frac{9s}{10}\cdot\frac{|C_i|}{n}\right] \leq \exp\left(-\frac{1}{200}\cdot\frac{ s \cdot s_{\min}}{k \cdot s_{\max}}\right)\text{.}\]
By union bound and since $s=500\cdot k\cdot \log k \cdot \frac{s_{\max}}{s_{\min}}$ we have
\[\Pr\left[\exists i\text{: }|S\cap C_i|<\frac{9s}{10}\cdot\frac{|C_i|}{n}\right] \leq k\cdot \exp\left(-\frac{1}{200}\cdot\frac{ s \cdot s_{\min}}{k \cdot s_{\max}}\right) \leq\frac{1}{10}\text{.}\]
Therefore with probability at least $\frac{9}{10}$ for all $i\in [k]$ we have 
\begin{align*}
|S\cap C_i|&\geq \frac{9\cdot s}{10}\cdot \frac{|C_i|}{n} \geq \frac{0.9\cdot s}{k}\cdot\frac{ s_{\min}}{ s_{\max}} 
\end{align*}
\end{proof}
\subsubsection{Approximate Centers are strongly orthogonal}
The main result of this section is Lemma \ref{lem:apxdosubspace} that generalizes Lemma \ref{lem:dosubspace} to the approximate of cluster means.

\begin{restatable}{lemma}{lemapxdosubspace}
\label{lem:apxdosubspace}
Let  $k \geq 2$ be an integer, $\varphi \in (0,1)$, and $\e \in (0,1)$. Let $G=(V,E)$ be a $d$-regular graph that admits a $(k,\varphi,\e)$-clustering $C_1, \ldots, C_k$. Let $0<\zeta<\frac{\sqrt{\epsilon}}{20\cdot k\cdot \varphi}$. Let $\mu_1,\ldots,\mu_k$ denote the cluster means of $C_1, \ldots, C_k$. Let $\wh{\mu}_1, \dots, \wh{\mu}_k \in \R^k$ denote an approximation of the cluster means such that for each $i\in[k]$, $||\mu_i- \wh{\mu}_i||_2\leq \zeta ||\mu_i||_2$. Let $S \subset \{\wh{\mu}_1, \dots, \wh{\mu}_k\}$ denote a subset of cluster means. 
Let $\wh{\Pi}\in \R^{k\times k}$ denote the orthogonal projection matrix into the $span({S})^{\perp}$. Then the following holds:
\begin{enumerate}
\item For all $\wh{\mu}_i \in \{\wh{\mu}_1, \ldots, \wh{\mu}_k\} \setminus S$  we have $\left| \|\wh{\Pi}\wh{\mu}_i \|_2^2 - ||\wh{\mu_i}||_2^2 \right| \leq \frac{20\sqrt{\e}}{\varphi}\cdot ||\wh{\mu_i}||_2^2 \text{.}$ \label{itm1}
\item For all $\wh{\mu}_i \neq \wh{\mu}_j\in \{\wh{\mu}_1, \ldots, \wh{\mu}_k\} \setminus S$  we have $|\langle \wh{\Pi}\wh{\mu}_i, \wh{\Pi}\wh{\mu}_j \rangle | \leq \frac{50\sqrt{\e}}{\varphi}\cdot\frac{1}{\sqrt{|C_i|\cdot|C_j|}} \text{.}$ \label{itm2}
\end{enumerate}
\end{restatable}
To prove Lemma \ref{lem:apxdosubspace} we use Lemma \ref{lem:apxQQ-1} from Section \ref{sebsec:moment-bounds} 
and we prove Lemma \ref{lem:apxQQT}.

\begin{lemma}
\label{lem:apxQQT}
Let  $k \geq 2$ be an integer, $\varphi \in (0,1)$, and $\e \in (0,1)$. Let $G=(V,E)$ be a $d$-regular graph that admits a $(k,\varphi,\e)$-clustering $C_1, \ldots, C_k$. Let $0<\zeta<\frac{\sqrt{\epsilon}}{20\cdot k\cdot \varphi}$. Let $\wh{\mu}_1, \dots, \wh{\mu}_k \in \R^k$ denote an approximation of the cluster means such that for each $i\in[k]$, $||\mu_i-\wh{\mu}_i||_2\leq \zeta ||\mu_i||_2$. Let $S=\{\wh{\mu}_1,\ldots,\wh{\mu}_k\}\setminus \{\wh{\mu}_i\}$. Let $\wh{H}=[\wh{\mu}_1,\wh{\mu}_2, \ldots, \wh{\mu}_{i-1}, \wh{\mu}_{i+1}, \ldots, \wh{\mu}_k]$ denote a matrix such that its columns are the vectors in $S$.  Let $\wh{W}\in \R^{(k-1)\times (k-1)}$ denote a diagonal matrix such that for all $j<i$ we have $\wh{W}(j,j)=\sqrt{|C_j|}$ and for all $j\geq i$ we have $\wh{W}(j,j)=\sqrt{|C_{j+1}|}$. Let $\wh{Z}=\wh{H}\wh{W}$. Then we have \[\wh{\mu}_i^T \wh{Z} \wh{Z}^{T} \wh{\mu}_i \leq \frac{10\sqrt{\epsilon}}{\varphi} \cdot ||\wh{\mu}_i||_2^2 \text{.}\] 
\end{lemma}
\begin{proof}
Note that $\wh{Z}\wh{Z}^T=(\sum_{j=1}^k |C_j|\wh{\mu}_j\wh{\mu}_j^T) -|C_i|\wh{\mu}_i \wh{\mu}_i^T$. Thus we have 
\begin{equation}
\label{eq:qqt-val}
\wh{\mu}_i^T \wh{Z} \wh{Z}^{T} \wh{\mu}_i = \wh{\mu}_i^T \left(\sum_{j=1}^k |C_j|\wh{\mu}_j\wh{\mu}_j^T \right)\wh{\mu}_i- |C_i|\cdot||\wh{\mu}_i||_2^4\text{.}
\end{equation}
By Lemma \ref{lem:spectraldistance} for any vector $x$ with $||x||_2=1$ we have 
\begin{equation}\label{eq:lem9Cm}
x^T\left(\sum_{j=1}^k |C_j|{\mu}_j{\mu}_j^T - I \right)x\leq \frac{4\sqrt{\epsilon}}{\varphi}
\end{equation}
Note that
\begin{align*}
&||\sum_{j=1}^k |C_j|\wh{\mu}_j\wh{\mu}_j^T - \sum_{j=1}^k |C_j|{\mu}_j{\mu}_j^T ||_2 \\
&\leq \sum_{j=1}^k |C_j|\cdot ||\wh{\mu}_j\wh{\mu}_j^T -{\mu}_j{\mu}_j^T ||_2 && \text{By triangle inequality} \\
&= \sum_{j=1}^k |C_j| \left( ||\left({\mu}_j+(\wh{\mu}_j-{\mu}_j)\right)\left({\mu}_j+(\wh{\mu}_j-{\mu}_j)\right)^T -{\mu}_j{\mu}_j^T ||_2 \right) \\
&\leq \sum_{j=1}^k |C_j| \left( ||\left(\wh{\mu}_j-{\mu}_j\right)\left(\wh{\mu}_j-{\mu}_j\right)^T ||_2 + ||{\mu}_j\left(\wh{\mu}_j-{\mu}_j\right)^T ||_2+ ||\left(\wh{\mu}_j-{\mu}_j\right){\mu}_j^T ||_2 \right) && \text{By triangle inequality} \\
&\leq \sum_{j=1}^k |C_j| \cdot (\zeta^2+2\zeta)\cdot||\mu_j||_2^2&& \text{Since }||\wh{\mu}_j-\mu_j||_2\leq \zeta ||\mu_j||_2 \\
&\leq \sum_{j=1}^k |C_j| \cdot 6\cdot \zeta \cdot\frac{1}{|C_j|}&& \text{By Lemma \ref{lem:dotmu} } ||\mu_j||^2_2\leq\frac{2}{|C_i|} \\
&\leq 6\cdot \zeta\cdot k \\
&\leq \frac{\sqrt{\epsilon}}{2\varphi} &&\text{Since } \zeta\leq  \frac{\sqrt{\epsilon}}{20\cdot k\cdot \varphi}
\end{align*}
Thus for any vector $x$ with $||x||_2=1$ we have 
\begin{equation}\label{eq:Cmux}
x^T\left( \sum_{j=1}^k |C_j|\wh{\mu}_j\wh{\mu}_j^T - \sum_{j=1}^k |C_j|{\mu}_j{\mu}_j^T \right)x\leq \frac{\sqrt{\epsilon}}{2\varphi}
\end{equation}
Putting \eqref{eq:Cmux} and \eqref{eq:lem9Cm}  for any vector any vector $x$ with $||x||_2=1$ we have that
\[x^T\left( \sum_{j=1}^k |C_j|\wh{\mu}_j\wh{\mu}_j^T - I \right)x\leq \frac{5\sqrt{\epsilon}}{\varphi}\] 
Hence we can write 
\begin{align*}
\wh{\mu}_i^T \left(\sum_{j=1}^k |C_j|\wh{\mu}_j\wh{\mu}_j^T \right) \wh{\mu}_i &= \wh{\mu}_i^T \left(\sum_{j=1}^k |C_j|\wh{\mu}_j\wh{\mu}_j^T -I \right)\wh{\mu}_i+ \wh{\mu}_i^T  \wh{\mu}_i 
\leq    \left(1+ \frac{5\sqrt{\epsilon}}{\varphi} \right) ||\wh{\mu}_i||_2^2  
\end{align*}
Therefore by \eqref{eq:qqt-val} we get
\begin{align*}
\wh{\mu}_i^T \wh{Z} \wh{Z}^{T} \wh{\mu}_i &= \wh{\mu}_i^T  \left(\sum_{j=1}^k |C_j|\wh{\mu}_j\wh{\mu}_j^T \right)\wh{\mu}_i- |C_i|\cdot||\wh{\mu}_i||_2^4 
\leq \left(1+ \frac{5\sqrt{\epsilon}}{\varphi} -|C_i|\cdot||\wh{\mu}_i||_2^2 \right) ||\wh{\mu}_i||_2^2 
\end{align*}
By Lemma \ref{lem:dotmu}, and since $||\wh{\mu_i}||\geq (1-\zeta)||\mu_i||_2$ and $\zeta \leq \frac{\sqrt{\epsilon}}{20\cdot k\cdot \varphi}$ we have that 
\[|C_i|\cdot||\wh{\mu}_i||_2^2 \geq \left(1- \frac{4\sqrt{\epsilon}}{\varphi} \right)  (1-\zeta)^2 \geq  1- \frac{5 \sqrt{\epsilon}}{\varphi} \]
Thus we get
\begin{align*}
\wh{\mu}_i^T \wh{Z}\wh{Z} ^{T} \wh{\mu}_i 
\leq \left(1+ \frac{5\sqrt{\epsilon}}{\varphi} -|C_i|\cdot||\wh{\mu}_i||_2^2 \right) ||\wh{\mu}_i||_2^2 
\leq \left(1+ \frac{5\sqrt{\epsilon}}{\varphi} -1+\frac{5\sqrt{\epsilon}}{\varphi} \right) ||\wh{\mu}_i||_2^2 
\leq \frac{10\sqrt{\epsilon}}{\varphi} \cdot ||\wh{\mu}_i||_2^2
\end{align*}
\end{proof}
We now prove the main result of this section (Lemma \ref{lem:apxdosubspace}).
\lemapxdosubspace*
\begin{proof}
\textbf{Proof of item \eqref{itm1}:} Since $\wh{\Pi}$ is a orthogonal projection matrix we have $||\wh{\Pi}||_2=1$. Hence, we have 
\[||\wh{\Pi}\wh{\mu}_i||_2^2\leq ||\wh{\mu}_i||_2^2 \leq \left(1+\frac{20\sqrt{\epsilon}}{\varphi}\right)||\wh{\mu}_i||_2^2\text{.}\]
Thus it's left to prove $||\wh{\Pi}\wh{\mu}_i||_2^2 \geq \left(1-\frac{20\sqrt{\epsilon}}{\varphi}\right)||\wh{\mu}_i||_2^2 $. Note that by Pythagoras $||\wh{\Pi}\wh{\mu}_i||^2_2 = ||\wh{\mu}_i||^2_2 - ||(I-\wh{\Pi})\wh{\mu}_i||^2_2$. We will prove $||(I-\wh{\Pi})\wh{\mu}_i||^2_2 \leq \frac{20\sqrt{\epsilon}}{\varphi}||\wh{\mu}_i||^2_2$ which implies
\[||\wh{\Pi}\wh{\mu}_i||^2_2 \geq \left(1 - 20\frac{\sqrt{\epsilon}}{\varphi} \right) ||\wh{\mu}_i||_2^2\text{.}\]

Thus in order to complete the proof we need to show $||(I-\wh{\Pi})\wh{\mu}_i||^2_2 \leq \frac{20\sqrt{\epsilon}}{\varphi}||\wh{\mu}_i||^2_2$. Let $S'=\{\wh{\mu}_1,\ldots,\wh{\mu}_k\}\setminus \{\wh{\mu}_i\}$. Let $\wh{\Pi}'$ denote the orthogonal projection matrix into $span(S')^{\perp}$. Note that $S\subseteq S'$, hence $span(S)$ is a subspace of $span(S')$, therefore we have $||(I-\wh{\Pi})\wh{\mu}_i||^2_2\leq ||(I-\wh{\Pi}')\wh{\mu}_i||^2_2$. Thus it suffices to prove $||(I-\wh{\Pi}')\wh{\mu}_i||^2_2\leq \frac{20\sqrt{\epsilon}}{\varphi}||\wh{\mu}_i||^2_2$. Let $\wh{H}=[\wh{\mu}_1,\wh{\mu}_2, \ldots, \wh{\mu}_{i-1}, \wh{\mu}_{i+1}, \ldots, \wh{\mu}_k]$ denote a matrix such that its columns are the vectors in $S'$.  Let $\wh{W}\in \R^{(k-1)\times (k-1)}$ denote a diagonal matrix such that for all $j<i$ we have $\wh{W}(j,j)=\sqrt{|C_j|}$ and for all $j\geq i$ we have $\wh{W}(j,j)=\sqrt{|C_{j+1}|}$. Let $\wh{Z}=\wh{H}\wh{W}$. Then the orthogonal projection matrix onto the span of $S'$ is defined as $(I-\wh{\Pi}')=\wh{Z}(\wh{Z}^T \wh{Z})^{-1} \wh{Z}^{T}$. By Lemma \ref{lem:apxQQ-1} item \eqref{itm2:apxZZ}, $(\wh{Z}^T \wh{Z})^{-1}$ is spectrally close to $I$, hence, $(\wh{Z}^T \wh{Z})^{-1}$ exists. Therefore we have 
\begin{align}
||(I-\wh{\Pi}')\wh{\mu}_i||^2_2&= \wh{\mu}_i^T \wh{Z}(\wh{Z}^T \wh{Z})^{-1} \wh{Z}^{T} \wh{\mu}_i \nonumber \\
&= \wh{\mu}_i^T \wh{Z}((\wh{Z}^T \wh{Z})^{-1}- I) \wh{Z}^{T} \wh{\mu}_i + \wh{\mu}_i^T \wh{Z} \wh{Z}^{T} \wh{\mu}_i \label{eq:hat-pi-11}
\end{align}
By Lemma \ref{lem:apxQQ-1} item \eqref{itm2:apxZZ}  we have 
\begin{equation}\label{eq:hat-pi-22}
\left| \wh{\mu}_i^T \wh{Z} \left((\wh{Z}^T \wh{Z})^{-1}-I \right)\wh{Z}^{T}  \wh{\mu}_i \right| \leq \frac{5\sqrt{\epsilon}}{\varphi}||\wh{Z}^{T} \wh{\mu}_i ||_2^2
\end{equation}
Thus we get
\begin{align*}
||(I-\wh{\Pi}')\wh{\mu}_i||^2_2 &\leq \wh{\mu}_i^T \wh{Z}((\wh{Z}^T \wh{Z})^{-1}- I) \wh{Z}^{T} \wh{\mu}_i + \wh{\mu}_i^T \wh{Z} \wh{Z}^{T} \wh{\mu}_i && \text{By \eqref{eq:hat-pi-11}}\\
&\leq \left(\frac{5\sqrt{\epsilon}}{\varphi}+1 \right)||\wh{Z}^{T} \wh{\mu}_i ||_2^2 &&\text{By \eqref{eq:hat-pi-22}} \\
&\leq 2 \cdot ||\wh{Z}^{T} \wh{\mu}_i ||_2^2 &&\text{For small enough }\frac{\epsilon}{\varphi^2} \\
\end{align*}
By Lemma \ref{lem:apxQQT} we have 
\[||\wh{Z}^{T} \wh{\mu}_i ||_2^2= \wh{\mu}_i^T \wh{Z} \wh{Z}^{T} \wh{\mu}_i \leq \frac{10\sqrt{\epsilon}}{\varphi} \cdot ||\wh{\mu}_i||_2^2\]
Therefore  we get
\begin{equation}\label{eq:Pdotbnd}
||(I-\wh{\Pi})\wh{\mu}_i||^2_2 \leq ||(I-\wh{\Pi}')\wh{\mu}_i||^2_2 \leq 2||\wh{Z}^{T} \wh{\mu}_i ||_2^2\leq \frac{20\sqrt{\epsilon}}{\varphi} ||\wh{\mu}_i||_2^2 
\end{equation} 
Hence,
\[||\wh{\Pi}\wh{\mu}_i||^2_2 \geq \left(1 - 20\frac{\sqrt{\epsilon}}{\varphi} ||\wh{\mu}_i||_2^2\right)\text{.}\]
\textbf{Proof of item \eqref{itm2}:} 
Note that
\begin{align*}
\langle \wh{\mu}_i, \wh{\mu}_j \rangle &= \langle (I-\wh{\Pi})\wh{\mu}_i+\wh{\Pi}\wh{\mu}_i ,  (I-\wh{\Pi})\wh{\mu}_j+\wh{\Pi}\wh{\mu}_j\rangle = \langle (I-\wh{\Pi})\wh{\mu}_i ,  (I-\wh{\Pi})\wh{\mu}_j\rangle + \langle \wh{\Pi}\wh{\mu}_i ,  \wh{\Pi}\wh{\mu}_j\rangle 
\end{align*}
Thus by triangle inequality we have
\[ |\langle \wh{\Pi}\wh{\mu}_i ,  \wh{\Pi}\wh{\mu}_j\rangle|  \leq |\langle \wh{\mu}_i, \wh{\mu}_j \rangle |+ |\langle (I-\wh{\Pi})\wh{\mu}_i ,  (I-\wh{\Pi})\wh{\mu}_j\rangle|\]
By Cauchy-Schwarz we have
\begin{align*}
|\langle (I-\wh{\Pi})\wh{\mu}_i ,  (I-\wh{\Pi})\wh{\mu}_j\rangle| 
&\leq || (I-\wh{\Pi})\wh{\mu}_i||_2|| (I-\wh{\Pi})\wh{\mu}_i||_2 \\
&\leq \frac{20\sqrt{\epsilon}}{\varphi} ||\wh{\mu}_i||_2  ||\wh{\mu}_j||_2 &&\text{By \eqref{eq:Pdotbnd}}\\
&\leq \frac{40\sqrt{\epsilon}}{\varphi}\cdot\frac{1}{\sqrt{|C_i||C_j|}} &&\text{By  Lemma \ref{lem:dotmu} and }  ||\wh{\mu_i}-\mu_i||_2 \leq \zeta||\mu_i||_2
\end{align*}
Also for any $i ,j\in [k]$ we have
\begin{align}
&\left|  \rdp{\wh{\mu}_i, \wh{\mu}_j} - \rdp{\mu_i, \mu_j}\right| \nonumber\\
&= \left| \rdp{ {\mu}_i+(\wh{\mu}_i- {\mu}_i),  {\mu}_j+(\wh{\mu}_j- {\mu}_j)} - \rdp{{\mu}_i,  {\mu}_j} \right|   \nonumber\\
&\leq  | \rdp{ \wh{\mu}_i- {\mu}_i , \wh{\mu}_j- {\mu}_j}| + |\rdp{ \wh{\mu}_i- {\mu}_i,  \mu_j}| + |\rdp{ \wh{\mu}_j- {\mu}_j,  \mu_i} |  &&\text{By triangle inequality}  \nonumber \\
&\leq \ || \wh{\mu}_i- {\mu}_i ||_2|| \wh{\mu}_j- {\mu}_j||_2 + || \wh{\mu}_i- {\mu}_i||_2||\mu_j||_2 + || \wh{\mu}_j- {\mu}_j||_2||  \mu_i||_2 &&\text{By Cauchy-Schwarz}  \nonumber \\
&\leq   (\zeta^2+2\zeta)\left( ||\mu_i||_2 ||\mu_j||_2   \right) &&\text{Since }   ||\wh{\mu_i}-\mu_i||_2 \leq \zeta||\mu_i||_2 \text{ for all }i \nonumber \\
&\leq   6\cdot \zeta \cdot \frac{1}{\sqrt{|C_i||C_j|}} && \text{By Lemma \ref{lem:dotmu} }  ||\mu_i||^2_2\leq \frac{2}{|C_i|} \text{ for all }i   \label{eq:6.4muijhats}
\end{align}
Note that
\begin{align*}
|\rdp{\wh{\mu}_i, \wh{\mu}_j}| &\leq |\rdp{{\mu}_i, {\mu}_j}| + | \rdp{{\mu}_i, {\mu}_j} - \rdp{\wh{\mu}_i, \wh{\mu}_j}  | && \text{By triangle inequality} \\
&\leq  \frac{8\sqrt{\epsilon}}{\varphi}\cdot\frac{1}{\sqrt{|C_i||C_j|}}   + 6\zeta \cdot\frac{1}{\sqrt{|C_i||C_j|}} && \text{By Lemma \ref{lem:dotmu} and \eqref{eq:6.4muijhats}} \\
&\leq \frac{10\sqrt{\epsilon}}{\varphi}\frac{1}{\sqrt{|C_i||C_j|}} && \text{Since }\zeta \leq \frac{\sqrt{\epsilon}}{20\cdot k\cdot \varphi}
\end{align*}
Therefore we get 
\[|\langle \wh{\Pi}\wh{\mu}_i ,  \wh{\Pi}\wh{\mu}_j\rangle|  \leq |\langle \wh{\mu}_i, \wh{\mu}_j \rangle |+ |\langle (I-\wh{\Pi})\wh{\mu}_i ,  (I-\wh{\Pi})\wh{\mu}_j\rangle| \leq \frac{50\sqrt{\epsilon}}{\varphi}\cdot \frac{1}{\sqrt{|C_i||C_j|}}\text{.}\]
\end{proof}

\newpage

\subsection{Partitioning scheme works with \textit{approximate} cluster means \& dot products}\label{sec:h_works}

In Section~\ref{sec:realcenterswork} we showed that the partitioning scheme works if we have access to real centers (i.e. $\mu_1, \dots, \mu_k$), to exact dot product evaluations (i.e $\rdp{\cdot, \cdot}$) and \textsc{OuterConductance} is precise.

In this section we show that approximations to all above is enough for the partitioning scheme to work. More precisely we show that if we have access only to $\adp{\cdot, \cdot} \approx \rdp{\cdot, \cdot}$, the search procedure finds $\hat{\mu}_i$'s that are only approximately equal to $\mu_i$'s and \textsc{OuterConductance} is only approximately correct then \textsc{FindCenters} still succeeds with high probability.

In order to prove such a statement we first show a technical Lemma (\textit{Lemma~\ref{lem:qualityofapproximation}}), that relates the approximate dot product with approximate centers to the dot product with the actual cluster centers. 

Note that the following Lemma~\ref{lem:qualityofapproximation} works for any $S \subset \{\mu_1, \dots, \mu_k \}$ and the corresponding $\wh{S}$. This is useful for application in Lemma~\ref{lem:ifclosetomuswin} because it allows to reason about candidate sets $\wh{C}^(T_1, \dots, T_b)_{\wh{\mu}}$, after we associate $\bigcup_{i \in [b]} T_i$ with $\wh{S}$.  

\begin{restatable}{lemma}{lemmqualityofapproximation}\label{lem:qualityofapproximation}
Let $k \geq 2$, $\varphi \in (0,1)$, $\frac{\e }{\varphi^2} $ be smaller than a sufficiently small constant.
Let $G=(V,E)$ be a $d$-regular graph that admits a $(k,\varphi,\epsilon)$-clustering $C_1,\dots,C_k$. 
Then conditioned on the success of the spectral dot product oracle the following conditions hold.

Let  $\widehat{\mu}_1, \widehat{\mu}_2,  \dots, \widehat{\mu}_k$ be such that for all $i \in [k]$ $ \|\wh{\mu}_i - \mu_i\|^2 \leq 10^{-12} \cdot \frac{\e}{\varphi^2 \cdot k^2}\|\mu_i\|^2$. Let $i \in[k]$ and $S \subseteq \{\mu_1, \dots, \mu_k \} \setminus \{\mu_i \}$ and $\wh{S} \subseteq \{\hat{\mu}_1, \dots, \hat{\mu}_k \} \setminus \{ \hat{\mu}_i \}$ be the corresponding subset to $S$. Let $\Pi$ be the orthogonal projection onto $span(S)^{\perp}$ and $\wh{\Pi}$ be the orthogonal projection onto $span(\wh{S})^{\perp}$. 
Let also $\pi_i : \mathbb{R}^k \xrightarrow{} \mathbb{R}^k$ be the projection onto the subspace spanned by $\Pi\mu_i$ and $\wh{\Pi} \widehat{\mu}_i$. Then if $\|\Pi_i f_x\|^2 \leq \frac{10^4}{ \min_{p \in [k]} |C_p|}$ then:

$$ \left|\frac{\langle f_x,\Pi\mu_i \rangle}{\|\Pi\mu_i\|^2} - \frac{\adp{f_x, \wh{\Pi} \widehat{\mu}_i}}{\an{\wh{\Pi} \widehat{\mu}_i}^2} \right| \leq 0.02 $$

Furthermore if $\widehat{\mu}_i$'s are averages of $s$ points, then $\frac{\adp{f_x, \wh{\Pi} \widehat{\mu}_i}}{\an{\wh{\Pi} \widehat{\mu}_i}^2} $ can be computed in $\widetilde{O}_{\varphi} \left(s^4 \cdot \left(\frac{k}{\e} \right)^{O(1)} \cdot n^{1/2 + O(\e/\varphi^2)} \right)$ time with preprocessing time of $\widetilde{O}_{\varphi} \left( \left(\frac{k}{\e} \right)^{O(1)} \cdot n^{1/2 + O(\e/\varphi^2)} \right)$ and space $ \widetilde{O}_{\varphi} \left( \left(\frac{k}{\e} \right)^{O(1)} \cdot n^{1/2 + O(\e/\varphi^2)} \right) $
\end{restatable}

\begin{proof}
First we prove the runtime guarantee and then we show correctness.
\noindent
\paragraph{Runtime.}
We first bound the running time. 
If we set the precision parameter of Algorithm~\ref{alg:dot-apx-pi} to $\xi = 10^{-6} \cdot \frac{\sqrt{\e}}{\varphi}$ then by Theorem~\ref{thm:dot} the preprocessing time takes 
$ \widetilde{O}_{\varphi} \left( \left(\frac{k}{\e} \right)^{O(1)} \cdot n^{1/2 + O(\e/\varphi^2)} \right)$ time, $\widetilde{O}_{\varphi} \left( \left(\frac{k}{\e} \right)^{O(1)} \cdot n^{1/2 + O(\e/\varphi^2)} \right)$ space, and by Corollary~\ref{corr:dotpi} 
computing $\frac{\adp{f_x, \wh{\Pi} \widehat{\mu}_i}}{\an{\wh{\Pi} \widehat{\mu}_i}^2}$ takes
$\widetilde{O}_{\varphi} \left(s^4 \cdot \left(\frac{k}{\e} \right)^{O(1)} \cdot n^{1/2 + O(\e/\varphi^2)} \right)$ time.

\paragraph{Correctness.} Now we show that we also obtain a good approximation. We will show it in two steps:
\begin{enumerate}
    \item\label{step:1} $\left|\frac{\langle f_x,\Pi\mu_i \rangle}{\|\Pi\mu_i\|^2} - \frac{\rdp{f_x, \wh{\Pi} \widehat{\mu}_i}}{\rn{\wh{\Pi} \widehat{\mu}_i}^2} \right|  \leq 0.01$
    \item\label{step:2} $\left|\frac{\langle f_x,\wh{\Pi} \widehat{\mu}_i \rangle}{\|\wh{\Pi} \widehat{\mu}_i\|^2} - \frac{\adp{f_x, \wh{\Pi} \widehat{\mu}_i}}{\an{\wh{\Pi} \widehat{\mu}_i}^2} \right| \leq 0.01 $
\end{enumerate}

If we are able to prove \ref{step:1} and \ref{step:2} then the claim of the Lemma follows from triangle inequality.

\noindent
Before we present the two proofs we show a useful fact:
\begin{align}
\rn{\wh{\Pi} \widehat{\mu}_i - \Pi\mu_i}
&\leq 
\rn{\wh{\Pi} \widehat{\mu}_i - \hat{\mu}_i} + \rn{\Pi\mu_i - \mu_i} + \rn{\wh{\mu}_i - \mu_i} && \text{By triangle inequality} \nonumber \\
&\leq \frac{20 \e^{1/4}}{\sqrt{\varphi}} \rn{\hat{\mu}_i} +
\frac{16 \e^{1/4}}{\sqrt{\varphi}} \rn{\mu_i} + 10^{-6} \cdot \frac{\sqrt{\e}}{\varphi \cdot k}\|\mu_i\| && \text{By Lemma~\ref{lem:apxdosubspace},~\ref{lem:dosubspace} and the bound on } \|\wh{\mu}_i - \mu_i\|^2 \nonumber \\
&\leq \frac{40 \e^{1/4}}{\sqrt{\varphi}} \rn{\mu_i}&& \text{As } \|\wh{\mu}_i - \mu_i\|^2 \leq 10^{-12} \cdot \frac{\e}{\varphi^2 \cdot k^2}\|\mu_i\|^2 \label{eq:pihatmuhat}
\end{align}

\noindent
\textbf{Proof of \ref{step:1}:} Notice that 
\begin{align}
\left|
\frac{\langle f_x,\Pi\mu_i \rangle}{||\Pi\mu_i||^2} - 
\frac{\langle f_x, \wh{\Pi} \wh{\mu}_i \rangle}{||\wh{\Pi} \wh{\mu}_i||^2}  
\right| 
&= 
\left| \rdp{ f_x, \frac{\Pi\mu_i}{||\Pi\mu_i||^2} - \frac{\wh{\Pi} \wh{\mu}_i}{||\wh{\Pi} \wh{\mu}_i||^2} } \right| \nonumber \\
&= 
\left| \rdp{ \Pi_i f_x, \frac{\Pi\mu_i}{||\Pi\mu_i||^2} - \frac{\wh{\Pi} \wh{\mu}_i}{||\wh{\Pi} \wh{\mu}_i||^2} } \right| && \text{By definition of } \pi_i \nonumber \\
&\leq
\rn{\Pi_i f_x} \left|\left|\frac{\Pi\mu_i}{\rn{\Pi\mu_i}^2} - \frac{\wh{\Pi} \wh{\mu}_i}{\rn{\wh{\Pi} \wh{\mu}_i}^2}\right|\right| && \text{By Cauchy-Schwarz} \label{eq:splittwo2}
\end{align}
First we will upper bound $\left|\left|\frac{\Pi\mu_i}{\rn{\Pi\mu_i}^2} - \frac{\wh{\Pi} \wh{\mu}_i}{\rn{\wh{\Pi} \wh{\mu}_i}^2}\right|\right|$. We split it into two cases:

\textbf{Case 1.} 
If $\frac{\Pi\mu_i}{\rn{\Pi\mu_i}^2} \geq \frac{\wh{\Pi} \wh{\mu}_i}{\rn{\wh{\Pi} \wh{\mu}_i}^2}$ then we have:
\begin{align}
\left|\left|\frac{\Pi\mu_i}{\rn{\Pi\mu_i}^2} - \frac{\wh{\Pi} \wh{\mu}_i}{\rn{\wh{\Pi} \wh{\mu}_i}^2}\right|\right|
&\leq 
\left|\left|\frac{\Pi\mu_i}{(1 - \frac{16 \sqrt{\e}}{\varphi} )\rn{\mu_i}^2} - \frac{\wh{\Pi} \wh{\mu}_i}{\rn{\wh{\Pi} \wh{\mu}_i}^2}\right|\right| && \text{By Lemma~\ref{lem:dosubspace}} \nonumber \\
&\leq
\left|\left|\frac{\Pi\mu_i}{(1 - \frac{16 \sqrt{\e}}{\varphi} )\rn{\mu_i}^2} - \frac{\wh{\Pi} \wh{\mu}_i}{(1 + \frac{20 \sqrt{\e}}{\varphi})(1 + 10^{-12} \cdot \frac{\e}{\varphi^2 \cdot k^2} )\rn{\mu_i}^2}\right|\right| &&  \text{Lemma~\ref{lem:apxdosubspace}, assumptions}  \nonumber \\
&\leq \frac{2}{\rn{\mu_i}^2} \left|\left|\Pi\mu_i - \left(1 - \frac{1600 \sqrt{\e}}{\varphi}\right)\wh{\Pi} \widehat{\mu}_i\right|\right| \nonumber \\
&\leq \frac{2}{\rn{\mu_i}^2} \left( \left|\left|\frac{1600 \sqrt{\e}}{\varphi} \Pi\mu_i \right|\right| +  \left(1 - \frac{1600 \sqrt{\e}}{\varphi}\right)\rn{\wh{\Pi} \widehat{\mu}_i - \Pi\mu_i} \right) && \text{By triangle inequality} \nonumber \\
&\leq \frac{12800 \sqrt{\e}}{\varphi} \frac{1}{\rn{\mu_i}} &&\text{By \eqref{eq:pihatmuhat} and Lemma~\ref{lem:dosubspace}} \nonumber
\end{align}

\textbf{Case 2.} 
If $\frac{\Pi\mu_i}{\rn{\Pi\mu_i}^2} < \frac{\wh{\Pi} \wh{\mu}_i}{\rn{\wh{\Pi} \wh{\mu}_i}^2}$ then we have:
\begin{align}
\left|\left|\frac{\Pi\mu_i}{\rn{\Pi\mu_i}^2} - \frac{\wh{\Pi} \wh{\mu}_i}{\rn{\wh{\Pi} \wh{\mu}_i}^2}\right|\right|
&\leq 
\left|\left|\frac{\Pi\mu_i}{(1 + \frac{16 \sqrt{\e}}{\varphi} )\rn{\mu_i}^2} - \frac{\wh{\Pi} \wh{\mu}_i}{\rn{\wh{\Pi} \wh{\mu}_i}^2}\right|\right| && \text{By Lemma~\ref{lem:dosubspace}} \nonumber \\
&\leq
\left|\left|\frac{\Pi\mu_i}{(1 + \frac{16 \sqrt{\e}}{\varphi} )\rn{\mu_i}^2} - \frac{\wh{\Pi} \wh{\mu}_i}{(1 - \frac{20 \sqrt{\e}}{\varphi})(1 - 10^{-12} \cdot \frac{\e}{\varphi^2 \cdot k^2} )\rn{\mu_i}^2}\right|\right| &&  \text{Lemma~\ref{lem:apxdosubspace}, assumptions}  \nonumber \\
&\leq \frac{2}{\rn{\mu_i}^2} \left|\left|\Pi\mu_i - \left(1 + \frac{1600 \sqrt{\e}}{\varphi}\right)\wh{\Pi} \widehat{\mu}_i\right|\right| \nonumber \\
&\leq \frac{2}{\rn{\mu_i}^2} \left( \left|\left|\frac{1600 \sqrt{\e}}{\varphi} \Pi\mu_i \right|\right| +  \left(1 + \frac{1600 \sqrt{\e}}{\varphi}\right)\rn{\wh{\Pi} \widehat{\mu}_i - \Pi\mu_i} \right) && \text{By triangle inequality} \nonumber \\
&\leq \frac{12800 \sqrt{\e}}{\varphi} \frac{1}{\rn{\mu_i}} &&\text{By \eqref{eq:pihatmuhat} and Lemma~\ref{lem:dosubspace}} \nonumber
\end{align}

\noindent
Combining the two cases we get:
$$\left|\left|\frac{\Pi\mu_i}{\rn{\Pi\mu_i}^2} - \frac{\wh{\Pi} \wh{\mu}_i}{\rn{\wh{\Pi} \wh{\mu}_i}^2}\right|\right| \leq \frac{12800 \sqrt{\e}}{\varphi} \frac{1}{\rn{\mu_i}} \text{.}
$$

\noindent Substituting into \eqref{eq:splittwo2} we get:

\begin{align}
\left|
\frac{\langle f_x,\Pi\mu_i \rangle}{||\Pi\mu_i||^2} - 
\frac{\langle f_x, \wh{\Pi} \wh{\mu}_i \rangle}{||\wh{\Pi} \wh{\mu}_i||^2}  
\right|
&\leq \rn{\Pi_i f_x} \cdot \frac{12800 \sqrt{\e}}{\varphi} \frac{1}{\rn{\mu_i}} \nonumber \\
&\leq \frac{100}{\sqrt{\min_{p \in [k]} |C_p|}}\cdot \frac{12800 \sqrt{\e}}{\varphi} \frac{1}{\rn{\mu_i}} &&\text{By assumption of the Lemma} \nonumber \\
&\leq 0.005 \frac{1}{\sqrt{\max_{p \in [k]} |C_p|} \cdot \rn{\mu_i}} && \text{As } \frac{ \e}{ \varphi^2} \text{ is sufficiently small and } \frac{\max_{p \in [k]} |C_p|}{\min_{p \in [k]} |C_p|} = O(1) \nonumber \\
&\leq 0.01 && \text{By Lemma~\ref{lem:dotmu}} \nonumber
\end{align}

\noindent
\textbf{Proof of \ref{step:2}:}
\begin{align}
\an{\wh{\Pi} \widehat{\mu}_i}^2
&\geq \rn{\wh{\Pi} \wh{\mu}_i}^2 - 10^{-6} \cdot \frac{\sqrt{\e}}{\varphi} \cdot n^{-1} &&\text{By Corollary~\ref{corr:dotpi}, setting of $\xi$ and assumptions}  \nonumber\\
&\geq \left(1 - \frac{20 \sqrt{\e}}{\varphi} \right) \cdot \rn{\wh{\mu}_i}^2 - 0.01 \cdot n^{-1} &&\text{By Lemma~\ref{lem:apxdosubspace}  and } \frac{\e}{\varphi^2} \text{ small} \nonumber \\
&\geq \left(1 -10^{-12} \frac{\e}{\varphi^2 \cdot k} \right) \cdot 0.99 \cdot \rn{\mu_i}^2 - 0.01 \cdot n^{-1} && \text{By } \rn{\wh{\mu}_i - \mu_i}^2 \leq 10^{-12} \frac{\e}{\varphi^2 \cdot k} \rn{\mu_i}^2 \text{ and } \frac{\e}{\varphi^2} \text{ small}\nonumber \\
&\geq \left(1 - \frac{4\sqrt{\e}}{\varphi} \right) \cdot 0.98 \cdot n^{-1} - 0.01 \cdot n^{-1} && \text{By Lemma~\ref{lem:dotmu}, } |C_i| \leq n \text{, } \frac{\e}{\varphi^2} \text{ small} \nonumber \\
&\geq 0.5 \cdot n^{-1} &&\text{As } \frac{\e}{\varphi^2} \text{ small} \label{eq:mihatapxnormlwrbnd}
\end{align}

\noindent                                    Next notice that:

\begin{align}
\left|\frac{\langle f_x,\wh{\Pi} \widehat{\mu}_i \rangle}{\|\wh{\Pi} \widehat{\mu}_i\|^2} - \frac{\adp{f_x, \wh{\Pi} \widehat{\mu}_i}}{\an{\wh{\Pi} \widehat{\mu}_i}^2} \right|
&\leq 
\left|\frac{\langle f_x,\wh{\Pi} \widehat{\mu}_i \rangle}{\|\wh{\Pi} \widehat{\mu}_i\|^2} - \frac{\rdp{f_x, \wh{\Pi} \widehat{\mu}_i}}{\an{\wh{\Pi} \widehat{\mu}_i}^2} \right| + \left|\frac{10^{-6} \cdot \frac{\sqrt{\e}}{\varphi} \cdot n^{-1}}{\an{\wh{\Pi} \widehat{\mu}_i}^2}\right|  &&\text{By Corollary~\ref{corr:dotpi}} \nonumber\\
&\leq \left|\rdp{f_x, \wh{\Pi} \wh{\mu}_i} \left( \frac{1}{\|\wh{\Pi} \widehat{\mu}_i\|^2} - \frac{1}{\an{\wh{\Pi} \widehat{\mu}_i}^2}\right) \right| + \left|\frac{10^{-6} \cdot n^{-1}}{0.5 \cdot n^{-1}}\right| && \text{By \eqref{eq:mihatapxnormlwrbnd} and } \frac{\e}{\varphi^2} \text{ small} \nonumber\\
&\leq \left|\rdp{f_x, \wh{\Pi} \wh{\mu}_i} \right| \left| \frac{1}{\|\wh{\Pi} \widehat{\mu}_i\|^2} - \frac{1}{\an{\wh{\Pi} \widehat{\mu}_i}^2} \right| + 10^{-5}
\label{eq:firstofthree}
\end{align}

\noindent Now we will separately bound $\left|\rdp{f_x, \wh{\Pi} \wh{\mu}_i} \right|$ and $\left| \frac{1}{\|\wh{\Pi} \widehat{\mu}_i\|^2} - \frac{1}{\an{\wh{\Pi} \widehat{\mu}_i}^2} \right|$ from \eqref{eq:firstofthree}. As $|\rdp{a,b}| \leq \rn{a} \cdot \rn{b}$ we get:
\begin{equation}\label{eq:dotproducttonorms}
\left|\rdp{f_x, \wh{\Pi} \wh{\mu}_i} \right|
\leq \rn{\Pi_i f_x} \cdot \rn{\wh{\Pi} \wh{\mu}_i}
\end{equation}
Now we bound the second term from \eqref{eq:firstofthree}:
\begin{align}
\left|\frac{1}{\|\wh{\Pi} \widehat{\mu}_i\|^2} - \frac{1}{\an{\wh{\Pi} \widehat{\mu}_i}^2} \right|
&= 
\left|
\frac{ \an{\wh{\Pi} \widehat{\mu}_i}^2 - ||\wh{\Pi} \widehat{\mu}_i||^2 }{\rn{\wh{\Pi} \widehat{\mu}_i}^2\an{\wh{\Pi} \widehat{\mu}_i}^2} 
\right| 
\nonumber\\
&\leq
\left|
\frac{10^{-6} \cdot \frac{\sqrt{\e}}{\varphi} \cdot n^{-1}}{\rn{\wh{\Pi} \widehat{\mu}_i}^2\an{\wh{\Pi} \widehat{\mu}_i}^2} 
\right| 
&& \text{Corollary~\ref{corr:dotpi}, setting of $\xi$ and assumptions} \nonumber\\
&\leq
10^{-5} \cdot \frac{\sqrt{\e}}{\varphi} \cdot \left|
\frac{ 0.5 \cdot n^{-1}}{\rn{\wh{\Pi} \widehat{\mu}_i}^2 \cdot 0.5 \cdot n^{-1}} 
\right|  && \text{By \eqref{eq:mihatapxnormlwrbnd} }  \nonumber\\
&\leq
10^{-5} \cdot \frac{\sqrt{\e}}{\varphi} \cdot \left|
\frac{1}{\rn{\wh{\Pi} \wh{\mu}_i} \cdot (\rn{\Pi\mu_i} - \frac{40\e^{1/4}}{\sqrt{\varphi}}\rn{\mu_i} ) } 
\right| && \text{By \eqref{eq:pihatmuhat}} \nonumber \\
&\leq 10^{-4} \cdot \frac{\sqrt{\e}}{\varphi} \cdot  \frac{ 1}{\rn{\wh{\Pi} \wh{\mu}_i}  \cdot \rn{\mu_i}  }  && \text{By Lemma~\ref{lem:dosubspace} and } \frac{\e}{\varphi^2} \text{ small}
\label{eq:secondofthree}
\end{align}

\noindent
Substituting \eqref{eq:dotproducttonorms} and \eqref{eq:secondofthree} in \eqref{eq:firstofthree} we get:

\begin{align*}
\left|\frac{\langle f_x,\wh{\Pi} \widehat{\mu}_i \rangle}{\|\wh{\Pi} \widehat{\mu}_i\|^2} - \frac{\adp{f_x, \wh{\Pi} \widehat{\mu}_i}}{\an{\wh{\Pi} \widehat{\mu}_i}^2} \right|
&\leq 
10^{-5} + 10^{-4} \cdot \frac{\sqrt{\e}}{\varphi} \cdot \frac{ \rn{\Pi_i f_x}}{\rn{\mu_i}  }  \\
&\leq 
10^{-5} + 10^{-4} \cdot \frac{\sqrt{\e}}{\varphi} \cdot \frac{100}{\sqrt{\min_{p \in [k]} |C_p|}} \cdot \frac{ 1}{\rn{\mu_i}  } && \text{By assumption} \\
&\leq 
10^{-5} + 10^{-3} \frac{1}{\sqrt{\max_{p \in [k]} |C_p|} \cdot \rn{\mu_i}} && \text{As } \frac{\e}{\varphi^2} \text{ small, } \frac{\max_{p \in [k]} |C_p|}{\min_{p \in [k]} |C_p|} = O(1) \\
&\leq 0.01 && \text{By Lemma~\ref{lem:dotmu}}
\end{align*}

\end{proof}

Now we are ready to show that there exist an algorithm (Algorithm~\ref{alg:estimateconductande}) that can estimate accurately the size of candidate clusters of the form $\wh{C}^{(T_1, \dots, T_b)}_{\wh{\mu}}$ and then, if the size is not too small, estimate outer-conductance of all candidate clusters. The proof of  correctness of the algorithm is based on applications of standard concentration bounds.

\begin{algorithm}
\caption{\textsc{OuterConductance}($G,\wh{\mu},(T_1,T_2,\ldots, T_{b}),S,s_1,s_2)$ \newline \text{ } \Comment $T_i$'s are sets of $\widehat{\mu}_j$ where $\widehat{\mu}_j$'s are given as sets of points \newline \text{ } \Comment see Section~\ref{sec:dotproductcomp} for the reason of such representation \newline \text{ } \Comment $s_1$ is \# sampled points for size estimation \newline \text{ } \Comment $s_2$ is \# sampled points for outer-conductance estimation}\label{alg:estimateconductande}
\begin{algorithmic}[1]
	\State $\text{cnt} := 0$
	\For {$t =1$ to $s_1$}
	    \State $x \sim \textsc{Uniform}\{1..n\}$ \Comment Sample a random vertex and test if it belongs to the cluster
	    \If{$\textsc{IsInside}(x, \widehat{\mu},(T_1,T_2,\ldots, T_{b}), S)$}
            \State	$\text{cnt} := \text{cnt} + 1$
	    \EndIf
	\EndFor
	\If{$\frac{n}{s_1} \cdot \text{cnt} < \min_{p \in [k]} |C_p|/2$}
	    \State \Return $\infty$ \Comment If the estimated size is too small return $\infty$
	\EndIf
	\State $e := 0, a := 0$
	\For {$t =1$ to $s_2$}
	    \State $x \sim \textsc{Uniform}\{1..n\}$
	    \State $y \sim \textsc{Uniform}\{w \in \mathcal{N}(u)\}$ \Comment $\mathcal{N}(u) = $ neighbors of $u$ in $G$
	     \If{$\textsc{IsInside}(x, \widehat{\mu},(T_1,T_2,\ldots, T_{b}), S)$}
            \State	$a := a + 1$
           \If{$ \neg\textsc{IsInside}(y, \widehat{\mu},(T_1,T_2,\ldots, T_{b}), S)$} \label{ln:istheneighboroutside}
                \State	$e = e + 1$
	        \EndIf
	    \EndIf
	\EndFor
	\State \Return $\frac{e}{a}$

\end{algorithmic}
\end{algorithm}

\begin{lemma}\label{lem:conductanceglued}
Let $k \geq 2$, $\varphi,\e,\gamma \in (0,1)$. 
Let $G=(V,E)$ be a $d$-regular graph that admits a $(k,\varphi,\epsilon)$-clustering $C_1,\dots,C_k$.

For a set of approximate centers $\{\wh{\mu}_1, \dots, \wh{\mu}_k\}$, where each $\wh{\mu}_i$ is represented as an average of at most $s$ embedded vertices (i.e $f_x$'s), an ordered partial partition $(T_1, \dots, T_b)$ of $\{\wh{\mu}_1, \dots, \wh{\mu}_k\}$ and $\wh{\mu} \in \{\wh{\mu}_1, \dots, \wh{\mu}_k\} \setminus \bigcup_{j \in [b]} T_i$ the following conditions hold.

If Algorithm~\ref{alg:estimateconductande} is invoked with $(G, \wh{\mu}, (T_1, \dots, T_b), \{\wh{\mu}_1, \dots, \wh{\mu}_k\} \setminus \bigcup_{j \in [b]} T_i, s_1, s_2)$ then it runs in $
\widetilde{O}_{\varphi} \left( (s_1 + s_2) \cdot s^4 \cdot \left(\frac{k}{\e} \right)^{O(1)} \cdot n^{1/2 + O(\e/\varphi^2)} \right)
$ time and if $s_1 =  \Theta( k \log(\frac{1}{\eta}))$ and $s_2 = \Theta(\frac{\varphi^2 \cdot  k}{\e}  \log(\frac{1}{\eta}))$ then with probability $1 - \eta$ it returns a value $q$ with the following properties.

\begin{itemize}
    \item If $|\wh{C}_{\wh{\mu}}^{(T_1, \dots, T_b)}| \geq \frac{3}{4} \min_{p \in [k]} |C_p|$ then 
    $q \in \left[\frac{1}{2}\phi \left(\wh{C}_{\wh{\mu}}^{(T_1, \dots, T_b)} \right) - \e/\varphi^2, \frac{3}{2}\phi \left(\wh{C}_{\wh{\mu}}^{(T_1, \dots, T_b)} \right) + \e/\varphi^2 \right] $,
    \item If $|\wh{C}_{\wh{\mu}}^{(T_1, \dots, T_b)}| < \frac{3 }{4} \min_{p \in [k]} |C_p|$ then 
    $q \geq \frac{1}{2}\phi \left(\wh{C}_{\wh{\mu}}^{(T_1, \dots, T_b)} \right) - \e/\varphi^2 \text{.}$
\end{itemize}
\end{lemma}

\begin{proof}

We start with the runtime analysis then follows the correctness analysis.

\paragraph{Runtime.} Algorithm~\ref{alg:estimateconductande} has two phases: one from line 1 to line 7 and second from line 8 to line 16. 

During the first phase Algorithm~\ref{alg:estimateconductande} calls Algorithm~\ref{alg:inside} $s_1$ times and Algorithm~\ref{alg:inside} runs in $\widetilde{O}_{\varphi} (s^4 \cdot \left(\frac{k}{\e} \right)^{O(1)} \cdot n^{1/2 + O(\e/\varphi^2)} ) $  
time as it computes $k^{O(1)}$ values of the form $\frac{\langle f_x, \hat{\mu}_i \rangle_{apx}}{||\hat{\mu}_i||^2_{apx}}$ which are computed in time $\widetilde{O}_{\varphi} (s^4 \cdot \left(\frac{k}{\e} \right)^{O(1)} \cdot n^{1/2 + O(\e/\varphi^2)} )$ by Lemma~\ref{lem:qualityofapproximation}, so in total the runtime of this phase is $\widetilde{O}_{\varphi} (s_1 \cdot s^4 \cdot \left(\frac{k}{\e} \right)^{O(1)} \cdot n^{1/2 + O(\e/\varphi^2)} ) $.

During the second phase Algorithm~\ref{alg:estimateconductande} calls Algorithm~\ref{alg:inside} $2s_2$ times so the runtime of this phase is $\widetilde{O}_{\varphi} (s_2 \cdot s^4 \cdot \left(\frac{k}{\e} \right)^{O(1)} \cdot n^{1/2 + O(\e/\varphi^2)} )$ in total.

So in total the runtime is $\widetilde{O}_{\varphi} ((s_1 +s_2) \cdot s^4 \cdot \left(\frac{k}{\e} \right)^{O(1)} \cdot n^{1/2 + O(\e/\varphi^2)} )$.

\paragraph{Correctness.} For simplicity we denote $\wh{C}_{\wh{\mu}}^{(T_1, \dots, T_b)}$ by $\wh{C}$ and $\min_{p \in [k]} |C_p|$ by $r_{\min}$ in this proof. Notice that the Algorithm~\ref{alg:estimateconductande} in the first phase computes $\text{cnt} = \sum_{i=1}^{s} X_i$, where $X_i$'s are independent Bernoulli trials with success probability $p = \frac{|\wh{C}|}{n}$. Let $z := \frac{n}{s_1} \sum_{i=1}^{s_1} X_i$. We introduce notation $x \approx_{\delta,\alpha} y$ to denote $x \in [ (1-\delta)y - \alpha, (1+ \delta)y + \alpha]$. By Chernoff-Hoeffding bounds we get that there exists a universal constant $\Gamma$ such that for all $0 < \delta \leq 1/2, \alpha > 0$ 
$$z \approx_{\delta, \alpha \cdot n} |\wh{C}| \text{ with probability } 
1 - 2^{-\Gamma s_1 \alpha \delta} \text{.}$$
Setting $\delta = 1/2, \alpha = \frac{r_{\min}}{8 n}$ we get that $z \approx_{1/2, r_{\min}/8} |\wh{C}|$ with probability 
$$
1 - 2^{-\Gamma s_1 \frac{r_{\min}}{32 n}} 
\geq 1 - 2^{-\Omega( s_1/k)} \text{,}
$$
as $\frac{\max_{p \in [k]} |C_p|}{\min_{p \in [k]} |C_p|} = O(1)$.
So if $s_1 = \Theta(k \log(1/\eta))$ then with probability $1 - \eta/2$ we have 
\begin{equation}\label{eq:firststageguaranteen}
z \approx_{1/2, r_{\min}/8} |\wh{C}| \text{.}
\end{equation}

Observe that if $\wh{C} < r_{\min}/4$ then by \eqref{eq:firststageguaranteen} we have that $z \leq (1+1/2)|\wh{C}| + r_{\min}/8 < r_{\min}/2$, which means that Algorithm~\ref{alg:estimateconductande} returns $\infty$. Note that it is consistent with the conclusion of the Lemma.

For the analysis of the second stage we assume that $|\wh{C}| \geq r_{\min}/4$.
We will analyze what value is returned in the second stage. First we will bound the probability that $a \leq  \frac{s_2 \cdot r_{\min}}{8 \cdot n}$. For $i \in [1 \dots s_2]$ let $X_i$ be a binary random variable which is equal $1$ iff in $i-th$ iteration of the for loop we increase the $a$ counter. We have that, for every $i$, $P[X_i = 1] = |\wh{C}|/n$ and the $X_i$'s are independent. Notice that $a = \sum_{i = 1}^{s_2} X_i$. From Chernoff bound we have that for $\delta < 1$:
\begin{equation}\label{eq:twosided1}
P\left[ \left|\sum_{i = 1}^{s_2} X_i - \mathbb{E}\left[\sum_{i=1}^{s_2} X_i\right] \right| > \delta \cdot \mathbb{E}\left[\sum_{i=1}^{s_2} X_i\right] \right] 
\leq
2e^{-\frac{\delta^2}{3} \mathbb{E}\left[\sum_{i=1}^{s_2} X_i\right]}, 
\end{equation}
Noticing that $\mathbb{E}\left[\sum_{i=1}^{s_2} X_i\right] = s_2 \frac{ |\wh{C}|}{n}$ if we set $\delta = 1/2$ we get that
\begin{equation}\label{eq:twosided}
P\left[ \left|\sum_{i = 1}^{s_2} X_i - s_2 \frac{ |\wh{C}|}{n} \right| > s_2 \frac{|\wh{C}|}{2n} \right] 
\leq
2e^{- s_2 \frac{|\wh{C}|}{12n}}
\leq 2e^{- \frac{s_2 \cdot r_{\min} }{48 \cdot n}}, 
\end{equation}
So with probability at least $1 - 2e^{- \frac{s_2 \cdot r_{\min}}{48 \cdot n}} \geq 1 - 2e^{-\Omega(s_2/k)}$ (as $\frac{\max_{p \in [k]} |C_p|}{\min_{p \in [k]} |C_p|} = O(1)$) we have that 
\begin{equation}\label{eq:manys}
a = \sum_{i = 1}^{s} X_i \geq \frac{1}{2} \cdot s_2 \cdot \frac{|\wh{C}|}{n} \geq \frac{s_2 \cdot r_{\min} }{8 \cdot n} \geq \Omega(s_2/k).
\end{equation}

Now observe that line~\ref{ln:istheneighboroutside} of \textsc{OuterConductance} is invoked exactly $a$ times. Let $Y_j$ be the indicator random variable that is $1$ iff $e$ is increased in the $j$-th call of line~\ref{ln:istheneighboroutside}. 
Notice that 
\begin{equation}
P[Y_i = 1] = \phi(\wh{C})
\end{equation}
That is because if $U_i$ is a random variable denoting a vertex $u$ sampled in $i$-th step then  $U_i$  is uniform on set $\wh{C}$ conditioned on $X_i = 1$ and the graph is regular. Now by the Chernoff-Hoeffging bounds we get that for all $0 < \delta \leq 1/2,\alpha > 0$ we have:
$$\frac{1}{a}\sum_{i = 1}^{a} Y_i \approx_{\delta, \alpha} \phi(\wh{C}) \text{ with probability } 1 - 2e^{-\Gamma a \alpha \delta}  \text{.}$$
Setting $\delta = 1/2, \alpha = \frac{\e}{\varphi^2}$ we get that $\frac{1}{a}\sum_{i = 1}^{a} Y_i \approx_{1/2,\e/\varphi^2} \phi(\wh{C})$ with probability:
\begin{equation}\label{eq:second2guarantee}
1 - 2e^{-\Gamma a \e /(4\varphi^2)} \geq 1 - 2e^{-\Omega(a \e /\varphi^2)} 
\end{equation}
Now taking the union bound over \eqref{eq:manys} and \eqref{eq:second2guarantee} we get that if we set $s_2 = \Theta(\frac{\varphi^2 \cdot k}{\e} \log(1/\eta))$ then  $\frac{1}{a}\sum_{i = 1}^{a} Y_i \approx_{1/2,\e/\varphi^2} \phi(\wh{C})$ with probability:
\begin{align*}
1 - 2e^{-\Omega(s_2/k)} - 2e^{-\Omega(a\e/\varphi^2)}
&\geq 1 - 2e^{-\Omega(s_2/k)} - 2e^{-\Omega(\frac{\e \cdot s_2 }{\varphi^2 \cdot k})} && \text{By \eqref{eq:manys}}  \\
&\geq 1 -\eta/2
\end{align*}
To conclude the proof we observe the following. 
\begin{itemize}
    \item If $|\wh{C}| < \frac{r_{\min}}{4}$ then with probability $1 - \eta/2$ the Algorithm returns $\infty$,
    \item If $|\wh{C}| \in [\frac{r_{\min}}{4},\frac{3\cdot r_{\min}}{4}) $ then either the Algorithm returns $\infty$ in the first stage or it reaches the second stage and with probability $1 - \eta$ it returns a value $\psi$ such that $\psi \approx_{1/2, \e/\phi^2} \varphi(\wh{C})$,
    \item If $|\wh{C}| \geq \frac{r_{\min}}{4}$ then by the union bound over the two stages with probability $1 - \eta$ it reaches the second stage and returns a value $\psi$ such that $\psi \approx_{1/2, \e/\varphi^2} \phi(\wh{C})$.
\end{itemize}
The above covers all the cases and is consistent with the conclusions of the Lemma.\

\end{proof}

Before we give the statement of the next Lemma we introduce some definitions. In Lemma~\ref{lem:conductanceglued} we proved that for every call to \textsc{OuterConductance} the value returned by the Algorithm~\ref{alg:estimateconductande} is, in a sense given by the conclusions of Lemma~\ref{lem:conductanceglued}, a good approximation to outer-conductance of $\wh{C}^{(T_1, \dots, T_b)}_{\wh{\mu}}$ (where $\wh{\mu}, (T_1, \dots, T_b)$ are the parameters of the call) with high probability. What follows is a definition of an event that the values returned by \textsc{OuterConductance} throughout the run of the final algorithm always satisfy one conclusion of  Lemma~\ref{lem:conductanceglued}. Later we use Definition~\ref{def:event1} in Lemma~\ref{lem:ifclosetomuswin} and then in the proof of Theorem~\ref{thm:uglythm} we will lower bound the probability of $\mathcal{E}_{\text{conductance}}$.

\begin{definition}[\textbf{Event $\mathcal{E}_{\text{conductance}}$}]\label{def:event1}
Let $k \geq 2$, $\varphi,\e,\gamma \in (0,1)$. 
Let $G=(V,E)$ be a $d$-regular graph that admits a $(k,\varphi,\epsilon)$-clustering $C_1,\dots,C_k$.

We define $\mathcal{E}_{\text{conductance}}$ as an event:

For every call to Algorithm~\ref{alg:estimateconductande} (i.e. \textsc{OuterConductance}$ $) that is made throughout the run of \textsc{FindCenters} the following is true. If Algorithm~\ref{alg:estimateconductande} is invoked with $(G, \wh{\mu}, (T_1, \dots, T_b), \{\wh{\mu}_1, \dots, \wh{\mu}_k\} \setminus \bigcup_{j \in [b]} T_i, s_1, s_2)$ then it returns a value $q$ with the following property.

\begin{itemize}
    \item If $|\wh{C}_{\wh{\mu}}^{(T_1, \dots, T_b)}| \geq \frac{3}{4} \min_{p \in [k]} |C_p|$ then $q \in \left[\frac{1}{2}\phi \left(\wh{C}_{\wh{\mu}}^{(T_1, \dots, T_b)} \right) - \e/\varphi^2, \frac{3}{2}\phi \left(\wh{C}_{\wh{\mu}}^{(T_1, \dots, T_b)} \right) + \e/\varphi^2 \right]$.
\end{itemize}
\end{definition}

The following Lemma is the key part of the corresponding proof of correctness of Algorithm~\ref{alg:testmus} (see Theorem~\ref{thm:uglythm} below).
It is a generalization of Lemma~\ref{lem:induction}.
We show that if $\wh{\mu}$'s are close to real centers and $\mathcal{E}$ and $\mathcal{E}_{\text{conductance}}$ hold then at every stage of the for loop from line~\ref{testcentersforloop} of Algorithm~\ref{alg:testmus} at least half of the candidate clusters:
$$\mathcal{C}_i := \bigcup_{\wh{\mu} \in S} \{ \wh{C}_{\wh{\mu}}^{(T_1, \dots, T_{i-1})} \} \text{,}$$
pass the test from line~\ref{testcenterscondactancetest} of Algorithm~\ref{alg:testmus}, which means that they have small outer-conductance and satisfy condition \eqref{eq:sec3outerconductancegoal}.

\begin{lemma}\label{lem:ifclosetomuswin}
Let $k \geq 2$, $\varphi \in (0,1)$, $\frac{\e}{\varphi^2} $ be smaller than a sufficiently small constant. 
Let $G=(V,E)$ be a $d$-regular graph that admits a $(k,\varphi,\epsilon)$-clustering $C_1,\dots,C_k$. 
Then conditioned on the success of the spectral dot product oracle there exists an absolute constant $\Upsilon$ such that the following conditions hold.

If \textsc{ComputeOrderedPartition}($G,\wh{\mu}_1,\wh{\mu}_2, \dots, \widehat{\mu}_k,s_1,s_2)$ is called with $(\wh{\mu}_1, \dots, \wh{\mu}_k)$ such that for every $i \in [k]$ we have $\|\wh{\mu}_i - \mu_i\|^2 \leq 10^{-12} \cdot \frac{\e}{\varphi^2 \cdot k^2}\|\mu_i\|^2$ then the following holds. Assume that at the beginning of the $i$-th iteration of the for loop from line~\ref{testcentersforloop} of Algorithm~\ref{alg:testmus} $|S| = b$ and, up to renaming of $\wh{\mu}$'s, $S = \{\wh{\mu}_1,\dots, \wh{\mu}_b\}$, the corresponding clusters are $\mathcal{C} = \{C_1, \dots, C_b \}$ respectively and the ordered partial partition of $\mu$'s is equal to $(T_1, \dots, T_{i-1})$. 
Then if for every $C \in \mathcal{C}$ we have that $|V^{(T_1, \dots, T_{i-1})} \cap C| \geq \left(1 - \Upsilon \cdot i \cdot \frac{\e}{\varphi^2} \right) |C|$ then at the beginning of $(i+1)$-th iteration:
\begin{enumerate}
    \item $|S| \leq b/2$ (that is at least half of the remaining cluster means were removed in $i$-th iteration), 
    \item\label{cond:second} for every $\mu \in S$ the corresponding cluster $C$ satisfies $|V^{(T_1, \dots, T_{i})} \cap C| \geq \left(1- \Upsilon \cdot (i+1) \cdot \frac{\e}{\varphi^2} \right) |C| $, where $(T_1, \dots, T_{i})$ is the ordered partial partition of $\mu$'s created in the first $i$ iterations. 
\end{enumerate}

\end{lemma}

\begin{proof}
\noindent
\textbf{Outline of the proof.} We start but defining a subset of vertices called outliers and then we show that the number of them is small. Next we prove that for vertices that are not outliers the evaluations of $\frac{\adp{f_x, \wh{\Pi} \widehat{\mu}_i}}{\an{\wh{\Pi} \widehat{\mu}_i}^2}$ are approximately correct (as in Lemma~\ref{lem:qualityofapproximation}). Next we mimic the structure, and on the high level the logic, of the proof of Lemma~\ref{lem:induction}: we first show the first conclusion of the Lemma and then the second one.

For simplicity we will denote $\min_{p \in [k]} |C_p|$ by $r_{\min}$ in this proof. Without loss of generality we can assume $S = \{\hat{\mu}_1,\dots, \hat{\mu}_b\}$ at the beginning of the $i$-th iteration of the for loop from line~\ref{testcentersforloop} of Algorithm~\ref{alg:testmus} and the corresponding clusters be $C_1, \dots, C_b$ respectively. Assume that for every $C \in \mathcal{C}$ we have that $|V^{(T_1, \dots, T_{i-1})} \cap C| \geq \left(1 - \Upsilon \cdot i \cdot \frac{\e}{\varphi^2} \right) |C|$.


Let $\wh{\Pi}$ be the projection onto the $\text{span} (\bigcup_{j<i} T_j)^{\perp}$. Recall that each $T_j$ is a subset of $\{\wh{\mu}_1, \dots, \wh{\mu}_k \}$. For every $j < i$ let 
$$T_j' := \bigcup_{\wh{\mu} \in T_j} \{\mu \} \text{.}$$
That is $T_j'$'s are $T_j$'s with $\wh{\mu}$'s replaced by the corresponding $\mu$'s. Now let $\Pi$ be the projection onto the $\text{span} (\bigcup_{j<i} T_j')^{\perp}$.

\paragraph{Outliers.} First we define a set of outliers, i.e. $X$, as the set of points with abnormally long projection onto the subspace spanned by $\{\Pi \mu_1, \dots, \Pi \mu_b, \wh{\Pi} \wh{\mu}_1, \dots, \wh{\Pi} \wh{\mu}_b \}$. Then we show that the number of outliers is small.

Let $Q$ be the orthogonal projection onto the $\text{span}(\{\Pi \mu_1, \dots, \Pi \mu_b, \wh{\Pi} \hat{\mu}_1, \dots, \wh{\Pi} \hat{\mu}_b \})$ and let:
$$X := \left\{ x \in V : \rn{Q f_x}^2 > \frac{10^4 }{r_{\min}} \right\} $$
By Lemma~\ref{lem:smallinsubspace} we get that 
\begin{equation}\label{eq:varianceinP}
\sum_{x \in V} \rn{Q f_x - Q \mu_x}^2 \leq O \left(b \cdot \frac{\e}{\varphi^2} \right) \text{.}
\end{equation}
Moreover for every $x \in X$:
\begin{align}
\rn{Q f_x - Q \mu_x} 
&\geq 
\rn{Q f_x} - \rn{Q \mu_x} && \text{By triangle  inequality} \nonumber \\
&\geq 
\rn{Q f_x} - \rn{\mu_x} && \text{As projection can only decrease the norm} \nonumber \\
&> 
\frac{10^2}{\sqrt{r_{\min}}} - \left(1 + O\left(\frac{\sqrt{\e}}{\varphi}\right) \right)\frac{1}{\sqrt{r_{\min}}} && \text{By Lemma~\ref{lem:dotmu} and Definition of $X$} \nonumber \\
&\geq \frac{90}{\sqrt{r_{\min}}} &&\text{For } \frac{\e}{\varphi^2} \text{ small enough}\label{eq:normofcentersmall}
\end{align}
Combining \eqref{eq:varianceinP} and \eqref{eq:normofcentersmall} we get:
\begin{equation}\label{eq:numberofoutliers}
|X| \leq O\left(b \cdot \frac{\e}{\varphi^2}\right) \cdot r_{\min} \leq O\left(b \cdot \frac{\e}{\varphi^2}\right) \cdot \frac{n}{k}
\end{equation}

\paragraph{Tests performed for non-outliers are approximately correct.} 
Observe that by the fact that spectral dot product succeeds we have by Lemma~\ref{lem:qualityofapproximation} that for all $x \in V \setminus X$ and for all $i \in [k]$:
\begin{equation}\label{eq:nonoutliers}
 \left|\frac{\langle f_x,\Pi\mu_i \rangle}{\|\Pi\mu_i\|^2} - \frac{\adp{f_x, \wh{\Pi} \widehat{\mu}_i}}{\an{\wh{\Pi} \widehat{\mu}_i}^2} \right| \leq 0.02 \text{,}
\end{equation}
as $\rn{Q f_x}^2 \leq \frac{10^4 }{r_{\min}}$ and the norm in any subspace can only be smaller and thus the assumption of Lemma~\ref{lem:qualityofapproximation} is satisfied.

\paragraph{1. At least half of the cluster means is removed from $S$.} Now we proceed with proving that most of the candidate clusters $\wh{C}_{\wh{\mu}}^{(T_1, \dots, T_{i-1})}$ have small outer-conductance and thus the corresponding $\wh{\mu}$'s are removed from set $S$ (see line~\ref{testcenterscondactancetest} of \textsc{ComputeOrderedPartition}). For brevity we will refer to $(T_1, \dots, T_{i-1})$ as $P$ in this proof.

Let $\mu \in S$. Let
$$ I :=  \bigcup_{\mu', \mu'' \in \{\mu_1, \dots, \mu_d\} } \Cr{\Pi  \mu', 0.9} \cap \Cr{\Pi \mu'', 0.9} \text{.}$$
\noindent
By Lemma~\ref{lem:pointsoutside} we have that 
\begin{equation}\label{eq:bndonintersections}
\left| I \right|  \leq O \left( b \cdot \frac{\e}{\varphi^2} \right) \cdot \frac{n}{k} \text{.}
\end{equation}

So by \eqref{eq:numberofoutliers} and \eqref{eq:bndonintersections} and Markov inequality we get that there exists a subset of clusters $\mathcal{R}  \subseteq \mathcal{C}$ such that $|\mathcal{R}| \geq b/2$ and for every $C \in \mathcal{R}$ we have that: 
\begin{equation}\label{eq:smallnumboutside}
|C \cap (I \cup X)| \leq O \left( \frac{\e}{\varphi^2} \right) \cdot \frac{n}{k}
\end{equation}
We will argue that for any order of the for loop from line~\ref{testcentersforloop} of Algorithm~\ref{alg:testmus} it is true that for every $C \in \mathcal{R}$ with corresponding means $\mu, \wh{\mu}$ the candidate cluster $\widehat{C}_{\wh{\mu}}^P$ satisfies the if statement from line~\ref{testcenterscondactancetest} of Algorithm~\ref{alg:testmus}. Recall that as per Definition~\ref{def:candidateclusters}:
$$\wh{C}_{\wh{\mu}}^{P} = \left\{ x \in V : \textsc{IsInside} \left(x,\wh{\mu},P, \{\wh{\mu}_1, \dots, \wh{\mu}_k \} \setminus \bigcup_{j \in [i-1]} T_j \right) = \textsc{True} \right\} \text{.}$$

First note that behavior of the algorithm is independent of the order of the for loop from line~\ref{testcentersforloop} of Algorithm~\ref{alg:testmus}  as by definition $\wh{C}_{\wh{\mu}}^P$'s for $\wh{\mu} \in S$ are pairwise disjoint. Now let $C \in \mathcal{R}$, $\mu, \wh{\mu}$ be the means corresponding to $C$ and $\wh{C}_{\wh{\mu}}^P$ be the candidate cluster corresponding to $\wh{\mu}$ with respect to $P = (T_1, \dots, T_{i-1})$. 

\textit{Now the goal is to show:
$$ |\widehat{C}_{\wh{\mu}}^P \triangle C| \leq O \left(\frac{\e}{\varphi^2}  \cdot \log(k) \right) \cdot |C| \text{,}$$
from which we will later conclude that the outer-conductance of the candidate set $\wh{C}_{\wh{\mu}}^P$ is small. Intuitively we would like to argue that
\begin{equation}\label{eq:informalgoal}
\Cr{\Pi \mu, 0.96} \stackrel{\sim}{\subseteq} \Ca{\wh{\Pi} \wh{\mu},0.93} \stackrel{\sim}{\subseteq} \Cr{\Pi \mu, 0.9} \text{,}
\end{equation}
and then use Lemmas from Section~\ref{sec:bound_int}. The equation \eqref{eq:informalgoal} is true up to the outliers as Lemma~\ref{lem:qualityofapproximation} guarantees a bound of $0.02$ for the test computations for vertices of small norm.} 

Now we give a formal proof, which is split into $2$ parts:

\paragraph{Showing $|\wh{C}_{\wh{\mu}}^P \cap C| \geq \left(1 - O \left(\frac{\e}{\varphi^2}  \cdot \log(k) \right) \right) |C|$.}
First we note that by \eqref{eq:nonoutliers} $\Cr{\Pi \mu,0.96}$ is mostly contained in $\Ca{\wh{\Pi} \wh{\mu},0.93}$. Recall that (see  Definition~\ref{def:apxthreshold sets} and Definition~\ref{def:thresholdsets}) we have:
$$\Ca{\wh{\Pi} \wh{\mu},0.93} = \left\{x \in V : \adp{f_x,\wh{\Pi} \wh{\mu}} \geq 0.93 \an{\wh{\Pi} \wh{\mu}}^2 \right\} \text{,}$$
$$\Cr{\Pi \mu,0.96} = \left\{x \in V : \rdp{f_x,\Pi \mu} \geq 0.96 \rn{\Pi \mu}^2 \right\} \text{.}$$
And \eqref{eq:nonoutliers} gives us that the errors for non-outliers are bounded by $0.02$, so formally we get:
\begin{equation}\label{eq:mostlytakes96}
\Cr{\Pi \mu,0.96} \setminus \Ca{\wh{\Pi} \wh{\mu},0.93} \subseteq X    
\end{equation}
Similarly, also by \eqref{eq:nonoutliers} we get that the intersections of candidate clusters $\Ca{\wh{\Pi} \wh{\mu},0.93}$ lie mostly in $I$. Formally:
\begin{equation}\label{eq:intersectionstruct}
\Ca{\wh{\Pi} \wh{\mu},0.93} \cap \bigcup_{\hat{\mu}' \neq \hat{\mu}} \Ca{\wh{\Pi} \hat{\mu}',0.93} \subseteq I \cup X 
\end{equation}
By Lemma~\ref{lem:mostinset} we get that
\begin{equation}\label{eq:goodintersection}
|C \cap \Cr{\Pi \mu, 0.96}| \geq \left(1 - O \left( \frac{\e}{\varphi^2} \right) \right) |C|
\end{equation}
\textit{Note that having two thresholds ($0.9$ and $0.96$) is very important here (see Remark~\ref{rem:twothresholds}). Intuitively we need some slack to show $\Cr{\Pi \mu, 0.96} \stackrel{\sim}{\subseteq} \Ca{\wh{\Pi} \wh{\mu},0.93} \stackrel{\sim}{\subseteq} \Cr{\Pi \mu, 0.9}$ as there is always some error in computation of $\frac{\adp{f_x, \wh{\Pi} \widehat{\mu}_i}}{\an{\wh{\Pi} \widehat{\mu}_i}^2}$.} 

Now combining inductive assumption  $|V^P \cap C| \geq \left(1 - \Upsilon \cdot i \cdot \frac{\e}{\varphi^2} \right) |C|$, \eqref{eq:smallnumboutside}, \eqref{eq:mostlytakes96}, \eqref{eq:intersectionstruct} and \eqref{eq:goodintersection} we get that:  
\begin{align}
|\widehat{C}_{\wh{\mu}}^P \cap C| 
&\geq 
\left(1 - \Upsilon \cdot i \cdot \frac{\e}{\varphi^2} \right) |C| - O \left( \frac{\e}{\varphi^2} \right) \cdot \frac{n}{k} - O \left( \frac{\e}{\varphi^2}  \right) \cdot |C| \nonumber \\
&\geq 
\left(1 - O \left(\frac{\e}{\varphi^2}  \cdot \log(k) \right)\right)|C| \label{eq:regioncontainslotofcluster}
\end{align}
\paragraph{Showing $|\wh{C}_{\wh{\mu}}^P \cap (V^P \setminus C)| \leq O\left( \frac{\e}{\varphi^2} \right) |C|$.}
Recall that as per Definition~\ref{def:candidateclusters} we have:
$$V^{P} = V \setminus \bigcup_{j < i} \bigcup_{\wh{\mu} \in T_j} \wh{C}_{\wh{\mu}}^{(T_1, \dots, T_{j-1})}$$
By Lemma~\ref{lem:notalostfromoutside} we get that:
\begin{equation}\label{eq:intersectionwithrest}
|\Cr{\Pi \mu, 0.9} \cap (V^P \setminus C)| \leq |\Cr{\Pi \mu, 0.9} \cap (V \setminus C)| \leq O \left(\frac{\e}{\varphi^2} \right) |C|
\end{equation}
By \eqref{eq:nonoutliers} we get:
\begin{equation}\label{eq:allinoutliersfirst}
\Ca{\wh{\Pi} \wh{\mu}, 0.93} \setminus \Cr{\Pi \mu, 0.9} \subseteq X    
\end{equation}

\noindent Let $\pi'$ be the projection onto the span of $\{ \Pi\mu, \wh{\Pi} \hat{\mu} \}$. Moreover let:
$$X' := \left\{x \in V : \|\pi' f_x\|^2 > \frac{10^4 }{r_{\min}} \right\} \text{.}$$
Note that by Lemma~\ref{lem:smallinsubspace} we have:
\begin{equation}\label{eq:boundin2dim}
\sum_{x \in V} \rn{\pi' f_x - \pi' \mu_x}^2 \leq O \left(\frac{\e}{\varphi^2} \right)    
\end{equation}
Moreover for every $x \in X'$ we have:
\begin{align}
\rn{\pi' f_x - \pi' \mu_x}
&\geq \rn{\pi' f_x} - \rn{\pi' \mu_x} && \text{By } \triangle \text{ inequality} \nonumber \\
&\geq \frac{10^2}{\sqrt{r_{\min}}} - \frac{2}{\sqrt{r_{\min}}} && \text{By Lemma~\ref{lem:dotmu}} \nonumber \\
&\geq \frac{90}{\sqrt{r_{\min}}} \label{eq:distin2dim}
\end{align}
Combining \eqref{eq:boundin2dim} and \eqref{eq:distin2dim} we get that:
\begin{equation}\label{eq:bndonoutliers}
|X'| \leq O\left(\frac{\e}{\varphi^2}\right) \cdot r_{\min} \leq O\left(\frac{\e}{\varphi^2}\right) \cdot \frac{n}{k}
\end{equation}
Then similarly to the analysis of \eqref{eq:nonoutliers} by Lemma~\ref{lem:qualityofapproximation} and the fact that spectral dot product succeeds we have that for every $x \in V \setminus X'$:
$$\left|\frac{\langle f_x,\Pi\mu \rangle}{\|\Pi\mu\|^2} - \frac{\adp{f_x, \wh{\Pi} \widehat{\mu}}}{\an{\wh{\Pi} \widehat{\mu}}^2} \right| \leq 0.02$$
Thus we get:
\begin{equation}\label{eq:allinoutliers}
\Ca{\wh{\Pi} \wh{\mu}, 0.93} \setminus \Cr{\Pi \mu, 0.9} \subseteq X' \text{,}    
\end{equation}
as for points not belonging to $X'$ the error in the tests performed by the Algorithm is upper bounded by $0.02$. Combining \eqref{eq:intersectionwithrest} and \eqref{eq:allinoutliers} we have:
\begin{equation}\label{eq:takenfromoutside}
|\wh{C}_{\wh{\mu}}^P \cap (V^P \setminus C)| \leq O\left( \frac{\e}{\varphi^2} \right) |C|
\end{equation}
And finally putting \eqref{eq:regioncontainslotofcluster} and \eqref{eq:takenfromoutside} together we have:
\begin{equation}\label{eq:smallsymdiff}
    |\widehat{C}_{\wh{\mu}}^P \triangle C| \leq O \left(\frac{\e}{\varphi^2}  \cdot \log(k) \right) \cdot |C|
\end{equation}

\paragraph{Outer-conductance of $\wh{C}_{\wh{\mu}}^P$ is small.}
Now we want to argue that $\wh{C}_{\wh{\mu}}^P$ passes the outer-conductance test from line~\ref{testcenterscondactancetest} in Algorithm~\ref{alg:testmus}. From the definition of outer-conductance:
\begin{align*}
\phi(\wh{C}_{\wh{\mu}}^P) 
&\leq 
\frac{E(C,V \setminus C) + d|\wh{C}_{\wh{\mu}}^P \triangle C|}{d(|C| - |\wh{C}_{\wh{\mu}}^P \triangle C|)} \\
&\leq 
\frac{E(C,V \setminus C) + d \cdot O \left(\frac{\e}{\varphi^2}  \cdot \log(k) \right) |C|}{d(|C| - O \left(\frac{\e}{\varphi^2}  \cdot \log(k) \right) |C|)} && \text{from \eqref{eq:smallsymdiff}} \\
&\leq 
\frac{O \left(\frac{\e}{\varphi^2} \right) + O \left(\frac{\e}{\varphi^2}  \cdot \log(k) \right)}{1-O \left(\frac{\e}{\varphi^2}  \cdot \log(k) \right)} && \text{because } \frac{E(C,V \setminus C)}{d|C|} \leq O \left(\frac{\e}{\varphi^2} \right)\\
&\leq 
O \left(\frac{\e}{\varphi^2}  \cdot \log(k) \right) && \text{for sufficiently small } \frac{\e}{\varphi^2}  \cdot \log(k)  
\end{align*}
 and it follows that
  $$\phi(\wh{C}_{\wh{\mu}}^P) \leq O \left(\frac{\e}{\varphi^2}  \cdot \log(k) \right) \text{,}$$ 
To conclude we notice that by \eqref{eq:smallsymdiff} we have $|\wh{C}_{\wh{\mu}}^P| > \frac{3 \cdot r_{\min}}{4}$, so as $\mathcal{E}_{\text{conductance}}$ is true we get that the candidate cluster $\wh{C}_{\wh{\mu}}^P$ passes the test.

\paragraph{2. Clusters corresponding to unremoved $\wh{\mu}$'s satisfy condition \ref{cond:second}.}

\noindent
Now we prove that for every $\wh{\mu}$ that was not removed from set $S$ only small fraction of its corresponding cluster is removed. 

Let $\wh{\mu} \in S$ be such that it is not removed in the $i$-th step and let $\mu$ be the corresponding real center. Let $C \in \mathcal{C}$ be the cluster corresponding to $\mu$. By assumption $|V^P \cap C| \geq \left(1- \Upsilon \cdot i \cdot  \frac{\e}{\varphi^2} \right) |C| $, where recall that $P = (T_1, \dots, T_{i-1})$. 

\textit{Now the goal is to show:
$$ |C \cap (V^{(T_1, \dots, T_{i-1})} \setminus V^{(T_1, \dots, T_{i})})| \leq O \left( \frac{\e}{\varphi^2} \right) |C| + O\left(\frac{\e}{\varphi^2}\right) \cdot \frac{n}{k} \leq  O \left( \frac{\e}{\varphi^2} \right) |C|  \text{,} $$
that is, that there is only a small number of vertices that were removed in the $i$-th stage and belong to $C$ at the same time. Intuitively we want to show that:
$$(V^{(T_1, \dots, T_{i-1})} \setminus V^{(T_1, \dots, T_{i})}) \cap \Cr{\Pi \mu, 0.96} \approx \emptyset \text{,}$$
and then use Lemmas from Section~\ref{sec:bound_int}. The equation above is true up to the outliers as Lemma~\ref{lem:qualityofapproximation} guarantees a bound of $0.02$ for the test computations for vertices of small norm.} 

Now we give a formal proof. Let $x \in V^{(T_1, \dots, T_{i-1})} \setminus V^{(T_1, \dots, T_{i})} = V^P \setminus V^{(T_1, \dots, T_{i})} $, where $(T_1, \dots, T_{i})$ is the partial partition of $\wh{\mu}$'s created in the first $i$ steps of the for loop of \textsc{ComputeOrderedPartition}. Then there exists $\wh{\mu'} \in \{\wh{\mu}_1, \dots, \wh{\mu}_b \}$ such that $x \in \wh{C}_{\wh{\mu}'}^P$ (recall that $\wh{C}_{\wh{\mu}'}^P$ is the candidate cluster corresponding to $\wh{\mu}'$ with respect to $P = (T_1, \dots, T_{i-1})$). Recall (Definition~\ref{def:candidateclusters}) that $\wh{C}_{\wh{\mu}'}^P$ is defined as:
$$\wh{C}_{\wh{\mu}'}^P = \left\{ x \in V : \textsc{IsInside} \left(x,\wh{\mu}',P, \{\wh{\mu}_1, \dots, \wh{\mu}_k \} \setminus \bigcup_{j \in [i-1]} T_j \right) = \textsc{True} \right\} \text{.}$$
This in particular means (see line~\ref{ln:dot-x-mu} of Algorithm \textsc{IsInside}) that:
$$\wh{C}_{\wh{\mu}'}^P \subseteq \Ca{\wh{\Pi} \wh{\mu}',0.93} \setminus \bigcup_{\wh{\mu}''\in  S\setminus \{\wh{\mu}'\}} \Ca{\wh{\Pi} \wh{\mu}'',0.93}  ,$$
which, as $\wh{\mu} \in S \setminus \{ \wh{\mu}' \}$, gives us that:
$$\wh{C}_{\wh{\mu}'}^P \cap  \Ca{\wh{\Pi} \wh{\mu},0.93} = \emptyset \text{,}$$
which using Definition~\ref{def:thresholdsets} gives that:
\begin{equation}\label{eq:propertyofreturnedeasy}
\adp{f_x, \wh{\Pi} \wh{\mu}} < 0.93 \an{\wh{\Pi} \wh{\mu}}^2\text{.}
\end{equation}

\noindent
We define $X'$ similarly as in point $1$. Let $\pi'$ be the projection onto the span of $\{ \Pi\mu, \wh{\Pi} \hat{\mu} \}$. Moreover let:
$$X' := \left\{x \in V : \|\pi' f_x\|^2 > \frac{10^4 }{r_{\min}} \right\} \text{.}$$
Similarly to the proof of \eqref{eq:bndonoutliers} we get
\begin{equation}\label{eq:bndoutliersu}
|X'| \leq O\left( \frac{\e}{\varphi^2} \right) \cdot r_{\min} \leq O\left( \frac{\e}{\varphi^2} \right) \cdot \frac{n}{k}
\end{equation}
Again similarly to the analysis of \eqref{eq:nonoutliers} we note that by Lemma~\ref{lem:qualityofapproximation} and the fact that spectral dot product succeeds:
\begin{equation}\label{eq:qualityfornonoutliers'}
\text{for every } y \in V \setminus X' \text{ we have  }\left|\frac{\langle f_y,\Pi\mu \rangle}{\|\Pi\mu\|^2} - \frac{\adp{f_y, \wh{\Pi} \widehat{\mu}}}{\an{\wh{\Pi} \widehat{\mu}}^2} \right| \leq 0.02
\end{equation}
Combining \eqref{eq:qualityfornonoutliers'} and \eqref{eq:propertyofreturnedeasy} we get that if $x \in V \setminus X'$ then
\begin{align*}
\frac{\langle f_x,\Pi\mu \rangle}{\|\Pi\mu\|^2}   
&\leq \frac{\adp{f_y, \wh{\Pi} \widehat{\mu}}}{\an{\wh{\Pi} \widehat{\mu}}^2} + 0.02 \\
&< 0.93 + 0.02 \\
&< 0.96
\end{align*}
which also means that $x \not\in \Cr{\Pi\mu,0.96}$. This means that:
\begin{equation}\label{eq:takeonlyoutliers}
(V^{(T_1, \dots, T_{i-1})} \setminus V^{(T_1, \dots, T_{i})} ) \cap \Cr{\Pi \mu, 0.96} \subseteq X'
\end{equation}

\noindent
But by Lemma~\ref{lem:mostinset}:
\begin{equation}\label{eq:smalloutsideeasy}
|\{x \in C :\rdp{\Pi f_x, \Pi\mu} < 0.96 \| \Pi\mu \|_2^2 \}| \leq O \left( \frac{\e}{\varphi^2} \right) \cdot |C| 
\end{equation}
Combining \eqref{eq:takeonlyoutliers}, \eqref{eq:bndoutliersu} and \eqref{eq:smalloutsideeasy} we get that: 
\begin{equation}\label{eq:ind2secondtolast}
|C \cap (V^{(T_1, \dots, T_{i-1})} \setminus V^{(T_1, \dots, T_{i})})| \leq O \left( \frac{\e}{\varphi^2} \right) |C| + O\left(\frac{\e}{\varphi^2}\right) \cdot \frac{n}{k} \leq  O \left( \frac{\e}{\varphi^2} \right) |C|  \text{.}
\end{equation}
By assumption that  $|V^{(T_1, \dots, T_{i-1})} \cap C| \geq \left(1 - \Upsilon \cdot i \cdot  \frac{\e}{\varphi^2} \right) |C| $  and \eqref{eq:ind2secondtolast} we get that: 
$$|V^{(T_1, \dots, T_{i})} \cap C| \geq \left(1 - \Upsilon \cdot (i+1) \cdot \frac{\e}{\varphi^2 } \right) |C| \text{,}$$
provided that $\Upsilon$ is bigger than the constant hidden under $O$ notation in \eqref{eq:ind2secondtolast}.

\end{proof}

The following Lemma is a generalization of Theorem~\ref{lem:realcenterswork} that uses Lemma~\ref{lem:ifclosetomuswin} as an inductive step to show that if \textsc{ComputeOrderedPartition} is called with $\wh{\mu}$'s that are good approximations to $\mu$'s then it returns an ordered partition that induces a good collection of clusters.

\begin{lemma}\label{thm:closeapxmuswin}
Let $k \geq 2$, $\varphi \in (0,1)$ and $\frac{\e}{\varphi^2}  \cdot \log(k) $ be smaller than a sufficiently small constant.
Let $G=(V,E)$ be a $d$-regular graph that admits a $(k,\varphi,\epsilon)$-clustering $C_1, \dots, C_k$.
Then conditioned on the success of the spectral dot product oracle the following conditions hold.

If \textsc{ComputeOrderedPartition}($G,\widehat{\mu}_1,\widehat{\mu}_2, \dots, \widehat{\mu}_k,s_1,s_2)$ is called with $(\hat{\mu}_1, \dots, \hat{\mu}_k)$ such that for every $i \in [k]$ we have $\|\wh{\mu}_i - \mu_i\|^2 \leq 10^{-12} \cdot \frac{\e}{\varphi^2 \cdot k^2}\|\mu_i\|^2$ then \textsc{ComputeOrderedPartition} returns $(\textsc{True}, (T_1 ,\dots, T_b))$ such that $(T_1, \dots, T_b)$ induces a collection of clusters $\{\wh{C}_{\wh{\mu}_1}, \dots, \wh{C}_{\wh{\mu}_k}\}$ such that there exists a permutation $\pi$ on $k$ elements such that for all $i \in [k]$:
$$\left|\wh{C}_{\wh{\mu}_{i}} \triangle C_{\pi(i)}\right| \leq O \left(\frac{\e}{\varphi^3}  \cdot \log(k) \right)|C_{\pi(i)}|$$ and 
 $$\phi(\wh{C}_{\wh{\mu}_i}) \leq O \left(\frac{\e}{\varphi^2}  \cdot \log(k) \right) \text{.}$$
\end{lemma}

\begin{proof}
Note that for $i=0$ in the for loop in line~\ref{ln:testcentersmainloop} of \textsc{ComputeOrderedPartition} $S$ and clusters $\{C_1, \dots, C_k \}$ trivially satisfy assumptions of Lemma~\ref{lem:ifclosetomuswin}. So using Lemma~\ref{lem:ifclosetomuswin} and induction we get that for every $i \in [0..\lceil \log(k) \rceil]$ at the beginning of the $i$-th iteration:
\begin{itemize}
    \item $|S| \leq k / 2^i$, 
    \item for every $\wh{\mu} \in S$ with corresponding $\mu$ and corresponding cluster $C$ we have $ |V^{(T_1, \dots, T_{i-1})} \cap C| \geq \left(1 - \Upsilon \cdot i \cdot \frac{\e}{\varphi^2} \right)|C|$ (where $\Upsilon$ is the constant from the statement of Lemma~\ref{lem:ifclosetomuswin}).
\end{itemize} 

In particular this means that after $O(\log(k))$ iterations set $S$ becomes empty. This also means that \textsc{ComputeOrderedPartition} returns in line~\ref{ln:testcenterstruereturn}, so it returns \textsc{True} and the ordered partial partition $(T_1, \dots, T_b)$ is in fact an ordered partition of $\{\wh{\mu}_1, \dots, \wh{\mu}_k\}$. 

Note that by definition (see Definition~\ref{def:implicitclustering}) all the approximate clusters $\{\wh{C}_{\wh{\mu}_1}, \dots, \wh{C}_{\wh{\mu}_k}\}$ are pairwise disjoint and moreover for every constructed cluster $\wh{C} \in \{\wh{C}_{\wh{\mu}_1}, \dots, \wh{C}_{\wh{\mu}_k}\}$ we have:
$$\phi(\widehat{C}) \leq O \left(\frac{\e}{\varphi^2}  \cdot \log(k) \right) \text{,}$$
as it passed the test in line~\ref{testcenterscondactancetest} of \textsc{ComputeOrderedPartition}. So by Lemma~\ref{lem:howtocluster} it means that there exists a permutation $\pi$ on $k$ elements such that for all $i \in [k]$:
$$\left|\wh{C}_{\wh{\mu}_{i}} \triangle C_{\pi(i)}\right| \leq  O \left(\frac{\e}{\varphi^3}  \cdot \log(k) \right)|C_{\pi(i)}| \text{.}$$

Recall Remark~\ref{rem:firstphithensym} for why the proof follows this framework of first arguing about outer-conductance and only after that, using Lemma~\ref{lem:howtocluster}, reasoning about symmetric difference.
\end{proof}

Now we present the final Theorem of this section which shows that \textsc{FindCenters} with high probability returns an ordered partition that induces a good collection of clusters. The proof is a careful union bound of error probabilities.

\begin{theorem}\label{thm:uglythm}
Let $k \geq 2$, $\varphi \in (0,1)$, $\frac{\e \log(k)}{\varphi^3}$  be smaller than a sufficiently small constant. Let $G=(V,E)$ be a $d$-regular graph that admits a $(k,\varphi,\epsilon)$-clustering $C_1,\dots,C_k$. Then Algorithm~\ref{alg:findrepresentatives} with probability $1 - \eta$ returns an ordered partition $(T_1 ,\dots, T_b)$ such that $(T_1, \dots, T_b)$ induces a collection of clusters $\{\wh{C}_{\wh{\mu}_1}, \dots, \wh{C}_{\wh{\mu}_k}\}$ such that there exists a permutation $\pi$ on $k$ elements such that for all $i \in [k]$:
$$\left|\wh{C}_{\wh{\mu}_{i}} \triangle C_{\pi(i)}\right| \leq  O \left(\frac{\e}{\varphi^3}  \cdot \log(k) \right)|C_{\pi(i)}|$$ and 
 $$\phi(\wh{C}_{\wh{\mu}_i}) \leq  O \left(\frac{\e}{\varphi^2}  \cdot \log(k) \right) \text{.}$$ 
Moreover
\begin{itemize}
    \item Algorithm~\ref{alg:findrepresentatives} (\textsc{FindCenters}) runs in time $$\widetilde{O}_{\varphi} \left(\log^2(1/\eta) \cdot 2^{O(\frac{\varphi^2}{\e} \cdot k^4 \log^2(k))} \cdot n^{1/2 + O(\e/\varphi^2)} \right) \text{,}$$ and uses $\widetilde{O}_{\varphi}\left( \left(\frac{k}{\e} \right)^{O(1)} \cdot n^{1/2 + O(\e/\varphi^2)} \right)$ space,
    \item Algorithm~\ref{alg:ballcarving} (\textsc{HyperplanePartitioning}) called with $(T_1, \dots, T_b)$ as a parameter runs in time $\widetilde{O}_{\varphi}\left( \left(\frac{k}{\e} \right)^{O(1)} \cdot n^{1/2 + O(\e/\varphi^2)} \right)$ per one evaluation.
\end{itemize}
\end{theorem}

\begin{proof}
We first prove the runtime guarantee and then we show correctness.

\paragraph{Runtime.} The first step of \textsc{FindCenters} (Algorithm~\ref{alg:findrepresentatives}) is to call \textsc{InitializeOracle($G,1/2$)} (Algorithm~\ref{alg:LearnEmbedding}) which by Lemma~\ref{lem:qualityofapproximation} runs in time $\widetilde{O}_{\varphi}\left( \left(\frac{k}{\e} \right)^{O(1)} \cdot n^{1/2 + O(\e/\varphi^2)} \right)$ and uses $\widetilde{O}_{\varphi}\left( \left(\frac{k}{\e} \right)^{O(1)} \cdot n^{1/2 + O(\e/\varphi^2)} \right)$ space (It's the preprocessing time in the statement of Lemma~\ref{lem:qualityofapproximation}). Then Algorithm~\ref{alg:findrepresentatives} repeats the following procedure $O(\log(1/\eta))$ times.

It tests all partitions of a set of sampled vertices of size $s = O( \frac{\varphi^2}{\e} \cdot k^4 \log(k) )$ 
into $k$ sets. There is at most $k^{s} = 2^{O(\frac{\varphi^2}{\e} \cdot k^4 \log^2(k))}$ of them. Notice that for each partition each $\widehat{\mu}_i$ is defined as 
$$\widehat{\mu}_i := \frac{1}{|P_i|}\sum_{x \in P_i} f_x \text{,}$$ 
so as the number of sampled points is $O( \frac{\varphi^2}{\e} \cdot k^4 \log(k) )$ then each $\widehat{\mu}_i$ is an average of at most $O( \frac{\varphi^2}{\e} \cdot k^4 \log(k) )$ points. To analyze the runtime notice that:

\begin{itemize}
    \item For each partition Algorithm~\ref{alg:findrepresentatives} runs Algorithm~\ref{alg:testmus}, 
    \item Algorithm~\ref{alg:testmus} invokes Algorithm~\ref{alg:estimateconductande} (\textsc{OuterConductance}) $k^{O(1)}$ times,
    \item \textsc{OuterConductance} takes, by Lemma~\ref{lem:conductanceglued}, $(s_1 + s_2) \cdot \frac{1}{\varphi^2} \cdot s^4 \cdot \left(\frac{\varphi^2}{\e} k \right)^{O(1)} \cdot n^{1/2 + O(\e/\varphi^2)} \log^2(n)$ time,
    \item $s_1 = \Theta( \frac{\varphi^2}{\e} k^5 \log^2(k) \log(1/\eta))$ and $s_2 = \Theta( \frac{\varphi^4}{\e^2} k^5 \log^2(k) \log(1/\eta))$. 
\end{itemize}


\noindent
So in total the runtime of \textsc{FindCenters} is 
$$\frac{1}{\varphi^2} \left(\frac{\varphi^2}{\e} k \right)^{O(1)} n^{1/2 + O(\e/\varphi^2)} \log^3(n) + \log(1/\eta) 2^{O(\frac{\varphi^2}{\e} \cdot k^4 \log^2(k))} k^{O(1)} (s_1 + s_2) \frac{s^4}{\varphi^2} \left(\frac{\varphi^2}{\e} k \right)^{O(1)} n^{1/2 + O(\e/\varphi^2)} \log^2(n)$$
Substituting for $s, s_1, s_2$ it simplifies to:
$$\frac{1}{\varphi^2}\log^2(1/\eta) \cdot 2^{O(\frac{\varphi^2}{\e} \cdot k^4 \log^2(k))} \cdot n^{1/2 + O(\e/\varphi^2)} \log^3(n)$$

\noindent
Runtime of Algorithm~\ref{alg:ballcarving}: Each $\wh{\mu}_i$ is an average of at most $s$ points, where $ s \leq O( \frac{\varphi^2}{\e} \cdot k^4 \log(k) )$, Algorithm~\ref{alg:ballcarving} performs $k^{O(1)}$ tests $\adp{f_x, \wh{\Pi}(\wh{\mu})} \geq 0.93 \rn{\wh{\Pi}(\wh{\mu})}^2$ and by Lemma~\ref{lem:qualityofapproximation} each test takes $\widetilde{O}_{\varphi}\left(s^4 \cdot \left(\frac{k}{\e} \right)^{O(1)} \cdot n^{1/2 + O(\e/\varphi^2)} \right)$ time. So in total the runtime of one invokation of \textsc{ClassifyByHyperplanePartioning}$(\cdot, (T_1, \dots, T_b))$ is in:

$$ \widetilde{O}_{\varphi}\left( \left(\frac{k}{\e} \right)^{O(1)} \cdot n^{1/2 + O(\e/\varphi^2)} \right) $$

\paragraph{Error of \textsc{OuterConductance} algorithm.} Now we analyze the error probabilities of \textsc{OuterConductance} across all the iterations of our algorithm. Note that we run the test for each cluster for each partition and for each of the $\log\left(2/\eta\right)$ iterations of the algorithm. So in total we run \textsc{OuterConductance} test $2^{O(\frac{\varphi^2}{\e} \cdot k^4 \log(k)^2)}k\log\left(\frac{2}{\eta}\right)$ times. By setting $s_1$ in 
$$O \left( k \left(\log(4/\eta)+\log(k \log(1/\eta))+ \frac{\varphi^2}{\e} \cdot k^4\log^2(k) \right) \right) \leq O \left( \frac{\varphi^2}{\e} \cdot k^5 \cdot \log^2(k) \cdot \log(1/\eta)\right) \text{,}$$ 
and $s_2$ in:
$$O \left( \frac{\varphi^2 \cdot k }{ \e } \left(\log(4/\eta)+\log(k \log(1/\eta))+ \frac{\varphi^2}{\e} \cdot k^4\log^2(k) \right) \right) \leq O \left( \frac{\varphi^4}{\e^2} \cdot k^5 \cdot \log^2(k) \cdot \log(1/\eta)\right) \text{,}$$ 
we get by Lemma~\ref{lem:conductanceglued} that the probability that the conclusion of Lemma~\ref{lem:conductanceglued} is not satisfied in a single run is bounded by $$\frac{\eta}{100\cdot2^{\Omega \left(\frac{\varphi^2}{\e} \cdot k^4 \log^2(k) \right)}k\log\left(\frac{1}{\eta}\right)}$$ 
So by union bound over the clusters, the partitions and the iterations we conclude that with probability $1-\frac{\eta}{50}$ the algorithm for every invokation returns a value satisfying the statement of Lemma~\ref{lem:conductanceglued}. Moreover observe that this also means that $\mathcal{E}_{\text{conductance}}$ is true as conclusions of Lemma~\ref{lem:conductanceglued} are stronger than the property required for event $\mathcal{E}_{\text{conductance}}$ to be true.  

\paragraph{W.h.p. every returned ordered partition defines a good clustering.} 
By the lower bound on the error probability of \textsc{OuterConductance} algorithm above we get that with probability $1 - \frac{\eta}{50}$ every cluster $\wh{C}$ that passes the test from line~\ref{testcenterscondactancetest} of Algorithm~\ref{alg:testmus} has to satisfy:
$$\phi(\wh{C}) \leq  O \left(\frac{\e}{\varphi^2}  \cdot \log(k) \right) \text{,}$$
as for $\wh{C}$ to pass the test the value $q$ returned by \textsc{OuterConductance} has to satisfy $q \leq O \left(\frac{\e}{\varphi^2}  \cdot \log(k) \right)$ but by Lemma~\ref{lem:conductanceglued} we have $q \geq \frac{1}{2}\phi \left(\wh{C}_{\wh{\mu}}^{(T_1, \dots, T_b)} \right) - \e/\varphi^2$. Now by Lemma~\ref{lem:howtocluster} this implies that if Algorithm~\ref{alg:findrepresentatives} returns an ordered partition, then with probability $1 - \frac{\eta}{50}$ the collection of clusters it defines satisfies the statement of the Theorem. 

\paragraph{Each iteration succeeds with constant probability.} In the remaining part of the proof we will show that a clustering is accepted with probability $1-\frac{\eta}{2}$. First note that from the paragraph \textbf{Error of \textsc{OuterConductance} algorithm} we know that $\mathcal{E}_{\text{conductance}}$ holds with probability $1 - \frac{\eta}{50}$. Next we show that in each iteration of the outermost for loop of Algorithm~\ref{alg:findrepresentatives} it succeeds with probability $1/2$ (conditioned on $\mathcal{E}_{\text{conductance}}$). By amplification this will imply our result. 

Now consider one iteration. Let $S$ be the set of sampled vertices. Observe that there exists a partition of $S = P_1 \cup P_2 \cup \dots \cup P_k$ such that for all $i\in [k]$, $P_i=S\cap C_i$. We set $s = 10^{15} \cdot \frac{\varphi^2}{\e} \cdot k^4 \log(k)$. Therefore by Lemma \ref{lem:s-eachCi} with probability at least $\frac{9}{10}$ we have for all $i \in [k]$
\[|S\cap C_i| \geq\frac{0.9 \cdot s }{k} \cdot \min_{p,q \in [k]} \frac{|C_p|}{|C_q|} \geq 9 \cdot 10^{14} \cdot \frac{\varphi^2}{\e} \cdot k^3 \log(k) \text{.}\]
Let $\delta=k^{-50}$ and $\zeta=\frac{10^{-6} \sqrt{ \e }}{\varphi \cdot k}$.
Therefore, we have
\[|S\cap C_i|\geq 9 \cdot 10^{14} \cdot \frac{\varphi^2}{\e} \cdot k^3 \log(k) \geq c\cdot\left(k\cdot \log \left(\frac{k}{\delta}\right) \cdot \left(\frac{1}{\delta}\right)^{(80\cdot\epsilon/\varphi^2)}\cdot\left(\frac{1}{\zeta}\right)^2 \right)^{1/(1-(80\cdot\epsilon/\varphi^2))}\]
where $c$ is the constant from Lemma \ref{lem:app-mu-norm}. The last inequality holds since $\frac{\e}{\varphi^2}  \log(k) $ is smaller than a sufficiently small constant, hence, $\left(\frac{\varphi^2}{\epsilon}\right)^{(\epsilon/\varphi^2)}\in O(1)$, and $k^{(\epsilon/\varphi^2)}\in O(1)$. Therefore by Lemma \ref{lem:app-mu-norm} for all $i \in [k]$ with probability at least $1-k^{-50}$ we have: 
$$\|\widehat{\mu}_i - \mu_i\|_2 \leq \zeta\cdot\|\mu_i\|_2 =\frac{10^{-6} \sqrt{ \e }}{\varphi \cdot k}\|\mu_i\|_2 \text{.}$$ 
Hence, by union bound over all sets $P_i$, with probability at least $\frac{9}{10}-k\cdot k^{-50} \geq \frac{7}{8}$ we get $\|\widehat{\mu}_i - \mu_i\|_2 \leq \frac{10^{-6} \sqrt{ \e }}{\varphi \cdot k}\|\mu_i\|_2$ for all $i \in [k]$ simultaneously.  

Now by Theorem~\ref{thm:dot} and the union bound we get that spectral dot product oracle succeeds with probability $1 - n^{-48}$. So by
Lemma~\ref{thm:closeapxmuswin} and the union bound \textsc{FindCenters} with probability $\frac{7}{8} - n^{-48} \geq \frac{1}{2}$ returns an ordered partition $(T_1, \dots, T_b)$ which induces a collection of clusters $\{\wh{C}_{\wh{\mu}_1}, \dots, \wh{C}_{\wh{\mu}_k}\}$ such that there exists a permutation $\pi$ on $k$ elements such that for all $i \in [k]$:
$$\left|\wh{C}_{\wh{\mu}_{i}} \triangle C_{\pi(i)}\right| \leq  O \left(\frac{\e}{\varphi^3}  \cdot \log(k) \right)|C_{\pi(i)}|$$ and 
 $$\phi(\wh{C}_{\wh{\mu}_i}) \leq  O \left(\frac{\e}{\varphi^2}  \cdot \log(k) \right) \text{.}$$ 
\end{proof}

\subsection{LCA}\label{sec:combining}

Now we prove the main result of the paper. Recall that a clustering oracle (Definition~\ref{def:oracle}) is a randomized algorithm that when given query access to  a $d$-regular graph $G = (V,E)$ that admits $(k,\varphi,\e)$-clustering $C_1, \dots, C_k$ it provides consistent access to a \textbf{partition} $\widehat{C}_1, \dots, \widehat{C}_k$ such that there exists a permutation $\pi$ on $k$ elements such that for all $i \in [k]$:
\begin{equation}\label{eq:symdiffguarantee}
    \left|\wh{C}_{\wh{\mu}_{i}} \triangle C_{\pi(i)}\right| \leq  O \left(\frac{\e}{\varphi^3}  \cdot \log(k) \right)|C_{\pi(i)}| \text{.}
\end{equation}
Consistency means that a vertex $x \in V$ is classified in the same way every time it is queried. 

First we will show a Proposition (Proposition~\ref{prop:partitioncollection}) that shows that it is enough to design an algorithm that returns a \textbf{collection of disjoint clusters} (not necessarily a partition) that satisfies \eqref{eq:symdiffguarantee} to get a clustering oracle. Using this Proposition as a reduction we then show Theorem~\ref{thm:mainresult}, which is the main Theorem of the paper.

\begin{proposition}\label{prop:partitioncollection}
If there exists a randomized algorithm $\mathcal{O}$ that when given query access to a $d$-regular graph $G=(V,E)$ that admits a $(k,\varphi,\epsilon)$-clustering  $C_1, \ldots , C_k$, the algorithm $\mathcal{O}$ provides
consistent query access to a \textbf{collection of disjoint clusters} $\mathcal{C}=(\widehat{C}_1,\ldots, \widehat{C}_k)$ of $V$. The collection $\mathcal{C}$ is determined solely by $G$ and the algorithm's random seed.
Moreover, with probability at least $9/10$ over the random bits of $\mathcal{O}$ the collection $\mathcal{C}$ has the following property:
for some permutation $\pi$ on $k$ elements one has for every $i \in [k]$:
$$|C_i \triangle \widehat{C}_{\pi(i)}| \leq O\left(\frac{\epsilon}{\varphi^3}\right) |C_i|\text{.}$$
Then if clusters have equal sizes and $\frac{\e \cdot n}{\varphi^3 \cdot k \log(k)}$ is bigger than a constant then there exists an algorithm $\mathcal{O}'$ that is a $(k,\varphi,\e)$-clustering oracle with the same running time and space up to constant factors.
\end{proposition}

\begin{proof}
The idea is to assign the points outside $\bigcup_{i \in [k]} \wh{C}_i$ randomly. That is to assign vertex $x \in V$, $\mathcal{O}'$ works exactly the same like $\mathcal{O}$ but if $\mathcal{O}$ left $x$ unassigned then $\mathcal{O}'$ assigns $x$ to a value chosen from $[k]$ uniformly at random.

Let $R = V \setminus \bigcup_{i \in [k]} \wh{C}_i$ and for every $i\in [k]$ let $S_i \subseteq R$ be the set of vertices that were randomly assigned to $\wh{C}_i$. By the fact that for every $i \in [k]$ $|C_i \triangle \widehat{C}_{\pi(i)}| \leq O\left(\frac{\epsilon}{\varphi^3}\right)  |C_i|$ we get that there exists a constant $C$ such that:
\begin{equation}\label{eq:outliers}
|R| \leq C \cdot \frac{\e}{\varphi^3}  \cdot  n \text{.}    
\end{equation}
Now let $i \in[k]$. By the Chernoff bound we have that for every $\delta \geq 1$:
\begin{equation}
P \left[\left| |S_i| - \frac{|R|}{k} \right| \geq  \delta \frac{|R|}{k} \right] \leq e^{- \delta \frac{|R|}{3 \cdot k}}    
\end{equation}
Setting $\delta = \frac{C \cdot \e \cdot n}{\varphi^3 \cdot |R|}$ we get:
\begin{equation}\label{eq:chernoff}
P \left[\left| |S_i| - \frac{|R|}{k} \right| \geq  C \cdot \frac{\e}{\varphi^3} \cdot \frac{n}{k} \right] \leq e^{- \frac{C \cdot \e \cdot n}{3 \cdot \varphi^3 \cdot k} }    
\end{equation}
Combining \eqref{eq:outliers} and \eqref{eq:chernoff} and the assumption that $\frac{\e \cdot n}{\varphi^3 \cdot k \log(k)}$ is bigger than a constant we get that 
$$P \left[|S_i| \geq 2C \cdot \frac{\e}{\varphi^3} \cdot \frac{n}{k} \right] \leq \frac{1}{100 \cdot k}$$
Using the union bound we get that with probability $9/10 - k \cdot \frac{1}{100 \cdot k} \geq 8/10$ we have that for every $i \in [k]$ $|S_i| \leq 2 C \cdot \frac{\e}{\varphi^3} \cdot \frac{n}{k}$.
So finally with probability $8/10$ for every $i \in [k]$:
\begin{align*}
|C_i \triangle (\wh{C}_{\pi(i)} \cup S_{\pi(i)}) | 
&\leq |C_i \triangle \wh{C}_{\pi(i)}| + |S_{\pi(i)}| \\
&\leq O \left(\frac{\e}{\varphi^3} \right) \cdot |C_i| + O \left(\frac{\e}{\varphi^3} \right) \cdot \frac{n}{k} && \text{By definition of } \mathcal{O}\\
&\leq O \left(\frac{\e}{\varphi^3} \right) \cdot |C_i| && \text{As } \frac{\max_{p \in [k]} |C_p|}{\min_{p \in [k]} |C_p|} = O(1)\text{,}  
\end{align*}
which means that $\mathcal{O}'$ is a $(k,\varphi,\e)$-clustering oracle.
\end{proof}

\thmmainresult*

\begin{proof}
By Theorem~\ref{thm:uglythm} we get that there exists an algorithm that runs in $\widetilde{O}_{\varphi} \left( 2^{O(\frac{\varphi^2}{\e} \cdot k^4 \log^2(k))} \cdot n^{1/2 + O(\e/\varphi^2)} \right)$ time and that with probability $9/10$ returns an ordered partition $(T_1, \dots, T_b)$ of $\{\wh{\mu}_1, \dots, \wh{\mu}_k \}$ such that the induced collection of clusters $\{ \wh{C}_{\wh{\mu}_1}, \dots, \wh{C}_{\wh{\mu}_k} \}$ satisfies the following.
There exists a permutation $\pi$ on $k$ elements such that for every $i \in [1,\dots,k]$:
$$|C_{\pi(i)} \triangle \wh{C}_{\wh{\mu}_i}| \leq O \left(\frac{\e}{\varphi^3}  \cdot \log(k) \right) |C_{\pi(i)}|$$
That algorithm is the preprocessing step of oracle $\mathcal{O}$. Then for each query $x_i \in V$ we run Algorithm~\ref{alg:ballcarving} which outputs $\wh{\mu}_j$ such that $x_i \in \wh{C}_{\wh{\mu}_j}$ (Note that $x_i$ might not belong to any of $\wh{C}_{\wh{\mu}_i}$, see Proposition~\ref{prop:partitioncollection} for how to deal with that). Algorithm~\ref{alg:ballcarving} by Theorem~\ref{thm:uglythm} runs in $\widetilde{O}_{\varphi}\left( \left(\frac{k}{\e} \right)^{O(1)} \cdot n^{1/2 + O(\e/\varphi^2)} \right)$ time.

\paragraph{Runtime tradeoff.}
Notice however that by Theorem~\ref{thm:dot} we can achieve a tradeoff in the preprocessing/query runtime and achieve $\widetilde{O}_{\varphi} \left( 2^{O(\frac{\varphi^2}{\e} \cdot k^4 \log^2(k))} \cdot n^{1 - \delta + O(\e/\varphi^2)} \right)$ for preprocessing time and $\widetilde{O}_{\varphi} ( \left(\frac{k}{\e} \right)^{O(1)} \cdot n^{1 - \delta + O(\e/\varphi^2)} )$ space and $\widetilde{O}_{\varphi}\left( \left(\frac{k}{\e} \right)^{O(1)} \cdot n^{\delta + O(\e/\varphi^2)} \right)$ for query time.

\paragraph{Random bits.} The only thing left to prove is to show that we can implement these two algorithms in LCA model using few random bits. There are couple of places in our Algorithms where we use randomness. 

First in \textsc{InitializeOracle} (Algorithm~\ref{alg:LearnEmbedding}) we sample 
$\widetilde{\Theta}(n^{O(\e/\varphi^2)}  \cdot k^{O(1)})$
random points. For that we need 
$\widetilde{\Theta}(n^{O(\e/\varphi^2)} \cdot k^{O(1)})$
random bits. 

For generating random walks in Algorithm~\ref{alg:LearnEmbedding} and Algorithm~\ref{alg:dotProduct} we need the following number of random bits. Notice that in all the proofs (see Lemma~\ref{lem:pairwise-collision}) we only need $4$-wise independence of random walks. That means that we can implement generating these random walks using a hash function $h(x)$ that for vertex $x \in V$ generates $O(\log(d) \cdot \frac{1}{\varphi^2} \cdot \log(n))$ bit string that can be interpreted as encoding a random walk of length $O(\frac{1}{\varphi^2} \cdot \log(n))$ (remember that graphs we consider are $d$-regular so $\log(d)$ bits is enough to encode a neighbour). It's enough for the hash function to be $4$-wise independent so it can be implemented using 
$O(\frac{1}{\varphi^2} \cdot \log(d) \cdot \log(n)) = \widetilde{O}_{\varphi}(1)$
random bits. 

The partitioning scheme (see Algorithm~\ref{alg:ballcarving}) works in $O(\log(k))$ adaptive stages. The stages are adaptive, that is why we use fresh randomness in every stage. For a single stage we observe that in the proof of Lemma~\ref{lem:conductanceglued} we only use Chernoff type bounds. So by~\cite{limitedIndependence} we don't need fully independent random variables. In our case it's enough to have 
$O(\log(n))$-wise independent random variables which can be implemented as hash functions using $O(\log^2(n))$ random bits. This means that in total we need 
$O(\log(k) \log^2(n)) = \widetilde{O}(1)$ 
random bits for this.

For sampling set $S$ in Algorithm~\ref{alg:findrepresentatives} we can use $O(\frac{\varphi^2}{\e} \cdot k^4 \log(k) \cdot \log(n)) = \widetilde{O}_{\varphi}(\frac{1}{\e} \cdot k^{O(1)} )$ fresh random bits. 

\noindent
So finally the total number of random bits we need is in: 
$$
\widetilde{O}_{\varphi}\left(n^{O(\e/\varphi^2)} \cdot k^{O(1)} + 1 + 1 + \frac{1}{\e} \cdot k^{O(1)} \right) \leq \widetilde{O}_{\varphi}\left(\frac{1}{\e} \cdot n^{O(\e/\varphi^2)} \cdot k^{O(1)} \right)$$
\end{proof}

\begin{remark}
Note that threshold sets $C_{y,\theta}$ (recall Definition~\ref{def:thresholdsets}) are well defined in LCA model because for all $x,y \in V$ whenever we compute $\adp{f_x,f_y}$ the result is the same as we use consistent randomness (see Definition~\ref{def:oracle}).
\end{remark}

\bibliographystyle{alpha}
\bibliography{ref.bib}
\end{document}